\documentclass [11pt]{article}

\usepackage{xspace} 
\newcommand{\fo}{\mathrm{FO}\xspace} 
\newcommand{\1}{\,|\,}
\usepackage{verbatim}
\usepackage{mathrsfs}
\usepackage{enumerate}
\usepackage{graphicx}% http://ctan.org/pkg/graphicx
\usepackage{amsmath,amssymb,amsthm} 
\theoremstyle{plain}
\newtheorem{theorem}{Theorem}
\newtheorem{lemma}[theorem]{Lemma}

\theoremstyle{theorem}
\newtheorem{claim}{Claim}

\newtheorem{proposition}[theorem]{Proposition}
\newtheorem{corollary}[theorem]{Corollary}
\theoremstyle{definition}
\newtheorem{definition}[theorem]{Definition}
\newtheorem{remark}[theorem]{Remark}
\newtheorem{example}{Example}

\theoremstyle{theorem}
\newtheorem{fact}{Fact}

\theoremstyle{theorem}
\newtheorem{assumption}{Assumption}

\usepackage{color}

\usepackage{tikz}
\usepackage{floatpag}
\floatpagestyle{empty}
\numberwithin{equation}{section} 
\usepackage{stmaryrd}
\usepackage{bm}
\usepackage{datetime}
\usepackage{dsfont}
\usepackage{hyperref}
 
\usepackage{appendix}

\DeclareMathAlphabet{\mathpzc}{OT1}{pzc}{m}{it}
\usetikzlibrary{decorations.pathreplacing}

\usepackage{textcomp}
\newcommand{\textapprox}{\raisebox{0.5ex}{\texttildelow}}

\newcommand{\ccirc}{\mathbin{\mathchoice
  {\xcirc\scriptstyle}
  {\xcirc\scriptstyle}
  {\xcirc\scriptscriptstyle}
  {\xcirc\scriptscriptstyle}
}}
\newcommand{\xcirc}[1]{\vcenter{\hbox{$#1\circ$}}}

\begin{document}
\title{$k$ variables are needed to define $k$-Clique in first-order logic}


\author{Yuguo He\footnote{Beijing Institute of Technology, Beijing 100081,
China, \href{mailto:hugo274@gmail.com}{hugo274@gmail.com}} 
}

\maketitle
\thispagestyle{plain}

\begin{abstract}
%The Clique problem has a long line of research in theoretical computer science. 
In an early paper, Immerman raised a proposal on developing  model-theoretic techniques to prove lower bounds on ordered structures, which represents a long-standing challenge in finite model theory. An iconic question standing for such a challenge is how many variables are needed to define $k$-Clique in first-order logic on the class of finite ordered graphs? If $k$ variables are necessary, as widely believed, it would imply that the bounded  (or finite) variable hierarchy in first-order logic is strict on the class of finite ordered graphs.  
In 2008, Rossman made a breakthrough by establishing an optimal average-case lower bound on the size of constant-depth unbounded fan-in circuits computing $k$-Clique. In terms of logic, this means that it needs greater than $\lfloor\frac{k} {4}\rfloor$ variables to describe the $k$-Clique problem in first-order logic on the class of finite ordered graphs, even in the presence of arbitrary arithmetic predicates. It follows, with an unpublished result of Immerman, that the bounded variable hierarchy in first-order logic is indeed strict. However, Rossman's methods come from circuit complexity and a novel notion of sensitivity by himself. And the challenge before finite model theory remains there. In this paper, we give an alternative proof for the strictness of bounded variable hierarchy in $\fo$ using pure model-theoretic toolkit, and anwser the question completely for first-order logic, i.e. $k$-variables are indeed needed to describe $k$-Clique in this logic. In contrast to Rossman's proof, our proof is purely constructive. Then we embed the main structures into a pure arithmetic structure to show a similar result where arbitrary arithmetic predicates are presented. 
Finally, we discuss its application in circuit complexity.  
\end{abstract}

%%%%\begin{IEEEkeywords}
%finite model theory; pebble games; $k$-Clique 

%%%%%%%\end{IEEEkeywords}

% \fi
%%%%%%%%%%%%%%\IEEEpeerreviewmaketitle

\section{Introduction}

The Clique problem has been studied extensively in theoretical computer science, not only in classical computational complexity, but also in parameterized complexity. In particular, in  circuit complexity, some remarkable lower bounds for $k$-Clique in restricted models have been established (for a short survey, cf., e.g., Rossman \cite{RossmanThesis}). 
Following \cite{Lynch86Circuit,Beame90Clique}, Rossman \cite{RossmanStoc} showed that, on average, 
no constant-depth unbouded fan-in circuits of size $O(n^\frac{k}{4})$ can recognize  
$k$-Clique. It represents a significant breakthrough in circuit complexity, because it is the first unconditional lower bound that cleverly breaks out of the traditional size-depth tradeoff, a well-known barrier to progress in the study of constant-depth circuits model for many years \cite{Dawar2011Ackermann}. 
We note that Rossman has achieved a stronger result in circuit complexity. But in the context of logic, his result is roughly the following: on the class of finite ordered graphs, any first-order logic sentence that defines $k$-Clique in the average case needs greater than $\lfloor\frac{k}{4}\rfloor$ variables, even when arbitrary arithmetic predicates are available. In addtion, a study of $k$-pebble games over random graphs was introduced in \cite{Rossman09Game} based on this result. 
Note that his lower bound is tight in the average case on a certain natural distribution \cite{Amano2010}. %But it is not the optimal one in the worst case.   
Rossman \cite{RossmanUpperBound} also gave the tight average-case lower bound for $k$-Clique on the class of finite ordered graphs, without arithmetic predicates other than the built-in linear order, by showing a tight upper bound $\frac{k}{4}+O(1)$ on the number of variables that is needed for defining $k$-Clique in the average case. In logic, Rossman's lower bound implies that the bounded  variable hierarchy (also called the finite variable hierarchy in the literature) in first-order logic ($\fo$, for short) will not collapse. Together with an unpublished result of Immerman (cf. Rossman \cite{RossmanThesis}), it follows that the bounded variable hierarchy in $\fo$ is strict, i.e. for any $k$, there is a property that can be defined by $k$ variables but is not definable by $k-1$ variables in $\fo$, thereby solved a long-standing question going back to Immerman \cite{Immerman1982Conj}. 
As mentioned in \cite{Dawar2011Ackermann}, this question on bounded variable hierarchy is ``deceptively simple'' in the appearance, but very hard to answer.
   Therefore, its settling by Rossman also represents ``one of the most significant breakthroughs in the field of finite model theory in many year'' \cite{Dawar2011Ackermann}. 

It is not supprising that a result in the field of circuit complexity would have such an impact on the field of finite model theory. See an explanation of such a connection in \cite{DawarHowmany}. 
The connection between computational complexity and finite model theory started by an early work of Fagin \cite{Fagin1974Start}, which showed the equivalence of NP and the class of properties that can be expressed by the existential second-order logic. Such kind of research was then carried on under the name of descriptive complexity, where Immerman has been playing a crucial role by making numerous fundamental contributions. To have a grasp of this  field, cf. the monograph \cite{Immerman1999Book}. 

The study of finite variable fragments of $\fo$ may be traced back to the 19th century \cite{Andreka1998Modal}. Its important role in model theory is well-known, cf. e.g. \cite{Dawar2011Ackermann, Grohe1998Variable, Immerman1982Conj}. 
Immerman began the study of syntactically uniform sequences of first-order formulas with constant number of variables. It is important because most of the well-known complexity classes can be characterized by such sequences. Moreover, the number of variables correspond to amount of hardware, e.g. the number of processors in $\mathrm{CRAM}$. In particular, it is well-known that $\mathrm{DSPACE}[n^k]=\mathrm{VAR}[k+1]$, which elegantly connects the number of variables in descriptions with the number of memory locations in deterministic Turing machines \cite{Immerman1991Variable}. It means that $k$ variables roughly correspond to $n^k$ memory locations.

In 1982, Immerman raised a proposal on developing techniques to prove lower bounds on finite structures \textit{with linear orders} (cf. \cite{Immerman1982Conj}, p.97), which stands for a challenge in finite model theory for decades. 
It is very important because all the known logical characterizations of well-known complexity classes inside NP rely on a built-in linear order over input finite 
structures. For NP and beyond, we can guess an order and rely on this order to simulate computation of Turing machines. In this sense, we still need linear orders in such cases, although implicitly. In short, computation needs orders to carry out. 
Note that  
his full proposal is very general and beyond our goal. Immerman considered syntactically uniform sequences of first-order formulas that define properties, whereas we only consider the most uniform one, i.e. all the formulas in a sequence are the same, which means the formulas in such a sequence have a constant size.  
In other words, we only consider the expressiveness of $\fo$, which has very limited  expressive power.
If we have techniques for uniform sequences of first-order formulas without such restriction, as Immerman asked for, this would lead to many profound results, including a possible settling of the well-known open problem on P vs. NP if we allow the formulas to have polynomial sizes in terms of the length of inputs. It is for this reason that we restrict our concern on $\fo$.

The strictness of bounded variable hierarchy was studied in various logics, such as modal logics and temporal logics (cf. \cite{Andreka1998Modal, Berwanger2007Calculus, Dawar2011Ackermann}). 
It is not obvious to tell how many variables are needed in $\fo$. For example, it is known that three variables suffice to describe any first-order property of linear orders with unary relations only \cite{Poizat1982ColorOrder}.\footnote{There is an alternative proof in the context of temporal logic, cf. p.4 of \cite{DawarHowmany} for a brief introduction.} Using similar ideas, we can show that it holds even in finite pure arithmetic structures (cf. Remark \ref{variable-hierarchy-collapse-arithm-struc}). 
In fact, it turns out to be an extremely hard problem to answer, when the structures are finite ordered graphs. It was observed that the $k$-Clique question may be the key to solve this problem \cite{DawarHowmany}.  
Intuitively, linear orders seem useless in reducing the number of variables that are needed to define $k$-Clique, which strongly suggests that \textit{precisely} $k$ variables are needed to define $k$-Clique in $\fo$ (cf. \cite{DawarHowmany} p.23, \cite{RossmanStoc} p.10, \cite{RossmanThesis} p.71). If this is true, then $\fo$ needs arbitrary number of variables.     
%When it is to be tied with questions of complexity theory, we usually consider linear orders, and the connection with the questions of logic and questions of complxity theory were highlighted in \cite{Immerman1982Conj}.  
The conjecture that bounded variable hierarchy over finite ordered graphs is strict in $\fo$ was first explicitly presented in \cite{DawarHowmany} (cf. p.3, Cojecture 2, the stronger version), which can go back to the proposal raised in Immerman's early paper \cite{Immerman1982Conj}, because the hierarchy is strict if $k$ variables are needed to define $k$-Clique over finite ordered graphs in $\fo$. 
In circuit complexity, this hierarchy corresponds to the size hierarchy. Rossman's result \cite{RossmanStoc} implies that this size hierarchy is infinite. It was later sharpened by Amano who showed that this size hierarchy is strict \cite{Amano2010,RossmanThesis}.

Here we should note that it is the finiteness and, \textit{in particular}, the linear orders that make Immerman's proposal, even in terms of $\fo$, very hard to take up, because it is  ``notoriously  difficult'' (cf. \cite{DawarHowmany}, p.23) to apply the standard tool in finite model theory on finite ordered structures (also cf. \cite{Immerman1999Book}). For example, it is very hard to show that $k$ variables are needed to define  $k$-Clique on finite ordered graphs. 
But, if either ``finite'' or ``order'' is dropped off, the statement turns out to be easy to prove.     
Therefore, on finite ordered graphs, showing that $k$ variables are needed to define  $k$-Clique, as well as showing the strictness of the bounded variable hierarchy in $\fo$, is an iconic problem that represents such a challenge \cite{Dawar2011Ackermann}. Hence, in this paper we shall refer to this particular question (i.e. how many variables are needed in $\fo$ to define $k$-Clique on the class of finite \textit{ordered} graphs) when we mention ``Immerman's question'' (originated from the proposal).   
Although it was proved that the bounded variable hierarchy is strict in $\fo$, it was solved mainly by techniques from circuit complexity and a novel notion called clique-sensitive core by Rossman.  And the challenge before finite model theory remains there.  
In this paper, based on standard finite model-theoretic toolkit, we develop novel notions and techniques to give an alternative proof for the strictness of bounded variable hierarchy in $\fo$, and fully answer 
Immerman's question. 
% using pure finite model-theoretic toolkit. %, i.e. precisely $k$ variables are needed to define $k$-Clique in $\fo$ on the class of finite ordered graphs.  
Note that our proof is completely constructive. Compared with existential proofs which merely show the existence of pursued mathematical objects, a constructive proof need show every bit of the objects clearly: in terms of our study, the structures should be constructed explicitly and the strategies should be fully revealed. 
It justifies the length of our proof and allows a straightforward application of the result to answer a related important question. That is, precisely $k$ variables are needed for $k$-Clique in $\fo(\mathbf{BIT})$. 
At the end of this paper, we discuss its application in circuit complexity.

For more background on the $k$-Clique problem and the linear order issue in finite model theory, the readers can confer the survey paper \cite{DawarHowmany} which collected related issues, connections, observations and results. In particular, Dawar showed that existential first-order formulas (and even existential infinitary logic formulas) require $k$ variables to define $k$-Clique on the class of finite ordered graphs, which we will introduce in section \ref{existential-case-section}. 

\section{Preliminaries}

Let $\mathbf{Z}, \mathbf{N}_0$ and $\mathbf{N}^+$ be the set of integer numbers, non-negative integer numbers and positive integer numbers respectively. By default, we assume that all the numbers mentioned in this paper are in $\mathbf{N}_0$.  Assume that $k$ is a fixed integer number where $k>1$. In this paper, we use semicolons to mean ``AND'' in definitions of sets. 
We use $|X|$ to denote the cardinality of the set $X$ and use $\wp(X)$ to denote the power set of $X$. A \textit{tuple} is a multiset whose elements are ordered. Hence we can compute the intersection or union of a set and a tuple (redundant elements are omitted). For a tuple $\bar c$, we use $\bar c(i)$ to denote the $i$-th element of $\bar c$.
%, where $\bar c(1)$ means the first element or the leftmost element of $\bar c$. 
Let $|\bar c|$ be the length of $\bar c$. 
If $|\bar c|\neq 0$, we use $\bar c\sqsubseteq \bar d$ to denote that $\bar c(i)=\bar d(i)$ for $0\leq i< |\bar c|$. If $|\bar c|=0$, $\bar c\sqsubseteq \bar d$ for any $\bar d$. By convention, we use ``$\lfloor x\rfloor$'' to denote the \textit{floor functions} $\mathrm{floor}(x)$, i.e. $\lfloor x\rfloor=max\{n\in \mathbf{Z}\mid n\leq x\}$, for any real number $x$. For $n\in\mathbf{N}^+$, we let $[n]$ be the set $\{0,\ldots,n-1\}$ and let $[1,n]$ be the set $\{1,\ldots,n\}$. And for $n_0,n_1\in\mathbf{N}^+$ we let $[n_0]\times [n_1]$ be the Cartesian product $[n_0]\times [n_1]$. Henceforth, we use a pair of integer to denote a point in a two dimension coordinate plane.  For a sequence of $n$ elements, a \textit{right circular shift} of this sequence is a permutation $\sigma$ such that the element in the $i$-th position is moved to the $\sigma(i)$-th position of the sequence, where  $\sigma(i)=(i+1)$ mod $n$.

In the following of this section we introduce some standard concepts and well-known results in finite model theory. Although we try to make it self-contained, the readers are assumed to have some elementary knowledge of first-order logic. %The readers may also cf., e.g., 

\subsection{Logic and structures}

Let $R_i$ be a relation symbol and $c_i$ a constant symbol. Let $\sigma=\langle R_1,\ldots,R_\ell,\\c_1,\ldots,c_n\rangle$ be a relational signature,  a $\sigma$-\textit{structure} $\mathfrak{A}$  consists of a universe $A$ together with an interpretation of $R_i$ and $c_i$ over $A$: 
\begin{itemize}
\item each $s$-ary relation symbol $R_i\in \sigma$ as a $s$-ary relation on $A$, usually written $R_i^\mathfrak{A}$;
\item each constant symbol $c_i\in \sigma$ as an element in $A$, usually written   $c_i^{\mathfrak{A}}$.
\end{itemize}  
The structure $\mathfrak{A}$ is a finite structure if $A$ is a finite set. 

For any $D\subseteq A$, a $\sigma$-\textit{substructure} of the $\sigma$-structure $\mathfrak{A}$ induced by $D$, denoted $\mathfrak{A}[D]$, is a $\sigma$-structure, wherein each constant symbol is interpreted as it is in $\mathfrak{A}$, and each $j$-ary relation symbol $R_i$ is interpreted by $R_i^\mathfrak{A}\cap D^j$.

For example, a \textit{finite digraph} $\mathcal{G}_d$ is a finite $\langle E\rangle$-structure where $E$ is a binary relation symbol; and a finite graph $\mathcal{G}$ is a digraph where $E^\mathcal{G}$ is symmetric.  By convention, the universe of a graph is called vertex set, denoted $V$, and the set of pairs (i.e. 2-tuples) that are used to interpret $\langle E\rangle$ are called edges, usually denoted by $E$ instead of $E^\mathcal{G}$. The fact that two vertices $a$ and $b$ are joined by an edge can be denoted by $Eab$, $E(a,b)$ or $(a,b)\in E$. Or we can say $a$ is \textit{adjacent} to $b$. 
%That is, by convention, we use $\langle V, E\rangle$ to denote a graph. 
A graph $\mathcal{G}^{\prime}=\langle V^{\prime},E^{\prime} \rangle$ is a \textit{subgraph} of a graph $\mathcal{G}=\langle V,E\rangle$ if $V^{\prime}\subseteq V$ and $E^{\prime}\subseteq E\cap (V^{\prime}\times V^{\prime})$; $\mathcal{G}^{\prime}$ is an \textit{induced subgraph} of $\mathcal{G}$, denoted by $\mathcal{G}[V^{\prime}]$, if $\mathcal{G}^{\prime}$ is a subgraph of $\mathcal{G}$ and $E^{\prime}=E\cap (V^{\prime}\times V^{\prime})$. %Abuse of notation, we also use $\mathcal{G}[a_0,\ldots,a_i]$ to denote the subgraph of $\mathcal{G}$ induced by the set of vertices $a_0,\ldots,a_i\in V$.
We also use $|\mathcal{G}|$ to denote the set of vertices of $|\mathcal{G}|$. 

Let $\sigma^{\prime}\subseteq \sigma$. The $\sigma^{\prime}$-reduct of  $\mathfrak{A}$, denoted $\mathfrak{A}|\sigma^{\prime}$, is obtained from $\mathfrak{A}$ by ``forgetting'' $\sigma\setminus\sigma^{\prime}$,  
i.e. leaving all the symbols in $\sigma\setminus\sigma^{\prime}$ uninterpreted.

Let $\mathfrak{A}$ and $\mathfrak{B}$ be two structures over the same finite signature $\sigma$. Say that $\mathfrak{A}$ and $\mathfrak{B}$ are isomorphic (or, there is an \textit{isomorphism} between them), if there is a bijection $f: A\rightarrow B$ such that 
\[
\left\{\begin{array}{ll}
f(c_i^\mathfrak{A})=c_i^\mathfrak{B}, & \mbox{ for any constant symbol }c_i\!\in\!\sigma;\\[8pt] 
 R_i^\mathfrak{A}(a_1,\cdots,a_j)\Leftrightarrow R_i^\mathfrak{B}(f(a_1),\cdots,f(a_j)), & \mbox{ for any } j\mbox{-ary relation symbol}\\& \hspace{3pt} R_i\in \sigma.
\end{array}\right.
 \]
For example, two graphs are isomorphic if they are the same after renaming of their vertices. 

Let $\mathcal{L}$ be a logic, e.g. $\fo$. 
Let $\mathfrak{A}$ be a $\sigma$-structure and $\psi$ be an $\mathcal{L}$-sentence. We use $\mathfrak{A}\models \psi$ to denote that $\psi$ is true in $\mathfrak{A}$, and we call $\mathfrak{A}$ a \textit{model} for $\psi$. Let $\mathrm{Mod}(\psi)$ be the set of models of $\psi$.
A \textit{property} $Q$ over $\sigma$ is a set of $\sigma$-structures closed under isomorphism. Say that $Q$ is \textit{expressible}, or \textit{definable}, in a logic $\mathcal{L}$ if there is a sentence $\varphi$ in $\mathcal{L}$  such that for every $\mathfrak{A}$,  $\mathfrak{A}\in\mathrm{Mod}(\varphi)$ iff $\mathfrak{A}\in Q$.

\href{http://en.wikipedia.org/wiki/Square_lattice}{\textit{Square lattice}}, denoted $\mathbf{Z}^2$, is the lattice in the two-dimensional Euclidean space whose lattice points are pairs of integers. A \textit{finite upright square lattice} is a finite square lattice that is isomorphic to the set of isolated vertices $[a]\times [b]$ for some $a,b\in\mathbf{N}^+$. In the sequel, when we mention ``square lattice'', we mean finite upright square lattice by default. Here, $a$ is the \textit{width} of the square lattice, and $b$ is the \textit{height} of the square lattice. Hence we can talk about the width and height of a graph structure whose universe is a square lattice. And we can talk about a row of this graph, which is composed of the vertices of the same second coordinate. We call the bottom row as the $0$-th row of the graph.  

A \textit{linear order} is a binary relation that is transitive, antisymmetric and total. We usually use the infix notation $\leq $ to denote a linear order. 
For example, we use $a\leq b$ to stand for $(a,b)\in \leq^\mathfrak{A}$ when the structure $\mathfrak{A}$ is clear from the context. 
We use $a<b$ to denote $a\leq b$ and $a\neq b$. Any linear order induces a natural distance measure, i.e. we can talk about $|a-b|$ when $a, b$ are two elements of a linear order. 

 A $k$-\textit{clique} of $\mathcal{G}$ is a complete subgraph of $\mathcal{G}$, containing $k$ vertices. 
That is, there is an edge between each pair of vertices in the $k$-clique. %Henceforth, we assume that $k>1$. 
The $k$-\textit{Clique problem} asks whether there is a $k$-clique in a given graph. By default, the graphs are ordered finite graphs. \label{k-clique-problem} 

For any $x, i\in \mathbf{N}_0$, let $\mathbf{BIT}(x,i)\in \{0,1\}$ be $1$ iff the $i$-th bit of the binary representation  of $x$ is $1$. Here we assume that the rightmost bit is the $0$-th bit. 
A \textit{pure arithmetic structure} is a structure whose signature contains only arithmetic predicates. 
It is finite by default. 
In particular, we assume that the signature contains the binary predicate $\mathbf{BIT}$, which can be used to define arbitrary arithmetic predicates, including linear orders.

The \textit{quantifier rank} of a formula $\phi\in\fo$, written $qr(\phi)$, is the maximum depth of nesting of its quantifiers.
Formally, it is defined inductively as follows:
\begin{enumerate}
\item If $\phi$ is atomic, then $qr(\phi)=0$.
\item $qr(\phi_1\lor\phi_2)=qr(\phi_1\land\phi_2)=max\{qr(\phi_1),qr(\phi_2)\}.$
\item $qr(\lnot\phi)=qr(\phi)$. 
\item $qr(\exists x\phi)=qr(\forall x\phi)=qr(\phi)+1$.
\end{enumerate}

Let $\mathcal{L}^k$ be the fragment of $\fo$ whose formulas have at most $k$ distinct variables, free or bound. \textit{Infinitary logic}, written $\mathcal{L}_{\infty\omega}$, is the closure of first-order logic with infinitary conjunctions and disjunctions. 
%That is, the formulas of $\mathcal{L}_{\infty\omega}$ are defined similarly as that of $\fo$, except that conjunctions and disjunctions are no longer binary operations. Indeed, they can act on arbitrary sets of $\mathcal{L}_{\infty\omega}$ formulas. 
$\mathcal{L}_{\infty\omega}^k$ is the fragment of $\mathcal{L}_{\infty\omega}$ whose formulas have at most $k$ distinct variables, free or bound. 
$\exists\mathcal{L}_{\infty\omega}^k$ is the fragment of $\mathcal{L}_{\infty\omega}^k$ in which no universal quantifiers appear in the formulas and in which all the existential quantifiers are within the scope of even number of negations. 

%We use $\oplus$ to denote XOR.  

\subsection{Pebble games}

A \textit{game board} consists of a pair of structures, e.g. $(\mathfrak{A},\mathfrak{B})$. 
%Assume that $\bar{a}$ interprets the constants of $\mathfrak{A}$, if there is any, and similarly $\bar{b}$ interprets the constants of $\mathfrak{B}$. 
 An $m$-round $(k-1)$-\textit{pebble game} over the game board $(\mathfrak{A},\mathfrak{B})$, written $\Game_m^{k-1}\!(\mathfrak{A},\mathfrak{B})$, is defined as the following. 

There are two players in the game, called Spoiler and Duplicator. 
There are $k-1$  pairs of pebbles, say $(e_1,f_1), \cdots, (e_k,f_k)$, available for the players, which are off the board at the begininig of the game.  
In each round, a pair of pebbles, say $(e_i,f_i)$, will be put on the structures wherein $e_i$ is put on an element of $\mathfrak{A}$ and $f_i$ is put on an element of $\mathfrak{B}$.   Spoiler first selects a structure and puts a pebble on one element of the selected structure; then Duplicator puts the other pebble in the same pair (matching pebble, for short) on one element of the other structure. If there is no pebble off the board, Spoiler can move a pebble to a new element; then Duplicator should move the  matching pebble to some  element in the other structure. 

 In the $\ell$-th round of the game, let $\overline{c_A}=(a_1,a_2,\ldots,a_n)$, where $n\leq k-1$, includes all the elements in $\mathfrak{A}$
 that have pebbles on them; let $\overline{c_B}=(b_1,b_2,\ldots,b_n)$ includes all the pebbled elements in $\mathfrak{B}$. And assume that, for any $i$,  $a_i$ and $b_i$ are the positions of $e_j$ and $f_j$ for some $j$ in this round. 
Sometimes  
we use $((\mathfrak{A}, \overline{c_A}),(\mathfrak{B},\overline{c_B}))$ to denote the game board in this round. But it does not mean that the game board is changed (we will mention a sort of imaginary games that allow changes of boards later). Only that some elements are pebbled.  
$((\mathfrak{A}, \overline{c_A}),(\mathfrak{B},\overline{c_B}))$ is in \textit{partial isomorphism} or $(\overline{c_A},\overline{c_B})$ defines an partial isomorphism between $\mathfrak{A}$ and $\mathfrak{B}$, if   $\{(a_1,b_1),\ldots,(a_n,b_n)\}$ defines an isomorphism between    $\mathfrak{A}[\overline{c_A}]$ and $\mathfrak{B}[\overline{c_B}]$.

 Spoiler wins the game if the game board is not in partial isomorphism in some round; otherwise, Duplicator wins the game. If Duplicator can guarantee a win after $m$ rounds of such $(k-1)$-pebble game, we say Duplicator has a \textit{winning strategy} in the $m$-round $(k-1)$-pebble game, denoted by $\mathfrak{A}\equiv_m^{k-1} \mathfrak{B}$.

The following holds in a $(k-1)$-pebble game.  
\begin{fact}\label{star_k-2}
In each round of a $(k-1)$-pebble game, either Duplicator can win this round by mimicking, or, there are at most $k-2$ pairs of pebbles on the game board at the start of this round.
\end{fact}

\noindent It is because that moving a pebble from one element to another element consists of two steps: first pick up the pebble off the board; afterwards put this pebble on the new element. We can always turn a game into such a game without losing anything. 
The first step will not violate the partial isomorphism of the game board, therefore when we talk about winning strategies we can skip such rounds safely and always assume that there are at most $k-2$ pairs of pebbles on the game board at the start of a round.

It is well-known that pebble games characterize the expressive power of finite variable logics.
\begin{theorem}
The following statements are equivalent.
\begin{itemize}
\item For any $\phi\in \mathcal{L}^{k-1}$ with $qr(\phi)\leq m,$
$\mathfrak{A}\models \phi \Leftrightarrow \mathfrak{B}\models \phi.$

\item $\mathfrak{A}\equiv_m^{k-1} \mathfrak{B}$.

\end{itemize}
\end{theorem}

Therefore, if for any $m$ we can find a pair of structures, e.g. $(\mathfrak{A},\mathfrak{B})$, such that 
$\mathfrak{A}$ satisfies some property while $\mathfrak{B}$ doesn't, and  $\mathfrak{A}\equiv_m^{k-1} \mathfrak{B}$, then this property is not expressible in $\mathcal{L}^{k-1}$.

To shorten description, usually we also say that a player \textit{picks} a vertex if the player puts a pebble on this vertex. If in some round of the game element $e$ has a pebble on it, we say $e$ is \textit{pebbled} in this round. Sometimes, we also use the verb ``pick'' (a pebbled vertex) to mean ``remove'' (the pebble from the vertex).  
Let $e$ and $f$ be a pair of vertices picked in some round of the game $\Game_m^{k-1}\!(\mathfrak{A},\mathfrak{B})$, with $e$ picked in $\mathfrak{A}$ and $f$ picked in $\mathfrak{B}$. We use $e\Vdash f$ to denote it. And for any two sets $X$ and $Y$, if there is a bijection $\eta: X\mapsto Y$ such that 
for any $e\in X$ $e\Vdash \eta(e)$,  then we use $X\Vdash\! Y$ to denote it. If the sets are ordered, i.e. they are tuples, then the bijection simply maps the $i$-th element (or called item) of $X$ to the $i$-th element of $Y$. 
Somtimes we also use $e\Vdash f$ to denote the \textit{pair of vertices} $e$ and $f$, where $e\Vdash f$. 
\label{def-Vdash}

If Spoiler can only put pebbles on elements of $\mathfrak{A}$ and can play for arbitrary number of rounds, then such $(k-1)$-pebble games characterize exactly the expressive power of $\exists\mathcal{L}_{\infty\omega}^{k-1}$. %i.e. when $\mathfrak{A}$ and $\mathfrak{B}$ satisfy the same set of $\mathcal{L}_{\infty\omega}^k$ formulas. 
On the other hand, if the players can use arbitrary number of pebbles in the games, such games are called ($m$-round) Ehrenfeucht-Fra\"iss\' e\xspace games, written $\Game_m\!(\mathfrak{A},\mathfrak{B})$.

\section{Bounded variable hierarchy in $\exists \mathcal{L}_{\infty\omega}^{k}$} \label{existential-case-section}
To help understand some bits of the main proof, we first consider a simpler problem and a simple structure $\mathcal{B}_k$.

In a two dimension coordinate plane, the \textit{coordinate congruence number} (or, coordinate residue class number)  of a vertex $(x,y)$ in the plane, denoted by $\mathbf{cc}(x,y)$, is defined as the following: %\\[-20pt]
\begin{eqnarray}\label{def-coordinate-congruence-number}
\mathbf{cc}(x,y):=x+y \mbox{ mod } k-1%\\[-18pt] 
\end{eqnarray} 
Note that there are $k-1$ different values for coordinate congruence numbers. 

\begin{definition}
 $\mathcal{B}_k$ is an ordered graph over the universe $[k-1]\times [k]$ and the linear order is defined by  the lexicographic ordering on the Cartesian product $[k]\times [k-1]$. That is, $(x_i,y_i)\leq(x_j,y_j)$ if $y_i<y_j$ or $y_i=y_j\land x_i\leq x_j$.  
%let $E_k^c$ be its edge set, then $((x_i,y_i),(x_j,y_j))\in E_k^c$ 
A vertex $(x_i,y_i)$ is adjacent to another vertex $(x_j,y_j)$ 
if and only if $y_i\neq y_j$ and $\mathbf{cc}(x_i,y_i)\neq \mathbf{cc}(x_j,y_j)$. 
\hfill \ensuremath{\divideontimes}
\end{definition}
%Note that, $\mathcal{B}_k$ is a graph in a coordinate plane.
%\begin{lemma}\label{no-k-clique-in-B-exists}
It is easy to see that  
$\mathcal{B}_k$ has no $k$-clique, by pigeonhole principle.

In the following we introduce a result by Dawar (cf. \cite{DawarHowmany},  p27) and use this chance to introduce some bits of the ideas that are used in our main proofs. 
\begin{theorem}  \label{existential-case}
For each $k$, there is a formula of $\exists \mathcal{L}^{k}$ that is not equivalent to any formula of $\exists \mathcal{L}^{k-1}_{\infty\omega}$ on ordered graphs.
\end{theorem} 
\begin{proof}
Our tool is the variant of $(k-1)$-pebble games for $\exists\mathcal{L}_{\infty\omega}^{k-1}$, where the game board is $(\mathcal{A}_k,\mathcal{B}_k)$. Here $\mathcal{A}_k$ is a $k$-clique that is composed of the vertices $(0,0),(0,1),\ldots,(0,k-1)$ with a linear order defined as that of $\mathcal{B}_k$. Recall that Spoiler is required to pick only in  $\mathcal{A}_k$. 
Observe that Duplicator is able to ensure that the subgraph of $\mathcal{B}_k$ induced by the pebbles in $\mathcal{B}_k$ is a complete graph in each round. In particular, to ensure a $(k-1)$-clique in $\mathcal{B}_k$ Duplicator needs only pick those vertices such that for any two picked vertices $(x_i,y_i), (x_j,y_j)$, $\mathbf{cc}(x_i,y_i)\neq\mathbf{cc}(x_j,y_j)$ if $y_i\neq y_j$.  
The \textit{main point} 
is that, in each round, Duplicator can always find a vertex $(x,y)$ for any $y$ such that $\mathbf{cc}(x,y)$ is different from that of all the pebbled vertices, if there are no more than $k-2$ pebbles on a structure (cf. Fact  \ref{star_k-2}). 
\end{proof}

We introduce this proof instead of others (e.g. a proof of algebraic flavor) because the ideas presented in this proof can shed some light on Lemma \ref{no-missing-edges_xi-1}, which will be used in the proof of the main Lemma \ref{main-lemma} (cf. Strategy \ref{xi-1}). 

Note that $\mathcal{B}_k$ is the same as the structure introduced by Dawar, if we circular shift the vertices of the $i$-th row $i$ times to the right. 
Moreover, we shall see that such right circular shifts can prevent Duplicator from the so called ``boundary checkout strategy'' of Spoiler (cf. p. \pageref{page-boundary-checkout-strategy}). Nevertheless, to help the readers understand the intuitions behind the constructions and to make the proofs as less involved as possible, in the following sections we first introduce the original structures, then shift the vertices afterwards.    

\section{Outline of the remainder of the paper}
The rest of the paper is organized as follows. We first introduce a proof for the special case where $k=3$ in section \ref{special-cases}. It is a good place to bring forword a key notion called ``(structural) abstraction''. We index the vertices of our graphs and view the graphs in different scales, each of which is a distinctive abstraction. A higher abstraction characterizes some key feature of lower abstractions. And Duplicator uses strategies over abstractions to decide her picks in the original games. In this viewpoint, 
we reduce the original games to games over abstractions. 
That is, the players are also playing a game over some specific abstraction in each round, in addition to the original game: each pick are projected to this abstraction and Duplicator need only ensure partial isomorphisms over this abstraction to win this round. 

In section \ref{section-structures}, we will introduce a pair of graphs $\mathfrak{A}_{k,m}^*$ and $\mathfrak{B}_{k,m}^*$. 
Before this, we introduce a notion called board history, which characterizes reasonable evolutions of a game board, and ``embed'' it into every vertex of $\mathfrak{A}_{k,m}^*$ and $\mathfrak{B}_{k,m}^*$ to construct a pair of ordered graphs $\mathfrak{A}_{k,m}$ and $\mathfrak{B}_{k,m}$, for the game board. In the process of creating $\mathfrak{B}_{k,m}^*$, we need a notion called congruence label, based on which the key notion type label is defined, which  \textit{roughly} tells us how a vertex of some label is connected to another vertex of other label. An element in the definition of type lable is a set $\Omega$, based on which we forbid some sort of edges in  $\mathfrak{B}_{k,m}^*$. And such missing of edges characterizes some global feature of some subgraphs of $\mathfrak{B}_{k,m}^*$, thereby distinguishing one subgraph from another. The notion ``abstraction'' is somehow based on such features. The pair of main structures are $\widetilde{\mathfrak{A}}_{k,m}$ and $\widetilde{\mathfrak{B}}_{k,m}$, which are obtained from $\mathfrak{A}_{k,m}$ and $\mathfrak{B}_{k,m}$ by right circular shifting of vertices. 

In section \ref{winning-strategy}, we use $\widetilde{\mathfrak{A}}_{k,m}$ and $\widetilde{\mathfrak{B}}_{k,m}$ for the game board to prove, by a simultaneous induction, that Duplicator has a winning strategy in the game, thereby 
$k$-variables are needed for $k$-Clique in $\fo$. But 
instead of studying it over the main structures directly, we study it over $\widetilde{\mathfrak{A}}_{k,m}^*$ and $\widetilde{\mathfrak{B}}_{k,m}^*$, and classify the vertices of these two graphs into $m$ sets (the $(i-1)$-th set $\mathbb{X}_{i-1}^*$ subsumes the $i$-th set $\mathbb{X}_{i}^*$), each of which induces a graph, i.e. an abstraction, that resembles $\mathcal{B}_k$ to some extent. The $i$-th abstraction is an induced subgraph of the $(i-1)$-th abstraction. 
 We shall see that the Duplicator has a winning strategy in the original game if she has a winning strategy in a so called associated game over abstractions and changing board.  
In such an associated game, if Duplicator is not able to respond Spoiler by picking a vertex in the $i$-th abstraction, then Duplicator resorts to the $(i-1)$-th  abstraction for a solution. 
In the games, Duplicator can \textit{force} the games played over some \textit{specific} abstraction, which, when necessary, enables herself to find a solution in the closest lower abstraction in each round.   

In section \ref{chapter-circuit-bound}, based on a not well-known but still reasonable assumption, we show that $n^{k-1}$ gates not suffice to compute $k$-Clique on \textit{DLOGTIME-uniform} families of constant-depth unbounded fan-in circuits. In section \ref{lowerbound-in-FO(BIT)} we first show that $k$-variables are needed to define $k$-Clique in $\fo(\mathbf{BIT})$, by embedding the main structures (cf. section \ref{section-structures}) in a pure arithmetic structure introduced by Schweikardt and Schwentick \cite{SchweikardtLinearOrder}. Afterwards, in section \ref{circuit-lowerbound} we translate this result in logic to a size lower bound in circuit complexity. We first 
show that $O(n^{\frac{k-3}{2}})$ gates not suffice to compute $k$-Clique, using the standard translation and an observation that the bounded variable hierarchy in $\fo$ collapses to $\fo^3$ in the pure arithmetic. Afterwards, based on an assumption, we get the believed tight lower bound via a notion called succinct regular circuits,  whose structures respect the ``logical structure'' of first-order formulas.

In section \ref{conclusions}, we summarize our main results and discuss future work and open questions.

\section{The structures}

\subsection{Vertex index, structural abstractions and games over abstractions: the case where  $k=3$}\label{special-cases}
In this section we prove that our main result holds in the special case where $k=3$. That is, $3$ variables are needed to define $3$-Clique in $\fo$. In other words, Duplicator has a winning strategy in the $2$-pebble games of arbitrary finite rounds. The case where $k=3$ is quite different from other cases: it is much simpler than the cases where $k>3$ (see the subsequent sections), but much more difficult than the case where $k=2$.   
Note that it is trivial when $k=2$. 
For any $m\in\mathbf{N}^+$, $\mathfrak{B}_{2,m}$ is simply a graph of two isolated vertices with arbitrary order on them. $\mathfrak{A}_{2,m}$ is built from $\mathfrak{B}_{2,m}$ by adding one edge between these two vertices. Duplicator simply mimics Spoiler, which is a winning strategy in an $m$-round $1$-pebble game over the game board $(\mathfrak{A}_{2,m},\mathfrak{B}_{2,m})$, for arbitrary $m$.

For the special case where $k=3$,    
 we introduce a proof that is most suitable to cast light on some of the concepts and ideas that will be used in  the subsequent sections. In particular, we introduce a key concept called ``(structural) abstraction'', as well as pebble games over abstractions, which is also crucial in the subsequent sections. In addition to giving a proof for this special case,  
we hope this can offer some intuition for the following more technical constructions and proofs. Note that almost all the lemmas introduced in this section will be used in section \ref{section-structures} and section \ref{winning-strategy}. 

Firstly, we construct a structure $\mathfrak{B}_{3,m}^\prime$ via a process that can be called ``(iterative) structural expansion''.\footnote{Note that it is different from the concepts ``expansion'' and ``extension'' in model theory, as defined in the classical textbook by Chang and Keisler.} Instead of a formal definition, which is easy to give, we explain it briefly by the following example. We first construct a structure, called $\mathfrak{B}_{3,m}^\prime[\mathbb{X}_m^*]$, whose universe is a square lattice and whose width is $\gamma_0^*$. Then, we use it as a ``skeleton'' or ``blueprint'' to build a lager structure, called $\mathfrak{B}_{3,m}^\prime[\mathbb{X}_{m-1}^*]$. The basic ``brick'' we shall use to build base on the blueprint can be anything. But here the brick we use is similar to $\mathfrak{B}_{3,m}^\prime[\mathbb{X}_m^*]$ itself. More precisely, we ``expand'', or replace, every vertex by a successive vertices. Hence any ``path'' (not necessary connected) of the ``skeleton'' that is from the bottom to the top corresponds to a set of vertices of  $\mathfrak{B}_{3,m}^\prime[\mathbb{X}_{m-1}^*]$, which is isomorphic to a square lattice. We call such a square lattice (not necessary upright) a ``brick''. Such bricks are either isomorphic or very similar. Once we get $\mathfrak{B}_{3,m}^\prime[\mathbb{X}_{m-1}^*]$, whose width is $\gamma_1^*$, we take it as a new ``skeleton'' and use it to build $\mathfrak{B}_{3,m}^\prime[\mathbb{X}_{m-2}^*]$, and so on, until we get $\mathfrak{B}_{3,m}^\prime[\mathbb{X}_1^*]$, i.e. the structure $\mathfrak{B}_{3,m}^\prime$ we want, whose width is $\gamma_{m-1}^*$. 
Once $\mathfrak{B}_{3,m}^\prime$ is obtained, we create a new structure $\mathfrak{B}_{3,m}$, as well as $\widetilde{\mathfrak{B}}_{3,m}$, based on it. 
In the following we define $\mathfrak{B}_{3,m}^\prime$ formally.       

For any $m, i\in \mathbf{N}^+$, where $m\geq 3$ and $0< i<m$, let 
\begin{align}
\gamma_0^* &:=4m \label{gamma_0-star-k3}\\
\gamma_i^* &:=4(m-i)\gamma_{i-1}^* \label{gamma_i-star-k3}
\end{align}

For $x\in [\gamma_{m-1}^*]$ and $1\leq i\leq j\leq m$, let 
\begin{align}
\beta_{m-j}^{m-i} &:=\frac{\gamma_{m-i}^*}{\gamma_{m-j}^*}\label{beta-func}\\
[x]_i &:=\lfloor x/\beta_{m-i}^{m-1}\rfloor \label{abstraction-func}\\ 
\llparenthesis x\rrparenthesis_i &:=[x]_i\beta_{m-i}^{m-1}+\frac{1}{2}\sum_{1<\ell\leq i}\beta_{m-\ell}^{m-1}\label{llprrp-func}
\end{align}  

Note that $\beta_{m-j}^{m-i}=\displaystyle\prod_{m-j\leq \ell <m-i} \frac{\gamma_{\ell+1}^*}{\gamma_{\ell}^*}=4^{j-i}\times\frac{(j-1)!}{(i-1)!}$. By convention, $0!=1!=1$. Hence $\gamma_{m-1}^*=\gamma_0^*\beta_0^{m-1}=m!\times 4^m$.

Obviously, the structure $\mathfrak{B}_{3,m}^\prime$ is big. So we put a remark in the appendix to illustrate some essence of the notion structural expansion. Cf. Remark \ref{remark-expansion}.

The readers should be aware of the difference between the notations $[x]_i$ and $[x]$. The latter is seldom used in our paper, which most often appears in the Cartesian product when we define the universe of a structure.\footnote{By Wikipedia, Gauss introduced the notation $[x]$ for the \href{http://en.wikipedia.org/wiki/Floor_and_ceiling_functions}{floor function} in 1808, which remained the standard until 1962 when there is need to distinguish the notation of ceiling functions from that of floor functions. In our paper, no ceiling functions are involved. Moreover, we need a notation to distinguish it from the standard notation $\lfloor x\rfloor$. Hence we adopt and alter Gauss's notation here, i.e. using $[x]_i$ to denote a special kind of floor functions, as defined in (\ref{abstraction-func}).}  

Let 
\begin{equation}
\mathbb{X}_1^*:=[\gamma_{m-1}^*]\times [3].\label{X_1-star-k3}
\end{equation}
For  $1<i\leq m$, let 
\begin{equation}\label{def-eqn-X_i-star}
\mathbb{X}_i^*:=\{(x,y)\in\mathbb{X}_1^*\mid x=\llparenthesis x\rrparenthesis_i\}.
\end{equation}

The structure  
$\mathfrak{B}_{3,m}^\prime$ is an ordered graph over the universe $\mathbb{X}_1^*$, wherein 
the linear order is defined as the lexicographic ordering over $[3]\times [\gamma_{m-1}^*]$. And 
for any pair of vertices $(x_i,y_i)$ and $(x_j,y_j)$, if $(x_i,y_i)\in\mathbb{X}_p^*$ implies $(x_j,y_j)\in\mathbb{X}_p^*$, and $\ell$ is the maximum in $[1,m]$ s.t. $(x_i,y_i)\in\mathbb{X}_\ell^*$, then $(x_i,y_i)$ is adjacent to $(x_j,y_j)$ if and only if $y_i\neq y_j$ and $\mathbf{cc}([x_i]_\ell,y_i)\neq\mathbf{cc}([x_j]_\ell,y_j)$.\footnote{We shall see that $\mathbb{X}_{t+1}^*\subseteq\mathbb{X}_{t}^*$ for any $t$, due to Lemma \ref{i=0theni-1=0}. We will define a concept called ``vertex index'' (cf. Definition \ref{vertex-index}) and we shall see that $(x,y)\in\mathbb{X}_{t}^*-\mathbb{X}_{t+1}^*$ if and only if the index of $(x,y)$ is $t$, for $1\leq t<m$.}  

We can regard the universe of $\mathfrak{B}_{3,m}^\prime$ as a square lattice, whose width is $\gamma_{m-1}^*$ and whose height is $3$. Its lattice points are the set of elements of $\mathbb{X}_1^*$.  
We can define the $i$-th abstraction of the structure as the induced graph $\mathfrak{B}_{3,m}^{\prime}[\mathbb{X}_i^*]$, whose universe is a square (sub)lattice of $\mathfrak{B}_{3,m}^\prime$ and whose width is $\gamma_{m-i}^*$. 
For instance, the $m$-th abstraction is $\mathfrak{B}_{3,m}^{\prime}[\mathbb{X}_m^*]$, who has width $\gamma_0^*$ and height $3$. We take it that the $i$-th abstraction is a higher abstraction relative to the $j$-th abstraction if $i>j$. 
%In the $m$-th abstraction, i.e. $\mathfrak{A}_{k,m}^*[\mathbb{X}_m^*]$, we regard $\mathpzc{U}_m^*$ successive vertices as one object. By definition, there are $2^m$ such objects in this abstraction. 
Note that, for any $0\leq j\leq i\leq m-1$, we can regard $\beta_j^i$ successive vertices in the $(m-i)$-th abstraction as one ``vertex''  in the $(m-j)$-th abstraction. More precisely, we can take it that each row of  $\mathfrak{B}_{3,m}^{\prime}[\mathbb{X}_{m-i}^*]$ is divided evenly into $\gamma_{m-j}^*$ intervals of the same length $\beta_j^i$, where each vertex in $\mathbb{X}_{m-j}^*$ is roughly in the middle of some interval that is composed of vertices in  $\mathbb{X}_{m-i}^*$.

So far we regard  $\mathfrak{B}_{3,m}^{\prime}[\mathbb{X}_{m-j}^*]$ as an abstraction of  $\mathfrak{B}_{3,m}^{\prime}[\mathbb{X}_{m-i}^*]$. 
Besides ``(structual) abstraction'', we may also consider the dual concept, i.e.  ``(structual) expansion'', which describes the reverse side. That is, the following two statements are equivalent: 
\begin{itemize}
\item $\mathfrak{B}_{3,m}^{\prime}[\mathbb{X}_i^*]$ is an abstraction of $\mathfrak{B}_{3,m}^{\prime}[\mathbb{X}_{i-1}^*]$: $\beta_{m-i}^{m-i+1}$ vertices in the $(i-1)$-th abstraction are encapsulated into one vertex in the $i$-th abstraction;

\item $\mathfrak{B}_{3,m}^{\prime}[\mathbb{X}_{i-1}^*]$ is an expansion of  $\mathfrak{B}_{3,m}^{\prime}[\mathbb{X}_i^*]$: every vertex of $\mathfrak{B}_{3,m}^{\prime}[\mathbb{X}_i^*]$ is replaced by $\beta_{m-i}^{m-i+1}$ successive vertices. 

\end{itemize}

Note that, for any $(x,y)\in\mathbb{X}_1^*-\mathbb{X}_i^*$, we can regard $[x]_i$ as a sort of ``abstraction'', which tells us the ``(relative) position'' of $(x,y)$ in  $\mathfrak{B}_{3,m}^{\prime}[\mathbb{X}_i^*]$. So we call $[x]_i$ the ``$i$-th relative first coordinate of $(x,y)$'' and $([x]_i,y)$ the ``$i$-th relative position of the vertex $(x,y)$''.\footnote{Imagining that, if we look at a picture from far away, then many vertices in a row might seem as one.}
% or, equivalently, we can regard that $(x,y)$ is generated from $([x]_i,y)$ of the square lattice that is isomorphic to $\mathfrak{B}_{3,m}^{\prime}[\mathbb{X}_i^*]$, from the viewpoint of expansion.  
And for any $(x,y)\in \mathbb{X}_j^*$ where $1\leq j<i\leq m$, the vertex  
$(\llparenthesis x\rrparenthesis_{i},y)$ , which is a lattice point of $\mathfrak{B}_{3,m}^{\prime}[\mathbb{X}_i^*]$, can be regarded as the \textit{projection} of $(x,y)$ (a lattice point of $\mathfrak{B}_{3,m}^{\prime}[\mathbb{X}_j^*]$) in the $i$-th abstraction,  because $(\llparenthesis x\rrparenthesis_{i},y)\in \mathbb{X}_{i}^*$ and $[\llparenthesis x\rrparenthesis_{i}]_{i}=[x]_{i}$, which is unique for $(x,y)$ by the following lemma.

\begin{lemma}\label{projection}
Let $1\leq i, j\leq m$. For any $(x,y)\in \mathbb{X}_j^*$ and  $(x^{\prime},y)\in \mathbb{X}_{i}^*$, if $[x^{\prime}]_{i}=[x]_{i}$, then $x^{\prime}=\llparenthesis x\rrparenthesis_{i}$.
\end{lemma}

From now on, we call $(\llparenthesis x\rrparenthesis_i,y)$ the projection of $(x,y)$
 in the $i$-th abstraction. 
This lemma says that, if $(x^\prime,y)$ is a lattice point of $\mathfrak{B}_{3,m}^\prime[\mathbb{X}_i^*]$, and $(x,y)$ has the same $i$-th relative position as $(x^\prime,y)$, then $(x^\prime,y)$ is the projection of $(x,y)$ in the $i$-th abstraction. 

The following lemma says that $\mathbb{X}_i^*$ subsumes $\mathbb{X}_j^*$ if $i\leq j$. Hence a vertex is in lower abstractions if it is in some higher abstraction. That is, for the square lattice $\mathfrak{B}_{3,m}^\prime$, a lattice point of the square (sub)lattice $\mathfrak{B}_{3,m}^\prime[\mathbb{X}_j^*]$ is also a lattice point of the square (sub)lattice $\mathfrak{B}_{3,m}^\prime[\mathbb{X}_i^*]$. 
\begin{lemma}\label{i=0theni-1=0}
For any $i$ where $1\leq i\leq m$, if $(x,y)\in\mathbb{X}_i^*$, then $(x,y)\in\mathbb{X}_{j}^*$ for any $1\leq j\leq i$.
\end{lemma}

In other words,  $(x,y)\in\mathbb{X}_i^*$ implies that $x=\llparenthesis x\rrparenthesis_j$ for $1\leq j\leq i$. 
Because of this lemma, it is meaningful to introduce the following important concept, by which we can index the vertices of $\mathfrak{B}_{3,m}^\prime$.  
\begin{definition}\label{vertex-index}
The index of $(x,y)\in \mathbb{X}_1^*$, written $\mathrm{idx}(x,y)$, is the maximum $i$, where $1\leq i\leq m$, such that $(x,y)\in \mathbb{X}_i^*$. 
\hfill\ensuremath{\divideontimes}
\end{definition}
 By Lemma  \ref{i=0theni-1=0}, $(x,y)$ has index $i$ if and only if $(x,y)\in \mathbb{X}_i^*-\mathbb{X}_{i+1}^*$, for $1\leq i<m$; and $\mathrm{idx}(x,y)\geq j$ if and only if  $(x,y)\in \mathbb{X}_j^*$, for $1\leq j\leq m$.

Note that $\mathfrak{B}_{3,m}^\prime$ has many $3$-cliques, i.e. triangles. We can index these triangles such that the index of a triangle is the smallest index of its vertices. We can generalize this concept to abitrary $k$ as the following.  

\begin{definition}
 A $k$-clique $C_k$, where $|C_k|\subset \mathbb{X}_1^*$, has index $i$ if $i$ is the maximum  in the range $[1,m]$ such that  $|C_k|\subset \mathbb{X}_i^*$.
\hfill\ensuremath{\divideontimes}
\end{definition}

By definitions, we have the following easy observations, whose proofs are straightforward.

\begin{lemma}\label{lattice-point-high-is-lower}
For any vertex $(x,y)$ of index $i$ and $j\leq i$, we have
\begin{equation*}
x=\llparenthesis x\rrparenthesis_j. 
\end{equation*}  
\end{lemma}
By lemma \ref{i=0theni-1=0}, if $(x,y)$ has index $i$ and $j\leq i$, then $(x,y)$ is already a vertex in $\mathbb{X}_j^*$. Therefore, by definition, the projection of $(x,y)$ in the $j$-th abstraction is itself.

\begin{lemma}\label{proj-abs}
For any $(x,y)\in\mathbb{X}_1^*$ and $i\leq j$, we have 
\begin{enumerate}[(1)]
\item $[\llparenthesis x\rrparenthesis_i]_j=[x]_j.$

\item $\llparenthesis\llparenthesis x\rrparenthesis_i\rrparenthesis_j=\llparenthesis x\rrparenthesis_j.$

\end{enumerate}
\end{lemma}
In particular, by \textit{(1)}, we have $[\llparenthesis x\rrparenthesis_i]_i=[x]_i$. 
\textit{(2)} says that the projection of $(x,y)$ to the $j$-th abstraction can be regarded as a process wherein we first project $(x,y)$ to the $(i+1)$-th abstraction, then to the $(i+2)$-th abstraction, and so on, until we project $(\llparenthesis x\rrparenthesis_{j-1},y)$ to the $j$-th abstraction, i.e. projecting to $(\llparenthesis x\rrparenthesis_j,y)$.

We immediately have the following observation, as a corollary of \textit{(2)} of Lemma \ref{proj-abs}.  
\begin{lemma}\label{proj-greater-index}
For any $(x,y)\in\mathbb{X}_1^*$ and $i\in [1,m]$, we have 
\begin{equation*}
\mathrm{idx}(\llparenthesis x\rrparenthesis_i,y)\geq i.
\end{equation*}
\end{lemma}
This lemma says that the projection of $(x,y)$ in the $i$-th abstraction is a vertex in $\mathbb{X}_i^*$, which is obvious. To prove it, we need only show that $(\llparenthesis x\rrparenthesis_i,y)\in\mathbb{X}_i^*$, i.e. $\llparenthesis x\rrparenthesis_i=\llparenthesis \llparenthesis x\rrparenthesis_i\rrparenthesis_i$.

\begin{fact}\label{specialcase-fact-surrounding}
For any vertex $(x,y)$ of index $i$ where $1<i\leq m$, there are exactly  $\beta_{m-i}^{m-i+1}-1$ vertices of index $i-1$ such that the projections of these vertices in the $i$-th abstraction is exactly $(x,y)$. Moreover, these vertices, together with $(x,y)$, are successive which form an interval. And $(x,y)$ is in the middle of this interval.   
\end{fact}
%Since proving it needs more than five lines, we give it briefly. 
\noindent\underline{\textit{Proof of Fact:}}\\[6pt]
\indent The first part of this claim is obvious: these vertices are exactly the set of vertices $([x]_i\beta_{m-i}^{m-i+1}+\ell,y)$ for any $0\leq \ell<\frac{1}{2}\beta_{m-i}^{m-i+1}$ or $\frac{1}{2}\beta_{m-i}^{m-i+1}<\ell<\beta_{m-i}^{m-i+1}$. For any other vertex, it is easy to verify that its projection in the $i$-th abstraction is either less than or greater than $(x,y)$ with respect to the linear order. 

To prove the second part of this claim, we observe that the (relative) position of $(x,y)$, i.e. $(\llparenthesis x\rrparenthesis_i,y)$ (cf. Lemma \ref{lattice-point-high-is-lower}), in the $(i-1)$-th abstraction is  $\lfloor \llparenthesis x\rrparenthesis_i/\beta_{m-i+1}^{m-1}\rfloor$.  Note that  
$\lfloor \llparenthesis x\rrparenthesis_i/\beta_{m-i+1}^{m-1}\rfloor=[x]_i\beta_{m-i}^{m-i+1}+\frac{1}{2}\beta_{m-i}^{m-i+1}$. 
Hence, there are $\frac{1}{2}\beta_{m-i}^{m-i+1}$ successive vertices of index $i-1$ that are on the left side of $(x,y)$ in the $y$-th row, and there are $\frac{1}{2}\beta_{m-i}^{m-i+1}-1$  successive vertices of index $i-1$ that are on the right side of $(x,y)$. This concludes the claim. 

\noindent\underline{\textit{Q.E.D. of Fact.}}\\[-5pt]

The following is a direct corollary of this fact.
\begin{fact}\label{specialcase-fact-surrounding-upto}
For any vertex $(x,y)$ of index $i$ where $1<i\leq m$, there are exactly  $\beta_{m-i}^{m-1}-1$ vertices of index up to $i-1$ such that the projections of these vertices in the $i$-th abstraction is exactly $(x,y)$. Moreover, these vertices, together with $(x,y)$, are successive which form an interval. And $(x,y)$ is roughly in the middle of this interval.   
\end{fact}
\noindent\underline{\textit{Proof of Fact:}}\\[6pt]
\indent Just note that $$\beta_{m-i}^{m-1}=\displaystyle\prod_{1<j\leq i}\beta_{m-j}^{m-j+1}.$$

\noindent\underline{\textit{Q.E.D. of Fact.}}\\[-5pt]

Assume that $0\leq j< i\leq m-1$. 
For any vertex $(x^\star,y)$ of index $m-j$, we call those $\beta_j^i-1$ vertices, whose indices are greater than or equal to $m-i$ but less than $m-j$ and whose projections in the $(m-j)$-th abstraction are  $(x^\star,y)$, \textit{the vertices in $\mathbb{X}_{m-i}^*$ that surround $(x^\star,y)$}. 
 For example, for any $(x,y)$ where $\mathrm{idx}(x,y)=m-i$, $(x,y)$ is a vertex that surrounds $(x^\star,y)$ if $[x]_{m-j}=[x^\star]_{m-j}$. 
And each vertex of index $m-j$ is surrounded by $\beta_j^{m-1}-1$ vertices of lower abstractions, i.e. the vertices in $\mathbb{X}_1^*$. 
Therefore, each vertex of index $m$ is surrounded by $\beta_0^1-1$ vertices of index $m-1$, where this vertex of index $m$ is in the middle of the interval that is composed of these surrounding vertices in $\mathbb{X}_{m-1}^*$; each vertex in $\mathbb{X}_{m-1}^*$, i.e. a vertex of index $m-1$ or $m$, is also surrounded by $\beta_1^2-1$ vertices of index $m-2$, and so on. 

A direct corollary of Fact \ref{specialcase-fact-surrounding-upto} is the following fact. 
\begin{fact}\label{unit-distance}
For any $(x,y),(x^\prime,y)\in\mathbb{X}_i^*$, we have
\begin{equation*}
|x-x^\prime|=c\beta_{m-i}^{m-1}, 
\end{equation*}
for some $c\in \mathbf{N}_0$. 
\end{fact}

We shall introduce pebble games over abstractions. The following observation is crucial for such games. 
\begin{lemma}\label{abstraction-strategy-premier} 
For any $1<\xi\leq m$ and $a,a^\prime\in [\gamma_{m-1}^*]$, if $a-\llparenthesis a\rrparenthesis_{\xi}=a^\prime-\llparenthesis a^\prime\rrparenthesis_{\xi}$, then the following hold:
\begin{enumerate}[(1)]
\item $\llparenthesis a\rrparenthesis_{\xi}-\llparenthesis a\rrparenthesis_{\xi-1}=\llparenthesis a^\prime\rrparenthesis_{\xi}-\llparenthesis a^\prime\rrparenthesis_{\xi-1}$

\item  $a-\llparenthesis a\rrparenthesis_{\xi-1}=a^\prime-\llparenthesis a^\prime\rrparenthesis_{\xi-1}$.\\[-1pt]

\end{enumerate}
\end{lemma}

For $i\in [k]$, let 
\begin{equation}
tr(i):=(i\,\, \mbox{mod } k-1)\times\sum_{1\leq p\leq m}\beta_{m-p}^{m-1}.
\end{equation}

\textit{The structure}  $\widetilde{\mathfrak{B}}_{3,m}$ is constructed from  
$\mathfrak{B}_{3,m}^\prime$ by 
\begin{enumerate}
\item removing a set of edges: for any vertex $(x,1)\in\mathbb{X}_\ell^*-\mathbb{X}_{\ell+1}^*$ (i.e. $\mathrm{idx}(x,1)=\ell<m$) where $[x]_\ell$ is even, we delete the following edges between $(x,1)$ and any vertex in $\Omega_{x}$ where  
\begin{equation}\label{def-Omega_xy}
\Omega_{x}:=\{(u,v)\in\mathbb{X}_{\ell+1}^* \mid (u,v) \mbox{ is not adjacent to } (\llparenthesis x\rrparenthesis_{\ell+1},1)\}; 
\end{equation}
 
\item circular shifting the vertices of the $i$-th row for $tr(i)$ times to the right. \label{def-circular-shifting}

\end{enumerate}

We call the structure constructed from $\mathfrak{B}_{3,m}^\prime$ after the first step (i.e. before the circular shifts) $\mathfrak{B}_{3,m}$. 
Note that, all the Lemmas mentioned so far continue to work with such adaption. 
In the following we are mainly interested in \textit{this} structure and all the results are created for \textit{this} structure. The shifts will be met and used \textit{only} when we discuss ``4$^\diamond$'' in the proof of Lemma \ref{winning-strategy-in-k=3} (cf. p. \pageref{page-4-diamond}) and  a strategy of Spoiler called ``boundary checkout strategy'' (cf. p. \pageref{page-boundary-checkout-strategy}).

%\begin{comment}
\begin{figure}[htbp]
%\hspace*{-3mm}
\centering
\includegraphics[trim = 0mm 0mm 0mm 0mm, clip, scale=0.57]{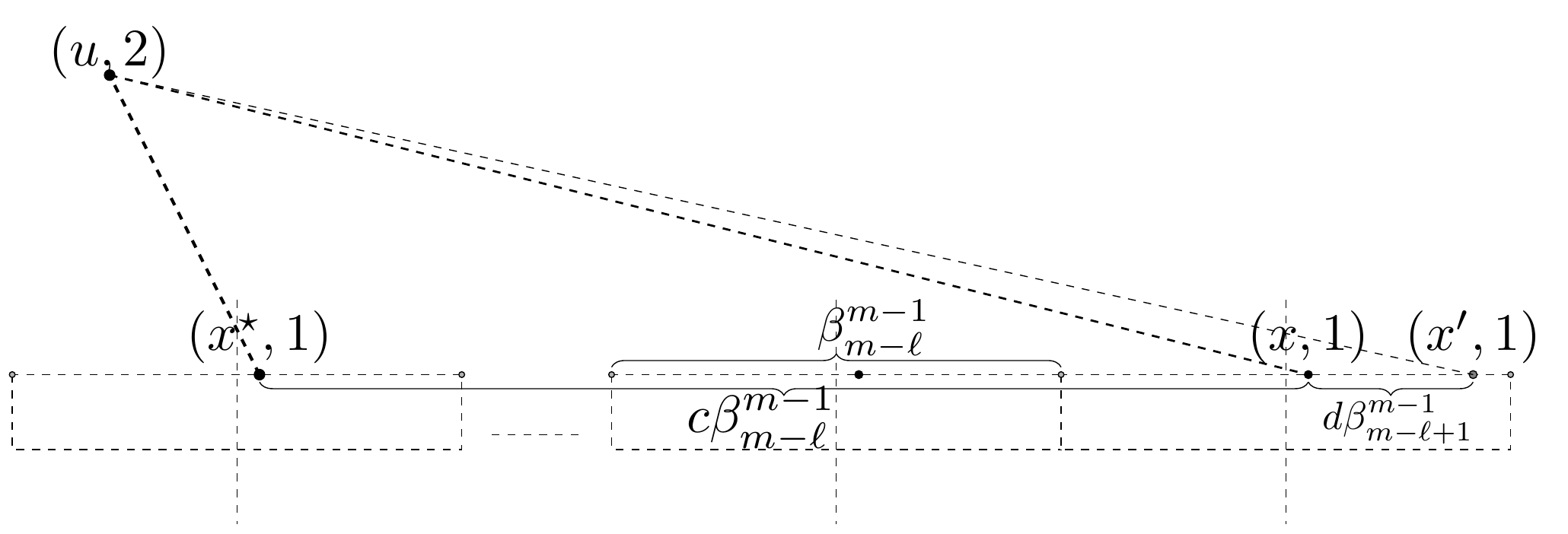}
%\scalebox{10}{}
\caption{From $\mathfrak{B}_{3,m}^\prime$ to $\mathfrak{B}_{3,m}$: some edges are forbidden. Here $x^\star=\llparenthesis x\rrparenthesis_{\ell+1}$. Suppose $\mathrm{idx}(x,1)=\ell$, $\mathrm{idx}(x^\prime,1)=\ell-1$, and  $(u,2),(x^\star,1)\in\mathbb{X}_{\ell+1}^*$. Assume $c$ and $d$ are even.}
\label{fig-k=3-trans-remove}
\end{figure}
%\end{comment}

In Fig. \ref{fig-k=3-trans-remove}, $c$ is even, i.e. $[x]_i$ is even ($\because (x^\star,1)\in\mathbb{X}_{\ell+1}^*$; we shall see it shortly). Then, by \eqref{def-Omega_xy}, $(x,1)$ is not adjacent to $(u,2)$ since $(x^\star,1)$ is not adjacent to $(u,2)$. Similarly, $(x^\prime,1)$ is not adjacent to $(u,2)$ since $(x,1)$ is not adjacent to $(u,2)$ and $d$ is even. Hence, the missing of an edge in higher abstraction (e.g. the one between $(x^\star,1)$ and $(u,2)$) will propagate to lower abstractions (e.g.  the one between $(x^\prime,1)$ and $(u,2)$).

By Fact \ref{specialcase-fact-surrounding-upto}, for any vertex $(x,y)$ of index $\ell$ where $1<\ell\leq m$, there are exactly  $\beta_{m-\ell}^{m-1}-1$ vertices of index up to $\ell-1$ such that the projections of these vertices in the $\ell$-th abstraction is exactly $(x,y)$. Moreover, these vertices, together with $(x,y)$, are successive which form an interval (i.e. the dashed rectangle in Fig. \ref{fig-k=3-trans-remove}). And $(x,y)$ is roughly in the middle of this interval.

\textbf{Note that}, for the sake of convenience, here we regard the leftmost vertex of the $\ell$-th row of $\widetilde{\mathfrak{B}}_{3,m}$ as $(\gamma_{m-1}^*-tr(\ell),\ell)$ instead of $(0,\ell)$.\footnote{\label{footnote-boundary} In this viewpoint, we regard  ``$(x,y)$'' as a name or label for the associated vertex. Then we can preserve the definitions, such as \eqref{def-coordinate-congruence-number}. A shortcoming of such treatment is that we have to be cautious when computing the distance of two vertices in a row. It is possibly no more the difference of the ``first coordinates''. Fortunately, most often we can think of $\mathfrak{B}_{3,m}$ instead of $\widetilde{\mathfrak{B}}_{3,m}$. Only when we discuss the reason ``4$^\diamond$ can be ensured'' (cf. p. \pageref{page-4-diamond}) or when we meet ``boundary checkout strategy'' (cf. p. \pageref{page-boundary-checkout-strategy}), we should switch to $\widetilde{\mathfrak{B}}_{3,m}$. }  
Because both $\gamma_{m-1}^*$ and $\beta_{m-p}^{m-1}$ are divisible by $(k-1)\beta_{m-i}^{m-1}$ for any $i<p$,  we have the following observation. 
\begin{lemma}\label{cc-boundary-vertex-is-0}
For any $\ell\in [k]$ and $1\leq i\leq m$,
$$\mathbf{cc}([\gamma_{m-1}^*-tr(\ell)]_i,\ell)=0.$$ 
\end{lemma}
It implies that, for any $i$ and $0<\ell<k-1$ (i.e. $\ell\not\equiv 0$ (mod $k-1$)),   
\begin{equation}\label{boundary-nequiv-0-mod-2}
[\gamma_{m-1}^*-tr(\ell)]_i\not\equiv 0 \hspace{3pt} (\mbox{mod }k-1).
\end{equation}

Only when $\ell=0$ or $\ell=k-1$, we have that, for any $i$,  
\begin{equation}\label{boundary-equiv-0-over-BottomCeiling}
[\gamma_{m-1}^*-tr(\ell)]_i\equiv 0 \hspace{3pt} (\mbox{mod }k-1).
\end{equation}

By Fact \ref{specialcase-fact-surrounding} and Fact \ref{specialcase-fact-surrounding-upto}, we know that the index of a vertex in $\mathbb{X}_i^*$ should be $i$, if it is a boundary vertex (i.e. either the leftmost or the rightmost) of the $i$-th abstraction. Similarly, we have the following observation.
\begin{lemma}\label{boundary-index-over-abstractions}
For any $1\leq i\leq m$ and $\ell\in [k]$,
$$\mathrm{idx}(\llparenthesis\gamma_{m-1}^*-tr(\ell)\rrparenthesis_i,\ell)=i.$$
\end{lemma}
\begin{proof}
The reason is simple. Just observe that, in $\mathfrak{B}_{k,m}$, $(\llparenthesis \gamma_{m-1}^*\rrparenthesis_i,\ell)$ is the rightmost vertex whose index is $i$,  
 and that the distance between $(\llparenthesis \gamma_{m-1}^*\rrparenthesis_i,\ell)$ and $(\llparenthesis\gamma_{m-1}^*-tr(\ell)\rrparenthesis_i,\ell)$ is only $\ell$, much less than $\frac{1}{2}\beta_{m-i-1}^{m-i}-1$. As a consequence, the index of $(\llparenthesis\gamma_{m-1}^*-tr(\ell)\rrparenthesis_i,\ell)$ cannot be $i+1$. On the other hand, by Lemma \ref{proj-greater-index}, $(\llparenthesis\gamma_{m-1}^*-tr(\ell)\rrparenthesis_i,\ell)\in \mathbb{X}_i^*$. This concludes the claim. 
\end{proof}

Let 
\begin{equation}
mid:=2m\beta_0^{m-1}+\frac{1}{2}\sum_{1<j\leq m}\beta_{m-j}^{m-1}.
\end{equation}

By the defintion, we know that $mid=\llparenthesis mid\rrparenthesis_m$, thereby $(mid,y)\in\mathbb{X}_m^*$ for any $y$. Note that $mid$ is roughly half of $\gamma_{m-1}^*$.\footnote{Just observe that $\gamma_{m-1}^*=4m\beta_0^{m-1}$ and that $\beta_j^{m-1}$ is much smaller than $\beta_0^{m-1}$ when $1<j<m$.} The structure 
$\mathfrak{A}_{3,m}$ is built from $\mathfrak{B}_{3,m}$ by adding an edge between $(mid,0)$ and $(mid,2)$. Call the endpoints of this edge critical points. 

The structure $\widetilde{\mathfrak{A}}_{3,m}$ is obtained from $\mathfrak{A}_{3,m}$ by the same circular shifts as the way we obtain $\widetilde{\mathfrak{B}}_{3,m}$. In other words, $\widetilde{\mathfrak{A}}_{3,m}$ is also obtained from $\widetilde{\mathfrak{B}}_{3,m}$ by adding an edge between $(mid,0)$ and $(mid,2)$. 

For each set $\mathbb{X}_i^*$, we define the $i$-th abstraction of the structure $\mathfrak{A}_{3,m}$ ($\mathfrak{B}_{3,m}$ resp.) by $\mathfrak{A}_{3,m}[\mathbb{X}_i^*]$ ($\mathfrak{B}_{3,m}[\mathbb{X}_i^*]$ resp.). For any $(u,v)$, call $\mathbf{cc}([u]_\ell,v)$ \textit{the coordinate congruence number of $(u,v)$ in the $\ell$-th abstraction}.  

The following statement is a straightforward but important observation. Recall that $k=3$ in this section.

\begin{lemma} \label{cm=depth}
If $1\leq j<i$ and $(x,y)\in \mathbb{X}_i^*$, then $[x]_{j}\equiv 0$ (mod $k-1$). 
\end{lemma}
In other words, 
$\mathbf{cc}([x]_{j},y)=y$ mod $k-1$. 
That is, for any vertex in higher abstraction, its coordinate congruence number in lower abstractions is completely determined by its second coordinate.   
Hence, in the case where $k=3$, for any $(u,v)\in\mathbb{X}_{\ell+1}^*$, $\mathbf{cc}([u]_\ell,v)=v$ mod $2$ since $[u]_\ell$ is even.  This is crutial for the following vital observation, which leads to the notion ``abstraction''.   

\begin{remark}\label{ExplanationOfAbstraction-specalcase}
%Recall that each vertex  of index $m-j$ is surrounded by $\beta_j^{m-1}$ vertices of lower abstractions. 
For any vertex $(x,y)\in\mathbb{X}_2^*$ and a number $i$ where $i<\mathrm{idx}(x,y)=t$, we call the set of successive vertices \textit{surrounding} $(x,y)$  the $i$-th \textit{complete expansion} of $(x,y)$, denoted $cex(x,y,i)$, if their relative positions in the $t$-th abstraction are the same as that of $(x,y)$. More precisely, $cex(x,y,i)=\{(u,y)\in\mathbb{X}_1^*\mid \mathrm{idx}(u,y)\leq i; \llparenthesis u\rrparenthesis_{t}=x\}\cup\{(x,y)\}$. 
 We can regard $cex(x,y,i)$ as \textit{one object} \label{def-special-object} 
that contains $(x,y)$. For example, the object $cex(mid,0,i)$  contains the critical point $(mid,0)$ and those vertices whose indices are no more than $i$ and their relative position in the $m$-th abstraction is $(mid,0)$.  
The reason we regard $\mathfrak{B}_{3,m}[\mathbb{X}_i^*]$ as an ``abstraction'' of $\mathfrak{B}_{3,m}$ is not only because we can regard $\beta_{m-i}^{m-1}$ elements in the first abstraction as one element in the $i$-th abstraction, but also because whether two vertices are adjacent in the $i$-th abstraction will determine the adjacency of some of the vertices in the lower  abstractions, due to the missing of edges in the process we produce $\mathfrak{B}_{3,m}$ from $\mathfrak{B}_{3,m}^{\prime}$ 
(cf. (\ref{def-Omega_xy}), the definition of $\Omega_{x}$; when $k>3$, cf. ``$\!\restriction\!\Omega$'' in Definition \ref{type-label}). What's more, for any $(x_0,y_0),(x_1,y_1),(x_0^\prime,y_0),(x_1^\prime,y_1)\in\mathbb{X}_{t}^*-\mathbb{X}_{t+1}^*$, the subgraph induced by $cex(x_0,y_0,t-1)$ and $cex(x_1,y_1,t-1)$ is isomorphic to the subgraph induced by  $cex(x_0^\prime,y_0,t-1)$ and $cex(x_1^\prime,y_1,t-1)$ if and only if the adjacency between $(x_0,y_0)$ and $(x_1,y_1)$ is the same as that between  $(x_0^\prime,y_0)$ and $(x_1^\prime,y_1)$.\footnote{That is, $(x_0,y_0)$ is adjacent to $(x_1,y_1)$ iff $(x_0^\prime,y_0)$ is adjacent to $(x_1^\prime,y_1)$.} Here we give a brief intuitive explanation. A strict proof is very verbos, which can be found in Remark \ref{special-locally-isom}. 
The isomorphism is defined by the bijection $h$ such that $h(x_i)=x_i^\prime$, where $i\in\{0,1\}$, and $u-x_i=h(u)-h(x_i)$ for any $(u,y_i)$ in the subgraph induced by $cex(x_0,y_0,t-1)$ and $cex(x_1,y_1,t-1)$. From Fact \ref{specialcase-fact-surrounding} and Fact \ref{specialcase-fact-surrounding-upto}, we \textit{can} see that $\mathrm{idx}(u,y_i)=\mathrm{idx}(h(u),y_i)$. We confess that this is word-of-mouth. %For a strict and formal proof, the reader can cf. Remark \ref{special-locally-isom}. 
Suppose that $\mathrm{idx}(u,y_i)=\ell$. 
%By lemma \ref{projection} and lemma \ref{i=0theni-1=0}, if there is no vertex between $(u,y_i)$ and $(x_i,y_i)$, whose index is in $[\ell+1,t-1]$, then $\llparenthesis u\rrparenthesis_j=x_i$ for any $j$ where $\ell<j\leq t$. 
By lemma \ref{cm=depth} and the fact that $(u-x_i)/\beta_{m-\ell}^{m-1}=(h(u)-h(x_i))/\beta_{m-\ell}^{m-1}$, we have $\mathbf{cc}([u]_\ell,y_i)=\mathbf{cc}([h(u)]_\ell,y_i)$. 
Moreover, we can show that $\mathbf{cc}([u]_j,y_i)=\mathbf{cc}([h(u)]_j,y_i)$ for $\ell<j\leq t$. For a strict proof,  cf. Remark \ref{special-locally-isom}. 
These facts justify the following observations.  
Firstly, observe that the subgraph induced by $cex(x_0,y_0,t-1)$ and $cex(x_1,y_1,t-1)$ is isomorphic to the subgraph induced by  $cex(x_0^\prime,y_0,t-1)$ and $cex(x_1^\prime,y_1,t-1)$, without considering the missing of edges due to the process we create $\mathfrak{B}_{3,m}$ from $\mathfrak{B}_{3,m}^\prime$. 
Secondly, by definition (\ref{def-Omega_xy}), we know that the \textit{nonadjacency} of $(x_0,y_0)$ and $(x_1,y_1)$ will propagate to lower abstractions.  Assume that $(x_0,y_0)$ is not adjacent to $(x_1,y_1)$ and so are $(x_0^\prime,y_0)$ and $(x_1^\prime,y_1)$. Let $(u,y_1)$ be a vertex of index $\ell=t-1$ such that $\llparenthesis u\rrparenthesis_t=x_1$ (note that it is in the subgraph induced by $cex(x_0,y_0,t-1)$ and $cex(x_1,y_1,t-1)$). 
%, and $(u^\prime,y_1)$ be the corresponding vertex of index $t-1$ such that  $\llparenthesis u^\prime\rrparenthesis_t=x_1^\prime$. Assume that  $x_1-\llparenthesis u\rrparenthesis_t=x_1^\prime-\llparenthesis u^\prime\rrparenthesis_t$. 
Then, by the previous observations that $\mathrm{idx}(u,y_1)=\mathrm{idx}(h(u),y_1)$ and $\mathbf{cc}([u]_{t-1},y_1)=\mathbf{cc}([h(u)]_{t-1},y_1)$, the vertex $(u,y_1)$ is not adjacent to $(x_0,y_0)$ iff either $[u]_{t-1}$ is even 
  or $\mathbf{cc}([x_0]_{t-1},y_0)=\mathbf{cc}([u]_{t-1},y_1)$. Similarly, $(h(u),y_1)$ is not adjacent to $(h(x_0),y_0)$ iff either $[h(u)]_{t-1}$ is even 
  or $\mathbf{cc}([h(x_0)]_{t-1},y_0)=\mathbf{cc}([h(u)]_{t-1},y_1)$. That is, the adjacency of $(x_0,y_0)$ and $(u,y_1)$ is the same as that of $(h(x_0),y_0)$ and $(h(u),y_1)$: the missing of edges is propagated from the $t$-th abstraction to the $(t-1)$-th abstraction. Such propagations will be the same in two isomorphic structures when they are toward to lower abstractions. In summary, the adjacency of $(x_0,y_0)$ and $(x_1,y_1)$ determines the unique feature of the subgraph induced by $cex(x_0,y_0,t-1)$ and $cex(x_1,y_1,t-1)$. We can generalize it by introducing more vertices in $\mathbb{X}_t^*$.
  For example, assume that the vertices $(x,0)$, $(x,1)$ and $(x,2)$ have the same adjacency as the vertices  $(x^\prime,0)$, $(x^\prime,1)$ and $(x^\prime,2)$; and assume that they are vertices in $\mathbb{X}_t^*-\mathbb{X}_{t+1}^*$. Then the  subgraph induced by $cex(x,0,t-1)$, $cex(x,1,t-1)$ and $cex(x,2,t-1)$ is isomorphic to the subgraph induced by $cex(x^\prime,0,t-1)$, $cex(x^\prime,1,t-1)$ and $cex(x^\prime,2,t-1)$.  
This somehow justifies the notion  ``abstraction''.  For any $i,j,\ell$ where $i>j>p$, $\mathfrak{B}_{3,m}[\mathbb{X}_i^*]$ is an ``abstraction'' of  $\mathfrak{B}_{3,m}[\mathbb{X}_j^*]$, which is also an ``abstraction'' of  $\mathfrak{B}_{3,m}[\mathbb{X}_p^*]$. Hence ``abstraction'' is relative. And each abstraction is a sketch of the structure with respect to some ``scale''.   

A strict argument for the above observation is verbose and involved, cf. the proofs of Lemma \ref{approxi-copy-cat}, \ref{approxi-copy-cat-1} and Lemma \ref{corollary-approxi-copy-cat} for insights, which are used in a more general and complicated setting. But this perhaps helps: the way we construct the structure via iterative structural expansion enforces the isomorphism of neighbourhoods of vertices of identical index. 
\hfill\ensuremath{\divideontimes}
\end{remark}

%++++++++++++++++++++++++++++

Clearly, $\mathfrak{A}_{3,m}$ has triangles, which implies that $\widetilde{\mathfrak{A}}_{3,m}$ has triangles. In particular, $\mathfrak{A}_{3,m}$ has 
 a triangle formed by the set of vertices $$\{(mid,0), (mid,1), (mid,2)\},$$ because all the vertices have index $m$, which implies that $\Omega_{mid}=\emptyset$, and $\mathbf{cc}([mid]_m,i)=i$ mod $2$, which implies that both $(mid,0)$ and $(mid,2)$ are adjacent to $(mid,1)$. In  contrast, we have the following observation. 
\begin{fact}\label{B_3_m-is-trianglefree}
$\mathfrak{B}_{3,m}$ has no triangle.  
\end{fact}
\noindent\underline{\textit{Proof of Fact:}}\\[4pt]
We prove it by contradiction. 
Assume that there are triangles in $\mathfrak{B}_{3,m}$ and $C_3$ is such a triangle that has the maximum index, say $t$. Note that $t$ cannot be $m$, for otherwise there are two vertices that have the same coordinate congruence number in the $m$-th abstraction  by  the pigeonhole principle. Similarly, 
$C_3$ must contain both vertices in $\mathbb{X}_t^*-\mathbb{X}_{t+1}^*$ and vertices in $\mathbb{X}_{t+1}^*$, due to the pigeonhole principle. Let $|C_3|=\{(a,0),(b,1),(c,2)\}$, inasmuch as the second coordinates of the vertices of $C_3$ must be different. 
  Let $P=\{(x,y)\in \mathbb{X}_{t}^*-\mathbb{X}_{t+1}^*\mid (x,y)\in |C_3|\}$. And  let 
 $Q=\{(x,y)\in \mathbb{X}_{t+1}^* \mid (x,y)\in |C_3|\}$. Note that $P\cap Q=\emptyset$. By Lemma \ref{i=0theni-1=0}, the set of vertices of $C_k$ is exactly $P\cup Q$.

Let $cC_3:=\{\mathbf{cc}([x]_t,y)\1 (x,y)\in |C_3|\}$.  
Since there are $3$ elements in $C_3$ and $|cC_3|\leq 2$, by pigeonhole principle, there are two vertices such that their coordinate congruence numbers in the $t$-th abstraction are the same. If one of them is in $P$, then there is no edge between them, by definition. Therefore, to have a triangle, both of them should be in $Q$. Recall that, by Lemma \ref{cm=depth}, $\mathbf{cc}([x]_t,y)=y$ mod $2$ for any $(x,y)\in \mathbb{X}_{t+1}^*$. Therefore, their coordinate congruence numbers in the $t$-th abstraction should be $0$. In other words, these two vertices are $(a,0)$ and $(c,2)$. Note that $(b,1)\in P$ since $P\neq \emptyset$ and $\mathbf{cc}([b]_t,1)$ should be $1$,  for otherwise $(b,1)$ is not adjacent to both $(a,0)$ and $(c,2)$. In other words, $[b]_t$ is even. Note that $(\llparenthesis b\rrparenthesis_{t+1},1)$ is either not adjacent to $(a,0)$ or not adjacent to $(c,2)$, for otherwise there is a $k$-clique whose index is greater than $t$. That is, either $(a,0)$ or $(c,2)$ is in $\Omega_{b}$. Therefore,  either $(a,0)$ or $(c,2)$ is not adjacent to $(b,1)$. A contradiction occurs. 

\noindent\underline{\textit{Q.E.D. of Fact.}}\\[-6pt]

As a direct corollary, we have  
\begin{equation}
\widetilde{\mathfrak{B}}_{3,m} \mbox{ has no triangle.}
\end{equation}

Note that the universes of $\mathfrak{A}_{3,m}$ and $\mathfrak{B}_{3,m}$ are square lattices that have the width of $\gamma_{m-1}^*$ and the height of $3$. For each row of a lattice, there is a linear order that is a segment of the original linear order of the universe of the structure.  
There are three such  linear orders in a structure or an abstraction, corresponding to three distinct rows. We can use three intervals to describe these orders. For $0\leq i\leq 2$, we use $[(0,i),(\gamma_{m-j}^*,i)]$ to describe the linear orders in the $j$-th abstraction. 

Since we can view the structures from different scales, which correspond to different abstractions, it is important to know how the linear orders in different abstractions are related to each other. The following lemma says that
an object is ahead of another one in the $(i-1)$-abstraction if it is ahead of that object in the $i$-th abstraction.  Recall that the following lemmas are about the structures $\mathfrak{A}_{3,m}$ and $\mathfrak{B}_{3,m}$, not about $\widetilde{\mathfrak{A}}_{3,m}$ and $\widetilde{\mathfrak{B}}_{3,m}$ unless explicitly stated. 
\begin{lemma}\label{HighOrder-is-LowOrder}
For  $2\leq i\leq m$ and any $x_1,x_2\in \mathbb{X}_1^*$, we have 
 $$[x_1]_i<[x_2]_i \Rightarrow [x_1]_{i-1}<[x_2]_{i-1}.$$ 
\end{lemma}

The intuition is simple. 
Imagine that we have a graph which is drawn on a grid. Making the cells of the  gird smaller by introducing more columns, which evenly divides a cell into a bunch of smaller cells, will not change the order of two points in the graph. 
Assume that $[x_1]_i<[x_2]_i$. Each row of the $(i-1)$-th abstraction can be divided evenly into $\gamma_{m-i}^*$ intervals of length $\beta_{m-i}^{m-i+1}$, and each interval has a single vertex of index $i$, which is roughly in the middle of this interval.  
We can regard each interval as a bucket. These buckets have a natural linear order induced from the original one. 
Then it is clear that a vertex in a ``bucket'' $\mathrm{BU}_1$ is ahead of a vertex in another ``bucket'' $\mathrm{BU}_2$ in the original linear order if $\mathrm{BU}_1$ is ahead of $\mathrm{BU}_2$ in the induced linear order. 

It implies that  
 an object is ahead of another one in some lower abstraction if it is ahead of that object in any higher abstraction.  So a direct corollary of Lemma \ref{HighOrder-is-LowOrder} is that  
\begin{equation}\label{corollary-HighOrder-is-LowOrder}
[x_1]_i<[x_2]_i \Rightarrow x_1<x_2.
\end{equation}

Recall that $k=3$. 
\begin{lemma}\label{conquer-boundary-strategy}
For any $(x,y)$,$(x^\prime,y)\in\mathbb{X}_1^*$ and for any $p$ where $1\leq p\leq m$, if 
\begin{enumerate}[(1)]
\item %\begin{equation}\label{premise-boundary-strategy-1}
$x-\llparenthesis x\rrparenthesis_p=x^\prime-\llparenthesis x^\prime\rrparenthesis_p$, 
%\end{equation}

\item %\begin{equation}\label{premise-boundary-strategy-2}
$[\llparenthesis x\rrparenthesis_p]_q\equiv [\llparenthesis x^\prime\rrparenthesis_p]_q \,\,(\mbox{mod }k-1), \mbox{ where }q\leq min\{\mathrm{idx}(\llparenthesis x\rrparenthesis_p,y),\mathrm{idx}(\llparenthesis x^\prime\rrparenthesis_p,y)\}$, 
%\end{equation}

\end{enumerate}
then, for any $i$ where $1\leq i\leq q$ and any $j$ where $1\leq j\leq p$, 
\begin{equation}\label{eqn-boundary-strategy-1} 
[\llparenthesis x\rrparenthesis_j]_i\equiv [\llparenthesis x^\prime\rrparenthesis_j]_i \hspace{3pt}(\mbox{mod }k-1).
\end{equation} 
\end{lemma}

In particular, when $i=j$, \eqref{eqn-boundary-strategy-1} is equivalent to the following: 
$$[x]_i\equiv [x^\prime]_i \hspace{3pt}(\mbox{mod }k-1). 
$$

Recall that the leftmost vertex of the $i$-th row of $\widetilde{\mathfrak{B}}_{3,m}$ is $(\gamma_{m-1}^*-tr(i),i)$. Cf. Footnote \ref{footnote-boundary}. 
Let $\widetilde{\mathfrak{A}}_{3,m}^+$ and $\widetilde{\mathfrak{B}}_{3,m}^+$ be built from $\widetilde{\mathfrak{A}}_{3,m}$ and $\widetilde{\mathfrak{B}}_{3,m}$ respectively  by adding a set of constants 
\begin{equation}\label{def-boundary-constants}
\left\{(a,b) \mid a=\gamma_{m-1}^*-tr(b)\mbox{ or } a=\gamma_{m-1}^*-tr(b)-1; b\in[0,2]\right\}.
\end{equation} 
We can take it that the constants are interpreted as extra immovable ``pebbles'' on the boundaries of rows of the structures. Call them \textit{boundary constants}. 
It is easy to see that 
\begin{equation}\label{specialcase-equiv-plus-2-equiv}
\widetilde{\mathfrak{A}}_{3,m}^+\equiv_m^{2} \widetilde{\mathfrak{B}}_{3,m}^+ \mbox{ implies } \widetilde{\mathfrak{A}}_{3,m}\equiv_m^{2} \widetilde{\mathfrak{B}}_{3,m}.
\end{equation} 

In the following we prove the main result of this section. 
\begin{lemma}\label{winning-strategy-in-k=3}
For any $m\geq 3$,
\[
\widetilde{\mathfrak{A}}_{3,m}^+\equiv_m^{2} \widetilde{\mathfrak{B}}_{3,m}^+.
\]
\end{lemma}
In each round, replying Spoiler's pick by a vertex of the same row is a basic element in the strategy of  Duplicator. Assume that Spoiler picks $(x,y)$ in some structure and Duplicator responds with $(x^\prime,y)$ in the other structure. 
Say that \textit{the game (and the board) is over the $i$-th abstraction},  
if $u-\llparenthesis u\rrparenthesis_i=u^\prime-\llparenthesis u^\prime\rrparenthesis_i$ for any pair of pebbled vertices $(u,v)\Vdash (u^\prime,v)$, and the projections of pebbled vertices in the $i$-th abstraction define a partial isomorphism. 

\begin{claim}\label{special-board-abstraction-downward-propagae}
For any $i\in [2,m]$, if the game board is over the $i$-th abstraction, then it is also over the $(i-1)$-th abstraction.
\end{claim}
\noindent\underline{\textit{Proof of Claim:}}\\[6pt]
\indent The argument is simple. First, by Lemma \ref{abstraction-strategy-premier}, we have that $u-\llparenthesis u\rrparenthesis_{i-1}=u^\prime-\llparenthesis u^\prime\rrparenthesis_{i-1}$ for any pair of pebbled vertices $(u,v)\Vdash (u^\prime,v)$. Second, by Remark \ref{ExplanationOfAbstraction-specalcase}, we have that the projections of pebbled vertices in the $(i-1)$-th abstraction define a partial isomorphism w.r.t edges. Moreover, it is also easy to see that partial isomorphism w.r.t order can also be preserved.  

\noindent\underline{\textit{Q.E.D. of Claim.}}\\[-3pt]

Duplicator's strategy works over  abstractions. 
That is, in each round Duplicator plays a related game over some specific abstraction and uses it to decide her pick in the original game.  
We use $\xi$ to remind Duplicator in which abstraction she should play in the current round of the game over abstractions. \textit{More precisely}, $\xi$ is the maximum $i$ such that the game board is over the $i$-th abstraction. 
At the beginning, $\xi=m$, i.e., in Duplicator's mind, 
the players are playing in the highest abstraction in the first round of the related game. We use $\theta$ to denote how many rounds are left at the start of the current round. At the beginning, $\theta=m$. After each round, $\theta$ decreases by one automatically and the game, both the original and the one above abstractions, moves to the next round. 
In each round, $\xi$ remains unchanged if Duplicator can respond properly such that the game board is still over the $\xi$-th abstraction. 
However, if Duplicator cannot do so, she tries to seek a solution in the closest lower abstraction, which will be explained in page \pageref{def-resort-to-closest-abs}.  

Occasionally, we say that ``Duplicator picks an object''. By this we mean that she picks a vertex, and this vertex is in this object, by default in the $\xi$-th abstraction (cf. p. \pageref{def-special-object}, Remark \ref{ExplanationOfAbstraction-specalcase}).  

To prevent from voilating partial isomorphism due to linear orders, in each round Duplicator should  ensure the following requirements in the first place, when she makes her picks.   
Assume that the current round is the $\ell$-th round. Although the game board is $(\widetilde{\mathfrak{A}}_{3,m}^+,\widetilde{\mathfrak{B}}_{3,m}^+)$, the following is stated w.r.t. $(\mathfrak{A}_{3,m},\mathfrak{B}_{3,m})$. Recall that, the circular shifts are introduced only to tackle 4$^\diamond$ (cf. p. \pageref{page-4-diamond}) and so called ``boundary checkout strategy'' of Spoiler (cf. p. \pageref{page-boundary-checkout-strategy}). 

\begin{enumerate}[(1)]
      \item If $[x]_\xi< m-\ell$ or  
            $\gamma_{m-\xi}^*-[x]_\xi< m-\ell$, then\\   
            \indent $\hspace{20pt}[x^\prime]_\xi=[x]_\xi$; 

      \item If  $m-\ell\leq [x]_\xi\leq \gamma_{m-\xi}^*-m+\ell$,  
            then\\ 
            \indent $\hspace{20pt} m-\ell\leq  [x^\prime]_\xi\leq  
            \gamma_{m-\xi}^*-m+\ell$. 

    \end{enumerate}

Call the above the requirement of linear order over abstractions, or \textit{abstraction-order-condition} (for $k=3$). 

We shall show that Duplicator can ensure this requirement in the first place, meanwhile the game board is in partial isomorphism, after each round. 

So far, we defaultly assume that Spoiler picks a vertex in the $\xi$-th abstraction in the current round of the original game, wherein two types of games, i.e. the original game and the corresponding game over abstractions, are coincided.   
However, suppose that Spoiler tries to break this assumption by picking a vertex in the $i$-th abstraction where $i<\xi$, in the original game. In such case, Duplicator regards it as if $(\llparenthesis x\rrparenthesis_\xi, y)$ were picked, and responds with $(x^\prime,y)$ such that $(\llparenthesis x^{\prime}\rrparenthesis_\xi, y)$ is the vertex she will pick to respond  the picking of $(\llparenthesis x\rrparenthesis_\xi, y)$ using her strategy that works in the $\xi$-th abstraction (or, when she cannot do it, responds  $(\llparenthesis x\rrparenthesis_{\xi-1}, y)$ using her strategy that works in the $(\xi-1)$-th abstraction, wherein a solution is ensured); meanwhile, Duplicator ensures that 
\begin{equation}\label{eqn-hr-copycat}
x^{\prime}-\llparenthesis x^{\prime}\rrparenthesis_\xi=x-\llparenthesis x\rrparenthesis_\xi \;(\mbox{or } x^{\prime}-\llparenthesis x^{\prime}\rrparenthesis_{\xi-1}=x-\llparenthesis x\rrparenthesis_{\xi-1}, \mbox{ in the other case)}.
\end{equation} 
 Duplicator is a copycat in the sense of \eqref{eqn-hr-copycat}, which is called \textit{horizontal-residue-copycat} (\textbf{hr-copycat}, in short).\footnote{Observation that  $x^{\prime}-\llparenthesis x^{\prime}\rrparenthesis_{\xi}=x-\llparenthesis x\rrparenthesis_{\xi}$ if and only if  $x^{\prime}-\llparenthesis x^{\prime}\rrparenthesis_\xi\equiv x-\llparenthesis x\rrparenthesis_\xi$ (mod $\beta_{m-\xi}^{m-1}$). As a consequence, we give this name. Similar thing can be find in Remark \ref{remark-ommit-mod}.}
\label{def-hr-copycat}
In other words, Duplicator resorts to the game over abstractions to determine her pick in the original game. 
In this way, Duplicator can reduce the game to a game over some specific abstraction in each round, i.e. either the $\xi$-th abstraction or the $(\xi-1)$-th abstraction. We shall see that Duplicator can win the games over abstractions in each round, provided that $\theta<\xi$.
And we shall see that $\theta<\xi$ is preserved throughout the game after the first round, which ensures that Duplicator is able to resort to the closest lower abstraction for a solution when necessary.   

In the following we explain Duplicator's strategy in more detail, using a simultaneous induction as follows, and show that it is a winning strategy. Note that, whenever we say that ``(Duplicator) wins this round'', we mean that she not only wins in the original pebble games, but also wins in the corresponding pebble games over the $\xi$-th abstraction.   
 
\begin{proof}
This proof is by induction, wherein we show that the follows are preserved after each round. (recall that we always assume that Spoiler picks $(x,y)$ and Duplicator responds with $(x^\prime,y)$ in the current round)
%; and we also require that $\mathrm{idx}(\llparenthesis x^\prime\rrparenthesis_\xi,y)\leq \mathrm{idx}(\llparenthesis x\rrparenthesis_\xi,y)$) 
\begin{enumerate}[1$^\diamond$]\label{special-winning-condtion-set}
\item $x-\llparenthesis x\rrparenthesis_\xi=x^\prime-\llparenthesis x^\prime\rrparenthesis_\xi$. 
\item The abstraction-order-condition holds. 

\item The board, without considering the (projections of) boundary constants (cf. \eqref{def-boundary-constants}), is in partial isomorphism over the $\xi$-th abstraction w.r.t.  edges. 

\item 3$^\diamond$ holds even if the boundaries of rows of the $\xi$-th abstraction are occupied with extra immovable pebbles.\footnote{They are not counted in the $k-1$ pairs of pebbles.}

\item $\theta<\xi$ after the first round.

\item The game board is in partial isomorphism. 
\end{enumerate}

We shall see that 1$^\diamond$\textapprox 4$^\diamond$ implies 6$^\diamond$, according to Remark \ref{ExplanationOfAbstraction-specalcase}.\footnote{In other words, Duplicator has a winning strategy in the orignial game if she has a winning strategy in the game over abstractions, provided that she is a hr-copycat in each round.}
 Moreover, although all of the conditions should be ensured simultaneously, in the game Duplicator will first try to ensure 2$^\diamond$, then 3$^\diamond$, then 4$^\diamond$, and then 1$^\diamond$ and 5$^\diamond$. 

In any round, Duplicator will first try to pick $(x^\prime,y)$ such that 1$^\diamond$\textapprox 4$^\diamond$ hold.\footnote{Because of 1$^\diamond$, Duplicator will first try to make it such that $(x^\prime,y)\in \mathbb{X}_\xi^*$ if $(x,y)\in\mathbb{X}_\xi^*$.} If she cannot find such a vertex, she resorts to the $(\xi-1)$-the abstraction.  
In the following of this section, whenever we say ``\textit{Duplicator resorts to the $(\xi-1)$-th abstraction}'', we mean that Duplicator tries to ensure 1$^\diamond$\textapprox 4$^\diamond$, wherein ``$\xi$'' is replaced by ``$\xi-1$'' in these requirements; 
 and $\xi:=\xi-1$ \textit{at the end of this round}.\label{def-resort-to-closest-abs}\footnote{Such a treatment makes it possible for a discussion involving both ``$\xi$'' and ``$\xi-1$'', without introducing additional symbols. 

Note that $\mathrm{idx}(\llparenthesis x^\prime\rrparenthesis_{\xi-1},y)=\xi-1$: by Lemma \ref{proj-greater-index}, $\mathrm{idx}(\llparenthesis x^\prime\rrparenthesis_{\xi-1},y)\geq\xi-1$; by Lemma \ref{proj-abs}, $\llparenthesis \llparenthesis x^\prime\rrparenthesis_{\xi-1}\rrparenthesis_{\xi}=\llparenthesis x^\prime\rrparenthesis_{\xi}$; then, by definition \eqref{def-eqn-X_i-star}, $\llparenthesis x^\prime\rrparenthesis_{\xi-1}=\llparenthesis x^\prime\rrparenthesis_{\xi}$ if $\mathrm{idx}(\llparenthesis x^\prime\rrparenthesis_{\xi-1},y)>\xi-1$.} Note that in this case Duplicator regards it as if Spoiler picked ``$(\llparenthesis x\rrparenthesis_{\xi-1}, y)$'' in current round and she replies in such a way that the projection of her pick in the $(\xi-1)$-th abstraction is her response over this abstraction.   
%she picks $(x^\prime,y)$ such that $\mathrm{idx}(\llparenthesis x^\prime\rrparenthesis_{\xi-1},y)=\xi-1$. When $x^\prime-\llparenthesis x^\prime\rrparenthesis_{\xi-1}=0$ (i.e. $(x^\prime,y)\in \mathbb{X}_{\xi-1}^*$), it means that $\mathrm{idx}(x^\prime,y)=\xi-1$. 
We shall see that, 4$^\diamond$ can be ensured if  if $[x^\prime]_{\xi-1}\equiv 0$ (mod $2$) when Duplicator resorts to the $(\xi-1)$-th abstraction: in this case $[x]_{\xi-1}\equiv[x^\prime]_{\xi-1}$ (mod 2), which meets \textit{(2)} of Lemma \ref{conquer-boundary-strategy}. 

\textbf{Basis}: In the first round, Duplicator simply mimics. Clearly, she wins this round. $\xi$ is unchanged, whereas $\theta:=\theta-1$. Therefore, $\theta<\xi$ after the first round, i.e. 5$^\diamond$ holds in the following rounds if $\xi$ is to be decreased by \textit{at most} one in each round. Obviously, the abstraction-order-condition holds at the end of the first round, and the other conditions hold.  

\textbf{Induction Step}: 
Suppose that Duplicator can win the first $\ell-1$ rounds where $1<\ell\leq m$, and 1$^\diamond$\textapprox 5$^\diamond$ hold,  
 we prove that she can also win the $\ell$-th round, i.e. 6$^\diamond$ holds, and  1$^\diamond$\textapprox 5$^\diamond$ are also preserved. Recall that we assume that Spoiler picks $(x,y)$ and Duplicator picks $(x^{\prime},y)$ in the $\ell$-th round, i.e. the current round. Moreover, we assume that there is one pair of pebbles on the board at the start of the $\ell$-th round. If there is no such a pair, Duplicator simply mimics Spoiler in this round, as in the first round. 
Assume that $(a,b), (a^{\prime},b)$  are the pair of pebbles on the board at the start of the $\ell$-th round, and that $(a,b)$, $(x,y)$ are in the same structure. By induction hypothesis, $a-\llparenthesis a\rrparenthesis_\xi=a^\prime-\llparenthesis a^\prime\rrparenthesis_\xi$. 

\textit{Assume that} $(x,y)\in \mathbb{X}_\xi^*$, which implies that $x-\llparenthesis x\rrparenthesis_\xi=0$. In such case Duplicator will first \textit{try to} reply with $(x^\prime,y)\in \mathbb{X}_\xi^*$ such that 2$^\diamond$\textapprox 4$^\diamond$ hold. 
% where $\mathrm{idx}(x^\prime,y)\leq \mathrm{idx}(x,y)$, which implies that 1$^\diamond$ holds. 
If she can do it, 1$^\diamond$ is also ensured since $x^\prime-\llparenthesis x^\prime\rrparenthesis_\xi=0$. And so is 5$^\diamond$. 

Suppose that $y=b$. It is clear that $3^\diamond$ holds.  Moreover, we have the following observation.   
\begin{claim}\label{SpecialCase-linear-order}
On condition that 2$^\diamond$ and 5$^\diamond$ hold at the start of the $\ell$-th round, $\mbox{2}^\diamond$ can be preserved after this round, at the price of decreasing $\xi$ by at most one.
\end{claim}

\noindent\underline{\textit{Proof of Claim: }}\\[-9pt]

By her strategy, if Duplicator cannot pick a vertex in the $\xi$-th abstraction that satisfies 2$^\diamond$, then she resorts to the $(\xi-1)$-th abstraction, where she can \textit{always} find a solution. Here is a brief argument. By definition, $\theta+\ell=m+1$. Hence $\xi>\theta= m-\ell+1$. 
Observe that $(\llparenthesis a\rrparenthesis_{\xi-1},b)$ and $(\llparenthesis a^\prime\rrparenthesis_{\xi-1},b)$ satisfy (2) of the requirement: the number of vertices in the $(\xi-1)$-th abstraction, which surround a vertex of index $\xi$, is $\beta_{m-\xi}^{m-\xi+1}=4(\xi-1)>4(m-\ell)$.\footnote{We also take this vertex into the account of the number.} 
Therefore, by Fact \ref{specialcase-fact-surrounding}, there are at least $2(m-\ell)$  vertices of index $\xi-1$ that are on the left side (right side, resp.) of the leftmost (rightmost, resp.) vertex of index $\xi$ in any row of $\mathfrak{B}_{3,m}$.   
That is, any pebbled vertex in $\mathbb{X}_\xi^*$ is away from the leftmost vertex or the rightmost vertex in the $(\xi-1)$-th abstraction, thereby satisfying (2) of the requirement, provided that Duplcator picks a vertex of index $\xi-1$ such that it is  away from both boundaries of a row in the $(\xi-1)$-th abstraction.

\noindent\underline{\textit{Q.E.D. of Claim}. }\\[-6pt]

By induction hypothesis, 2$^\diamond$ also holds for $a$ and $a^\prime$, if $x$ is substituted with $a$ and $x^\prime$ is substituted with $a^\prime$. 
Therefore it is easy for Duplicator to ensure that $[x]_\xi\leq [a]_\xi \mbox{ iff } [x^\prime]_\xi\leq [a^\prime]_\xi$: if $[x]_\xi\leq [a]_\xi$ ($[x]_\xi\geq [a]_\xi$ resp.) and $[a]_\xi<m-\ell$ ($\gamma_{m-\xi}^*-[a]_\xi< m-\ell$, resp.) then Duplicator picks $(x^\prime,y)$ s.t. $[x^\prime]_\xi=[x]_\xi$; otherwise $[a]_\xi$ and $[a^\prime]_\xi$ are far away from the two boundaries of the $b$-th row of the $\xi$-th abstractions. In the former case, $2^\diamond$ clearly holds after the $\ell$-th round.
However, in the latter case, $2^\diamond$ not necessarily holds after the $\ell$-th round. For example, there are more than $m-\ell+1$ vertices of index $\xi$ on the left side of $(\llparenthesis a\rrparenthesis_\xi,b)$, whereas there are exactly $m-\ell$ such vertices on the left side of $(\llparenthesis a^\prime\rrparenthesis_\xi,b)$. If Spoiler picks $(x,y)$ s.t. $[x]_\xi=[a]_\xi-1$, then no matter how Duplicator responds, $2^\diamond$ will never hold in the $\xi$-th abstraction.  
 Nevertheless, in such case, if Duplicator resorts to the $(\xi-1)$-th abstraction, then $2^\diamond$ can be ensured in the $(\xi-1)$-th abstraction after the $\ell$-th round. That is, $[x]_{\xi-1}\leq [a]_{\xi-1} \mbox{ iff } [x^\prime]_{\xi-1}\leq [a^\prime]_{\xi-1}$ and (2) of abstraction-order-condition is met by Claim \ref{SpecialCase-linear-order}. 
If $[x]_{\xi-1}\neq [a]_{\xi-1}$, then by \eqref{corollary-HighOrder-is-LowOrder}, either $x<a$  and $x^\prime<a^\prime$, or $x>a$ and $x^\prime>a^\prime$.  If $[x]_{\xi-1}=[a]_{\xi-1}$, then $[x^\prime]_{\xi-1}=[a^\prime]_{\xi-1}$.   
By induction hypothesis, $a-\llparenthesis a\rrparenthesis_\xi=a^\prime-\llparenthesis a^\prime\rrparenthesis_\xi$, which also implies that $a-\llparenthesis a\rrparenthesis_{\xi-1}=a^\prime-\llparenthesis a^\prime\rrparenthesis_{\xi-1}$, due to Lemma \ref{abstraction-strategy-premier}.  Duplicator can ensure that $x-\llparenthesis x\rrparenthesis_\xi=x^\prime-\llparenthesis x^\prime\rrparenthesis_\xi$ or $x-\llparenthesis x\rrparenthesis_{\xi-1}=x^\prime-\llparenthesis x^\prime\rrparenthesis_{\xi-1}$ depending on whether Duplicator has to resort to the $(\xi-1)$-th abstraction for a solution. That is, 1$^\diamond$ can be ensured.\footnote{If $(x^\prime,y)\in\mathbb{X}_\xi^*$, then  $x-\llparenthesis x\rrparenthesis_\xi=x^\prime-\llparenthesis x^\prime\rrparenthesis_\xi=0$; if $(x^\prime,y)\in\mathbb{X}_{\xi-1}^*$, then $x-\llparenthesis x\rrparenthesis_{\xi-1}=x^\prime-\llparenthesis x^\prime\rrparenthesis_{\xi-1}=0$. To see the latter, we need only show that $x-\llparenthesis x\rrparenthesis_{\xi-1}=0$, which is obvious. Recall that $(x,y)\in\mathbb{X}_\xi^*$. By Lemma \ref{i=0theni-1=0}, $(x,y)\in\mathbb{X}_{\xi^*-1}$. Therefore, $x-\llparenthesis x\rrparenthesis_{\xi-1}=0$.} 
Therefore, 
in both of the cases, $x\leq a \mbox{ iff }x^\prime\leq a^\prime$. 

Moreover, 4$^\diamond$ should be ensured. 
If it cannot, again, Duplicator resorts to the $(\xi-1)$-th abstraction, picking $(x^\prime,y)\in \mathbb{X}_{\xi-1}^*-\mathbb{X}_\xi^*$ such that $[x^\prime]_{\xi-1}\equiv 0$ (mod $2$). By Lemma \ref{proj-abs}, it implies that $[\llparenthesis x^\prime\rrparenthesis_{\xi-1}]_{\xi-1}\equiv 0$ (mod $2$). 
Then by Lemma \ref{cm=depth}, which tells us that $[x]_{\xi-1}\equiv 0$ (mod $2$) and hence $[\llparenthesis x\rrparenthesis_{\xi-1}]_{\xi-1}\equiv 0$ (mod $2$),\footnote{Because $(x,y)\in\mathbb{X}_{\xi-1}$ by Lemma \ref{i=0theni-1=0}, $\llparenthesis x\rrparenthesis_{\xi-1}=x$.} and by Lemma \ref{conquer-boundary-strategy},  $[\llparenthesis x\rrparenthesis_j]_i\equiv [\llparenthesis x^\prime\rrparenthesis_j]_i\equiv 0$ (mod $2$), for $1\leq i\leq \xi-1$ and $1\leq j\leq \xi-1$. Let $n:=\xi-1$.  \label{page-4-diamond}
It means that once $[\llparenthesis x\rrparenthesis_{\xi-1}]_n\equiv [\llparenthesis x^\prime\rrparenthesis_{\xi-1}]_n$ (mod $2$) holds, it still holds for them  in the following rounds, despite of how much $\xi$ and $n$ are decreased. 
It implies that, for  $i<\xi$ and any $(c,d)$ where $([c]_i,d)$ is a boundary vertex (either the leftmost or the rightmost) of the $i$-th abstraction, $(c,d)\in \Omega_x$ iff $(c,d)\in \Omega_{x^\prime}$ if $y=1$ and $d\neq y$.\footnote{To see it, we need analyze two cases. Note that it is easier to do it when $k=3$. But the following arguments works even when $k>3$, wherein the structures are generated in the similar way as we construct  $\widetilde{\mathfrak{B}}_{3,m}^+$. 
If $0<d<k-1$,  $(c,d)\notin \Omega_x$ and $(c,d)\notin \Omega_{x^\prime}$: it is trivial when $d=1$; if $d\neq 1$ (i.e. $k>3$), it is because of \eqref{boundary-nequiv-0-mod-2}. 
Suppose that $d=0$ or $d=k-1$. By definition,   $\Omega_x\cap\mathbb{T}=\Omega_{x^\prime}\cap\mathbb{T}=\emptyset$ where $\mathbb{T}$ is the set of vertices $(x,0)$ or $(x,k-1)$ that has index $\xi-1$, which implies that $(\llparenthesis c\rrparenthesis_{\xi-1},d)\notin \Omega_x$ and $(\llparenthesis c\rrparenthesis_{\xi-1},d)\notin \Omega_{x^\prime}$, by Lemma \ref{boundary-index-over-abstractions}.}  
Moreover, by Lemma \ref{cc-boundary-vertex-is-0}, $\mathbf{cc}([c]_t,d)=0$ mod $2$ for any $t$; and by Lemma \ref{cm=depth}, $\mathbf{cc}([\llparenthesis x\rrparenthesis_{\xi-1}]_p,y)=\mathbf{cc}([\llparenthesis x^\prime\rrparenthesis_{\xi-1}]_p,y)=y$ mod $2$, if $p=\mathrm{idx}(c,d)<\xi-1$.  
Therefore, 4$^\diamond$ holds. Then, by definition, $(c,d)$ is adjacent to $(x,y)$ if and only if $(c,d)$ is adjacent to $(x^\prime,y)$ in the other structure. 

In all the cases, $\theta$ is decreased by one, whereas $\xi$ is at most decreased by one. Hence 5$^\diamond$ is preserved.  
To summarize, when $(x,y)\in \mathbb{X}_\xi^*$, 
Duplicator has a winning strategy in the case  $y=b$, and 1$^\diamond$\textapprox 5$^\diamond$ hold.  

Now assume that $b\neq y$.  By definition,  $(a,b)<(x,y)$ iff $b<y$, and $(a^\prime,b)<(x^\prime,y)$ iff $b<y$. 

\textbf{Firstly}, suppose that Spoiler picks the vertex $(x,y)$ in $\mathbb{X}_\xi^*$. \label{page-special-strategy-firstcase}
Duplicator \textit{first tries to find} all the vertices,  whose index is in the range $[\xi,\mathrm{idx}(x,y)]$, that can ensure the abstraction-order-condition. These vertices are the candidates that Duplicator will possibly pick. 
Then she chooses the subset of vertices of them such that 3$^\diamond$ and 4$^\diamond$ hold (i.e. $(x^\prime,y)$ is in this subset). Note that 3$^\diamond$ says that $(\llparenthesis x^\prime\rrparenthesis_\xi,y)$ is adjacent to $(\llparenthesis a^\prime\rrparenthesis_\xi,b)$ if and only if $(\llparenthesis x\rrparenthesis_\xi,y)$ is adjacent to $(\llparenthesis a\rrparenthesis_\xi,b)$. And 4$^\diamond$ says that this holds even if extra pebbles  are put on the boundaries of rows  of the $\xi$-th abstraction(recall that when we talk about 4$^\diamond$, the structures involved are $\widetilde{\mathfrak{A}}_{3,m}^+$ and $\widetilde{\mathfrak{B}}_{3,m}^+$ instead of $\mathfrak{A}_{3,m}^+$ and $\mathfrak{B}_{3,m}^+$).  %i.e. $(x^\prime,y)$ is adjacent to $(a^\prime,b)$ if and only if $(x,y)$ is adjacent to $(a,b)$. 
If Duplicator can respond in this way, then she obviously wins this round, provided that 1$^\diamond$ holds, and $\xi$ remains unchanged. Hence $\theta<\xi$ is preserved. Note that 1$^\diamond$ holds because $(x,y),(x^\prime,y)\in\mathbb{X}_\xi^*$ which implies that $x^\prime-\llparenthesis x^\prime\rrparenthesis_\xi=x-\llparenthesis x\rrparenthesis_\xi=0$. 

\textbf{Secondly}, \textit{suppose that she cannot do so}.\label{page-special-strategy-secondcase} Then Duplicator resorts to the $(\xi-1)$-th abstraction (considering the projections of pebbled vertices in the $(\xi-1)$-th abstraction), wherein she can \textit{always} do it in this way because $\theta<\xi$. 
%We distinguish two cases. The first case is relatively simpler: Duplicator cannot do it because of violation of 4$^\diamond$. In such case, Duplicator simply let $\xi$ be decreased by one, and not change other choices leading to a valid pick except of violation of 4$^\diamond$. In other words, 
%Duplicator can first focus on other conditions, then consider 4$^\diamond$. 
%The second case is that Duplicator cannot do it because of violation of other conditions. 
In such case, Duplicator picks a vertex of index $\xi-1$, i.e.  
$\mathrm{idx}(x^\prime,y)=\xi-1$. 1$^\diamond$ holds since $x^\prime-\llparenthesis x^\prime\rrparenthesis_{\xi-1}=x-\llparenthesis x\rrparenthesis_{\xi-1}=0$ ($\because (x,y)\in\mathbb{X}_{\xi-1}^*$ by Lemma \ref{i=0theni-1=0}). 
Recall that she will ensure  2$^\diamond$ in the first place, which is easy. 
%It is easy, which is explained in the case where $y=b$.  
It remains to show that Duplicator has a strategy to ensure 3$^\diamond$ in the same time. 
%such that $(x,y)$ is adjacent to $(a,b)$ if and only if $(x^\prime,y)$ is adjacent to $(a^\prime,b)$.    
Since Duplicator cannot respond properly in the $\xi$-th abstraction, it means that one of the following two cases holds (note that $x=\llparenthesis x\rrparenthesis_\xi$ if $(x,y)\in \mathbb{X}_\xi^*$), if $(x^\prime,y)\in\mathbb{X}_\xi^*$: 
\begin{itemize}
\item $(\llparenthesis x\rrparenthesis_\xi,y)$ is adjacent to $(\llparenthesis a\rrparenthesis_\xi,b)$, whereas $(\llparenthesis x^\prime\rrparenthesis_\xi,y)$ is not adjacent to $(\llparenthesis a^\prime\rrparenthesis_\xi,b)$;
 
\item $(\llparenthesis x\rrparenthesis_\xi,y)$ is not adjacent to $(\llparenthesis a\rrparenthesis_\xi,b)$, whereas $(\llparenthesis x^\prime\rrparenthesis_\xi,y)$ is adjacent to $(\llparenthesis a^\prime\rrparenthesis_\xi,b)$. 

\end{itemize}
It happens when Duplicator has to pick such a vertex due to (1) of the abstraction-order-condition (recall that she will first try to ensure 2$^\diamond$ before ensuring 3$^\diamond$). In particular, Spoiler can choose the leftmost (or rightmost) vertex  of a row in the $\xi$-th abstraction. %That is, now we consider the case where  $[x]_\xi=0$ (or $[x]_\xi=\gamma_{m-\xi}^*-1$). 
For example, Spoiler can simply pick the leftmost vertex of a row in a structure, and Duplicator \textit{has to} pick the leftmost vertex of the same row in the other structure. We call this strategy of Spoiler ``\textit{boundary checkout strategy}''. \label{page-boundary-checkout-strategy} 
Clearly such strategy will not violate 4$^\diamond$ since we have shown that, in the worst case, Duplicator can resort to the $(\xi-1)$-th abstraction such that 4$^\diamond$ can be ensured (cf. p.  \pageref{page-4-diamond}).   
%\footnote{If $y=1$ or $b=1$,  $y$ mod $2=\mathbf{cc}([x]_1,y)=\mathbf{cc}(x,y)\neq \mathbf{cc}(\llparenthesis a\rrparenthesis_\xi,b)=\mathbf{cc}([\llparenthesis a\rrparenthesis_\xi]_1,b)=b$ mod $2$. Otherwise, $\mathbf{cc}(x,y)=\mathbf{cc}(\llparenthesis a\rrparenthesis_\xi,b)=\mathbf{cc}(x^\prime,y)=\mathbf{cc}(\llparenthesis a^\prime\rrparenthesis_\xi,b)=0$ mod $2$.} 
Moreover, in the following we show that 2$^\diamond$ and, in particular, 3$^\diamond$ can also be ensured if Duplicator resorts to the $(\xi-1)$-th abstraction.  
 
Assume that $([x]_\xi,y)$ is \textit{not} the leftmost vertex or the rightmost vertex of the $y$-th row of the $\xi$-th abstraction.   
Recall that Duplicator has to resort to the $(\xi-1)$-th abstraction and that  $\frac{1}{2}\beta_{m-\xi}^{m-\xi+1}>2(m-\ell)$. 
In this case, there are so many vertices of index $\xi-1$, which surround $(x,y)$, that $(x,y)$ (i.e. the vertex $(\llparenthesis x\rrparenthesis_{\xi-1},y)$) is away from the two boundaries of the $y$-th row of the $(\xi-1)$-th abstraction. 
That is,  there are at least $m-\ell$ vertices of index $\xi-1$ that are between the vertex $([x]_{\xi-1},y)$ and a boundary vertex of $y$-th row of $(\xi-1)$-th abstraction. It 
implies that the abstraction-order-condition can be met only if the same thing holds for $([x^\prime]_{\xi-1},y)$  (cf. (2) of the abstraction-order-condition, where the structures are $\mathfrak{A}_{3,m}$ and $\mathfrak{B}_{3,m}$ instead of $\widetilde{\mathfrak{A}}_{3,m}$ and $\widetilde{\mathfrak{B}}_{3,m}$), which is easy.

Suppose that $\mathrm{idx}(\llparenthesis a^\prime\rrparenthesis_\xi,b)=t.$ By Lemma \ref{proj-greater-index}, $t\geq \xi$. Recall that Duplicator has to 
 resorts to the $(\xi-1)$-th abstraction, wherein Spoiler uses the boundary checkout strategy. That is, in all the cases she picks $(x^\prime,y)$ such that $\mathrm{idx}(x^\prime,y)=\xi-1$ and $[x^\prime]_{\xi-1}\equiv 0$ (mod $2$). 
By Lemma \ref{proj-abs}, $[\llparenthesis x^\prime\rrparenthesis_{\xi-1}]_{\xi-1}=[x^\prime]_{\xi-1}$. It implies that 4$^\diamond$ holds, as have been explained (cf. p. \pageref{page-4-diamond}). In the following we show that 3$^\diamond$ also can be ensured. Note that 3$^\diamond$ means that  $(\llparenthesis x^\prime\rrparenthesis_{\xi-1},y)$ is adjacent to  $(\llparenthesis a^\prime\rrparenthesis_{\xi-1},b)$ if and only if $(\llparenthesis x\rrparenthesis_{\xi-1},y)$ is adjacent to $(\llparenthesis a\rrparenthesis_{\xi-1},b)$. 
In other words, by Lemma \ref{lattice-point-high-is-lower}, Duplicator need to show that 
\begin{equation}\label{eqn-3-diamond-in-xi-1}
(x^\prime,y) \mbox{ is adjacent to } (\llparenthesis a^\prime\rrparenthesis_{\xi-1},b)\Leftrightarrow (x,y) \mbox{ is adjacent to }(\llparenthesis a\rrparenthesis_{\xi-1},b).
\end{equation}
We explain it case by case. 
\begin{enumerate}[(I)] 
\item Both $y$ mod $2=0$ and $b$ mod $2=0$. 

 By definition,  $\Omega_x=\Omega_{x^\prime}=\Omega_a=\Omega_{a^\prime}=\emptyset$.  Moreover, $[\llparenthesis a^\prime\rrparenthesis_\xi]_{\xi-1}$ $\equiv$ $[\llparenthesis a \rrparenthesis_\xi]_{\xi-1}$ $\equiv$ $[x]_{\xi-1}$ $\equiv 0$ (mod $2$), by Lemma \ref{cm=depth}. By Lemma \ref{conquer-boundary-strategy}, $[\llparenthesis a^\prime\rrparenthesis_{\xi-1}]_{\xi-1}\equiv [\llparenthesis a \rrparenthesis_{\xi-1}]_{\xi-1}$ (mod $2$). 
 
\item $y=1$ and $b$ mod $2=0$.

\begin{itemize}
\item $(x,y)$ is adjacent to $(\llparenthesis a\rrparenthesis_{\xi-1},b)$: 

We shall see shortly that Duplicator can pick $(x^\prime,y)$ such that, for any $i$ where $\xi\leq i\leq m$, 
\begin{equation}\label{eqn-ajacency-primise}
\mathbf{cc}([x^\prime]_i,y)\neq \mathbf{cc}([a^\prime]_i,b) \mbox{ and } \mathrm{idx}(\llparenthesis x^\prime\rrparenthesis_i,y)=i.
\end{equation}

Provided that \eqref{eqn-ajacency-primise} holds, in the following we show that 
\begin{equation}\label{eqn-adjacency-special-xi}
(\llparenthesis x^\prime\rrparenthesis_\xi,y) \mbox{ is adjacent to } (\llparenthesis a^\prime\rrparenthesis_{\xi},b),
\end{equation}
and use this to show that
\begin{equation}\label{eqn-adjacency-special_xi-1}
(x^\prime,y) \mbox{ is adjacent to } (\llparenthesis a^\prime\rrparenthesis_{\xi-1},b). 
\end{equation}

Firstly, note that $\mathrm{idx}(\llparenthesis x^\prime\rrparenthesis_i,y)=i$ for $\xi-1\leq i\leq m$, because $\mathrm{idx}(\llparenthesis x^\prime\rrparenthesis_{\xi-1},y)=\xi-1$ ($\because$ $\mathrm{idx}(x^\prime,y)=\xi-1$, $(x^\prime,y)\in\mathbb{X}_{\xi-1}^*$. It means that $x^\prime=\llparenthesis x^\prime\rrparenthesis_{\xi-1}$); moreover, $\mathrm{idx}(\llparenthesis a^\prime\rrparenthesis_i,b)\geq i=\mathrm{idx}(\llparenthesis x^\prime\rrparenthesis_i,y)$. 

Secondly, by Lemma \ref{proj-abs},  $\mathbf{cc}([x^\prime]_i,y)\neq \mathbf{cc}([a^\prime]_i,b)$ implies that $\mathbf{cc}([\llparenthesis x^\prime\rrparenthesis_i]_i,y)\neq \mathbf{cc}([\llparenthesis a^\prime\rrparenthesis_i]_i,b)$. 
Also by Lemma \ref{proj-abs}, for any $i$ where $\xi\leq i\leq m$, $\llparenthesis \llparenthesis a^\prime\rrparenthesis_{\xi}\rrparenthesis_i=\llparenthesis a^\prime\rrparenthesis_i$. 
Therefore, $$\mathbf{cc}([\llparenthesis x^\prime\rrparenthesis_i]_i,y)\neq \mathbf{cc}([\llparenthesis \llparenthesis a^\prime\rrparenthesis_{\xi-1}\rrparenthesis_i]_i,b),$$ for $\xi\leq i\leq m$. It means that $(\llparenthesis a^\prime\rrparenthesis_{\xi},b)\notin \Omega_{\llparenthesis x^\prime\rrparenthesis_\xi}$. 

As a consequence, \eqref{eqn-adjacency-special-xi} holds since  $\mathbf{cc}([\llparenthesis x^\prime\rrparenthesis_\xi]_\xi,y)\neq \mathbf{cc}([\llparenthesis a^\prime\rrparenthesis_\xi]_\xi,b)$ and $(\llparenthesis a^\prime\rrparenthesis_{\xi},b)\notin \Omega_{\llparenthesis x^\prime\rrparenthesis_\xi}$. This implies that $(\llparenthesis a^\prime\rrparenthesis_{\xi},b)\notin \Omega_{x^\prime}$.

 Recall that $\mathrm{idx}(\llparenthesis a^\prime\rrparenthesis_\xi,b)\geq \xi$. Hence, $\mathbf{cc}([\llparenthesis a^\prime\rrparenthesis_{\xi}]_{\xi-1},b)=0$, by Lemma \ref{cm=depth}. Recall that $[x^\prime]_{\xi-1}\equiv 0$ (mod $2$). As a consequence, $\mathbf{cc}([\llparenthesis a^\prime\rrparenthesis_{\xi}]_{\xi-1},b)\neq\mathbf{cc}([x^\prime]_{\xi-1},y)$. Hence, $(x^\prime,y)$ is adjacent to $(\llparenthesis a^\prime\rrparenthesis_{\xi},b)$.  
  
Recall that $\mathrm{idx}(\llparenthesis a^\prime\rrparenthesis_{\xi-1},b)\geq \xi-1$.    
Suppose that $\mathrm{idx}(\llparenthesis a\rrparenthesis_{\xi-1},b)=\xi-1$. \label{page-special-strategy-arg} 
Then $[\llparenthesis a\rrparenthesis_{\xi-1}]_{\xi-1}\equiv 0$ (mod $2$), for otherwise $(x,y)$ is \textit{not} adjacent to $(\llparenthesis a\rrparenthesis_{\xi-1},b)$.\footnote{Note that $\mathbf{cc}([\llparenthesis a\rrparenthesis_{\xi-1}]_{\xi-1},b)=1$ if $[\llparenthesis a\rrparenthesis_{\xi-1}]_{\xi-1}\equiv 1$ (mod $2$). On the other hand, by Lemma \ref{cm=depth}, $\mathbf{cc}([x]_{\xi-1},y)=1$, because $(x,y)\in \mathbb{X}_\xi^*$. Therefore, $(x,y)$ is \textit{not} adjacent to $(\llparenthesis a\rrparenthesis_{\xi-1},b)$, since $\mathbf{cc}([x]_{\xi-1},y)=\mathbf{cc}([\llparenthesis a\rrparenthesis_{\xi-1}]_{\xi-1},b)$.\label{footnote-special-stra-noadj}} Then, similarly, $[\llparenthesis a^\prime\rrparenthesis_{\xi-1}]_{\xi-1}\equiv 0$ (mod $2$).\footnote{Recall that $a-\llparenthesis a\rrparenthesis_\xi=a^\prime-\llparenthesis a^\prime\rrparenthesis_\xi$. Let $\Delta:=a-\llparenthesis a\rrparenthesis_\xi$. By definition, we have  $[a]_{\xi-1}-[\llparenthesis a\rrparenthesis_\xi]_{\xi-1}=\lfloor \Delta/\beta_{m-\xi+1}^{m-1}\rfloor$ because it is easy to see that $[\llparenthesis a\rrparenthesis_\xi]_{\xi-1}\in\mathbb{N}^+$ (cf. \eqref{llprrp-func}). Similarly,  $[a^\prime]_{\xi-1}-[\llparenthesis a^\prime\rrparenthesis_\xi]_{\xi-1}=\lfloor \Delta/\beta_{m-\xi+1}^{m-1}\rfloor$. 
Hence $[a]_{\xi-1}-[\llparenthesis a\rrparenthesis_\xi]_{\xi-1}=[a^\prime]_{\xi-1}-[\llparenthesis a^\prime\rrparenthesis_\xi]_{\xi-1}$. By Lemma \ref{proj-greater-index}, $(\llparenthesis a\rrparenthesis_\xi,b),(\llparenthesis a^\prime\rrparenthesis_\xi,b)\in\mathbb{X}_\xi^*$. Hence, $[\llparenthesis a\rrparenthesis_\xi]_{\xi-1}\equiv [\llparenthesis a^\prime\rrparenthesis_\xi]_{\xi-1}\equiv 0$ (mod $2$), by Lemma \ref{cm=depth}. Hence, $[a^\prime]_{\xi-1}\equiv [a]_{\xi-1}$ (mod $2$). By Lemma \ref{proj-abs}, $[\llparenthesis a^\prime\rrparenthesis_{\xi-1}]_{\xi-1}\equiv [\llparenthesis a\rrparenthesis_{\xi-1}]_{\xi-1}$ (mod $2$).\label{footnote-isom-for-copycat}}
Recall that $\mathrm{idx}(x^\prime,y)=\mathrm{idx}(\llparenthesis a^\prime\rrparenthesis_{\xi-1},b)=\xi-1$ and $[x^\prime]_{\xi-1}\equiv 0$ (mod $2$). 
Hence $\mathbf{cc}([\llparenthesis a^\prime\rrparenthesis_{\xi-1}]_{\xi-1},b)\neq \mathbf{cc}([x^\prime]_{\xi-1},y)$ determines the adjacency of $(x^\prime,y)$ and $(\llparenthesis a^\prime\rrparenthesis_{\xi-1},b)$, which implies that \eqref{eqn-adjacency-special_xi-1} holds.
 
Now suppose that $\mathrm{idx}(\llparenthesis a\rrparenthesis_{\xi-1},b)>\xi-1$. 
%It means that $\llparenthesis a\rrparenthesis_{\xi-1}=\llparenthesis a\rrparenthesis_{\xi}$, due to \eqref{def-eqn-X_i-star} and \textit{(2)} of Lemma \ref{proj-abs}. 
Because $a-\llparenthesis a\rrparenthesis_\xi=a^\prime-\llparenthesis a^\prime\rrparenthesis_\xi$, by Lemma \ref{abstraction-strategy-premier},  $a^\prime-\llparenthesis a^\prime\rrparenthesis_{\xi-1}=a-\llparenthesis a\rrparenthesis_{\xi-1}\neq 0$. Hence $\mathrm{idx}(\llparenthesis a^\prime\rrparenthesis_{\xi-1},b)>\xi-1$.  Therefore, by Lemma \ref{cm=depth}, $\mathbf{cc}([\llparenthesis a^\prime\rrparenthesis_{\xi-1}]_{\xi-1},b)\\\neq \mathbf{cc}([x^\prime]_{\xi-1},y)$, which implies that \eqref{eqn-adjacency-special_xi-1} holds.

We get the desired result on condition that \eqref{eqn-ajacency-primise} holds. Now we give a process, by which Duplicator does can choose a vertex for $(x^\prime,y)$ to satisfy \eqref{eqn-ajacency-primise}, meanwhile satisfying 1$^\diamond$, 2$^\diamond$ and 4$^\diamond$. Duplicator first chooses a vertex of index $m$, say $(x_m^\prime,y)$, such that $\mathbf{cc}([x_m^\prime]_m,y)\neq \mathbf{cc}([a^\prime]_m,b)$. Afterwards, she chooses a vertex of index $m-1$, say $(x_{m-1}^\prime,y)$, from $cex(x_m^\prime,y,m-1)$ such that $\mathbf{cc}([x_{m-1}^\prime]_{m-1},y)\neq \mathbf{cc}([a^\prime]_{m-1},b)$. Then she chooses a vertex  of index $m-2$, say $(x_{m-2}^\prime,y)$, from $cex(x_{m-1}^\prime,y,m-2)$ such that $\mathbf{cc}([x_{m-2}^\prime]_{m-2},y)\neq \mathbf{cc}([a^\prime]_{m-2},b)$, and so on. Finally, she chooses $(x^\prime,y)\in\mathbb{X}_{\xi-1}^*-\mathbb{X}_{\xi}^*$ from the object $cex(x_{\xi}^\prime,y,\xi-1)$ such that $\mathbf{cc}([x^\prime]_{\xi-1},y)\neq \mathbf{cc}([\llparenthesis a^\prime\rrparenthesis_\xi]_{\xi-1},b)$. Note that this final step implies that $[x^\prime]_{\xi-1}\equiv 0$ (mod $2$).\footnote{To see it, just note that only when $[x^\prime]_{\xi-1}\equiv 0$ (mod $2$) it is possible that $\mathbf{cc}([x^\prime]_{\xi-1},y)$$\neq$$\mathbf{cc}([\llparenthesis a^\prime\rrparenthesis_\xi]_{\xi-1},b)$: $\mathbf{cc}([x^\prime]_{\xi-1},y)=1$ and $\mathbf{cc}([\llparenthesis a^\prime\rrparenthesis_\xi]_{\xi-1},b)=0$ (by Lemma \ref{cm=depth}). It means that $[\llparenthesis x\rrparenthesis]_{\xi-1}\equiv [\llparenthesis x^\prime\rrparenthesis]_{\xi-1}\equiv 0$ (mod $2$), which ensures 4$^\diamond$. Moreover, 1$^\diamond$ holds since $x^\prime-\llparenthesis x^\prime\rrparenthesis_{\xi-1}=x-\llparenthesis x\rrparenthesis_{\xi-1}=0$.} 
In the process Duplicator \textit{can} choose a vertex that is at least one vertex away from the boundaries of the $y$-th row of the $i$-th abstraction where $\xi\leq i\leq m$, i.e. $0<[x_i^\prime]_i<\gamma_{m-i}^*-1$ (for the sake of convenience, suppose that we are talking about $\mathfrak{A}_{3,m}$ and $\mathfrak{B}_{3,m}$ instead of $\widetilde{\mathfrak{A}}_{3,m}$ and $\widetilde{\mathfrak{B}}_{3,m}$)). It implies that the  abstraction-order-condition can be ensured because $\frac{1}{2}\beta_{m-\xi}^{m-\xi+1}>2(m-\ell)$.  By definition (cf. Remark \ref{ExplanationOfAbstraction-specalcase} for the definition of ``$cex(x,y,i)$''),  $[x^\prime]_i=[x_i^\prime]_i$ for $\xi-1\leq i\leq m$. Then by Lemma \ref{projection}, we have  $\llparenthesis x^\prime\rrparenthesis_i=x_i^\prime$. As a consequence, Duplicator can pick $(x^\prime,y)$ to satisfy \eqref{eqn-ajacency-primise}. Moreover, in the process Duplicator has the freedom to pick a vertex that is not a critical point or abstractions of a critical point, i.e. $[x_i^\prime]_i\neq [mid]_i$ (note that the current round is not the first round). 

\item $(x,y)$ is not adjacent to $(\llparenthesis a\rrparenthesis_{\xi-1},b)$: 

By Lemma \ref{proj-greater-index}, $\mathrm{idx}(\llparenthesis a^\prime\rrparenthesis_{\xi-1},b)\geq\xi-1$. 
Assume that 
$$\mathrm{idx}(\llparenthesis a^\prime\rrparenthesis_{\xi-1},b)=\xi-1.$$ Then we have  $\mathrm{cc}([\llparenthesis a\rrparenthesis_{\xi-1}]_{\xi-1},b)=1$, for otherwise $(x,y)$ is  adjacent to $(\llparenthesis a\rrparenthesis_{\xi-1},b)$: by Lemma \ref{cm=depth}, we have  $\mathrm{cc}([x]_{\xi-1},y)=1$; moreover, $(\llparenthesis a\rrparenthesis_{\xi-1},b)\notin \Omega_{x}$ since $\mathrm{idx}(x,y)>\mathrm{idx}(\llparenthesis a\rrparenthesis_{\xi-1},b)$. 
By Remark \ref{ExplanationOfAbstraction-specalcase} (also cf. 
Fact \ref{specialcase-fact-surrounding} and Fact \ref{specialcase-fact-surrounding-upto}), we have  $\mathrm{idx}(\llparenthesis a^\prime\rrparenthesis_{\xi-1},y)=\mathrm{idx}(\llparenthesis a\rrparenthesis_{\xi-1},y)$.   
By Footnote \ref{footnote-isom-for-copycat}, we have $[\llparenthesis a^\prime\rrparenthesis_{\xi-1}]_{\xi-1}\equiv [\llparenthesis a\rrparenthesis_{\xi-1}]_{\xi-1}$ (mod $2$). 
Hence, $\mathrm{cc}([\llparenthesis a^\prime\rrparenthesis_{\xi-1}]_{\xi-1},b)=1$. 
Note that $\mathrm{cc}([x^\prime]_{\xi-1},y)=\mathrm{cc}([x]_{\xi-1},y)=1$. Therefore, $(x^\prime,y)$ is not adjacent to $(\llparenthesis a^\prime\rrparenthesis_{\xi-1},b)$ since $\mathrm{cc}([x^\prime]_{\xi-1},y)=\mathrm{cc}([\llparenthesis a^\prime\rrparenthesis_{\xi-1}]_{\xi-1},b)$. Now suppose that  
$$\mathrm{idx}(\llparenthesis a^\prime\rrparenthesis_{\xi-1},b)=t^\prime>\xi-1.$$  
Duplicator first finds a vertex $(x_{t^\prime}^\prime,y)$ such that its index is $t^\prime$ and $\mathbf{cc}([x_{t^\prime}^\prime]_{t^\prime},y)=\mathbf{cc}([\llparenthesis a^\prime\rrparenthesis_{\xi-1}]_{t^\prime},b)$. The vertex $(x_t^\prime,y)$ can be one that is neither the leftmost nor the rightmost vertex of index $t^\prime$.   
  Then she chooses $(x_{t^\prime-1}^\prime,y)$ from $cex(x_{t^\prime}^\prime,y,t^\prime-1)$ such that its index is $t^\prime-1$ and $[x_{t^\prime-1}^\prime]_{t^\prime-1}\equiv 0$ (mod $2$); and so on, until she chooses $(x_{\xi}^\prime,y)$ from $cex(x_{\xi+1}^\prime,y,\xi)$ such that its index is $\xi$ and  $[x_{\xi}^\prime]_\xi\equiv 0$ (mod $2$). 
Finally, she chooses  $(x^\prime,y)$ from $cex(x_{\xi}^\prime,y,\xi-1)$ such that  $\mathrm{idx}(x^\prime,y)=\xi-1$ and $[x^\prime]_{\xi-1}\equiv 0$ (mod $2$). 
By definition and Lemma \ref{projection}, $(\llparenthesis a^\prime\rrparenthesis_{\xi-1},b)\in \Omega_{x^\prime}$. 
Therefore, $(x^\prime,y)$ is not adjacent to $(\llparenthesis a^\prime\rrparenthesis_{\xi-1},b)$.   

\end{itemize}

\item   $y$ mod $2=0$ and $b=1$.

If $(x,y)$ is not adjacent to $(\llparenthesis a\rrparenthesis_{\xi-1},b)$, Duplicator can use the same processe introduced in (II) to pick a vertex for $(x^\prime,y)$ such that $(x^\prime,y)$ is not adjacent to $(\llparenthesis a^\prime\rrparenthesis_{\xi-1},b)$. In the following we assume that $(x,y)$ is adjacent to $(\llparenthesis a\rrparenthesis_{\xi-1},b)$. 

Note that $\mathbf{cc}([x^\prime]_{\xi-1},y)=0$.  The following arguments are similar to that in (II), cf. page \pageref{page-special-strategy-arg} and related footnotes, i.e. Footnote \ref
{footnote-special-stra-noadj} and Footnote \ref{footnote-isom-for-copycat}. 
Suppose that $\mathrm{idx}(\llparenthesis a\rrparenthesis_{\xi-1},b)=\xi-1$. Then $[\llparenthesis a^\prime\rrparenthesis_{\xi-1}]_{\xi-1}\equiv 0$ (mod $2$). Hence, $\mathbf{cc}([\llparenthesis a^\prime\rrparenthesis_{\xi-1}]_{\xi-1},b)=1\neq \mathbf{cc}([x^\prime]_{\xi-1},y)$. Now suppose that $\mathrm{idx}(\llparenthesis a\rrparenthesis_{\xi-1},b)>\xi-1$. It implies that $\mathrm{idx}(\llparenthesis a^\prime\rrparenthesis_{\xi-1},b)>\xi-1$. Therefore, by Lemma \ref{cm=depth}, $\mathbf{cc}([\llparenthesis a^\prime\rrparenthesis_{\xi-1}]_{\xi-1},b)\neq \mathbf{cc}([x^\prime]_{\xi-1},y)$. 
%By Lemma \ref{cm=depth}, $\mathbf{cc}([\llparenthesis a^\prime\rrparenthesis_\xi]_{\xi-1},b)=1$. 
%By the argument in Footnote \ref{footnote-isom-for-copycat}, we have [$[\llparenthesis a^\prime\rrparenthesis_{\xi-1}]_{\xi-1}\equiv [\llparenthesis a\rrparenthesis_{\xi-1}]_{\xi-1}$ (mod $2$). 
Moreover, $(x^\prime,y)\notin \Omega_{\llparenthesis a^\prime\rrparenthesis_{\xi-1}}$  because, by Lemma \ref{proj-greater-index}, $\mathrm{idx}(\llparenthesis a^\prime\rrparenthesis_{\xi-1},b)\geq \mathrm{idx}(x^\prime,y)=\xi-1$.  
 It implies that $(x^\prime,y)$ is adjacent to $(\llparenthesis a^\prime\rrparenthesis_{\xi-1},b)$ if $(x,y)$ is adjacent to $(\llparenthesis a\rrparenthesis_{\xi-1},b)$.

\end{enumerate} 
By Lemma \ref{abstraction-strategy-premier}, $a-\llparenthesis a\rrparenthesis_{\xi}=a^\prime-\llparenthesis a^\prime\rrparenthesis_{\xi}$ implies that $a-\llparenthesis a\rrparenthesis_{\xi-1}=a^\prime-\llparenthesis a^\prime\rrparenthesis_{\xi-1}$.  
 Note that all the vertices $(x,y)$, $(x^\prime,y)$, $(\llparenthesis a\rrparenthesis_{\xi-1},b)$, and $(\llparenthesis a^\prime\rrparenthesis_{\xi-1},b)$ are in $\mathbb{X}_{\xi-1}^*$. Then by Remark \ref{ExplanationOfAbstraction-specalcase}, we have that $(x^\prime,y)$ is  adjacent to $(a^\prime,b)$  if and only if  $(x,y)$ is adjacent to $(a,b)$. \\[-3pt]

\textbf{Thirdly}, \label{page-third-case-k=3}
if Spoiler picks a vertex $(x,y)$ in $\mathbb{X}_{\xi-1}^*-\mathbb{X}_\xi^*$. The ideas are very similar to the last argument, i.e. Duplicator resorts to the $(\xi-1)$-th abstraction for a solution. Just note that Duplicator need ensure that $\mathrm{idx}(x^\prime,y)=\xi-1$, which means that it is not a critical point, and $\mathbf{cc}([x^\prime]_{\xi-1},y)=\mathbf{cc}([x]_{\xi-1},y)$. 

\textbf{Fourthly}, so far we assume that Spoiler picks a vertex $(x,y)$ in $\mathbb{X}_\xi^*$. 
%\textcolor{red}{If $(x,y)\in\mathbb{X}_{\xi-1}^*$, then $\xi:=\xi-1$ and Duplicator uses the strategy introduce before.} 
Now suppose that he picks a vertex in $\mathbb{X}_1^*-\mathbb{X}_{\xi-1}^*$, i.e. a vertex of index less than $\xi-1$. 
Duplicator has a simple strategy that we have mentioned briefly before. That is, she regards it as if $(\llparenthesis x\rrparenthesis_\xi,y)$, or $(\llparenthesis x\rrparenthesis_{\xi-1},y)$, were picked, and responds with $(x^\prime,y)$ such that\label{k=3-strategy-3}\footnote{In fact, (i) and (ii) include all the cases needed to discuss  because we can take it that the first case (cf. p. \pageref{page-special-strategy-firstcase}, $b\neq y$), i.e. Spoiler picks a a vertex $(x,y)$ in $\mathbb{X}_\xi^*$ and Duplicator can reply properly with $(x^\prime,y)\in\mathbb{X}_\xi^*$, is a special case of (i) where $x^\prime-\llparenthesis x^\prime\rrparenthesis_\xi=x-\llparenthesis x\rrparenthesis_\xi=0$.} 
\begin{enumerate}[(i)]
\item if she can respond properly to the picking of $(\llparenthesis x\rrparenthesis_\xi,y)$ in the $\xi$-th abstraction:\\
 $(\llparenthesis x^\prime\rrparenthesis_\xi,y)$ is the vertex she would pick to respond to the ``picking'' of  $(\llparenthesis x\rrparenthesis_\xi,y)$;  meanwhile, she let $x^\prime-\llparenthesis x^\prime\rrparenthesis_\xi=x-\llparenthesis x\rrparenthesis_\xi$; 

\item otherwise, she resorts to the $(\xi-1)$-th abstraction for a solution:\\
$(\llparenthesis x^\prime\rrparenthesis_{\xi-1},y)$ is the vertex she would pick to respond the ``picking'' of $(\llparenthesis x\rrparenthesis_{\xi-1},y)$, and $x^\prime-\llparenthesis x^\prime\rrparenthesis_{\xi-1}=x-\llparenthesis x\rrparenthesis_{\xi-1}$. 
%$\xi:=\xi-1$; resort to the last strategy. 
\end{enumerate} 

Observe that such strategy implies that $\mathrm{idx}(x,y)=\mathrm{idx}(x^\prime,y)$ in this case (cf. 
Fact \ref{specialcase-fact-surrounding} and Fact \ref{specialcase-fact-surrounding-upto}, or Remark \ref{ExplanationOfAbstraction-specalcase}). 
In the case (i), Duplicator's strategy can help her win this round: as have been explained in Remark \ref{ExplanationOfAbstraction-specalcase}, it means that $(x,y)$ is adjacent to $(a,b)$ if and only if  $(x^\prime,y)$ is adjacent to $(a^\prime,b)$  because  $(\llparenthesis x\rrparenthesis_{\xi},y)$ is adjacent to $(\llparenthesis a\rrparenthesis_{\xi},b)$ if and only if  $(\llparenthesis x^\prime\rrparenthesis_{\xi},y)$ is adjacent to $(\llparenthesis a^\prime\rrparenthesis_{\xi},b)$, by her strategy.

Similarly, recall that Duplicator has a winning strategy over the $(\xi-1)$-th abstraction, if we regard it as if $(\llparenthesis x\rrparenthesis_{\xi-1},y)$ instead of $(x,y)$ were picked in this round (cf. p. \pageref{page-special-strategy-secondcase}).  That is, the projection of all the pebbled vertices to the $(\xi-1)$-th abstraction, including $(\llparenthesis x\rrparenthesis_{\xi-1},y)$ and $(\llparenthesis x^\prime\rrparenthesis_{\xi-1},y)$, will form a new game board that is in partial isomorphism. Recall that $x-\llparenthesis x\rrparenthesis_{\xi-1}=x^\prime-\llparenthesis x^\prime\rrparenthesis_{\xi-1}$.  According to Remark \ref{ExplanationOfAbstraction-specalcase},  $(x,y)$ is adjacent to $(u,v)$ if and only if  $(x^\prime,y)$ is adjacent to $(u^\prime,v)$.

All in all, her strategy  is a winning strategy in such 2-pebble games. 
\end{proof}

From Duplicator's strategy we can see that, at the \textit{end} of the current round,   
$\xi$ is the maximum number in $[1,m]$ that makes 1$^\diamond$ hold. 

By \eqref{specialcase-equiv-plus-2-equiv}, we have  
$$\widetilde{\mathfrak{A}}_{3,m}\equiv_m^{2} \widetilde{\mathfrak{B}}_{3,m}.$$
As a corollary, it is easy to see that it needs and only needs $3$ variables to define $3$-clique in $\fo$ on finite ordered graphs (cf. the proof of Theorem \ref{main-theorem} in Section \ref{winning-strategy} for the details).\\[2pt] 

\textit{Could we make the structures smaller}, such that we can have a feeling of how the structures look like? The answer is yes, at the price of complicating the arguments a little bit. With further thinking, we can remove all the vertices of index $1$ from $\mathfrak{A}_{3,m}$ and $\mathfrak{B}_{3,m}$ and the result still holds. It is because Duplicator can ensure that $\xi=m$ after the first \textit{two} rounds. 

Essentially there are only two changes in the structures. Firstly, there are $m-1$ abstractions in a structure. Secondly, there are $4(m-1)$ vertices in the $m$-th abstraction. 

For any $m, i\in \mathbf{N}^+$, where $m\geq 3$ and $0< i<m-1$, let 
\begin{align*}
\gamma_0^* &:=4(m-1)\\
\gamma_i^* &:=4(m-i-1)\gamma_{i-1}^*
\end{align*}

For $x\in[\gamma_{m-2}^*]$ and $2\leq i\leq j\leq m$, let 
\begin{align*}
\beta_{m-j}^{m-i} &:=\frac{\gamma_{m-i}^*}{\gamma_{m-j}^*}\\
[x]_i &:=\lfloor x/\beta_{m-i}^{m-2}\rfloor\\
\llparenthesis x\rrparenthesis_i &:=[x]_i\beta_{m-i}^{m-2}+\frac{1}{2}\sum_{2<\ell\leq i}\beta_{m-\ell}^{m-2}
\end{align*}  

Now $\beta_{m-j}^{m-i}=\displaystyle\prod_{m-j\leq \ell <m-i} \frac{\gamma_{\ell+1}^*}{\gamma_{\ell}^*}=4^{j-i}\times\frac{(j-2)!}{(i-2)!}$. And  $\gamma_{m-2}^*=\gamma_0^*\beta_0^{m-2}=(m-1)!\times 4^{m-1}$. 

Recall that Duplcator simply mimics Spoiler in the first round. In the following rounds, Duplicator continues mimicking until both of $cex(mid,0,m-1)$ and $cex(mid,2,m-1)$ have pebbled vertices in \textit{one}  structure and it is Spoiler who picks one of them \textit{in this round}. Recall that we always assume that Spoiler picks $(x,y)$ in this round (i.e. current round). 
Suppose w.l.o.g. that $y=2$. In this ``icebreaking'' round, if $(x,y)$ is a vertex of $\mathfrak{A}_{3,m}$, Duplicator need only ensure that 
\begin{itemize}
\item $\mathbf{cc}([x^\prime]_m,y)\neq \mathbf{cc}([mid]_m,0)=0$; 
\item  $x^\prime-\llparenthesis x^\prime\rrparenthesis_m=x-\llparenthesis x\rrparenthesis_m$;
\item $[x^\prime]_m\neq 0$ and $[x^\prime]_m\neq \gamma_0^*-1$ (that is, $(\llparenthesis x^\prime\rrparenthesis_m,y)$ is away from the boundaries of the $y$-th row of the $m$-th abstraction).  

\end{itemize}

Note that there are more than one vertex that Duplicator can choose to satisfy these conditions.

If $(x,y)$ is a vertex of $\mathfrak{B}_{3,m}$, Duplicator need only ensure that
\begin{itemize}
\item $\mathbf{cc}([x^\prime]_m,y)=\mathbf{cc}([mid]_m,0)=0$; 
\item $\llparenthesis x^\prime\rrparenthesis_m\neq mid$;
\item  $x^\prime-\llparenthesis x^\prime\rrparenthesis_m=x-\llparenthesis x\rrparenthesis_m$;
\item $[x^\prime]_m\neq 0$ and $[x^\prime]_m\neq \gamma_0^*-1$. 

\end{itemize}  
Recall that we right circular shift the middle row such that 4$^\diamond$ holds. In the previous analysis, we assume that immovable ``pebbles'' are put on the boundaries of rows at the start of games. Now we take away such assumption and study directly the game $\Game_m^{2}(\widetilde{\mathfrak{A}}_{3,m},\widetilde{\mathfrak{B}}_{3,m})$ instead of the game 
$\Game_m^{2}(\widetilde{\mathfrak{A}}_{3,m}^+,\widetilde{\mathfrak{B}}_{3,m}^+)$. 
 As a consequence, $\xi=m$ after the first two rounds.\footnote{Here is an \underline{\textit{EXAMPLE}} for the first two rounds of the games. In the first round, Spoiler picks $(mid,0)$; Duplicator replies with $(mid,0)$ in the other structure. In the second round, Spoiler picks $(mid,2)$ in $\widetilde{\mathfrak{A}}_{k,m}$; Duplicator responds by picking $(x,2)$ in $\widetilde{\mathfrak{B}}_{3,m}$ where $(x,2)\in\mathbb{X}_m^*$, $\mathbf{cc}([x]_m,2)=1$ and $[x]_m\neq 0,\gamma_0^*-1$. Therefore, after these two rounds, $\xi$ is still $m$. Although 4$^\diamond$ no more holds if there are additional immovable pebbles on the boundaries of rows, Spoiler can ``show'' it only when he picks $(c,1)$ where $c$ is the projection of the leftmost vertex in the $m$-th abstraction, because there are no such additional pebbles on the game board $(\widetilde{\mathfrak{A}}_{3,m},\widetilde{\mathfrak{B}}_{3,m})$.  
This will use one more round. And Duplicator can resort to the $(m-1)$-th abstraction for a solution. 
Note that $(x,2)$, as well as $(mid,2)$, is adjacent to $(c^\prime,1)$ where $c^\prime$ is the projection of the leftmost vertex in the $(m-1)$-th abstraction.}  
Therefore, the game (over abstractions) will never go into the first abstraction.  
Such treatment saves one more round for Duplicator, i.e. she can win $i+1$ rounds in  $\Game_m^{2}(\widetilde{\mathfrak{A}}_{3,m},\widetilde{\mathfrak{B}}_{3,m})$ if she can win $i$ rounds in $\Game_m^{2}(\widetilde{\mathfrak{A}}_{3,m}^+,\widetilde{\mathfrak{B}}_{3,m}^+)$.

Moreover, the time Duplicator stops mimicking is the time when those two pebbles are put in different rows: one is on a vertex of the bottom row and the other is on a vertex of the top row in a structure. 
Hence we can make the length of a row a bit smaller. 
It is especially useful when we try to draw a picture for the structures. 
It is for these two reasons that we can change the previous definitions a little bit, while almost all the arguments remain the same, except that 
\begin{itemize}
\item we substitute $\beta_{m-i}^{m-1}$, $\mathbb{X}_1^*$, $\gamma_{m-1}^*$ with $\beta_{m-i}^{m-2}$, $\mathbb{X}_2^*$, $\gamma_{m-2}^*$;
\item we need to take the second round into account when proving the induction basis, as have just been introduced;
\item the abstraction-order-condition is adapted as follows:
    \begin{enumerate}
      \item If $[x]_\xi< m-1-\ell$ or  
            $\gamma_{m-\xi}^*-[x]_\xi< m-1-\ell$, then\\   
            \indent $\hspace{20pt}[x^\prime]_\xi=[x]_\xi$; 

      \item If  $m-1-\ell\leq [x]_\xi\leq \gamma_{m-\xi}^*-m+1+\ell$,  
            then\\ 
            \indent $\hspace{20pt} m-1-\ell\leq  [x^\prime]_\xi\leq  
            \gamma_{m-\xi}^*-m+1+\ell$. 

    \end{enumerate}

\item $mid:=2(m-1)\beta_0^{m-2}+\frac{1}{2}\sum_{2<j\leq m}\beta_{m-j}^{m-2}$.
\end{itemize}
 Note that, now  $\xi>\theta+1= m-\ell+2$ for $2<\ell\leq m$ since $\xi=m$ after the first two rounds, therefore Claim \ref{SpecialCase-linear-order} still holds.

%To understand the notion ``abstraction'', probably the case where $k=3$ and $m=4$ is more interesting than the simplest case. Nevertheless, the structures in this case is a little big. 

However, the structures are still a bit big even for the simplest cases.  Fig. \ref{fig:A_3_3-and-3rd-abstraction}  shows $\mathfrak{A}_{3,3}$ on the left side and the third abstraction  $\mathfrak{A}_{3,3}[\mathbb{X}_3^*]$ on the right side, both of which are rotated $90^{\circ}$ counterclockwise. Note that the vertex with a label ``$\mathrm{mid}0$'' is just the vertex $(mid,0)$, and the vertex with a label ``$\mathrm{mid}1$'' is the vertex $(mid,2)$. 

The black nodes in the graph represent the vertices of the third abstracton of $\mathfrak{A}_{3,3}$, whereas the grey nodes represent those vertices in the second abstraction. 
We group the verties of each row of the structure $\mathfrak{A}_{3,3}$ by blue dashed rectangles. The vertices in the same rectangle have the same ``position'' in the third abstraction. That is, for any $(u, v)$ and $(u^\prime,v)$, $[u]_3=[u^\prime]_3$ if they are in the same rectangle. To simplify the picture, we only show the blue dashed rectangles in the highest row of  $\mathfrak{A}_{3,3}$ and ommit all the others.

%Therefore we only draw the third and the fourth abstractions of $\mathfrak{A}_{3,3}$,  which is shown in Fig. \ref{}. The graph on the left side of the figure is the third abstraction $\mathfrak{A}_{3,m}[\mathbb{X}_3^*]$, rotated $90^\ccirc$ counterclockwise, and the graph on the right side is the fourth abstraction  $\mathfrak{A}_{3,m}[\mathbb{X}_4^*]$.

 By a simple counting, we know that there are $8$ triangles in $\mathfrak{A}_{3,3}$.\footnote{There are four vertices, whose index is $3$, whose coordinate congruence number is $0$ in the third abstraction, and whose second coordinate is $1$: that is, the vertices $(2,1)$, $(10,1)$ and $(18,1)$ and $(26,1)$. The two vertices with labels ``$\mathrm{mid0}$'', ``$\mathrm{mid1}$'', and any one of these four  vertices can form a triangle. In addition, there are four vertices, whose index is $2$, whose coordinate congruence number is $0$ in the second abstraction, and whose second coordinate is $1$: that is, the vertices $(0,1)$, $(8,1)$, $(16,1)$ and $(24,1)$.} If we remove the red edge of $\mathfrak{A}_{3,3}$, we obtain the structure $\mathfrak{B}_{3,3}$, which is triangle-free by Fact \ref{B_3_m-is-trianglefree}. Obviously, for any $m\geq 3$, the girth of $\mathfrak{A}_{3,3}$ is $3$ and the girth of $\mathfrak{B}_{3,3}$ is $4$. See, e.g., the shortest cycle that consists of the vertices $(mid,0)$, $(mid,1)$, $(mid,2)$ and $(2,1)$ in $\mathfrak{B}_{3,3}$.

%+++++++++++++++++++++++++++ begin Fig A_{3,3} +++++++++++++++++++++++++++++++
%\begin{comment}
\begin{figure}
\centering
\vspace*{-20mm}
\begin{tikzpicture}[scale=0.14]

%+++++++++++++++++++ 画最上面的row与最下面的row之间的边+++++++++++++++++++
%最上面row的点index为1

\foreach \i in {0,2,...,10,12,14}{
    \foreach \j in {0,1,...,15} {
           \draw[color=black!50]  plot[smooth, tension=.7] coordinates {(0,2+10*\i) (24,7+10*\j)};
    }      
}

\foreach \i in {0,1,...,15}{
   \foreach \j in {0,1,...,15} {
           \draw[color=black!50]  plot[smooth, tension=.7] coordinates {(0,7+10*\i) (24,2+10*\j)};
     }      
}

%最上面row的点index为2， 最下面的点index为1
\foreach \i in {1,3,...,15}{
    \foreach \j in {0,1,...,15} {
           \draw[color=black!50]  plot[smooth, tension=.7] coordinates {(0,2+10*\i) (24,7+10*\j)};
    }      
}

%========画中间row的所有点与其他row的边 ==================================================
%第二抽象层：中间row cc{}=1，下面row cc{}=0； 

\foreach \i in {2,42,82,122 }{
% the middle row, "index=1",cc(x,y)=1  on the one part; the botom row arround a vertex of index 2 whose cc{}=0 in the second abstraction
%第二抽象层：中间row cc{}=1，下面row cc{}=0； 
 \foreach \j in {2,42,82,122} { 
       \draw[color=black!50]  plot[smooth, tension=.7] coordinates {(12,\i) (24,\j)};     
       \draw[color=black!50]  plot[smooth, tension=.7] coordinates {(12,\i) (24,\j+10)};            
      }
 }
 
 \foreach \i in {12,52,92,132}{
% the middle row, "index=2",cc(x,y)=1  on the one part; the botom row arround a vertex of index 2 whose cc{}=0 in the second abstraction
%第二抽象层：中间row cc{}=1，下面row cc{}=0；  
 \foreach \j in {12,52,92,132}{ 
       \draw[color=black!50]  plot[smooth, tension=.7] coordinates {(12,\i) (24,\j)};     
       \draw[color=black!50]  plot[smooth, tension=.7] coordinates {(12,\i) (24,\j-10)};            
      }
 }
 
\foreach \i in {7,47,87,127 }{
% the middle row, index=1,cc(x,y)=0 on the one part;  the botom row arround a vertex of index 2 whose cc{}=0 in the second abstraction
%第二抽象层：中间row cc{}=1，下面row cc{}=0；  
     \foreach \j in {7,47,87,127} {
       \draw[color=black!50]  plot[smooth, tension=.7] coordinates {(12,\i) (24,\j)};     
       \draw[color=black!50]  plot[smooth, tension=.7] coordinates {(12,\i) (24,\j+10)};            
      }
 } 
 
 \foreach \i in {17,57,97,137 }{
% the middle row, index=1,cc(x,y)=0 on the one part;  the botom row arround a vertex of index 2 whose cc{}=0 in the second abstraction
%第二抽象层：中间row cc{}=1，下面row cc{}=0；  
     \foreach \j in  {17,57,97,137 }{
       \draw[color=black!50]  plot[smooth, tension=.7] coordinates {(12,\i) (24,\j)};     
       \draw[color=black!50]  plot[smooth, tension=.7] coordinates {(12,\i) (24,\j-10)};            
      }
 } 
 
%===========第二抽象层：中间row cc{}=0，下面row cc{}=1； 
\foreach \i in {22,62,102,142}{
% the middle row, "index=1",cc(x,y)=1  on the one part; the botom row arround a vertex of index 2 whose cc{}=1 in the second abstraction
%第二抽象层：中间row cc{}=0，下面row cc{}=1； 
 \foreach \j in {22,62,102,142}{ 
       \draw[color=black!50]  plot[smooth, tension=.7] coordinates {(12,\i) (24,\j)};     
       \draw[color=black!50]  plot[smooth, tension=.7] coordinates {(12,\i) (24,\j+10)};            
      }
 }
 
 \foreach \i in {32,72,112,152}{
% the middle row, "index=2",cc(x,y)=1  on the one part; the botom row arround a vertex of index 2 whose cc{}=1 in the second abstraction
%第二抽象层：中间row cc{}=0，下面row cc{}=1；  
 \foreach \j in  {32,72,112,152}{ 
       \draw[color=black!50]  plot[smooth, tension=.7] coordinates {(12,\i) (24,\j)};     
       \draw[color=black!50]  plot[smooth, tension=.7] coordinates {(12,\i) (24,\j-10)};            
      }
 }
 
\foreach \i in {27,67,107,147 }{
% the middle row, index=1,cc(x,y)=0 on the one part;  the botom row arround a vertex of index 2 whose cc{}=1 in the second abstraction
%第二抽象层：中间row cc{}=0，下面row cc{}=1；  
     \foreach \j in {27,67,107,147 } {
       \draw[color=black!50]  plot[smooth, tension=.7] coordinates {(12,\i) (24,\j)};     
       \draw[color=black!50]  plot[smooth, tension=.7] coordinates {(12,\i) (24,\j+10)};            
      }
 } 
 
 \foreach \i in {37,77,117,157 }{
% the middle row, index=1,cc(x,y)=0 on the one part;  the botom row arround a vertex of index 2 whose cc{}=1 in the second abstraction
%第二抽象层：中间row cc{}=0，下面row cc{}=1；  
     \foreach \j in  {37,77,117,157}{
       \draw[color=black!50]  plot[smooth, tension=.7] coordinates {(12,\i) (24,\j)};     
       \draw[color=black!50]  plot[smooth, tension=.7] coordinates {(12,\i) (24,\j-10)};            
      }
 } 
 
%===========
%++++++++++++++++++++++++++++++
%========
%第二抽象层：上面row cc{}=1，中间row cc{}=0； 

\foreach \i in {2,42,82,122 }{
% the middle row, "index=1",cc(x,y)=1  on the one part; the botom row arround a vertex of index 2 whose cc{}=0 in the second abstraction
%第二抽象层：上面row cc{}=1，中间row cc{}=0；
 \foreach \j in {2,42,82,122} { 
       \draw[color=black!50]  plot[smooth, tension=.7] coordinates {(0,\i) (12,\j)};     
       \draw[color=black!50]  plot[smooth, tension=.7] coordinates {(0,\i) (12,\j+10)};            
      }
 }
 
 \foreach \i in {12,52,92,132}{
% the middle row, "index=2",cc(x,y)=1  on the one part; the botom row arround a vertex of index 2 whose cc{}=0 in the second abstraction
%第二抽象层：上面row cc{}=1，中间row cc{}=0；  
 \foreach \j in {12,52,92,132}{ 
       \draw[color=black!50]  plot[smooth, tension=.7] coordinates {(0,\i) (12,\j)};     
       \draw[color=black!50]  plot[smooth, tension=.7] coordinates {(0,\i) (12,\j-10)};            
      }
 }
 
\foreach \i in {7,47,87,127 }{
% the middle row, index=1,cc(x,y)=0 on the one part;  the botom row arround a vertex of index 2 whose cc{}=0 in the second abstraction
%第二抽象层：上面row cc{}=1，中间row cc{}=0；
     \foreach \j in {7,47,87,127} {
       \draw[color=black!50]  plot[smooth, tension=.7] coordinates {(0,\i) (12,\j)};     
       \draw[color=black!50]  plot[smooth, tension=.7] coordinates {(0,\i) (12,\j+10)};            
      }
 } 
 
 \foreach \i in {17,57,97,137 }{
% the middle row, index=1,cc(x,y)=0 on the one part;  the botom row arround a vertex of index 2 whose cc{}=0 in the second abstraction
%第二抽象层：上面row cc{}=1，中间row cc{}=0；
     \foreach \j in  {17,57,97,137 }{
       \draw[color=black!50]  plot[smooth, tension=.7] coordinates {(0,\i) (12,\j)};     
       \draw[color=black!50]  plot[smooth, tension=.7] coordinates {(0,\i) (12,\j-10)};            
      }
 } 
 
%===========第二抽象层：中间row cc{}=0，下面row cc{}=1； 
\foreach \i in {22,62,102,142}{
% the middle row, "index=1",cc(x,y)=1  on the one part; the botom row arround a vertex of index 2 whose cc{}=1 in the second abstraction
%第二抽象层：上面row cc{}=0，中间row cc{}=1； 
 \foreach \j in {22,62,102,142 }{ 
       \draw[color=black!50]  plot[smooth, tension=.7] coordinates {(0,\i) (12,\j)};     
       \draw[color=black!50]  plot[smooth, tension=.7] coordinates {(0,\i) (12,\j+10)};            
      }
 }
 
 \foreach \i in {32,72,112,152}{
% the middle row, "index=2",cc(x,y)=1  on the one part; the botom row arround a vertex of index 2 whose cc{}=1 in the second abstraction
%第二抽象层：上面row cc{}=0，中间row cc{}=1；  
 \foreach \j in  {32,72,112,152}{ 
       \draw[color=black!50]  plot[smooth, tension=.7] coordinates {(0,\i) (12,\j)};     
       \draw[color=black!50]  plot[smooth, tension=.7] coordinates {(0,\i) (12,\j-10)};            
      }
 }
 
\foreach \i in {27,67,107,147 }{
% the middle row, index=1,cc(x,y)=0 on the one part;  the botom row arround a vertex of index 2 whose cc{}=1 in the second abstraction
%第二抽象层：上面row cc{}=0，中间row cc{}=1；  
     \foreach \j in {27,67,107,147 } {
       \draw[color=black!50]  plot[smooth, tension=.7] coordinates {(0,\i) (12,\j)};     
       \draw[color=black!50]  plot[smooth, tension=.7] coordinates {(0,\i) (12,\j+10)};            
      }
 } 
 
 \foreach \i in {37,77,117,157 }{
% the middle row, index=1,cc(x,y)=0 on the one part;  the botom row arround a vertex of index 2 whose cc{}=1 in the second abstraction
%第二抽象层：上面row cc{}=0，中间row cc{}=1； 
     \foreach \j in  {37,77,117,157 }{
       \draw[color=black!50]  plot[smooth, tension=.7] coordinates {(0,\i) (12,\j)};     
       \draw[color=black!50]  plot[smooth, tension=.7] coordinates {(0,\i) (12,\j-10)};            
      }
 }

%===========

%++++++++++++++++++++++++++++++

%   \foreach \y in {2,7,...,117} {
 
 %             \draw[color=black!50] (0,\y) -- (24,\y);       
  % }

% color=red， 第0层和第二层之间的粗直线
   \foreach \y in {12,32,...,152} {
 
              \draw[very thick] (0,\y) -- (24,\y);       
   }

%漏斗曲线
  \foreach \n in {-2,-1,0,1,2,3,4}{
  
            \draw[color=black,very thick]  plot[smooth, tension=.7] coordinates {(0,112-20*\n) (19,100-20*\n) (24,92-20*\n)};
            \draw[color=black,very thick]  plot[smooth, tension=.7] coordinates {(0,92-20*\n) (5,100-20*\n) (24,112-20*\n)};
}

%画第二抽象层中的点之间的边：一个点与（x）相距为2个单位的点之间
\foreach \i in {0,1,2,3,4,5}{
       \draw[color=black,very thick]    plot[smooth, tension=.7] coordinates {(12,52+20*\i) (24,12+20*\i)};
}

\foreach \i in {0,1,2,3,4,5}{
       \draw[color=black,very thick]    plot[smooth, tension=.7] coordinates {(0,52+20*\i) (12,12+20*\i)};
}

%\draw [color=black,very thick]  plot[smooth, tension=.7] coordinates {(24,92) (12,32)};
%\draw [color=black,very thick]   plot[smooth, tension=.7] coordinates {(24,112) (12,52)};
%\draw [color=black,very thick]  plot[smooth, tension=.7] coordinates {(12,92) (0,32)};
%\draw [color=black,very thick]  plot[smooth, tension=.7] coordinates {(12,112) (0,52)};

\draw[color=black,very thick]    plot[smooth, tension=.7] coordinates {(24,52) (12,12)  };
\draw[color=black,very thick]    plot[smooth, tension=.7] coordinates {(24,72) (12,32)  };
\draw[color=black,very thick]    plot[smooth, tension=.7] coordinates {(24,92) (12,52)  };
\draw[color=black,very thick]    plot[smooth, tension=.7] coordinates {(24,112) (12,72)  };
\draw[color=black,very thick]    plot[smooth, tension=.7] coordinates {(24,132) (12,92)  };

\draw[color=black,very thick]    plot[smooth, tension=.7] coordinates {(12,52) (0,12)  };
\draw[color=black,very thick]    plot[smooth, tension=.7] coordinates {(12,72) (0,32)  };
\draw[color=black,very thick]    plot[smooth, tension=.7] coordinates {(12,92) (0,52)  };
\draw[color=black,very thick]    plot[smooth, tension=.7] coordinates {(12,112) (0,72)  };
\draw[color=black,very thick]    plot[smooth, tension=.7] coordinates {(12,132) (0,92)  };

\draw[color=yellow,very thick]   plot[smooth, tension=.7] coordinates {(12,132) (24,52)};
\draw[color=yellow,very thick]   plot[smooth, tension=.7] coordinates {(12,132) (0,52)};
\draw[color=yellow,very thick]   plot[smooth, tension=.7] coordinates {(52,132) (64,52)};
\draw[color=yellow,very thick]   plot[smooth, tension=.7] coordinates {(52,132) (40,52)};
\draw [color=yellow,very thick] plot[smooth, tension=.7] coordinates {(0,132) (12,52)};
\draw[color=yellow,very thick]  plot[smooth, tension=.7] coordinates {(24,132) (12,52)};
\draw [color=yellow,very thick] plot[smooth, tension=.7] coordinates {(40,132) (52,52)};
\draw[color=yellow,very thick]  plot[smooth, tension=.7] coordinates {(64,132) (52,52)};

\draw[color=yellow,very thick]   plot[smooth, tension=.7] coordinates {(12,72) (24,152)};
\draw[color=yellow,very thick]   plot[smooth, tension=.7] coordinates {(12,72) (0,152)};
\draw[color=yellow,very thick]   plot[smooth, tension=.7] coordinates {(12,152) (24,72)};
\draw[color=yellow,very thick]   plot[smooth, tension=.7] coordinates {(12,152) (0,72)};

%+++++++画中间分割竖直线，第二抽象层++++++++++++++++++++++++++++++++++++++++++++++++++++++++

\draw[blue, thick,dashed] (31.5,12) -- (31.5,132);

% color=red， 第1row和第3row之间的粗直线
   \foreach \y in {12,32,...,132,152} {
 
              \draw[very thick] (40,\y) -- (64,\y);       
   }

%漏斗曲线，距离1个单位
  \foreach \n in {-2,-1,0,1,2,3,4}{
  
            \draw[color=black,very thick]  plot[smooth, tension=.7] coordinates {(40,112-20*\n) (59,100-20*\n) (64,92-20*\n)};
            \draw[color=black,very thick]  plot[smooth, tension=.7] coordinates {(40,92-20*\n) (45,100-20*\n) (64,112-20*\n)};
}
%漏斗曲线，距离3个单位
  \foreach \n in {0,1,2,3,4}{
          \draw[color=green,very thick]  plot[smooth, tension=.7] coordinates {(64,72+20*\n) (57,56+20*\n) (52,49+20*\n) (49,44+20*\n) (47,38+20*\n) (40,12+20*\n)};
          \draw[color=green,very thick]   plot[smooth, tension=.7] coordinates {(40,72+20*\n) (47,56+20*\n) (52,49+20*\n) (55,44+20*\n) (57,38+20*\n) (64,12+20*\n)};

 }
%\draw[color=green,very thick]  plot[smooth, tension=.7] coordinates {(24,112) (12,59) (5,41) (0,12)};
%\draw[color=green,very thick]  plot[smooth, tension=.7] coordinates {(0,112) (12,59) (19,41) (24,12)};

%画第二抽象层中的点之间的边：一个点与（x）相距为2个单位的点之间
\foreach \i in {0,1,2,3,4,5}{
       \draw[color=black,very thick]    plot[smooth, tension=.7] coordinates {(52,52+20*\i) (64,12+20*\i)};
}

\foreach \i in {0,1,2,3,4,5}{
       \draw[color=black,very thick]    plot[smooth, tension=.7] coordinates {(40,52+20*\i) (52,12+20*\i)};
       
}
%\draw [color=brown,very thick]  plot[smooth, tension=.7] coordinates {(64,92) (52,32)};
%\draw [color=brown,very thick]   plot[smooth, tension=.7] coordinates {(64,112) (52,52)};
%\draw [color=brown,very thick]  plot[smooth, tension=.7] coordinates {(52,92) (40,32)};
%\draw [color=brown,very thick]  plot[smooth, tension=.7] coordinates {(52,112) (40,52)};

\draw[color=black,very thick]    plot[smooth, tension=.7] coordinates {(64,52) (52,12)  };
\draw[color=black,very thick]    plot[smooth, tension=.7] coordinates {(64,72) (52,32)  };
\draw[color=black,very thick]    plot[smooth, tension=.7] coordinates {(64,92) (52,52)  };
\draw[color=black,very thick]    plot[smooth, tension=.7] coordinates {(64,112) (52,72)  };
\draw[color=black,very thick]    plot[smooth, tension=.7] coordinates {(64,132) (52,92)  };

\draw[color=black,very thick]    plot[smooth, tension=.7] coordinates {(52,52) (40,12)  };
\draw[color=black,very thick]    plot[smooth, tension=.7] coordinates {(52,72) (40,32)  };
\draw[color=black,very thick]    plot[smooth, tension=.7] coordinates {(52,92) (40,52)  };
\draw[color=black,very thick]    plot[smooth, tension=.7] coordinates {(52,112) (40,72)  };
\draw[color=black,very thick]    plot[smooth, tension=.7] coordinates {(52,132) (40,92)  };

%画第二抽象层中的点之间的边：一个点与（x）相距为4个单位的点之间
\draw[color=yellow,very thick]   plot[smooth, tension=.7] coordinates {(52,92) (64,12)};
\draw[color=yellow,very thick]   plot[smooth, tension=.7] coordinates {(52,112) (64,32)};

\draw[color=yellow,very thick]   plot[smooth, tension=.7] coordinates {(40,92) (52,12)};
\draw[color=yellow,very thick]   plot[smooth, tension=.7] coordinates {(40,112) (52,32)};

\draw[color=yellow,very thick]    plot[smooth, tension=.7] coordinates {(64,92) (52,12)};
\draw[color=yellow,very thick]   plot[smooth, tension=.7] coordinates {(64,112) (52,32)};
\draw[color=yellow,very thick]  plot[smooth, tension=.7] coordinates {(52,92) (40,12)};
\draw[color=yellow,very thick]   plot[smooth, tension=.7] coordinates {(52,112) (40,32)};

\draw[color=yellow,very thick]   plot[smooth, tension=.7] coordinates {(52,72) (64,152)};
\draw[color=yellow,very thick]   plot[smooth, tension=.7] coordinates {(52,72) (40,152)};
\draw[color=yellow,very thick]   plot[smooth, tension=.7] coordinates {(52,152) (64,72)};
\draw[color=yellow,very thick]   plot[smooth, tension=.7] coordinates {(52,152) (40,72)};

%第一row和第3row之间相距5个单位的点间
\foreach \i in {0,1,2}{
\draw [blue, very thick]  plot[smooth, tension=.7] coordinates {(40,112+\i*20) (46,77+\i*20) (52,60+\i*20) (58,45+\i*20) (64,12+\i*20)};
\draw [blue, very thick]   plot[smooth, tension=.7] coordinates {(64,112+\i*20) (58,77+\i*20) (52,60+\i*20) (46,45+\i*20) (40,12+\i*20)};
}
%+++++++++++第二抽象层+++++++++++++++++++++++++++++++++++++++++++++++++++++++++++++++++++
%第二抽象层
\draw[color=yellow,very thick]   plot[smooth, tension=.7] coordinates {(12,92) (24,12)};
\draw[color=yellow,very thick]   plot[smooth, tension=.7] coordinates {(12,112) (24,32)};

\draw[color=yellow,very thick]   plot[smooth, tension=.7] coordinates {(0,92) (12,12)};
\draw[color=yellow,very thick]   plot[smooth, tension=.7] coordinates {(0,112) (12,32)};

  \foreach \n in {0,1,2,3,4}{
          \draw[color=green,very thick]  plot[smooth, tension=.7] coordinates {(24,72+20*\n) (17,56+20*\n) (12,49+20*\n) (9,44+20*\n) (7,38+20*\n) (0,12+20*\n)};
          \draw[color=green,very thick]   plot[smooth, tension=.7] coordinates {(0,72+20*\n) (7,56+20*\n) (12,49+20*\n) (15,44+20*\n) (17,38+20*\n) (24,12+20*\n)};

 }
%\draw[color=green,very thick]  plot[smooth, tension=.7] coordinates {(24,112) (12,59) (5,41) (0,12)};
%\draw[color=green,very thick]  plot[smooth, tension=.7] coordinates {(0,112) (12,59) (19,41) (24,12)};

\draw[color=yellow,very thick]    plot[smooth, tension=.7] coordinates {(24,92) (12,12)};
\draw[color=yellow,very thick]   plot[smooth, tension=.7] coordinates {(24,112) (12,32)};
\draw[color=yellow,very thick]  plot[smooth, tension=.7] coordinates {(12,92) (0,12)};
\draw[color=yellow,very thick]   plot[smooth, tension=.7] coordinates {(12,112) (0,32)};

%第一row和第3row之间相距5个单位的点间
\foreach \i in {0,1,2} {
\draw [blue, very thick]  plot[smooth, tension=.7] coordinates {(0,112+\i*20) (6,77+\i*20) (12,60+\i*20) (18,45+\i*20) (24,12+\i*20)};
\draw [blue, very thick]   plot[smooth, tension=.7] coordinates {(24,112+\i*20) (18,77+\i*20) (12,60+\i*20) (6,45+\i*20) (0,12+\i*20)};
}

\draw[color=red,very thick]  (39.671,92) node (v1) {} arc (100:80:71);

\foreach \i in {0,1,2,3,4,5,6,7} {
        \draw[color=blue, dashed]  (39,19+20*\i) rectangle (41,0+20*\i);    
}

\draw[color=red,very thick]  (-0.329,92) node (v1) {} arc (100:80:71);
%\draw  plot[smooth, tension=.7] coordinates {(v1) (12,22)}; 
% in the second row, and second column of first abstracton,
%there is a surrounding node, the first node in the second abstraction, which is adjacent to the top critical point.  

\foreach \i in {0,1,2,3,4,5,6,7} {
        \draw[color=blue, dashed]  (-1,19+20*\i) rectangle (1,0+20*\i);    
}

%========draw  node  labels =============

%\tikzstyle{mynodestyle} = [draw,outer sep=0,inner sep=1,minimum size=10]
%\node [bordered] (v1) at (4,9) {node};

\node [draw,outer sep=0,inner sep=1,minimum size=10] at (-5,92) {mid1};
\node [draw,outer sep=0,inner sep=1,minimum size=10] at (28,92) {mid0};
\node [draw,outer sep=0,inner sep=1,minimum size=10] at (35,92) {mid1};
\node [draw,outer sep=0,inner sep=1,minimum size=10] at (68,92) {mid0};

\node [draw,outer sep=0,inner sep=1,minimum size=10] at (27,52) {a};
\node [draw,outer sep=0,inner sep=1,minimum size=10] at (27,132) {b};
\node [draw,outer sep=0,inner sep=1,minimum size=10] at (67,52) {a};
\node [draw,outer sep=0,inner sep=1,minimum size=10] at (67,132) {b};

              \path[draw,fill=red] (0,92) circle (10pt);     
              \path[draw,fill=red] (12,92) circle (10pt);     
              \path[draw,fill=red] (24,92) circle (10pt);     
              \path[draw,fill=red] (40,92) circle (10pt);   
              \path[draw,fill=red] (52,92) circle (10pt);   
              \path[draw,fill=red] (64,92) circle (10pt);

%++++++++++++++++++++++++++++++++++++++++++++++++++++++++++++++++++++++++++              
\draw[color=black,very thick]    plot[smooth, tension=.7] coordinates {(64,152) (52,112)};
\draw[color=black,very thick]    plot[smooth, tension=.7] coordinates {(52,152) (40,112)};

\draw[color=black,very thick]    plot[smooth, tension=.7] coordinates {(24,152) (12,112)};
\draw[color=black,very thick]    plot[smooth, tension=.7] coordinates {(12,152) (0,112)};

\draw[color=brown,very thick]    plot[smooth, tension=.7] coordinates {(40,152) (46,104) (52,66) (58,47) (64,12)};
\draw[color=brown,very thick]    plot[smooth, tension=.7] coordinates {(64,152) (58,104) (52,66) (46,47) (40,12)};

\draw[color=brown,very thick]    plot[smooth, tension=.7] coordinates {(0,152) (6,104) (12,66) (18,47) (24,12)};
\draw[color=brown,very thick]    plot[smooth, tension=.7] coordinates {(24,152) (18,104) (12,66) (6,47) (0,12)};

\draw[color=pink,very thick]   plot[smooth, tension=.7] coordinates {(52,12) (64,132)};
\draw[color=pink,very thick]   plot[smooth, tension=.7] coordinates {(52,12) (40,132)};

\draw[color=pink,very thick]   plot[smooth, tension=.7] coordinates {(52,32) (64,152)};
\draw[color=pink,very thick]   plot[smooth, tension=.7] coordinates {(52,32) (40,152)};

\draw[color=pink,very thick]   plot[smooth, tension=.7] coordinates {(52,152) (64,32)};
\draw[color=pink,very thick]   plot[smooth, tension=.7] coordinates {(52,152) (40,32)};

\draw[color=pink,very thick]   plot[smooth, tension=.7] coordinates {(52,132) (64,12)};
\draw[color=pink,very thick]   plot[smooth, tension=.7] coordinates {(52,132) (40,12)};

%+++++++++

\draw[color=pink,very thick]   plot[smooth, tension=.7] coordinates {(12,12) (24,132)};
\draw[color=pink,very thick]   plot[smooth, tension=.7] coordinates {(12,12) (0,132)};

\draw[color=pink,very thick]   plot[smooth, tension=.7] coordinates {(12,32) (24,152)};
\draw[color=pink,very thick]   plot[smooth, tension=.7] coordinates {(12,32) (0,152)};

\draw[color=pink,very thick]   plot[smooth, tension=.7] coordinates {(12,152) (24,32)};
\draw[color=pink,very thick]   plot[smooth, tension=.7] coordinates {(12,152) (0,32)};

\draw[color=pink,very thick]   plot[smooth, tension=.7] coordinates {(12,132) (24,12)};
\draw[color=pink,very thick]   plot[smooth, tension=.7] coordinates {(12,132) (0,12)};

% ++++++++++ draw nodes +++++++++
\foreach \x in {0,12,24} {
   \foreach \y in {2,7,...,157} {
 
              \path[draw,fill=black!30](\x,\y) circle (8pt);       
   }
}

\foreach \x in {0,12,24} {
   \foreach \y in {12,32,...,152} {
   
              \path[draw,fill=black](\x,\y) circle (8pt);       
   }
}

%画右边的点
\foreach \x in {40,52,64} {
   \foreach \y in {2,7,...,157}{ 
 
              \path[draw,fill=black!30](\x,\y) circle (8pt);       
    }
}

\foreach \x in {40,52,64} {
   \foreach \y in {12,32,...,152}  {
       {

              \path[draw,fill=black](\x,\y) circle (8pt);       

    }
}
}
\end{tikzpicture}
\caption{\label{fig:A_3_3-and-3rd-abstraction} $\mathfrak{A}_{3,3}$ (left side) and its third abstraction $\mathfrak{A}_{3,3}[\mathbb{X}_3^*]$. Removing the red edge from $\mathfrak{A}_{3,3}$, we obtain the structure $\mathfrak{B}_{3,3}$.}
\end{figure}
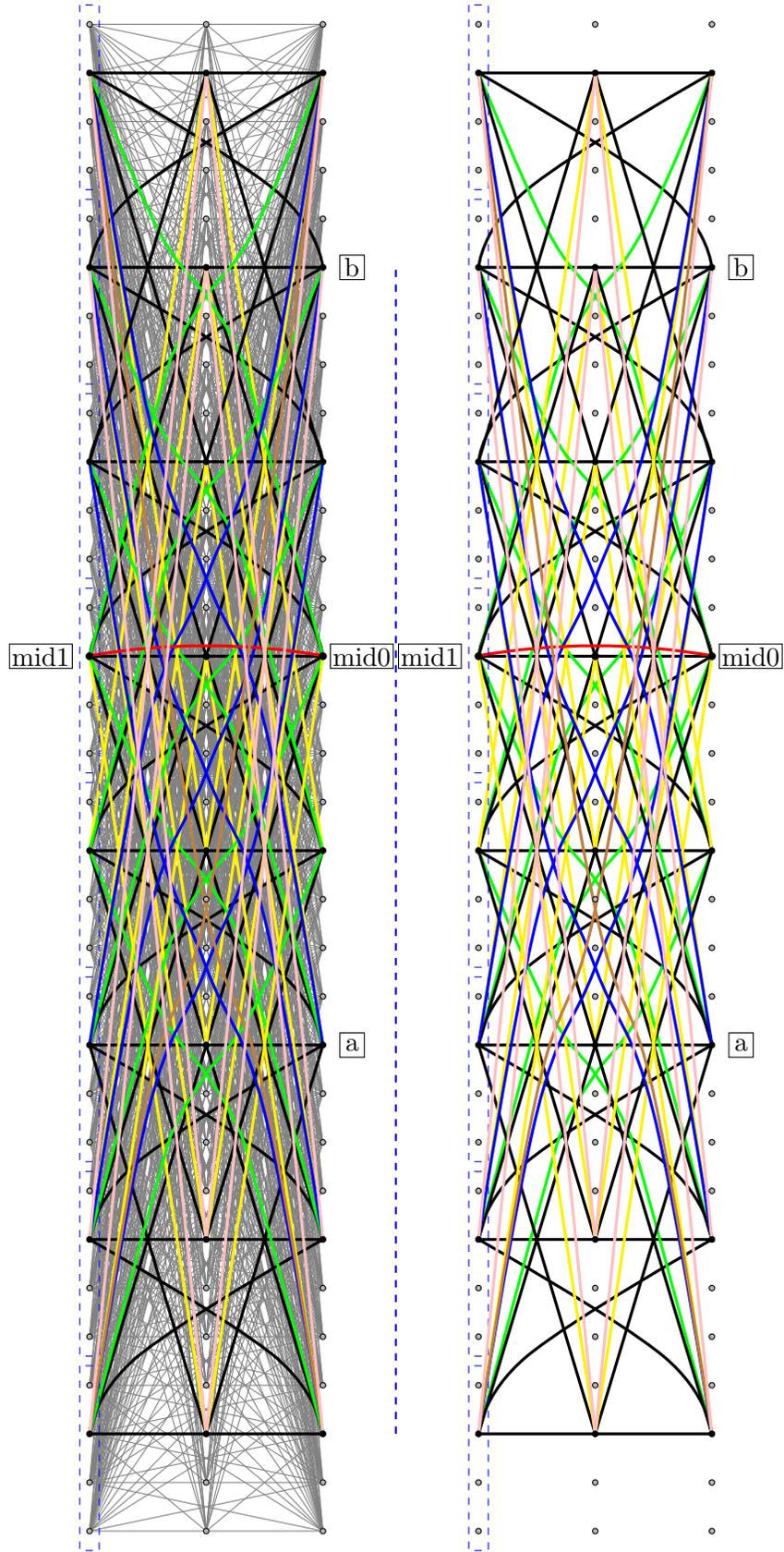

%\end{comment}
%+++++++++++++++++++++++++++ end Fig A_{3,3} ++++++++++++++++++++++++++++++++++

It is easy to see that Duplicator has a winning strategy in $\Game_3^2(\widetilde{\mathfrak{A}}_{3,3},\widetilde{\mathfrak{B}}_{3,3})$. In other words, in this special case we don't have to right circular shift the mid row (the $1$-th row) of the structures. Nevertheless,  Fig. \ref{fig:A_3_3-and-3rd-abstraction} is still not easy to handle directly insomuch as edges crisscross one another in a fashion that deters ``observation'', let alone to play a game over this game board. 
Hence it is better to study a piece of the structures to understand the key concept ``abstraction'', as shown in Fig. \ref{piece-of-structure}.

%\begin{comment}
\begin{figure}[htbp]
\centering
\includegraphics[trim= 25mm 0mm 0mm 0mm, scale=0.7]{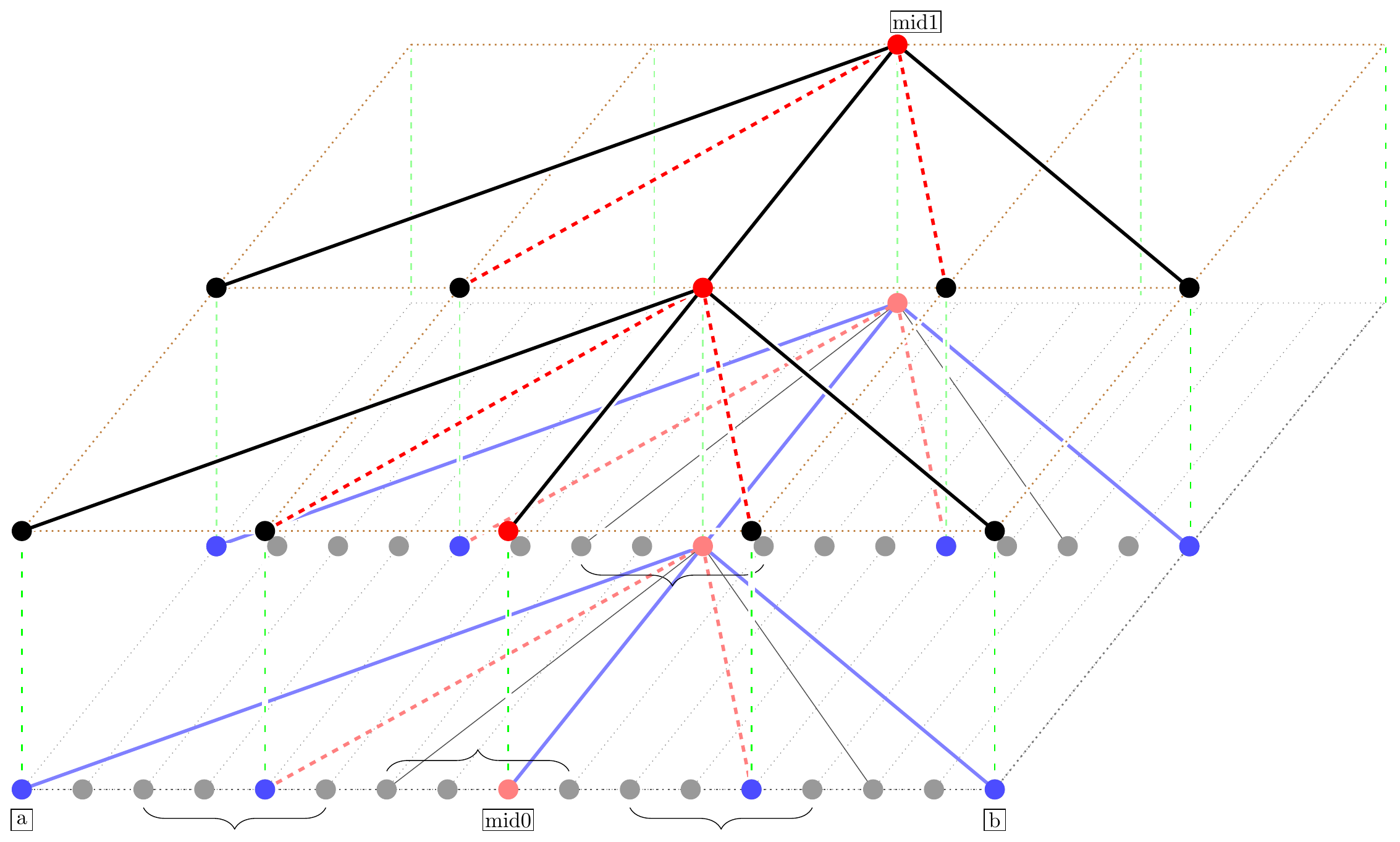}
%\scalebox{10}{}
\caption{A piece of $\mathfrak{A}_{3,3}$. To simplify the figure, in this small piece we only show the edges between $(mid,2)$ (i.e. the vertex with label ``$\mathrm{mid}1$'') and those vertices whose second coordinate is $1$ and the edges between $(mid,1)$ and those vertices whose second coordinate is $0$. For example, the red edge between $(mid,2)$ and $(mid,0)$ is ommitted in this figure. }
\label{piece-of-structure}
\end{figure}
%\end{comment}

In Fig. \ref{piece-of-structure}, we not only give a small piece of $\mathfrak{A}_{3,3}$, but also give its third abstraction and overlap them in a way that delivers a bit  intuition: we can watch the graph from different scales. Compared with the higher abstractions, the lower abstractions show more details of the original structure, thereby in the finer scales. Note that we use braces instead of dashed rectangles in Fig. \ref{piece-of-structure}. The subgraph induced by $cex(14,0,2)$ and $cex(mid,1,2)$ is isomorphic to that induced by $cex(22,0,2)$ and $cex(mid,1,2)$, for both $(14,0)$ and $(22,0)$ are not adjacent to $(mid,1)$. By contrast, The subgraph induced by $cex(14,0,2)$ and $cex(mid,1,2)$ is \textit{not} isomorphic to that induced by $cex(mid,0,2)$ and $cex(mid,1,2)$, because $(mid,0)$ is adjacent to $(mid,1)$ while $(14,0)$ is not. In general, the adjacency of vertices of the $i$-th abstraction determines the adjacency of all the middle row vertices of the $(i-1)$-th abstraction  whose $(i-1)$-th relative first coordinates are even.\footnote{The ``adjacency of a vertex'' tells us how the vertex is connected to the other vertices.}\\[-2pt]   

\noindent\textit{\textbf{Further Remark}}    

We create the structures $\mathfrak{A}_{3,3}$ and $\mathfrak{B}_{3,3}$ to ensure that Duplicator has a winning strategy in a $3$-round $2$-pebble game.%, even when the boundaries of rows are occuppied with additional immovable pebbles. 
We can further simplify the structures. In fact, the structure on the right side of Fig. \ref{fig:A_3_3-and-3rd-abstraction} (the third abstraction) can be taken as the structure $\mathfrak{A}_{3,3}$. What's more, we even don't have to right circular shift the mid row of the structures in this very special case. It is not only because the additional immovable pebbles on the boundaries are not necessary a part of the games, but also based the simple observation that, in the last round, Duplicator need only ensure that the game board is in partial isomorphism, either w.r.t. edges or w.r.t. orders. 
 
It is possible to make the structures even smaller. For example, in the case where $k=m=3$, there exists a pair of structures $\mathfrak{A}_{3,3}$ and $\mathfrak{B}_{3,3}$, each of which only has $12$ vertices. See Fig. \ref{fig:A_3_3-and-B_3_3}. 
 Duplicator's strategy in this special case is similar to the one we introduced before. Duplicator simply mimics Spoiler in the first round. And in the following rounds  Duplicator continues mimicking until both of $(mid,0)$ and $(mid,2)$ are  ``pebbled'' %(a complete expansion of a vertex is pebbled if one of the vertices in the expansion is pebbled) 
in one structure and it is Spoiler who picked one of them in this round. Suppose that in this icebreaking round Spoiler picked a critical point in $\mathfrak{A}_{3,3}$.  
In such case, Duplicator can pick either $(1,0)$ or $(1,2)$ in $\mathfrak{B}_{3,3}$, depending on which critical point is picked by Spoiler. Clearly Duplicator wins this round.  
Hence there is at most one round left, which is easy for Duplicator to handle. For example, assume that  both of $(mid,0)$ and $(mid,2)$ are pebbled in $\mathfrak{A}_{3,3}$ and $(1,2)$ is pebbled in $\mathfrak{B}_{3,3}$. Then, if Spoiler moves the pebble on $(mid,0)$ to $(0,1)$ in $\mathfrak{A}_{3,3}$, Duplicator need only move the pebble on $(mid,0)$ to  $(1,1)$ in $\mathfrak{B}_{3,3}$. 

%\begin{comment}
\begin{figure}%[trim = 0mm 50mm 50mm 0mm, clip]
\centering
\hspace*{-13mm}
\begin{tikzpicture}[xscale=0.20,yscale=0.23]

\draw[blue, thick,dashed] (31.5,32) -- (31.5,92);

\foreach \z in {0,1}{ 

% color=red， 第1row和第3row之间的粗直线
   \foreach \y in {32,52,...,92} {
 
              \draw[very thick] (40-40*\z,\y) -- (64-40*\z,\y);       
   }

%漏斗曲线，距离1个单位
  \foreach \n in {1,2,3}{
  
            \draw[color=black,very thick]  plot[smooth, tension=.7] coordinates {(40-40*\z,112-20*\n) (64-40*\z,92-20*\n)};
            \draw[color=black,very thick]  plot[smooth, tension=.7] coordinates {(40-40*\z,92-20*\n) (64-40*\z,112-20*\n)};
}
%漏斗曲线，距离3个单位
  \foreach \n in {1}{
          \draw[color=green,very thick]  plot[smooth, tension=.7] coordinates {(64-40*\z,72+20*\n)  (40-40*\z,12+20*\n)};
          \draw[color=green,very thick]   plot[smooth, tension=.7] coordinates {(40-40*\z,72+20*\n) (64-40*\z,12+20*\n)};

 }

%画第三抽象层中的点之间的边：一个点与（x）相距为2个单位的点之间
\foreach \i in {1,2}{
       \draw[color=black,very thick]    plot[smooth, tension=.7] coordinates {(52-40*\z,52+20*\i) (64-40*\z,12+20*\i)};
}

\foreach \i in {1,2}{
       \draw[color=black,very thick]    plot[smooth, tension=.7] coordinates {(40-40*\z,52+20*\i) (52-40*\z,12+20*\i)};
       
}

\draw[color=black,very thick]    plot[smooth, tension=.7] coordinates {(64-40*\z,72) (52-40*\z,32)  };
\draw[color=black,very thick]    plot[smooth, tension=.7] coordinates {(64-40*\z,92) (52-40*\z,52)  };

\draw[color=black,very thick]    plot[smooth, tension=.7] coordinates {(52-40*\z,72) (40-40*\z,32)  };
\draw[color=black,very thick]    plot[smooth, tension=.7] coordinates {(52-40*\z,92) (40-40*\z,52)  };

} %\end z-for loop

\draw[color=red,very thick]  (-0.329,71.9214) node (v1) {} arc (100:80:71);

% ++++++++++ draw nodes begin +++++++++

\foreach \x in {0,12,24} {
   \foreach \y in {32,52,...,92} {
   
              \path[draw,fill=black](\x,\y) circle (8pt);       
   }
}

\foreach \x in {40,52,64} {
   \foreach \y in {32,52,...,92}  {
       {

              \path[draw,fill=black](\x,\y) circle (8pt);       

    }
}
}

%========draw  node  labels end =============
\node [draw,outer sep=0,inner sep=1,minimum size=10] at (-3,72) {mid1};
\node [draw,outer sep=0,inner sep=1,minimum size=10] at (27,72) {mid0};
\node [draw,outer sep=0,inner sep=1,minimum size=10] at (37,72) {mid1};
\node [draw,outer sep=0,inner sep=1,minimum size=10] at (67,72) {mid0};

\draw [red,very thick,dotted] (0.5392,32.0027) arc (-100:-80:66);

\end{tikzpicture}
\caption{\label{fig:A_3_3-and-B_3_3} $\mathfrak{A}_{3,3}$ (left side) and $\mathfrak{B}_{3,3}$ (right side). Note that $(x,0)$ is \textit{not} adjacent to $(x,2)$ for any $x$, which is exemplified by a \textit{dotted} red arc in $\mathfrak{A}_{3,3}$. We omit all other such dotted arcs in the figure.  
This game board is simplified to the extent that permits the readers to actually play and check the outcome of the $3$-round $2$-pebble game over it.%, at the price of losing some crucial insights for more general cases.
} 
\end{figure}
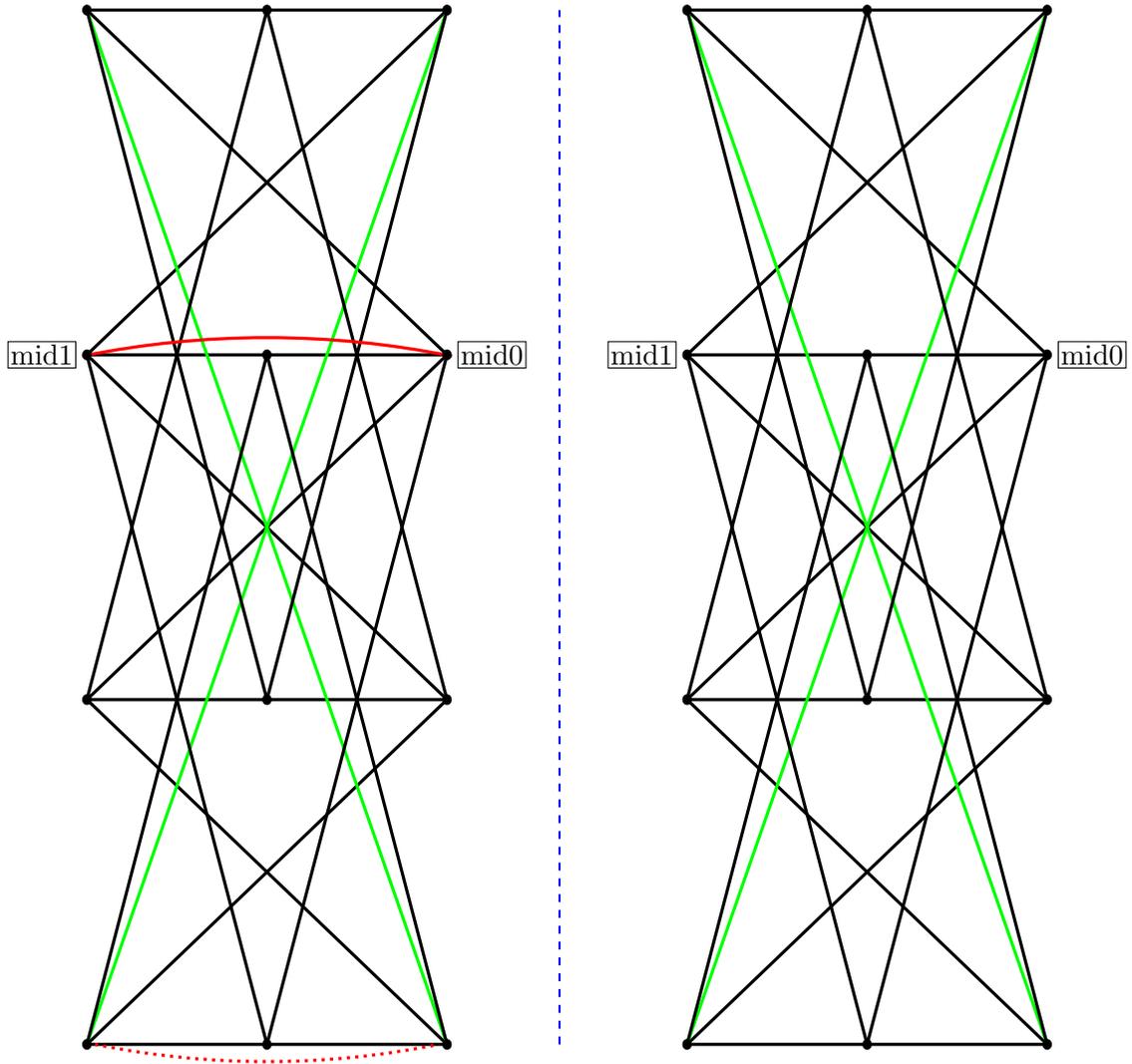

%\end{comment}

%+++++++++++
% picture for B_{3,4} here 
%+++++++++++

Similarly, we can construct $\mathfrak{A}_{3,4}$ and $\mathfrak{B}_{3,4}$, each of which only contains $45$ vertices. However, on the one hand it is still a bit complicated; on the other hand, it is hard to generalize. Therefore we do not put the picture here. 

In conclusion, it is possible to make the structures smaller in some special cases. Nevertheless, we shall not consider it in the following more general cases, to make the arguments as simple as possible.

\subsection{Type labels, board histories, and structures with a structure} \label{section-structures}
Since we have already proved the special cases where $k=2, 3$ in the last section,   
henceforth we assume that $k\geq 4$. Compared with the special cases, including the case where Spoiler is only allowed to pick in $\mathcal{A}_k$ (cf. Section \ref{existential-case-section}), the general cases suddenly become complicated for the following reasons. 
\begin{enumerate}[(1$^\star$)]
\item The two graphs in the game board should be sufficiently similar so that the choosing of structure in each round will not immediately help Spoiler find the difference. %Although the graphs should be similar enough, one of them should have no $k$-clique, while the other should have at least one such clique. 
By contrast, the graphs we constructed for the special case werein Spoiler only picks in $\mathcal{A}_k$ is so different that Spoiler can find the difference immediately if he is allowed to choose $\mathcal{B}_k$.  

\item In the case where $k=3$, there are two pebbles in total, which can only tell us whether two pebbled vertices are adjacent or not. With more pebbles, the adjacency of a vertex to its neighbours that can be detected is greatly complicated.  

\item In the case where Spoiler only picks in $\mathcal{A}_k$, each vertex in $\mathcal{A}_k$ has the same type and the pebbled vertices induce only cliques. That is, the newly picked vertex is adjacent to all the pebbled vertices. In the general case, as we have many different types of vertices, pebbled vertices may induce different subgraphs, possibly not in favor of Duplicator. 
  
\item With more ``types'' of vertices, it is much easier for Spoiler to find the difference simply by picking in a relative ``small'' range (w.r.t. the linear order) of a row persistently, when the graphs are ordered.  
It is particularly for this reason that makes the linear order issue notoriously infavorable for explicit constructions. 
\end{enumerate}

As a consequence, we need to create some new concepts and techniques to deal with the difficulty. For (1$^\star$), we forbbid some edges to ensure that no $k$-clique (cf., e.g., $\mathrm{RngNum}(\cdot,\cdot)$ and $\mathrm{sgn}(\cdot,\cdot)$ functions) exists in one of the structure. For (2$^\star$), we introduce complicated notion of ``type'' labels (cf. Definition \ref{type-label-star}). For (3$^\star$), we can first find a vertex that is adjacent to (the projection of) all the pebbled vertices (in some specific abstraction) (cf. Lemma \ref{no-missing-edges_xi-1}); afterwards, with ``small'' adjustion, we can find the vertex that has the adjacency to the pebbled vertices in the way we want, cf. Strategy 2, (\ref{xi-1-simuluation-1}). For (4$^\star$), we introduce novel concepts and techniques, e.g. board histories, to tackle this issue. 

Note that, the concepts and lemmas introduced in the last section will continue to be used in the following proofs.

For the sake of strictness, in the following we define some numbers that will be used later (cf. Definition \ref{type-label}, \ref{iterative-expansion}, \ref{type-label-star}, \ref{B_km-star}). %We give the intuitive meanings of these numbers by a remark following Definition \ref{structural-abstraction}. 
The readers may choose to skip these definitions temporarily and recall them only when necessary. 
%(for example, in the case when we need to understand Lemma \ref{B_k-has-no-k-clique}). 
The intuition behind these somewhat strange and elabrated definitions is that we want to use the values of $x, y$ (especially the value of $x$) of a vertex $(x,y)$ to determine the type label of $(x,y)$, which tells us how $(x,y)$ is connected to the other vertices.

In the following we define  
 $\mathpzc{U}_i^*$, $\eta_i^*$, $\gamma_i^*$, $\beta_{j}^{i}, [x]_{i},  \llparenthesis x\rrparenthesis_{i}$,  and $\mathbb{X}_i^*$
 using simultaneous induction, which means that some notations will possibly be used before they are defined. 

We use a $\mathpzc{U}_i^*$-tuple to denote \textit{one} unit in the $i$-th abstraction of the structures to be constructed.  
Let $\mathpzc{U}_m^*:=k-1$. Note that the factor ``$k-1$'' is the number of distinct coordinate congruence numbers. %In fact, in most cases when we see ``$k-1$'' we mean this number. 
And we usually regard every $k-1$ successive vertices as one object. 
%We introduce the other factor ``$k-1$'' to create necessary redundancy, which is explained in footnote \ref{footnote-free-will-1}. 

For any $1\leq i<m$, let 
\begin{align} \label{hbar}
%\eta_i^* &:=k-1\cdot k-1 \mbox{ if }i=m; \mbox{ otherwise, }\nonumber\\
\eta_i^* &:=(k-1)\times cl_{i+1}^*, \mbox{ where } \\
&cl_{i+1}^*=\left(2k\cdot\gamma_{m-i-1}^*+\Sigma_{j=1}^{k-2} (k\cdot\gamma_{m-i-1}^*)^j\right); \label{def-num-cl_i}\\
%\rho^* &:= 2^{\binom{k-2}{2}}\times 2^{10(m-i)} \times\eta_i^*;\\
\mathpzc{U}_i^* &:=3\times 2^{\binom{k-2}{2}}\times\eta_i^*. \label{def-hbar-i}
\end{align}

To understand what the value $cl_i^*$ stands for, cf. page \pageref{page-def-L-lists-for-cl}. And cf. page \pageref{{page-def-U_m}} for the role of $\mathpzc{U}_i^*$. Cf. page \pageref{page-SW-fun} 
for an explanation of the factors  $2^{\binom{k-2}{2}}$, and page \pageref{page-constant5} 
for an explanation of the constant ``$3$'' in (\ref{def-hbar-i}).    
%**************************************\\

 Let $\gamma_{i}^*$ ($0\leq i<m$) be defined as the following:
\begin{align}  %\label{-r-} 
%\alpha_i &:=& 2^{i+1+\sum_{m-i-1\leq j\leq m-1} j}\\
%\pi_0 &:=& |\mathbf{Cl}_m|=k^2\\
%\pi_{i+1} &:=& k^2\times \mu_{S_1,i} \\
%\tau_i &:=& \left\{\begin{array}{ll}
%\mathpzc{U}_m, & \quad\qquad i=0 \\[10pt]
%\displaystyle (2\mu_{S_2}\times\mu_{S_3,i})^i\times\pi_i, &  \quad\qquad i\geq 1\\
%\end{array}\right.\\ 
%\gamma_{i} &:=&\alpha_i\tau_i \\
%\frac{1}{2}\gamma_{i} &:=& \frac{1}{2}\alpha_i\tau_i\\
\gamma_{0}^* &:= 2^m\times \mathpzc{U}_m^*; \label{gamma_0-star}\\
\gamma_{i}^* &:=2^{m-i}\times \mathpzc{U}_{m-i}^*\times \gamma_{i-1}^*.\label{gamma_i-gamma_i-1}\\
 \mathbb{X}_1^* &:=[\gamma_{m-1}^*]\times [k]. \label{X_1-star}
\end{align}  

Note that ``$2^m$'', as well as $2^{m-i}$, comes from Fact \ref{linear-orders-1}. It helps to ensure that Spoiler cannot win the game simply by exploiting pure linear orders.\footnote{Later we shall see that Duplicator will resort to the $(m-i)$-th abstraction if she cannot respond properly without violating a linear order requirement in the $(m-i+1)$-th abstraction. From the definition (\ref{gamma_i-gamma_i-1}), we can see that $2^{m-i}$ objects (each object contains $\mathpzc{U}_{m-i}^*$ objects in the $(m-i)$-th abstraction, which include all the different types of vertices of index $m-i$. Cf. Definition \ref{vertex-index} for the concept vertex index) are ``generated (or expanded)'' from one object of the $(m-i+1)$-th abstraction. As a consequence,  Duplicator will find that every interval (delimited by pebbles or boundaries of the structures) in the $(m-i)$-th abstraction is sufficiently large even if the interval is of length one in the $(m-i+1)$-th abstraction. It is for this reason that we use $2^{m-i}$ instead of $2^m$ in definition (\ref{gamma_i-gamma_i-1}).} Compared with the corresponding definitions \ref{gamma_0-star-k3}, \ref{gamma_i-star-k3} (cf. the special case where $k=3$), where $4m$ suffice to make it and $ \mathpzc{U}_j^*=1$ for any $j$, in the general cases we need $2^m$ and $\mathpzc{U}_j^*$ needs to be sufficiently large to take account of all \textit{types} of vertices.

We shall define a notion called board history in Definition \ref{def-board-history}. In the following we use $bh^\#$ to denote the number of all possible board histories (including invalid ones). %Cf. page \pageref{page-def-J-lists-for-bh-sharp} for an explanation.
\begin{equation}\label{def-num-bh}
bh^\# :=m(k\cdot\gamma_{m-1}^*+1)^{k-1}.
\end{equation}

Similarly  we can define $\gamma_i$ and $\mathbb{X}_1$  as follows.
\begin{align}
\gamma_{0} &:= 2^m\times \mathpzc{U}_m^*\times m\times bh^\#; \label{gamma_0-star}\\
\gamma_{i} &:=2^{m-i}\times \mathpzc{U}_{m-i}^*\times \gamma_{i-1}^*\times m\times bh^\#, \mbox{ where } 0<i<m.\label{gamma_i-gamma_i-1}\\
 \mathbb{X}_1&:=[\gamma_{m-1}^*]\times [k]\times [m]\times [bh^\#]. \label{X_1-star}
\end{align}   

 By the definition, we have that $\gamma_i=m\times bh^\#\times\gamma_i^*$.

The notations $\beta_j^i$, $[x]_i$ and $\llparenthesis x\rrparenthesis_i$ are defined in (\ref{beta-func})\textapprox(\ref{llprrp-func}). We copy them here to remind the readers. 

For any $0\leq j\leq i\leq m-1$, \\[-8pt]
\begin{equation*}
\beta_{j}^{i}:=\frac{\gamma_{i}^*}{\gamma_{j}^*}\\[-3pt]
\end{equation*}

Note that $\frac{\gamma_{\ell}^*}{\gamma_{\ell-1}^*}=2^{m-\ell}\times \mathpzc{U}_{m-\ell}^*$ for any $1\leq \ell<m$, by (\ref{gamma_i-gamma_i-1}). That is, $\frac{\gamma_{\ell}^*}{\gamma_{\ell-1}^*}\in\mathbf{N}^+$ since, obviously, $\mathpzc{U}_{m-\ell}^*\in\mathbf{N}^+$.\footnote{
By definition, $\mathpzc{U}_{m}^*\in\mathbf{N}^+$. $\therefore \gamma_0^*\in\mathbf{N}^+$. Note that $\gamma_0^*$ is the width of the $m$-th abstraction of the structure $\mathfrak{B}_{k,m}^*$ we are going to construct. $\mathpzc{U}_{m-1}^*$ is then defined based on $\gamma_0^*$, which is in $\mathbf{N}^+$. Then $\gamma_1^*$ is defined based on $\mathpzc{U}_{m-1}^*$. Then based on $\gamma_1^*$, we can define  $\mathpzc{U}_{m-2}^*$. Afterwards,  $\gamma_2^*$ is defined based on $\mathpzc{U}_{m-2}^*$, which is in $\mathbf{N}^+$. 
And so on. In this process, 
$\gamma_{m-i-1}^*\in\mathbf{N}^+$ for any $i$, $\therefore cl_{i+1}^*\in\mathbf{N}^+$; $\therefore$ $\eta_i^*\in\mathbf{N}^+$. As a consequence, $\mathpzc{U}_{m-\ell}^*\in\mathbf{N}^+$ for any $\ell$.} 
Therefore, 
$\beta_j^i\in\mathbf{N}^+$ because $\beta_j^i=\displaystyle\prod_{j\leq \ell <i} \frac{\gamma_{\ell+1}^*}{\gamma_{\ell}^*}$.  

For any $0\leq x<\gamma_{m-1}^*$ and $1\leq i\leq m$, let 
\begin{align*}
[x]_{i} &:=\lfloor x/\beta_{m-i}^{m-1} \rfloor\\
%x|_+^{*i} &:=x\mbox{ mod } \beta_{m-i}^{*m-1}\\
\displaystyle\llparenthesis x\rrparenthesis_{i} &:=[x]_{i}\beta_{m-i}^{m-1}+\frac{1}{2}\sum_{1<j\leq i}\beta_{m-j}^{m-1} %\nonumber %\\[-18pt] 
\end{align*}

 For $2\leq i\leq m$, let %\\[-9pt] 
\begin{equation}\label{X_i}
\mathbb{X}_i^*:=\{(x,y)\in\mathbb{X}_1^* \1  x=\llparenthesis x\rrparenthesis_{i}\}. %\\[-7pt]
\end{equation}

 Note that $|\mathbb{X}_i^*|=k\times\gamma_{m-1}^*/\beta_{m-i}^{m-1}=k\times\gamma_{m-i}^*$. 
That is, $\gamma_{m-i}^*$ describes the width of the finite upright square lattice whose lattice points are the set of elements of $\mathbb{X}_i^*$. 
%Obviously, the depth of the finite square lattice is $k$.   
We will define a pair of structures $\mathfrak{A}_{k,m}^*$ and $\mathfrak{B}_{k,m}^*$ in Definition \ref{B_km-star}, whose universe are $\mathbb{X}_1^*$. Hence $\gamma_{m-1}^*$ is the width of the universe of $\mathfrak{A}_{k,m}^*$ and $\mathfrak{B}_{k,m}^*$. 
These two structures have $m$ abstractions as $\mathfrak{A}_{3,m}$ and $\mathfrak{B}_{3,m}$ do. In some sense, $\mathfrak{B}_{k,m}^*[\mathbb{X}_m^*]$ is  similar to $\mathfrak{B}_{3,m}[\mathbb{X}_m^*]$, except that we forbid some edges (for the definition of $\mathfrak{B}_{k,m}^*$, cf. Definition \ref{B_km-star}). Each vertex of $\mathfrak{B}_{k,m}^*$, as well as $\mathfrak{A}_{k,m}^*$, has an index that is defined in the same way as Definition \ref{vertex-index}.  
From the viewpoint of iterative structural expansion, we can  construct $\mathfrak{B}_{k,m}^*[\mathbb{X}_{m-1}^*]$ from $\mathfrak{B}_{k,m}^*[\mathbb{X}_m^*]$ and so on, until we obtain $\mathfrak{B}_{k,m}^*[\mathbb{X}_1^*]$, i.e. $\mathfrak{B}_{k,m}^*$. But now the situation is much more complicated than the special case:  in the $i$-th  abstraction, we need to create all \textit{types} of vertices in this abstraction by forbidding some edges, and put these varied distinct types of vertices (in the same row) into one ``\textit{u}nit'' or object, i.e. an interval of $\mathpzc{U}_{i}^*$ vertices.  
Note the \textbf{resemblance} between \ref{gamma_0-star}\textapprox  \ref{X_1-star} and \ref{gamma_0-star-k3}, \ref{gamma_i-star-k3}, \ref{X_1-star-k3}, if 
we regard every successive squence of $\mathpzc{U}_{i}^*$ vertices as one object in a row of the $i$-th abstraction of the structures. For instance, we can regard every $\mathpzc{U}_m^*$ successive vertices in a row of the $m$-th abstraction of $\mathfrak{A}_{k,m}^*$ as one object, or one \textit{u}nit. 
%\eqref{def-hbar-i} tells us the size of $cex(x,y,i-1)$ for any $(x,y)\in\mathbb{X}_i^*$ is $\mathpzc{U}_i^*$.  
We can call it an $\mathpzc{U}_m^*$-tuple. \label{{page-def-U_m}} We shall see, in Definition \ref{type-label-star}, how the value $[x]_i$ mod $\mathpzc{U}_i^*$ determines the way $(x,y)$ is connected to the other vertices, provided that  $\mathrm{idx}(x,y)=i$. 

For any $(x,y)\in \mathbb{X}_p^*$, we  use $(\llbracket x\rrbracket_p,y)$ to denote the $\mathpzc{U}_p^*$-tuple of vertices $\{(x_{min},y),(x_1,y),\ldots,(x_{_{\mathpzc{U}_p^*-1}},y)\}$ wherein, for any $i\in [1,\mathpzc{U}_p^*-1]$,   
\begin{enumerate}[(1)] \label{llbracket-min}
\item $(x_{min},y)\in \mathbb{X}_p^*$;  
\item  $\lfloor [x_{min}]_p/\mathpzc{U}_{p}^*\rfloor=\lfloor [x]_p/\mathpzc{U}_{p}^*\rfloor=[x_{min}]_p/\mathpzc{U}_{p}^*$;
\item $\mathrm{idx}(x_i,y)=p$; %$(x_i,y)\in \mathbb{X}_p^*-\mathbb{X}_{p+1}^*$; 
\item $[x_i]_p=[x_{min}]_p+i$. 
\end{enumerate}
(1) means that $\mathrm{idx}(x_{min},y)\geq p$; (2) implies that $\gamma_{m-p}^*$ can be divided by $\mathpzc{U}_{p}^*$, and $[x_{min}]_p$ and $[x]_p$ is in the same $\mathpzc{U}_p^*$-tuple where $[x_{min}]_p$ is the first element of this tuple; (4) means that, in the $p$-th abstraction, the distance between $(x_i,y)$ and $(x_{min},y)$ is $i$.  
We use $\llbracket x\rrbracket_p^{min}$ to denote $[x_{min}]_p$. 
We introduce $(\llbracket x\rrbracket_p,y)$ to denote such an interval that contains all the different types of vertices of index $p$.

Once $\mathfrak{A}_{k,m}^*$ and $\mathfrak{B}_{k,m}^*$ are constructed, whose universes are $\mathbb{X}_1^*$, another pair of structures $\mathfrak{A}_{k,m}$ and $\mathfrak{B}_{k,m}$ can be constructed based on them, whose universes are $\mathbb{X}_1$.  
In the sequel, we define a notation $x^\flat$ for any $x\in[\gamma_{m-1}]$ as follows  %\\[-5pt] 
\begin{equation}
x^\flat:=x \mbox{ mod } (\gamma_{m-1}^*\times k)
%x^\flat:=\lfloor x/(m\times bh^\#)\rfloor. %\\[-1pt]
\end{equation}
We associate a vertex $(x,y)$ in $\mathfrak{A}_{k,m}$ with a vertex $(x^\flat,y)$  in a ``flat'' structure $\mathfrak{A}_{k,m}^*$ that ``forgets'' \textit{board history} information (cf. Definition \ref{def-board-configuration} and Definition \ref{def-board-history}) of vertices of $\mathfrak{A}_{k,m}$. We can take it that each vertex $(x,y)$ in $\mathfrak{A}_{k,m}$ is a pair $((x^\flat,y),h_{xy})$ where $(x^\flat,y)$ is a vertex in $\mathfrak{A}_{k,m}^*$ and $h_{xy}$ describes the board history associated with $(x^\flat,y)$. Later, we shall see that all the vertices associated with the same board history are arranged together w.r.t. the first coordinate.

For $1<i\leq m$, let 
\begin{equation}
\mathbb{X}_i:=\{(x,y)\in\mathbb{X}_1\mid (x^\flat,y)\in \mathbb{X}_i^*\}. %\\[-7pt]  
\end{equation}

We can also assign an index for each vertex $(x,y)$ of $\mathfrak{A}_{k,m}$, as well as  $\mathfrak{B}_{k,m}$, with $\mathrm{idx}(x^\flat,y)$. Therefore, $\mathrm{idx}(x,y)=i$ if and ony if $(x,y)\in \mathbb{X}_i-\mathbb{X}_{i+1}$, provided that $i<m$. 

\begin{definition}\label{def-board-configuration}
A \textbf{board configuration}
 is a $(k-1)$-tuple $\{\mathbb{X}_1^*\cup\{(*,*)\}\}^{k-1}$.\footnote{An alternative definition is the following. A board configuration is a $(k-1)$-tuple $\{(\mathbb{X}_1^*,n)\cup\{((*,*),0)\}\}^{k-1}$, where $n\in [1,k-1]$, such that the sum of the second item of the pairs is no more than $k-1$. This version of definition tells us how many pebbles are put on the same vertex.}
  A \textit{valid board configuration} is a board configuration where ``$(*,*)$''s can only appear in the tail of the tuple. Let $\mathbb{BC}$ be the set of all valid configurations.  
For any $(x,y)\in\mathbb{X}_1$, we use $(x,y)[\mathrm{BC}]$
to denote the board configuration associating with $(x,y)$. Moreover, for any valid board configuration $\mathbb{Z}$, we use $|\mathbb{Z}|$ to denote the number of elements of $\mathbb{Z}$  that is not $(*,*)$. $\mathbb{Z}$ is an \textit{empty board configuration}, denoted by $\mathrm{BC}_\emptyset$, if $|\mathbb{Z}|=0$.  
\hfill\ensuremath{\divideontimes}
\end{definition}

In this paper, when we talk about a board configuration, it is a valid one by default. 
We use $(x,y)[\mathrm{BC}]\ccirc(u,v)$, where $(u,v)\in\mathbb{X}_1$, to denote that the first ``$(*,*)$'' in the tuple $(x,y)[\mathrm{BC}]$, if there is a  $(*,*)$ and $(u^\flat,v)$ is not in the tuple, is replaced by $(u^\flat,v)$.
%Note that we use a board configuration to describe the set of ``pebbled'' vertices in either $\mathfrak{A}_{k,m}^*$ or $\mathfrak{B}_{k,m}^*$. 
$(x_i,y_i)[\mathrm{BC}]$ describes the setting when exactly those vertices $(u,v)$ in  $(x_i,y_i)[\mathrm{BC}]$ are ``supposed to'' be pebbled, though not necessary true.  
A  configuration $(x_i,y_i)[\mathrm{BC}]$ can evolve to another configuration $(x_j,y_j)[\mathrm{BC}]$ in one round of a reasonable game over board $(\mathfrak{A}_{k,m}^*,\mathfrak{B}_{k,m}^*)$, denoted $(x_i,y_i)\triangleright(x_j,y_j)$, if and only if either $(x_i,y_i)[\mathrm{BC}]\ccirc(x_i,y_i)=(x_j,y_j)[\mathrm{BC}]$, 
  or   $(x_i,y_i)[\mathrm{BC}]\!=\!(x_j,y_j)[\mathrm{BC}]\ccirc(x_i,y_i)$. That is, it is either $(x_i,y_i)[\mathrm{BC}]\!=\!(x_j,y_j)[\mathrm{BC}]$, or $(x_j,y_j)[\mathrm{BC}]$  evolves from $(x_i,y_i)[\mathrm{BC}]$ by adding or removing $(x_i^\flat,y_i)$. And if $|(x_i,y_i)[\mathrm{BC}]|=k-1$, then $(x_i,y_i)[\mathrm{BC}]$ could only evolve to  $(x_j,y_j)[\mathrm{BC}]$ by removing $(x_i^\flat,y_i)$; if $|(x_i,y_i)[\mathrm{BC}]|=0$, then  $(x_i,y_i)[\mathrm{BC}]$ could only evolve to  $(x_j,y_j)[\mathrm{BC}]$ by  adding $(x_i^\flat,y_i)$.

For any board configuration $\mathbb{Z}$, we can use $\mathbb{Z}-\{(*,*)\}$ to denote the tuple that is obtained from $\mathbb{Z}$ by removing all the ``$(*,*)$''. For convenience, we also use $(x,y)\mathrm{BC}$ to denote $(x,y)\mathrm{BC}-\{(*,*)\}$. Hence such truncated board configurations may have less than $k-1$ elements.  

In the following we define an important concept that reflects intrinsic ``logic'' of evolution of games, but \textit{not} the \textit{real} evolution of games. 
\begin{definition}\label{def-board-history}
For any $(x,y)\in\mathbb{X}_1$, the \textbf{board history} of $(x,y)$, denoted by $\chi(x,y)\!\restriction\! \mathrm{BH}$,  consists of a sequence of $m$ board configurations, written $(\mathrm{BC}_0,\ldots,\mathrm{BC}_{m-1})$. A valid board history need satisfy the following requirements. 
\begin{itemize}
\item $\mathrm{BC}_i$ is a valid board configuration, for any $i$;
\item let   $\mathrm{i}_{\mathrm{cur}}^{x,y}\!:=\!\chi(x,y)\!\!\restriction\!\! \mathrm{bc}+1$ and $(x_{_{\mathrm{i}_{\mathrm{cur}}^{x,y}}},y_{_{\mathrm{i}_{\mathrm{cur}}^{x,y}}}):=(x,y)$; for $0\!\leq\! j\!<\! \mathrm{i}_{\mathrm{cur}}^{x,y}$, $\mathrm{BC}_j\!=\!(x_{j+1},y_{j+1})[\mathrm{BC}]$, for some $(x_{j+1},y_{j+1})\!\in\! \mathbb{X}_1$; in particular, 
\begin{enumerate}[(i)]
\item $\chi(x,y)\!\restriction\! \mathrm{BH}(\mathrm{i}_{\mathrm{cur}}^{x,y}-1)=(x,y)[\mathrm{BC}]$;
 
\item 
$\mathrm{BC}_0=\mathrm{BC}_\emptyset$;
\end{enumerate}   

\item for $1\leq j<\mathrm{i}_{\mathrm{cur}}^{x,y}$, $(x_j,y_j)\triangleright(x_{j+1},y_{j+1})$; henceforth, for any $i$ we use  $(x,y)\!\!\restriction\!\!_{\mathrm{H}}^i$ to denote $(x_i,y_i)$, e.g.  $(x,y)\!=\!(x,y)\!\!\restriction\!\!_{\mathrm{H}}^{\mathrm{i}_{\mathrm{cur}}^{x,y}}$;  
\item for $\mathrm{i}_{\mathrm{cur}}^{x,y}\leq j\leq m-1$, $\mathrm{BC}_j=((*,*),\ldots,(*,*))$. 
\hfill $\divideontimes$
\end{itemize}  
\end{definition}

For convenience, the denotation ``$\chi(x,y)\!\restriction\! \mathrm{BH}$'' coincides with that of type label, a notion which will be introduced later (cf. Definition \ref{type-label}). By default,  a board history is a valid board history. The (actual) \textit{length} of the board history of $(x,y)$ is $\chi(x,y)\!\!\restriction\!\! \mathrm{bc}$, i.e. $\mathrm{i}_{\mathrm{cur}}^{x,y}-1$. 
Say that $\chi(x,y)\!\!\restriction\!\! \mathrm{BH}$ is a \textit{void board history}, 
if all board configurations in $\chi(x,y)\!\!\restriction\!\! \mathrm{BH}$ are empty, i.e.  $\chi(x,y)\!\!\restriction\!\! \mathrm{bc}=0$. 
 We use $\chi(x,y)\!\!\restriction\!\! \mathrm{BH}\!\ccirc\! \mathrm{BC}_l$, where $\mathrm{BC}_l$ is a nonempty board configuration (i.e. not a sequence of $((*,*),\!\ldots,\!(*,*))$), to denote that the board history of $(x,y)$ is extended by one more board configuration, i.e. the first  $((*,*),\!\ldots,\!(*,*))$ in the tail is replaced by $\mathrm{BC}_l$. For convenience, we also use $\chi(x,y)\!\restriction\! \mathrm{BH}$ to denote the board history of $(x,y)$ wherein all the empty board configurations in the tail are removed. Hence such board histories may have less than $m$ elements.

A board history $\chi(x_j,y_j)\!\!\restriction\!\! \mathrm{BH}$ can legally evolve from another board history $\chi(x_i,y_i)\!\!\restriction\!\! \mathrm{BH}$ in one step, denoted $(x_i,y_i)\!\xrightarrow[\mathrm{BC}]{1}\! (x_j,y_j)$, if and only if
\begin{enumerate}[(1)]
\item $\mathrm{i}_{\mathrm{cur}}^{x_j,y_j}=\mathrm{i}_{\mathrm{cur}}^{x_i,y_i}$+1; and 
 
\item 
$\chi(x_i,y_i)\!\restriction\! \mathrm{BH}\sqsubseteq\chi(x_j,y_j)\!\restriction\! \mathrm{BH}$; and

\item $(x_i,y_i)\triangleright (x_j,y_j)$.

\end{enumerate}
If $\chi(x_p,y_p)\!\!\restriction\!\! \mathrm{BH}$ can evolve from  $\chi(x_i,y_i)\!\!\restriction\!\! \mathrm{BH}$ using $l$ steps, and   $(x_p,y_p)\xrightarrow[\mathrm{BC}]{1} (x_j,y_j)$, then $\chi(x_j,y_j)\!\!\restriction\!\! \mathrm{BH}$ can evolve from  $\chi(x_i,y_i)\!\!\restriction\!\! \mathrm{BH}$ using $l+1$ steps.  
  We use  $(x_i,y_i)\xrightarrow[\mathrm{BC}]{*} (x_j,y_j)$ to denote that   $\chi(x_j,y_j)\!\!\restriction\!\! \mathrm{BH}$ can evolve from  $\chi(x_i,y_i)\!\!\restriction\!\! \mathrm{BH}$ using  $\mathrm{i}_{\mathrm{cur}}^{x_j,y_j}-\mathrm{i}_{\mathrm{cur}}^{x_i,y_i}$ steps.  

For valid configurations $(x_i,y_i)[\mathrm{BC}]$ and $(x_j,y_j)[\mathrm{BC}]$, if $(x_i,y_i)\xrightarrow[\mathrm{BC}]{*} (x_j,y_j)\\\land (x_i,y_i)[\mathrm{BC}]\sqsubseteq 
\!(x_j,y_j)[\mathrm{BC}]$, or $(x_j,y_j)\xrightarrow[\mathrm{BC}]{*} (x_i,y_i)\land (x_j,y_j)[\mathrm{BC}]\sqsubseteq\! (x_i,y_i)[\mathrm{BC}]$, we say that the board histories are in \textit{continuity} (in the default direction\footnote{That is,  $\chi(x_j,y_j)\!\!\restriction\!\! \mathrm{BH}$ can be evolved from $\chi(x_i,y_i)\!\!\restriction\!\! \mathrm{BH}$ if $\chi(x_j,y_j)\!\!\restriction\!\! \mathrm{bc}>\chi(x_i,y_i)\!\!\restriction\!\! \mathrm{bc}$, and vice versa. Moreover, we can also define a notion ``$\sqsubset$'' akin to $\subset$ and give an alternative definition for continuity based on it.}), written $(x_i,y_i)\xrightarrow[\mathrm{BC}]{con.} (x_j,y_j)$ or $(x_j,y_j)\xrightarrow[\mathrm{BC}]{con.} (x_i,y_i)$ respectively. 
Note that only valid board histories can be in continuity. By the transitivity of the binary relations $\xrightarrow[\mathrm{BC}]{*}$ and $\subseteq$, the binary relation $\xrightarrow[\mathrm{BC}]{con.}$ is also transitive. 

The \textit{initial segment of board history} $\chi(x,y)\!\restriction\! \mathrm{BH}$ of length $\ell$, written by $\chi(x,y)\!\restriction\! \mathrm{IBH}[\ell]$, is composed of the first $\ell+1$ board configurations in $\chi(x,y)\!\restriction\! \mathrm{BH}$. We use $(x,y)_{[\ell]}\xrightarrow[\mathrm{BC}]{con.} (x^\prime,y)$ to denote that $\chi(x,y)\!\restriction\! \mathrm{IBH}[\ell]$ can evolve to and is in continuity with  $\chi(x^\prime,y)\!\restriction\! \mathrm{BH}$.

In the following we introduce some constructions that explain the numbers $cl_i^*$ in \eqref{def-num-cl_i}.  
We can create a row of vertices that is isomorphic to $\mathfrak{B}_{k,m}^*[\mathbb{X}_i^*]|\langle \leq \rangle$.   
We denote such a tuple by $\mathfrak{L}_i$. Note that the elements of $\mathfrak{L}_i$ have the same second coordinate, whereas those of $\mathfrak{B}_{k,m}^*[\mathbb{X}_i^*]|\langle \leq \rangle$  have $k$ different second coordinates.  
Now we make a larger tuple of vertices in a row, denoted $\mathfrak{L}_i^+$, by concatenating two copies of $\mathfrak{L}_i$ with $k-2$ more lists $\mathfrak{L}_i^j$, where $1\leq j\leq k-2$ and the productions are Cartesian products. \textbf{Note that} $|\mathfrak{L}_i^+|=cl_i^*$. \label{page-def-L-lists-for-cl} That is,  
we want to use a number modulo $cl_i^*$ to encode game boards with up to $k-2$ ``pebbled'' vertices over the $i$-th abstraction.  And we shall see what the number means in Definition \ref{type-label}.  
Note that, $\mathfrak{L}_i^+$ is only a small piece of an $\mathpzc{U}_i^*$-tuple in $\mathfrak{B}_{k,m}^*$.  

%% 
%For any $1\leq t,l\leq m$ and any $(x_i,y_i)\in \mathbb{X}_t-\mathbb{X}_{t+1}$,  if $l\geq t$ then let 
We define a function $\mathrm{RngNum}(\cdot,\cdot)$ as the following.

For any $1\leq l\leq m$ and any $(x_i,y_i)\in \mathbb{X}_1$, 
\begin{equation}
\mathrm{RngNum}(x_i^\flat,l):=
%\left\{\begin{array}{ll}
\left\lfloor \frac{[x_i^\flat]_{l}\mbox{ mod }\mathpzc{U}_{l}^*}{\frac{1}{3}\mathpzc{U}_{l}^*} \right\rfloor-1. %&\mbox{ if } [x_i^\flat]_{l}\mbox{ mod }\mathpzc{U}_{l}^*\geq \frac{1}{3}\mathpzc{U}_{l}^*;\\[15pt]
%-1, & \mbox{ otherwise}. 
%\end{array}\right.
\end{equation}

%If $l<t$, then let 
%\begin{equation}
%\mathrm{RngNum}(x_i^\flat,l):=-1.
%\end{equation} 

Note that there are exactly $3$ different values for $\mathrm{RngNum}$. It explains \textit{the constant ``$3$''} that appears in (\ref{def-hbar-i}) \label{page-constant5}

Because $[x_i^\flat]_{l}\mbox{ mod }\mathpzc{U}_{l}^*=0$ if $l<\mathrm{idx}(x_i,y_i)$, we immediately have the following observation. 
\begin{lemma}\label{rngnum-is_-1}
If $\mathrm{idx}(x_i,y_i)>l$, then 
$$\mathrm{RngNum}(x_i^\flat,l)=-1.$$ 
\end{lemma}

Likewise, for any $(x_i,y_i)\in \mathbb{X}_1^*$, we can define $\mathrm{RngNum}$ in the similar way except that the superscript ``$\flat$'' should be removed from the definition. 

%================ begin comment ===================
\begin{comment}
We have the following observation. 

\begin{lemma}\label{RngNum-xi-2_xi-1}
For any $(x,y),(x^\prime,y)\in\mathbb{X}_1^*$ and $1<\xi\leq m$, if $\llparenthesis x\rrparenthesis_\xi-\llparenthesis x\rrparenthesis_{\xi-1}=\llparenthesis x^\prime\rrparenthesis_\xi-\llparenthesis x^\prime\rrparenthesis_{\xi-1}$, then 
\begin{equation*}
\mathrm{RngNum}(\llparenthesis x\rrparenthesis_{\xi-1},\xi-1)=\mathrm{RngNum}(\llparenthesis x^\prime\rrparenthesis_{\xi-1},\xi-1).
\end{equation*}
\end{lemma}
\begin{proof}
By Lemma \ref{proj-greater-index} and Lemma \ref{i=0theni-1=0}, $(\llparenthesis x\rrparenthesis_\xi,y)\in\mathbb{X}_{\xi-1}^*$. Hence, by Lemma \ref{unit-distance},  $|\llparenthesis x\rrparenthesis_\xi-\llparenthesis x\rrparenthesis_{\xi-1}|=c\beta_{m-\xi+1}^{m-1}$ for some $c\in\mathbf{N}_0$. Likewise,  $|\llparenthesis x^\prime\rrparenthesis_\xi-\llparenthesis x^\prime\rrparenthesis_{\xi-1}|=c\beta_{m-\xi+1}^{m-1}$. 
Moreover, $[\llparenthesis x\rrparenthesis_\xi]_{\xi-1}$ mod $\mathpzc{U}_{\xi-1}^*=0$, because $\mathrm{idx}(\llparenthesis x\rrparenthesis_\xi,y)\geq \xi$ (cf. Lemma \ref{proj-greater-index}). Similarly, $[\llparenthesis x^\prime\rrparenthesis_\xi]_{\xi-1}$ mod $\mathpzc{U}_{\xi-1}^*=0$. 
Therefore, the claim holds. 
\end{proof}

\end{comment}
%================= end comment ====================

Assume that, like $(x_i,y_j)$, $(x_j,y_j)$ is also in $\mathbb{X}_1$ and that  $min\{\mathrm{idx}(x_i^\flat,y_i),\\\mathrm{idx}(x_j^\flat,y_j)\}=t$ and $y_i\neq y_j$, we define a function $\mathrm{sgn}: \mathbb{X}_1^* \times \mathbb{X}_1^*\mapsto\{0,1\}$ as follows. This function will be used in Definition \ref{iterative-expansion}. Note that, when $y_i,y_j\in [1,k-2]$, the vaule of $\mathrm{sgn}((x_i^\flat,y_i),(x_j^\flat,y_j))$ is meaningless since it will not be used.  From here on we assume that either $y_i\equiv 0$ mod $k-1$ or $y_j\equiv 0$ mod $k-1$. 
Let $\mathrm{sgn}((x_i^\flat,y_i),(x_j^\flat,y_j))=\mathrm{sgn}((x_j^\flat,y_j),(x_i^\flat,y_i))$, where  
\begin{itemize}
\item  $\mathrm{idx}(x_i^\flat,y_i)=\mathrm{idx}(x_j^\flat,y_j)=t$:\\ 
$\mathrm{sgn}((x_i^\flat,y_i),(x_j^\flat,y_j))=0$. 

\item $\mathrm{idx}(x_i^\flat,y_i)>t=\mathrm{idx}(x_j^\flat,y_j)$ (it is symmetric when $\mathrm{idx}(x_j^\flat,y_j)>t$):\\  
$\mathrm{sgn}((x_i^\flat,y_i),(x_j^\flat,y_j))=1$ if and only if 
one of the following holds
\begin{itemize}
\item $k-1=y_i>y_j>0$ and $\mathrm{RngNum}(x_j^\flat,t)=0$; 

\item $k-1>y_j>y_i=0$ and $\mathrm{RngNum}(x_j^\flat,t)=1$. 

\end{itemize}

\end{itemize}

The following observation is direct. 
\begin{lemma}\label{sgn-equal-0}
 %If $\mathrm{idx}(x_i^\flat,y_i)\leq\mathrm{idx}(x_j^\flat,y_j)$ 
Suppose that $min\{\mathrm{idx}(x_i^\flat,y_i),\mathrm{idx}(x_j^\flat,y_j)\}=t$  and $\mathrm{RngNum}(x_i^\flat,t)=\mathrm{RngNum}(x_j^\flat,t)=-1$. Then $$\mathrm{sgn}((x_i^\flat,y_i),(x_j^\flat,y_j))=0.$$  
\end{lemma} 
%\begin{proof}
%If $\mathrm{idx}(x_i^\flat,y_i)=\mathrm{idx}(x_j^\flat,y_j)$, then $\mathrm{sgn}((x_i^\flat,y_i),(x_j^\flat,y_j))=0$. If $t=\mathrm{idx}(x_i^\flat,y_i)<\mathrm{idx}(x_j^\flat,y_j)$, then by Lemma \ref{rngnum-is_-1}, $\mathrm{RngNum}(x_j^\flat,t)=-1$. Again, $\mathrm{sgn}((x_i^\flat,y_i),(x_j^\flat,y_j))=0$. 
%\end{proof}

By definition, we know that, for any $(x,y)$ where $0<y<k-1$, $\mathrm{idx}(x,y)=i$  and $\mathrm{RngNum}(x,i)\neq -1$, $(x,y)$ is either not adjacent to $(a,k-1)$ or not adjacent to $(b,0)$ where $\mathrm{idx}(a,k-1)>i$ and $\mathrm{idx}(b,0)>i$. Note that in this case $\mathrm{RngNum}(a,i)=\mathrm{RngNum}(b,i)=-1$. %It is important since it helps ensure that $\mathfrak{B}_{k,m}^*$ has no $k$-clique. Cf. Lemma \ref{B_k-has-no-k-clique}, the case (3). 

%For a pair of vertices $(x,y)$ and $(u,v)$ in $\mathbb{X}_1^*$, suppose that $y$ or $v$ is $0$ or $k-1$, $y\neq v$,  $\mathrm{idx}(x,y)=i<\mathrm{idx}(u,v)$, and $\mathbf{cc}([x]_i,y)\neq \mathbf{cc}([u]_i,v)$. We use a list $\mathfrak{R}_{xy}$ to encode the information that tells us whether $(u,v)$ should \textit{not} be adjacent to $(x,y)$ simply based on the triple $(v,\mathrm{idx}(u,v),\mathrm{RngNum}(u,\mathrm{idx}(u,v))$. %Suppose without loss of generality that $0<y<k-1$.\footnote{Here we regard ``$0<y<k-1$'' as one possible value, just as $y=0$ and $y=k-1$. Hence, there are three possible values in total, just as shown in Fig. \ref{sgn-RangNum}} 
%Let the list $\mathfrak{P}_y$ be the tuple of elements $(j,p,q)$, where $j\in \{0,k-1\}$, $p=\{i+1,i+2,\ldots,m\}$ and $q=\{-1,0,\ldots,3\}$, which is ordered by the lexicographic ordering on the Cartesian product $\{0,k-1\}\times \{i+1,i+2,\ldots,m\} \times \{-1,0,\ldots,3\}$. Then $\mathfrak{R}_{xy}$ is the power set of $\mathfrak{P}_y$ with its elements in the natural lexicographic ordering. 
%\textbf{That is}, $|\mathfrak{R}_{xy}|$ is $2^{10(m-i)}$, and we shall see in Definition \ref{type-label-star} and Definition \ref{B_km-star} that, in $\mathfrak{B}_{k,m}^*$, any edge between $(x,y)$ and $(u,v)$ is forbidden if the triple  $(v,\mathrm{idx}(u,v),\mathrm{RngNum}(u,\mathrm{idx}(u,v))$ is in the $\ell$-th element of $\mathfrak{R}_{xy}$ where $\ell=\lfloor [x]_i/(\eta_i^*/2^{10(m-i)})\rfloor$ mod $2^{10(m-i)}$. \label{page-R_y}     

Let $\mathbb{X}_1^\downarrow:=\{0,\ldots,\gamma_{m-1}^*-1\}\times \{1,\ldots,k-2\}$.

In this paper, for a natural number $n$, we use $(n)_{2;\binom{k-2}{2}}$ to denote the binary representation of $n$, a $0\textnormal{-}1$ string of length exactly $\binom{k-2}{2}$.

\label{page-SW-fun} If $\mathrm{idx}(x_i^\flat,y_i)=\mathrm{idx}(x_j^\flat,y_j)=t$, 
we define a function $\mathrm{SW}: \mathbb{X}_1^\downarrow\times \mathbb{X}_1^\downarrow\mapsto \{0,1\}^{\binom{k-2}{2}}$. 
For any $(x,y)\in\mathbb{X}_t^*-\mathbb{X}_{t+1}^*$, $g(x):=0$ if 
 $[x]_{t}\mbox{ mod }\mathpzc{U}_{t}^*\!<\!\frac{1}{3}\mathpzc{U}_{t}^*$;  $g(x):=\lfloor[x]_{t}/\eta_t^*\rfloor\mbox{ mod } 2^{\binom{k-2}{2}}$, otherwise.
Let 
\begin{equation}\label{def-SW-func}
\mathrm{SW}\left((x_i^\flat,y_i),(x_j^\flat,y_j)\right):=\left(\left|
g(x_i^\flat)-g(x_j^\flat)\right|\right)_{2;\binom{k-2}{2}}
\end{equation}
$\mathrm{SW}\left((x_i^\flat,y_i),(x_j^\flat,y_j)\right)$ is undefined if  $\mathrm{idx}(x_i^\flat,y_i)\neq\mathrm{idx}(x_j^\flat,y_j)$.

For any $(x,y),(x^\prime,y^\prime)\in\mathbb{X}_1^*$ where $0<y<y^\prime<k-1$ and $\mathrm{idx}(x,y)=\mathrm{idx}(x^\prime,y^\prime)$, we use  a $(k-2)$-by-$(k-2)$ symmetric $0$-$1$ matrix, say $M=(a_{i,j})$, to encode $\mathrm{SW}((x,y),(x^\prime,y^\prime))$ \textit{succinctly}. That is, only the entry $a_{i,j}$ where $k-2\geq j> k-1-i \geq 1$ has a valid value. 
Let $\hat{q}(y,y^\prime):=y^\prime-y+\Sigma_{s=0}^{y-2} (k-3-s)$ if $y<y^\prime$; undefined, otherwise. We use $a_{k-1-y,y^\prime}$ to denote the $(\hat{q}(y,y^\prime)-1)$-th bit of $\mathrm{SW}((x,y),(x^\prime,y^\prime))$.  In other words, 
we use $a_{k-2,2}$ to denote the $0$-th bit of $\mathrm{SW}((x,y),(x^\prime,y^\prime))$, and $a_{k-2,3}$ for the $1$-th bit, $\ldots$, and $a_{k-3,3}$ for the $(k-3)$-th bit, and so on. Note that $a_{2,k-2}$ is the leftmost bit of $\mathrm{SW}((x,y),(x^\prime,y^\prime))$. 
We use $a_{k-1-y,y^\prime}\!\in\!\{0,1\}$ to tell whether the edge between $(x,y)$ and $(x^{\prime},y^\prime)$ should be ``switched  off'' or not (cf. Definition \ref{iterative-expansion}, 2 (e)).  
%which tells whether we should remove the edge between $(x,y_i)$ and $(x^{\prime},y_j)$ depending on the values of $y_i$ and $y_j$. \\[-12pt] 

%Here we use $(z)_{2;\binom{k-2}{2}}$ to denote the $\binom{k-2}{2}$ bits binary representation of a natural number $z$, 

\begin{example}\label{example-0-1-matrix-SW}
There are several ways to encode $\mathrm{SW}(\cdot,\cdot)$. Perhaps our choice is not the most standard one. Hence we give an example in the following.

Assume that $k=7$ and that the pair of vertices $(x,2),(x^\prime,4)\in \mathbb{X}_1^*$, where $\mathrm{idx}(x,2)=\mathrm{idx}(x^\prime,4)$, satisfy that
$$\mathrm{SW}((x,2),(x^\prime,4))=1011100011.$$
 
The matrix $(a_{ij})$ is shown in the following.

\begin{equation*}
\left(\begin{array}{ccccc}
\times &\times &\times &\times &\times\\
\times &\times &\times &\times &1\\
\times &\times &\times &1 &0\\
\times &\times &0 &\textbf{\textit{1}} &1\\
\times &1 &1 &0 &0\\
\end{array}\right) 
\end{equation*}

Note that, the topmost ``$1$'' in the position of the second row from the top and the fifth column from the left is the value of $a_{2,5}$. 

By definition,  
$\hat{q}(2,4)=6$. From the matrix, we can see that the $5$-th bit of $\mathrm{SW}((x,2),(x^\prime,4))$ is $1$. 

Hence $\mathbf{BIT}(\mathrm{SW}((x,2),(x^\prime,4)),\hat{q}(2,4))=1$.  
\hfill\ensuremath{\divideontimes}
\end{example}

We will define the structures $\mathfrak{A}_{k,m}^*$ and $\mathfrak{B}_{k,m}^*$ in Definition \ref{B_km-star}. We will use $E_*^A$ to denote the set of edges of $\mathfrak{A}_{k,m}^*$ and $E_*^B$ to denote that of $\mathfrak{B}_{k,m}^*$. Moreover, we use 
 $E_*$ to denote the set of edges of either $\mathfrak{A}_{k,m}^*$ or 
 $\mathfrak{B}_{k,m}^*$.

We associate every vertex of $\mathfrak{B}_{k,m}^*$ with a label, called congruence label, that is related to the coordinate congruence number in the $i$-th abstraction where $i$ is the index of this vertex. Assume that $m\geq k$.
\begin{definition}\label{def-congruence-labels-B}
The \textbf{\textit{set} of congruence labels} of $\mathfrak{B}_{k,m}^*$, denoted by $\mathbf{Cl}_i$, are defined as  follows. Note that we use underline to stand for a \textit{string}.  For example, $\underline{x}$ stands for some sort of encoding of $x$. %\\[-18pt]
\begin{equation}\label{def-Cl_m}
\mathbf{Cl}_m :=\displaystyle\{\underline{n,j;m;\mathsf{R};\emptyset}\mid n\in[k-1]; j\in [k];\mathsf{R}\in\{-1,0,1\}\}; 
\end{equation} %\\[-10pt]
For $1\!<i\leq\!m$,%\\[-22pt]
\begin{multline}\label{def-Cl_i}
\mathbf{Cl}_{i-1}:=\{\underline{n,j;i-1;\mathsf{R};M} \1 n\in[k-1];j\in [k];\\\mathsf{R}\in\{-1,0,1\};M\subseteq \mathbf{Cl}_i\}\cup \mathbf{Cl}_i.
\end{multline}%\\[-35pt]
\hfill \ensuremath{\divideontimes}
\end{definition}

\begin{remark}\label{remark-congruence-labels}
 In  (\ref{def-Cl_m}), $n$ is intended to denote the coordinate congruence number of a vertex in the $m$-th abstraction; while $j$ is the second coordinate of the vertex and $m$ is its index. $\mathsf{R}$ is intended to denote a value for $\mathrm{RngNum}(\cdot,\cdot)$. 
In a moment we shall introduce a related notion, i.e. a congruence label associating with a vertex, cf. Definition \ref{congruence-label} and Definition \ref{type-label-star}. Every vertex is associated with a congruence label. 
For example, in $\mathfrak{B}_{k,m}^*$, if a vertex $(x,y)$ has a congruence label of $\underline{1,2;m;-1;\emptyset}$, it implies that $\mathrm{idx}(x,y)=m$, $y=2$, $\mathbf{cc}([x]_m,y)=1$ and $\mathrm{RngNum}(x,m)=-1$.  
In (\ref{def-Cl_i}), ``$i-1$'' denotes the index of a  vertex, and we use $M$ to denote that the vertex is \textit{not} adjacent to any vertex whose congruence label is in $M$. Note that, the index of any vertex, whose congruence label is in $\mathbf{Cl}_i$, is greater than or equal to $i$.  
%We also use the notation $\underline{n,j;*;M}$, wherein $*$ stands for arbitrary value in $[1,m]$.  
%Note that we use an underline below 0 in $(\underline{0},0)$ to distinguish it from the vertex $(0,0)$.
\hfill \ensuremath{\divideontimes}
\end{remark}

Note that the size of $\mathbf{Cl}_i$ is completely determined by $k$, which is \textit{not} equivalent to $cl_i^*$ or $cl_i$. 

In the following, for the sake of conciseness, in the first place we define  $\mathfrak{B}_{k,m}$ (cf. Definition \ref{iterative-expansion}), and closely related concepts ``type label'' and ``congruence label'' for elements in $\mathbb{X}_1$. 
But the readers can choose to read and understand Definition \ref{type-label-star} and Definition \ref{B_km-star} first, which are relatively more simple, and come back here only when necessary. Although reading them will require an understanding of Definition \ref{type-label}\textapprox Definition \ref{iterative-expansion}, the readers can ignore the parts related to board histories at the moment. After understanding these notions and structures, the readers should continue to read and understand Definition \ref{type-label}\textapprox Definition \ref{iterative-expansion}. Note that we define the congruence label and type label of a vertex, as well as the (edges of) structure $\mathfrak{B}_{k,m}^*$,   \textit{simultaneously}. Cf. page \pageref{page-def-L-lists-for-cl} for the definition of $\mathfrak{L}_{i}^+$. 

\begin{definition} \label{type-label}
For any $(x,y)\in \mathbb{X}_1$, the\textbf{ type label} of $(x,y)$, denoted by
 $\chi(x,y)$, is defined as follows. 

Assume that $\mathrm{idx}(x^\flat,y)=i$. Let $\chi(x,y):=\underline{\underline{\mathfrak{a},y;i;S};\mathrm{BH};\mathrm{bc};\Omega}$ where \\[-20pt]
\begin{itemize}
\item 
\begin{equation}  \label{i:b}
\mathfrak{a}=\mathbf{cc}([x^\flat]_i,y);\\[-7pt]
\end{equation}

\item  $\mathrm{BH}$ (board history) is the $j^{\prime}$-th element of $\mathfrak{J}^+$ where%\\[-12pt] 
\begin{equation}\label{def-value-board-histroy}
j^{\prime}=\lfloor x/(\gamma_{m-1}^*\times k)\rfloor \mbox{ mod } bh^\#;\\[-12pt]
\end{equation}

\item  $\mathrm{bc}$ is the (actual) length of the board history of $(x,y)$ where 
\begin{equation}\label{def-length-history}
\mathrm{bc}=\lfloor x/(\gamma_{m-1}^*\times k\times bh^\#)\rfloor \mbox{ mod } m;
\end{equation}

\item if $i=m$ then 
$S=\Omega=\emptyset$; otherwise (i.e. $1\leq i<m$), 

\item  $\Omega$ is defined as follows:
\begin{itemize}
\item  $y\in\{0,k-1\}$:  $\Omega=\emptyset$;
 
\item  $y\in [1,k-2]$: $\Omega =\{(u^\flat,v)\in\mathbb{X}_i^*\1 (u,v)\in\mathbb{X}_i; v\neq y; [x^\flat]_i\equiv  [u^\flat]_i\equiv 0 \hspace{3pt}(\mbox{mod } k-1); \left((\llparenthesis u^\flat\rrparenthesis_{i+1},v),(\llparenthesis x^\flat\rrparenthesis_{i+1},y)\right)\!\notin\! E_*;  
v\in\{0,k-1\}\rightarrow (u^\flat,v)\in\mathbb{X}_{i+1}^*\land\mathrm{sgn}((\llparenthesis x^\flat\rrparenthesis_{i+1},y),(\llparenthesis u^\flat\rrparenthesis_{i+1},v))=0\}$; %\\[-25pt]

\end{itemize}

\item 
 $S\subseteq \mathbf{Cl}_{i+1}$ is determined by the $j$-th element of $\mathfrak{L}_{i+1}^+$ where %\\[-13pt] 
\begin{itemize}
\item 
\begin{equation}%\\[-30pt]
 j=\lfloor [x^\flat]_i/(k-1) \rfloor \mbox{ mod } cl_{i+1}^*;
\end{equation}
\item if $\!j\!<\!\gamma_{m-i-1}^*$ or $j\geq cl_{i+1}^*\!-\!\gamma_{m-i-1}^*$ then $S\!:=\!\emptyset$; otherwise, 
assume that this $j$-th element is a $d$-tuple (note that  $1\!\leq\! d\!\leq\! k\!-\!2$) $((u_1,\!v_1),\cdots,(u_d,\!v_d))\in(\mathbb{X}_{i+1}^*)^d$, then%\\[-10pt]
\begin{equation}
S:=\displaystyle\bigcup_{1\leq i\leq d}\{\mathbf{cl}(u_i,v_i) \}.
\end{equation} 

\end{itemize}

%\item $\Lambda\subseteq \{(row,index,rng) \mid 0<y<k-1\rightarrow row\in \{0,k-1\}; y=0\rightarrow row\in \{mid,k-1\}; y=k-1\rightarrow row\in \{0,mid\}; index>i; rng\in \{-1,0,\cdots,3\}\}$ is determined by the $\ell$-th element of $\mathfrak{R}_{xy}$ where\footnote{Here ``$mid$'' is just a symbol, corresponding to ``$0<y<k-1$''.}
%\begin{equation}
%\ell=\lfloor [x^\flat]_i/(\eta_i^*/2^{10(m-i)}) \rfloor \mbox{ mod } \mathfrak{R}_{xy}.
%\end{equation}

\end{itemize}
\hfill\ensuremath{\divideontimes}
\end{definition}

%\begin{remark}
%Note that $E^B$ will be defined in Definition \ref{iterative-expansion}.
%\end{remark}

%\begin{remark}
%Since $\frac{1}{2}\beta_{m-i-1}^{m-i}$ is divisible by $2(k-1)2^{k-3}$, $\lfloor \frac{[x]_i \mbox{ mod } \beta_{m-i-1}^{m-i}-\frac{1}{2}\beta_{m-i-1}^{m-i}}{2(k-1)}\rfloor$ can be replaced by $\lfloor \frac{[x]_i \mbox{ mod } \beta_{m-i-1}^{m-i}}{2(k-1)}\rfloor$ in the definition of  type label (cf.  ``$S_3$'').
%\end{remark}

%\begin{remark}
%In $\chi(x,y)\!\!\restriction\!\! S^{\prime\prime}$,  since $[x]_{i+1}\!\equiv\! 0 \!\mbox{ mod }\! k\!-\!1$, $\mathbf{cc}([u]_{i+1}\!,v)\neq\mathbf{cc}([x]_{i+1}\!,y)$ implies  $\mathbf{cc}([u]_{i+1}\!,v)\neq y$ mod $k-1$.
%\end{remark}

In Definition \ref{iterative-expansion} where $\mathfrak{B}_{k,m}$ is defined, we shall see that type labels determine whether a vertex having some type lable is \textit{NOT} adjacent to another vertex that has another type label.  
Henceforth, we use $\chi(x,y)\!\!\restriction\!\! S$  to denote $S$ in $\chi(x,y)$, and  similarly for $\chi(x,y)\!\!\restriction\!\! \Omega$,  $\chi(x,y)\!\!\restriction\!\!\mathrm{BH}$, and $\chi(x,y)\!\!\restriction\!\!\mathrm{bc}$. Note that, \eqref{def-value-board-histroy} says that, in a row of the structure, every successive sequence of vertices of length $\gamma_{m-1}^*\times k$ have been associated with the same board history and $j^\prime$ determines the board history associated with $(x,y)$; \eqref{def-length-history} tells us about the length of board history that is associated with $(x,y)$. Moreover, in Definition \ref{iterative-expansion} we shall see how edges are forbbiden based on $\chi(x,y)\!\!\restriction\!\! \Omega$,  $\chi(x,y)\!\!\restriction\!\! S$ and $\chi(x,y)\!\!\restriction\!\! \mathrm{BH}$.

\begin{definition}\label{congruence-label}
In the structure $\mathfrak{B}_{k,m}$, for any $(x,y)\in \mathbb{X}_1$ and any $1\leq i\leq m$, let $\mathrm{idx}(x^\flat,y)=i$.  The \textit{congruence label}  of $(x,y)$, denoted by
 $\mathbf{cl}(x,y)$, is defined as the following:%\\[-15pt]
\begin{equation}\label{upper-congruence-label}
\mathbf{cl}(x,y):=\underline{\mathbf{cc}([x^\flat]_i,y),y;i;\mathsf{R}_i;\chi(x,y)\!\!\restriction\!\! S} \mbox{ where }\mathsf{R}_i=\mathrm{RngNum}(x^\flat,i).
\end{equation}
%The \textit{congruence superscript}  $\!$of $(x,y)$ is $\chi(x,y)\!\!\restriction\!\! S$.
\hfill\ensuremath{\divideontimes}
\end{definition}
The readers can cf. Example \ref{example-congr-label} for an explanation on what a congruence label is meant in some structure (i.e. $\mathfrak{B}_{k,m}^*$) that will be defined soon. 

We use $(u,v)\twoheadrightarrow (x,y)$ to denote the formula $\psi$, where 
\begin{equation}
\psi=(u,v)\xrightarrow[\mathrm{BC}]{con.}(x,y)\land (u^\flat,v)\in (x,y)[\mathrm{BC}]
\end{equation}

Note that $\twoheadrightarrow$ is a strict preorder, i.e. a transitive relation. 

We use $(x_j,y_j)\rightsquigarrow (x_i,y_i)$ to denote the formula $\varphi$, where \\[-15pt]
\begin{equation}\label{varphi-for-rightsquiggarrow}
\varphi=(x_i,y_i)\twoheadrightarrow (x_j,y_j) 
\land \mathbf{cl}(x_i,y_i)\in\chi(x_j,y_j)\!\!\restriction\!\!S \land (x_j^\flat,y_j)\notin (x_i,y_i)[\mathrm{BC}]
\end{equation}

%Let $\mathbb{X}_1^\mathrm{r}:=[k]\times [\gamma_{m-1}]$; replacing $\gamma_{m-1}$ by $\gamma_{m-1}^*$ in $\mathbb{X}_1^\mathrm{r}$, we get $\mathbb{X}_1^\mathrm{*r}$.  

Now we introduce the pair of important structures, a sort of structures with (temporal) ``structures''.
\begin{definition} \label{iterative-expansion}
$\mathfrak{B}_{k,m}$  is a $\langle E,\leq\rangle$-\textbf{structure}, i.e. ordered graph, over the universe $\mathbb{X}_1$ and the linear order is defined as the follows. 
For any vertices $(x_i,y_i)$ and $(x_j,y_j)$ of $\mathfrak{B}_{k,m}$, $(x_i,y_i)<(x_j,y_j)$ if $y_i<y_j$. Suppose that $y_i=y_j$. 
For any vertex $(x,y)$ of $\mathfrak{B}_{k,m}$, the value $\lfloor x/(\gamma_{m-1}^*\times k)\rfloor$ mod $bh^\#$ is supposed to determine the board history (cf. Definition \ref{type-label}, (\ref{def-value-board-histroy})) associated with the vertex $(x,y)$.   %\label{page-def-J-lists-for-bh-sharp} 
Hence, we need an \textit{ordering of different board histories}.\label{page-history-order} 
 Here, we define the ordering as the following: firstly we can define a lexicographic ordering on the game configurations: $\mathrm{BC}_1$ is less than $\mathrm{BC}_2$ if $|\mathrm{BC}_1|<|\mathrm{BC}_2|$ (hence the empty game board is the minimal element in the order); now for any invalid board history $h_1$ and any valid one $h_2$, $h_1<h_2$;  
for valid board histories in $N_1\times \cdots\times N_m$ where $N_i$ stands for the set of game configurations, the linear order is defined by the lexicographic ordering on the Cartesian product $N_1\times \ldots\times N_m$. Note that we only take the \textit{actual} board histories into account. Hence a shorter board history is ahead of a longer board history.  

Then, on condition that $(x_i,y_i)$ and $(x_j,y_j)$ have different board histories, $(x_i,y_i)<(x_j,y_j)$ if $\chi(x_i,y_i)\!\!\restriction\!\!\mathrm{BH}<\chi(x_j,y_j)\!\!\restriction\!\!\mathrm{BH}$. 

Suppose that $\chi(x_i,y_i)\!\!\restriction\!\!\mathrm{BH}=\chi(x_j,y_j)\!\!\restriction\!\!\mathrm{BH}$,\footnote{In other words, $\lfloor x_i/(\gamma_{m-1}^*\times k)\rfloor \mbox{ mod } bh^\#=\lfloor x_j/(\gamma_{m-1}^*\times k)\rfloor \mbox{ mod } bh^\#$.} then $(x_i,y_i)\leq (x_j,y_j)$ if $x_i$ mod $\gamma_{m-1}^*\times k$ is no more than  $x_j$ mod $\gamma_{m-1}^*\times k$.  

The edge set $E^B$ of $\mathfrak{B}_{k,m}$ is defined as follows.  

For any $1\leq t\leq m$ and  for any vertices $(x_i,y_i)$ and $(x_j,y_j)$ 
\begin{enumerate}[1]
\item $(x_i,y_i)$ is not adjacent to $(x_j,y_j)$ if $y_i=y_j$.\\[5pt] \textit{In the following} assume that $y_i\neq y_j$. 

\item If $\mathrm{min}\{\mathrm{idx}(x_i,y_i), \mathrm{idx}(x_j,y_j)\}=t$, then $((x_i,y_i),(x_j,y_j))\in E^B$ if and only if all the following conditions hold: 
\begin{enumerate}
\item $((x_i^\flat,y_i)(x_j^\flat,y_j))\in E_*^B$ (cf. Definition \ref{B_km-star});

\item Either $(x_i,y_i)\twoheadrightarrow (x_j,y_j)$, or $(x_j,y_j)\twoheadrightarrow (x_i,y_i)$; 
%Either $(x_i,y_i)\xrightarrow[\mathrm{BC}]{con.} (x_j,y_j)\land(x_i^\flat,y_i)\in (x_j,y_j)[\mathrm{BC}]$, or \\$(x_j,y_j)\xrightarrow[\mathrm{BC}]{con.} (x_i,y_i)\land(x_j^\flat,y_j)\in (x_i,y_i)[\mathrm{BC}]$; 

\item $(x_j,y_j)\!\notin\! \{(u,v)\!\in\! \mathbb{X}_1 \1  (x_i,y_i)\rightsquigarrow (u,v)\}$, and\\ $(x_j,y_j)\!\notin\! \{(u,v)\!\in\! \mathbb{X}_1 \1 (u,v)\rightsquigarrow (x_i,y_i)\}$;

\item $(x_j^\flat,y_j)\notin \chi(x_i,y_i)\!\!\restriction\!\! \Omega$ and $(x_i^\flat,y_i)\notin \chi(x_j,y_j)\!\!\restriction\!\! \Omega$;  

%\item $(x_j^\flat,y_j)\notin \{(u,v)\in \mathbb{X}_1^* \mid (v,\mathrm{idx}(u,v),\mathrm{RngNum}(u,\mathrm{idx}(u,v)))\in \chi(x_i,y_i)\!\!\restriction\!\! \Lambda\}$, and\\ $(x_i^\flat,y_i)\notin \{(u,v)\in \mathbb{X}_1^* \mid (v,\mathrm{idx}(u,v),\mathrm{RngNum}(u,\mathrm{idx}(u,v)))\in \chi(x_j,y_j)\!\!\restriction\!\! \Lambda\}$;

\item \underline{If either $y_i\in \{0,k-1\}$ or $y_j\in \{0,k-1\}$}: %\\[-18pt]
\begin{equation*} %\label{structure-down-neq}
%E^B(x_i,y_i)(x_j,y_j) \Leftrightarrow\\ 
\mathbf{cc}([x_i^\flat]_{t},y_i)\neq \mathbf{cc}([x_j^\flat]_{t},y_j) \mbox{ and }
    \mathrm{sgn}((x_i^\flat,y_i),(x_j^\flat,y_j))=0;
\end{equation*}
\underline{If $y_i,y_j\in [1,k-2]$}:
\begin{itemize}
\item If $\mathrm{idx}(x_i^\flat,y_i)\neq\mathrm{idx}(x_j^\flat,y_j)$: 

\indent $\hspace{10pt}\mathbf{cc}([x_i^\flat]_{t},y_i)\neq \mathbf{cc}([x_j^\flat]_{t},y_j)$;

\item If $\mathrm{idx}(x_i^\flat,y_i)=\mathrm{idx}(x_j^\flat,y_j)=t$:\\[-18pt]
 \begin{multline}%\label{structure-down-leq}
(\mathbf{cc}([x_i^\flat]_{t},y_i)-\mathbf{cc}([x_j^\flat]_{t},y_j))\times(y_i-y_j)\\ \hspace{15pt} \times(-1)^{\mathbf{BIT}(\mathrm{SW}((x_i^\flat,y_i),(x_j^\flat,y_j)),\hat{q}(y_i,y_j))}\!>\!0. \nonumber  
\end{multline}

\end{itemize}

\end{enumerate}

\end{enumerate}

$\mathfrak{A}_{k,m}$ is constructed from $\mathfrak{B}_{k,m}$
 by adding  a set $E^+$ of edges where  
$E^+=\{((x_i,y_i),(x_j,y_j))\mid (x_i^\flat,y_i),(x_j^\flat,y_j)\!\in\! \mathbb{X}_m^*; y_i,y_j=0\mbox{ or } k-1; y_i\neq y_j; [x_i^\flat]_m\equiv [x_j^\flat]_m\equiv 0 \hspace{3pt}(\mbox{mod }k-1); (x_i,y_i)\twoheadrightarrow (x_j,y_j)\mbox{ or } (x_j,y_j)\twoheadrightarrow (x_i,y_i)\}$.% We call endpoints of such edges \textit{critical points}.%\\[-18pt]
%, and call a critical point $c_+$ if its second coordinate is 0; otherwise call it $c^+$.
%between $(\frac{1}{2}\gamma_{m-1},0)$ and $(\frac{1}{2}\gamma_{m-1},k-1)$.
\hfill\ensuremath{\divideontimes}
\end{definition}

%Let $E_*$ be the set of edges of either $\mathfrak{A}_{k,m}^*$ or  $\mathfrak{B}_{k,m}^*$.
As the way we define type label, congruence label and $\mathfrak{B}_{k,m}$ over the universe $\mathbb{X}_1$, we  define the dual notions over $\mathbb{X}_1^*$ \textit{simultaneously}. 
\begin{definition} \label{type-label-star}
In the structure $\mathfrak{B}_{k,m}^*$, for any $(x,y)\in \mathbb{X}_1^*$, the \textit{congruence label} of $(x,y)$, also denoted $\mathbf{cl}(x,y)$, is defined as the dual one in Definition \ref{congruence-label}, except that ``$\flat$''s are removed. 
%: $\mathbf{cl}(x,y)=\mathbf{cl}(x^\flat,y)$ for any $(x,y)\in \mathbb{X}_1$. 
The readers are suggested to confer Example \ref{example-congr-label} in the appendix. 

For any $(x,y)\in \mathbb{X}_1^*$, the\textit{ type label} of $(x,y)$, also denoted $\chi(x,y)$, is defined as follows.  

Assume that $\mathrm{idx}(x,y)=i$. Let $\chi(x,y):=\underline{\underline{\mathfrak{a},y;i;S};\Omega}$ where\footnote{Here $\underline{\mathfrak{a},y;i;S}$ represents a congruence label. Also confer Remark \ref{remark-congruence-labels}.} \\[-14pt]
\begin{itemize}
\item 
$\mathfrak{a}=\mathbf{cc}([x]_i,y)$;
\item If $i=m$  then $S=\Omega=\emptyset$; otherwise (i.e. $1\leq i<m$), 
\begin{itemize}
\item  if  $y\in\{0,k-1\}$ then  $\Omega=\emptyset$; otherwise, $\Omega\!=\!\{(u,v)\!\in\!\mathbb{X}_i^*\1 [x]_i\!\equiv\! [u]_i\!\equiv\! 0 \mbox{ mod } k-1; \left((\llparenthesis u\rrparenthesis_{i+1},v),(\llparenthesis x\rrparenthesis_{i+1},y)\right)\notin E_*; v\in \{0,k-1\}\rightarrow  (u,v)\in\mathbb{X}_{i+1}^*\land\mathrm{sgn}((\llparenthesis x\rrparenthesis_{i+1},y),(\llparenthesis u\rrparenthesis_{i+1},v))=0\};$ %\\[-25pt]

\item 
 $S\!\subseteq\! \mathbf{Cl}_{i+1}$ is determined by the $j$-th element of $\mathfrak{L}_{i+1}^{*+}$ where
 $j\!=\!\lfloor [x]_i/(k-1)\rfloor \!\mbox{ mod }\! cl_{i+1}^*$; if $j<\gamma_{m-i-1}^*$  or $j\geq cl_{i+1}^*-\gamma_{m-i-1}^*$ then $S:=\emptyset$; otherwise, 
 assume that this $j$-th element is $((u_1,\!v_1),\cdots,\\(u_d,\!v_d))\!\in\!(\mathbb{X}_{i+1}^*)^d$, then 
$S:=\bigcup_{1\leq i\leq d} \{\mathbf{cl}(u_i,v_i)\}$.

%\item  $\Lambda\subseteq \{(row,index,rng) \mid 0<y<k-1\rightarrow row\in \{0,k-1\}; y=0\rightarrow row\in \{mid,k-1\}; y=k-1\rightarrow row\in \{0,mid\}; index>i; rng\in \{-1,0,\cdots,3\}\}$ is determined by the $\ell$-th element of $\mathfrak{R}_{xy}$ where$\ell=\lfloor [x]_i/\eta_i^* \rfloor \mbox{ mod } \mathfrak{R}_{xy}$.

\end{itemize}

\end{itemize}

\hfill\ensuremath{\divideontimes}
\end{definition}

 Note that, $\mathbf{cl}(u,v)=\mathbf{cl}(u^\flat,v)$ and  
 $\chi(x,y)\!\!\restriction\!\! S=\chi(x^\flat,y)\!\!\restriction\!\! S$ for a vertex $(x,y)\in \mathbb{X}_1$. In fact, the two version of ``type label'' introduced in Definition \ref{type-label} and Definition \ref{type-label-star} are also very similar except that the former one also take into account the board history associating with a vertex. %(recall that we can regard $\mathbb{X}_1$ as an expansion of $\mathbb{X}_1^*$ such that every vertex in $\mathbb{X}_1^*$ is expanded by a factor $m\times bh^\#$).   

\begin{definition}\label{B_km-star}
$\mathfrak{B}_{k,m}^*$ is defined similarly as $\mathfrak{B}_{k,m}$, except that
\begin{itemize}
\item $\mathbb{X}_i$, for any $i$, is replaced by $\mathbb{X}_i^*$; 
%\item remove 2) a) in Definition \ref{iterative-expansion};
\item remove 2 (a), (b) in Definition \ref{iterative-expansion}; 

\item revise 2 (c) as the following:\\ $(x_j,y_j)\!\notin\! \{(u,v)\!\in\! \mathbb{X}_1^* \1  \mathbf{cl}(u,v)\!\in\! \chi(x_i,y_i)\!\!\restriction\!\!S\}$;\\ $(x_j,y_j)\!\notin\! \{(u,v)\!\in\! \mathbb{X}_1^* \1 \mathbf{cl}(x_i,y_i)\!\in\! \chi(u,v)\!\!\restriction \!\!S\}$;

\item remove all the superscripts ``$\flat$'' in Definition \ref{iterative-expansion}.
\end{itemize}
$\mathfrak{A}_{k,m}^*$ is constructed from $\mathfrak{B}_{k,m}^*$ 
 by adding  a set $E_*^+$ of edges where  
$E_*^+=\{((x_i,y_i),(x_j,y_j))\mid (x_i,y_i),(x_j,y_j)\in \mathbb{X}_m^*; y_i,y_j=0\mbox{ or } k-1; y_i\!\neq\! y_j; [x_i]_m\!\equiv\! [x_j]_m\equiv 0 \hspace{3pt}(\mbox{mod } k-1)\}$. We call endpoints of such edges \textit{critical points}.

For each $t$ in $[1,m]$ we define a function $\gimel_t^b: |\mathfrak{B}_{k,m}^*|\mapsto \mathbb{BC}$, whose concrete definition will be decided later (cf. page \pageref{def-virtual-game}). In addition,   
we define $\mathfrak{B}_{k,m}^{*\gimel_t^b}$ similarly as $\mathfrak{B}_{k,m}^*$, except that we revise 2) c) as the following:\\
$(x_j,y_j)\notin \{(u,v)\!\in\! \mathbb{X}_1^* \1 (u,v)\in \gimel_t^b(x_i,y_i) \land \mathbf{cl}(u,v)\in \chi(x_i,y_i)\!\!\restriction\!\!S\}$,  and\\ $\indent\hspace{20pt}(x_j,y_j)\notin \{(u,v)\!\in\! \mathbb{X}_1^* \1 (x_i,y_i)\in \gimel_t^b(u,v)\land \mathbf{cl}(x_i,y_i)\in \chi(u,v)\!\!\restriction\!\!S\}$. 

%Assume that the edge set of $\mathfrak{B}_{k,m}^{*\gimel_t^b}$  is $E_{*,B}^{\gimel_t^b}$.

Similarly, we can define $\gimel_t^a$ and $\mathfrak{A}_{k,m}^{*\gimel_t^a}$.
% and $E_{*,A}^{\gimel_t^a}$.
%$\mathfrak{A}_{k,m}^{*\mathbb{Z}}$ is constructed from $\mathfrak{B}_{k,m}^{*\mathbb{Z}}$ as $\mathfrak{A}_{k,m}^*$ is from $\mathfrak{B}_{k,m}^*$.  
\hfill\ensuremath{\divideontimes}
\end{definition}
 
 Note that $\gimel_t^b(x,y)$ is different from $(x,y)[\mathrm{BC}]$. The latter is a notion defined w.r.t. $\mathfrak{A}_{k,m}$ or $\mathfrak{B}_{k,m}$. 

We shall find in section \ref{winning-strategy} that $\mathfrak{A}_{k,m}^*$ and $\mathfrak{B}_{k,m}^*$, as well as $\mathfrak{A}_{k,m}^{*\gimel_t^a}$ and $\mathfrak{B}_{k,m}^{*\gimel_t^b}$, are the structures that we will study in the pebble games, which are the ``associated'' structures of $\mathfrak{A}_{k,m}$ and $\mathfrak{B}_{k,m}$. In fact, it is a good idea to understand the definition of them before reading Definition \ref{iterative-expansion}. 

%Note that vertices that have the same type label do not necessarily have the same isomorphism type, though they do if under some conditions. Such relaxation make it easier to handle the issues, not only for simpler proofs but also for simplify structure constructions.

%We use ``$\mathrm{nil}$'' to denote a tower of $(\lbag \emptyset\rbag,i)$. For example, we use $(\lbag 0\rbag,i_0)^{\mathrm{nil}}$ to denote $(\lbag 0\rbag,i_0)^{(\lbag\emptyset\rbag,i_1)^{(\lbag\emptyset\rbag,i_2)}}$ for any $i_0,i_1$ and $i_2$.

%For any $w\in \mathbf{Cl}_i$, where $w=u^M$, let  $\mathbf{HC}_i(w):=w$ if $u\neq \lbag \emptyset\rbag$; $\mathbf{HC}_j(w)$ is undefined if $j<i$; $\mathbf{HC}_{i+1}(w):=\bigcup_{w_i\in M} \mathbf{HC}_{i+1} w_i$. Moreover, if for some $j\geq i$ $\mathbf{HC}_j(w)\neq \{\mathrm{nil}\}$ and $T=\{\mathrm{nil}\}$ for any element $v^T$ of $\mathbf{HC}_j(w)$, we call $\mathbf{HC}_j(w)$ the ``ceiling'' of $w$ and the value of $j$ the ``height of ceiling'' of $w$. Furthermore, assume that $\mathbf{cl}(x,y)=w$, then we call $\{(a,b)\mid \mathbf{cl}(a,b)\in\mathbf{HC}_j(w)\}-\{\mathrm{nil}\}$ the \textbf{influence boundary} of $(x,y)$.   

\begin{definition}\label{structural-abstraction}
For any $1\leq t\leq m$, the $t$-th \textbf{abstraction} of $\mathfrak{B}_{k,m}$ ($\mathfrak{A}_{k,m}$ resp.), denoted   $\mathfrak{B}_{k,m}^{(t)}$ ( $\mathfrak{A}_{k,m}^{(t)}$ resp.), is $\mathfrak{B}_{k,m}[\mathbb{X}_{t}]$ ($\mathfrak{A}_{k,m}[\mathbb{X}_{t}]$ resp.).
\hfill\ensuremath{\divideontimes}
\end{definition}

Note that the graph $\mathfrak{A}_{k,m}$ and  $\mathfrak{B}_{k,m}$ are just   $\mathfrak{A}_{k,m}^{(1)}$ and $\mathfrak{B}_{k,m}^{(1)}$ respectively. 
We use $\mathbb{X}_{t}^A$, $E^A$ to denote the set $\mathbb{X}_{t}$, $E$ of $\mathfrak{A}_{k,m}$, and use $\mathbb{X}_{t}$, $E$ to denote the set $\mathbb{X}_{t}$, $E$ of either $\mathfrak{A}_{k,m}$ or $\mathfrak{B}_{k,m}$. \textit{The abstractions of  $\mathfrak{A}_{k,m}^*$ and  $\mathfrak{B}_{k,m}^*$ can be defined similarly}. We call $\mathbf{cc}([x]_t,y)$ the ``coordinate congruence number in the $t$-th abstraction'', for any $(x,y)\!\in\!\mathbb{X}_1$. 

We get a pair of structures $\widetilde{\mathfrak{A}}_{k,m}^*$ and $\widetilde{\mathfrak{B}}_{k,m}^*$ from $\mathfrak{A}_{k,m}^*$ and $\mathfrak{B}_{k,m}^*$, just like the way we obtain $\widetilde{\mathfrak{A}}_{3,m}$ and $\widetilde{\mathfrak{B}}_{3,m}$ from $\mathfrak{A}_{3,m}$ and $\mathfrak{B}_{3,m}$, by circular shifting the vertices of the $i$-th row for $tr(i)$ times to the right (cf. p. \pageref{def-circular-shifting}).  Similarly, we can define $\widetilde{\mathfrak{B}}_{k,m}^{*\gimel_t^b}$ etc.   
Recall that, for $i\in [k]$,  
$tr(i)=(i\,\, \mbox{mod } k-1)\times\sum_{1\leq p\leq m}\beta_{m-p}^{m-1}$. 
We obtain a pair of structures $\widetilde{\mathfrak{A}}_{k,m}$ and $\widetilde{\mathfrak{B}}_{k,m}$ from $\mathfrak{A}_{k,m}$ and $\mathfrak{B}_{k,m}$, by moving each vertex $(x,y)$ to $(x^\prime,y)$ where $x^\prime=(x+tr(y))\mbox{ mod }(\gamma_{m-1}^*\times k)$. They are the pair of main structures that form the game board. 
We shall meet these newly created structures in the next section, i.e. Section \ref{winning-strategy}. 
Abuse of denotations, we still call the edge sets of $\widetilde{\mathfrak{B}}_{k,m}^*$ and $\widetilde{\mathfrak{B}}_{k,m}$ as $E_*^B$ and $E^B$ respectively.  

Recall that we regard every $\mathpzc{U}_{i}^*$-tuple of vertices as one object in $\mathfrak{A}_{k,m}^*$, and such an object includes all types of vertices. 
%We can also regard  every successive squence of $\mathpzc{U}_{m-i}$ vertices as one object in a row of the $i$-th abstraction of $\mathfrak{A}_{k,m}$. 

%The following observation is straightforward according to Definition \ref{iterative-expansion}. 
\begin{lemma}\label{universal-simulator}
Let $1\leq r<i\leq m$. For any $(e,f)\in \mathbb{X}_r^*$,  $\mathfrak{a}\!\in\! [k-1]$, $\ell\in\{-1,0,1\}$, 
and $w\!\in\! \wp(\mathbf{Cl}_{i})$, there is  $(e^{\prime},f)\in\mathbb{X}_r^*$ in the $\mathpzc{U}_{r}^*$-tuple $(\llbracket e\rrbracket_{r},f)$ such that 
$$\mathbf{cl}(e^{\prime},f)=\underline{\mathfrak{a},f;r;\ell;w}.$$ 
\end{lemma}
We shall see in the Main Lemma \ref{main-lemma} (cf. Strategy 2, (2-5)) that Lemma \ref{universal-simulator} gives Duplicator the freedom to ensure her picked vertex satisfying some conditions: Lemma \ref{universal-simulator} allows Duplicator to choose $\mathbf{cl}(e^{\prime},f)$ freely. Lemma \ref{universal-simulator} is also used in the proof of Lemma \ref{no-missing-edges_xi-1}.

%==============================================================================

It is clear that $\mathfrak{A}_{k,m}$ contains $k$-cliques. In particular, 
the following lemma says that there is a $k$-clique roughly in the middle of the structure with respect to the first coordinate.
\begin{lemma}\label{A-km-has-clique}
The subgraph of $\mathfrak{A}_{k,m}$ induced by a set of vertices $(x_0,0),\ldots,$\\$(x_{k-1},k-1)$ is a $k$-clique, where 
$x_i^\flat=\frac{1}{2}\gamma_{m-1}^*+\frac{1}{2}\sum_{1<j\leq m}\beta_{m-j}^{m-1}$ for any $i$, and $(x_0,0)\twoheadrightarrow (x_1,1)\twoheadrightarrow\cdots
\twoheadrightarrow (x_{k-1},k-1)$. 
\end{lemma}
\begin{proof}
First, we show that, for any $i$, $(x_i^\flat,i)\in\mathbb{X}_m^*$ and $\mathbf{cc}([x_i^\flat]_m,i)=i$ mod $k-1$. 

It is because $\llparenthesis x_i^\flat\rrparenthesis_m=\lfloor(\frac{1}{2}\gamma_{m-1}^*+\frac{1}{2}\sum_{1<j\leq m}\beta_{m-j}^{m-1})/\beta_0^{m-1}\rfloor\beta_0^{m-1}+\\\frac{1}{2}\sum_{1<j\leq m}\beta_{m-j}^{m-1}=\lfloor\frac{1}{2}\gamma_{m-1}^*/\beta_0^{m-1}\rfloor\beta_0^{m-1}+\frac{1}{2}\sum_{1<j\leq m}\beta_{m-j}^{m-1}=x_i^\flat$. The last equation is based on the easy observation that $\frac{1}{2}\gamma_{m-1}^*$, which equals $\frac{1}{2}\gamma_0^*\beta_0^{m-1}$, is divisible by $\beta_0^{m-1}$, since $\gamma_0^*$ is even. Therefore, $(x_i^\flat,i)\in\mathbb{X}_m^*$. 
  Moreover, $[x_i^\flat]_m=\frac{1}{2}\gamma_0$, which is divisible by $k-1$. Hence, $\mathbf{cc}([x_i^\flat]_m,i)=i$ mod $k-1$. 

Second, $g(x_i)=0$, for $[x_i^\flat]_{m}\mbox{ mod }\mathpzc{U}_{m}^*=\frac{1}{2}\gamma_0 \mbox{ mod }\mathpzc{U}_{m}^*=2^{m-1}\mathpzc{U}_m^*$  mod $\mathpzc{U}_{m}^*$ $=0<2^{\binom{k-2}{2}}\times\eta_m^*$. Therefore, for any $i,j$, $\mathrm{SW}((x_i^\flat,i),(x_j^\flat,j))=0$. If $i,j\in [1,k-2]$, then the condition  e) of 2) in Definition \ref{iterative-expansion} clearly holds. Assume that either $i$ or $j$ is either $0$ or $k-1$. In this case, we need only notice that $\mathrm{RngNum}(x_i^\flat,m)=\mathrm{RngNum}(x_j^\flat,m)=-1$ because $[x_i^\flat]_{m}\mbox{ mod }\mathpzc{U}_{m}^*=[x_j^\flat]_{m}\mbox{ mod }\mathpzc{U}_{m}^*=0<2^{\binom{k-2}{2}}\times\eta_m^*$.  

Third, for any $i,j$, $\mathbf{cl}(x_i^\flat,i)\!\notin\! (x_j^\flat,j)\!\!\restriction\!\!S$ and $(x_i^\flat,i)\!\notin\! (x_j^\flat,j)\!\!\restriction\!\!\Omega$, since we've already proved that $(x_i^\flat,i),(x_j^\flat,j)\in\mathbb{X}_m^*$ for any $i, j$. Note that $\mathbf{cl}(x_i^\flat,i)\notin (x_j^\flat,j)\!\!\restriction\!\!S$  means that the formula $(x_j^\flat,j)\rightsquigarrow (x_i^\flat,j)$ does not hold. 

Fourth, the binary relation ``$\twoheadrightarrow$'' is transitive. Therefore, $(x_0,0)\twoheadrightarrow(x_1,1)\twoheadrightarrow\cdots
\twoheadrightarrow (x_{k-1},k-1)$ implies that $(x_i,i)\twoheadrightarrow (x_j,j)$ for any $0\leq i<j\leq k-1$. 

Then by Definition \ref{B_km-star}, the set of vertices $(x_0^\flat,0),\ldots,(x_{k-1}^\flat,k-1)$ is a $k$-clique in $\mathfrak{A}_{k,m}^*$. As a consequence, by Definition \ref{iterative-expansion}, the set of vertices $(x_0,0),\ldots,(x_{k-1},k-1)$ is also a $k$-clique in $\mathfrak{A}_{k,m}$, since $(x_i,i)\twoheadrightarrow (x_j,j)$ for any $0\leq i<j\leq k-1$. 
\end{proof}
Note that $(x_0,0)\twoheadrightarrow(x_1,1)\twoheadrightarrow\cdots
\twoheadrightarrow (x_{k-1},k-1)$ implies that $\chi(x_0,0)\!\!\restriction\!\!\mathrm{BH}$ is a void board history, and $\chi(x_i,i)\!\!\restriction\!\!\mathrm{BH}$ consists of $i$ nonempty board configurations where $\chi(x_i,i)\!\!\restriction\!\!\mathrm{BH}(j)=\chi(x_i,i)\!\!\restriction\!\!\mathrm{BH}(j-1)\ccirc (x_{j-1},j-1)$ for $0<j\leq i$.  

Lemma \ref{A-km-has-clique} implies that $\widetilde{\mathfrak{A}}_{k,m}$ has $k$-cliques. 

 Although it is relatively easy to see that $\mathfrak{A}_{k,m}$ has $k$-cliques, it is not obvious whether $\mathfrak{B}_{k,m}$ has a $k$-clique or not. The following lemma answers this question.

\begin{lemma}\label{B_k-has-no-k-clique}
$\mathfrak{B}_{k,m}$ has no $k$-clique. 
\end{lemma}
\begin{proof}
Suppose that $\mathfrak{B}_{k,m}^*$ has no $k$-clique, then $\mathfrak{B}_{k,m}$ also has no $k$-clique, by virtue of 2) a) in Definition \ref{iterative-expansion}. Therefore, we need only prove that $\mathfrak{B}_{k,m}^*$ has no $k$-clique. 

Assume for the purpose of a contradiction that there are $k$-cliques in $\mathfrak{B}_{k,m}^*$ and $C_k$ is such a $k$-clique that has the maximum index, say $t$, among all the $k$-cliques. $t$ cannot be $m$, for otherwise there are two vertices that have the same coordinate congruence number in the $m$-th abstraction, 
by the pigeonhole principle. 
According to the definition, the second coordinates of the vertices of $C_k$ must be different. That is, for each $i\in [0,k-1]$, there is a unique vertex whose second coordinate is $i$.  Let $P=\{(x,y)\!\in\! |C_k|\mid (x,y)\!\in\! \mathbb{X}_{t}^*\!-\!\mathbb{X}_{t+1}^*\}$. And  let 
 $Q=\{(x,y)\in |C_k| \mid (x,y)\in \mathbb{X}_{t+1}^*\}$. Note that, $P\cap Q=\emptyset$,  %By Lemma \ref{i=0theni-1=0}, 
and the set of vertices of $C_k$ is exactly $P\cup Q$.

Let $cC_k\!:=\!\{\mathbf{cc}([x]_t,y)\1 (x,y)\!\in\! |C_k|\}$. 
Since there are $k$ elements in $C_k$ and $|cC_k|\leq k-1$, by pigeonhole principle, there are two vertices $(a^{\star},b^{\star}),(c^{\star},d^{\star})$ such that $\mathbf{cc}([a^{\star}]_t,b^{\star})\!=\!\mathbf{cc}([c^{\star}]_t,d^{\star})$. If $(a^{\star},b^{\star})\!\in\! P$ or $(c^{\star},d^{\star})\!\in\! P$, then by Definition \ref{iterative-expansion}, there is no edge between these two vertices. Therefore, to have a $k$-clique, both $(a^{\star},b^{\star})$ and $(c^{\star},d^{\star})$ should be in $Q$. Recall Lemma \ref{cm=depth}, for any $(x,y)\in Q$, $\mathbf{cc}([x]_t,y)=y$ mod $k-1$. Therefore, $\mathbf{cc}([a^{\star}]_t,b^{\star})=\mathbf{cc}([c^{\star}]_t,d^{\star})=0$. In other words, $\{b^{\star},d^{\star}\}=\{0,k-1\}$. 
Assume without loss of generality that
\begin{equation}\label{no-k-clique-eqn1}
b^{\star}=0 \mbox{ and } d^{\star}=k-1.
\end{equation}  
There are three cases need to consider. 
\begin{enumerate}[(1)]
\item $Q=\emptyset$: As have just explained, we have $(a^{\star},b^{\star})$ in Q. Hence  a contradiction occurs.
\item $Q\neq\emptyset$, and for any vertex $(x,y)$ of $C_k$, $[x]_t\equiv 0$ (mod $k-1$):\\
%By Lemma \ref{cm=depth}, for any $(u,v)\in Q$, $\mathbf{cc}([u]_t,v)=v$ mod $k-1$. 
Let $P^{\prime}\!=\!\{(u,\!v)\!\in\!\mathbb{X}_{t+1}^*\1 \exists (u^{\prime},\!v)\!\in\! P \mbox{ s.t. }  [u^{\prime}]_{t+1}\!=\![u]_{t+1}\}$.
 % Then for any $i\in [0,k-1]$, there is a unique vertex $(x,y)\in P^{\prime}\cup Q$ such that $y=i$. In other words, 
The second coordinates of the $k$ vertices of $P^{\prime}\cup Q$, all of which are in $\mathbb{X}_{t+1}^*$, are different. Since $C_k$ is the $k$-clique that has the maximum index, hence  $\mathfrak{B}_{k,m}[P^{\prime}\!\cup\! Q]$ cannot be a $k$-clique and there are two vertices, say $(a,b)$ and $(c,d)$, of $P^{\prime}\cup Q$ such that $((a,b)(c,d))\notin E_*^B$. 
\begin{enumerate}[i)]
\item $(a,b)$, $(c,d)$ are vertices of $C_k$: Straightforward contradiction.
\item $(a,b)\in |C_k|$ while $(c,d)\notin |C_k|$ (the case when $(c,d)\in |C_k|$ while $(a,b)\notin |C_k|$ is symmetric): Because $(c,d)\notin |C_k|$, $(c,d)\notin Q$. 
Then by (\ref{no-k-clique-eqn1}), $d\in [1,k-2]$.  
Let $(c^{\prime},d)\in P$ where $[c^{\prime}]_{t+1}=[c]_{t+1}$. By Lemma \ref{projection}, $c=\llparenthesis c^{\prime}\rrparenthesis_{t+1}$. 
%Note that, if $\mathrm{RngNum}(c^\prime,t)\neq -1$, then either $(a^\star,b^\star)$ or $(c^\star,d^\star)$ is not adjacent to $(c^\prime,d)$. Hence, we can assume that $\mathrm{RngNum}(c^\prime,t)=-1$. In addition, by Lemma \ref{rngnum-is_-1}, $\mathrm{RngNum}(a,t)=-1$. Therefore, $\mathrm{sgn}((),())=0$. 
By Lemma \ref{cm=depth}, $[a]_t\equiv [c]_t\equiv 0$ (mod $k-1$). 
Note that $\llparenthesis a\rrparenthesis_{t+1}=a$. 
Then by the definition of $\restriction\!\! \Omega$ and Definition \ref{iterative-expansion},  $(a,b)\in\chi(c^\prime,d)\!\!\restriction\!\!\Omega$ (cf. Definition \ref{type-label}), which means that $(a,b)$ is not adjacent to $(c^\prime,d)$. A contradiction occurs. 

\item $(a,b),(c,d)\notin |C_k|$: By (\ref{no-k-clique-eqn1}), we have $b,d\in [1,k-2]$.  
Let $(a^{\prime},b),(c^{\prime},d)\in P$ such that $[a^{\prime}]_{t+1}=[a]_{t+1}$ and $[c^{\prime}]_{t+1}=[c]_{t+1}$. 
%Firstly, if either $\mathrm{RngNum}(a^\prime,t)\neq -1$ or $\mathrm{RngNum}(c^\prime,t)\neq -1$, then either $(a^\star,b^\star)$ or $(c^\star,d^\star)$ is not adjacent to $(a^\prime,d)$ or $(c^\prime,d)$. Hence, we can assume that $\mathrm{RngNum}(a^\prime,t)=\mathrm{RngNum}(c^\prime,t)=-1$. Therefore, 
By definition, either $(a^\prime,b)\in \chi(c^\prime,d)\!\!\restriction\!\!\Omega$ or $(c^\prime,d)\in \chi(a^\prime,b)\!\!\restriction\!\!\Omega$. In other words, $(a^\prime,b)$ is not adjacent to $(c^\prime,d)$. 
 we arrive at a contradiction.

\end{enumerate} 

\item $Q\neq\emptyset$ and there exists a vertex $(x,y)\in P$ such that $[x]_t\not\equiv$ 0 mod $k-1$:

Recall (\ref{no-k-clique-eqn1}) that if  $\mathbf{cc}([a^{\star}]_t,b^{\star})=\mathbf{cc}([c^{\star}]_t,d^{\star})$ then $(a^{\star},b^{\star}),(c^{\star},d^{\star})$ are in $Q$ and  $\{b^{\star},d^{\star}\}=\{0,k-1\}$. Hence $0<y<k-1$.   
Now consider the vertices in $P$. Their second coordinates are in the range $[1,k-2]$.  
Imagine that we have a sequence of slots numbered by $0,1,\ldots,k-1$, some of which are already occupied with billiards balls, i.e. a ball with a number $i$ is filled in the $i$-th slot. In particular, the 0-th and $(k-1)$-th slots are filled. And we want to fill the left slots with balls in the same way, i.e. we want to fill the $i$-th empty slot with a ball labelled with $i$. If we put a ball to a slot in some wrong way, then we can find two slots whose balls are in disorder: there are $l_1,l_2,s_1,s_2\in [1,k-2]$ such that $(l_1-l_2)(s_1-s_2)<0$ and a ball with label $l_1$ is filled in the $s_1$-th slot and a ball with label $l_2$ is filled in the $s_2$-th slot. In our context,  a vertex $(u,v)\in P$ is a ``ball'', and the number $\mathbf{cc}([u]_t,v)$ is the label on it. The $i$-th row of $\mathfrak{B}_{k,m}$ is the $i$-th slot. Since there is a vertex $(x,y)$ ($0<y<k-1$) which is not in the right ``slot'', i.e. $[x]_t\not\equiv 0$ (mod $k-1$) (hence $(x,y)$ must be a vertex in $P$), the sequence of the vertices of $P$ is in disorder. It means that there is another vertex $(x^{\prime},y^{\prime})\in |C_k|$ ($0<y^{\prime}<k-1$) such that $[x^{\prime}]_t\not\equiv 0$ (mod $k-1$) (hence $(x^{\prime},y^{\prime})$ must be a vertex in $P$, and  $(x^{\prime},y^{\prime})$ is also not in the right ``slot'') and $(\mathbf{cc}([x]_t,y)-\mathbf{cc}([x^{\prime}]_t,y^{\prime}))(y-y^{\prime})<0$. 
If $\mathrm{SW}((x,\!y),(x^{\prime},\!y^{\prime}))=0$ then by Definition \ref{iterative-expansion} there is no edge between $(x,y)$ and $(x^{\prime},y^{\prime})$. So we arrive at a contradiction to the assumption that $C_k$ is a clique. Hence we assume that $\mathrm{SW}((x,y),(x^{\prime},y^{\prime}))\!\neq\! 0$, which implies that either $\mathrm{RngNum}(x,t)\!\neq\! -1$ or $\mathrm{RngNum}(x^{\prime},t)\!\neq\! -1$. Assume without loss of generality that $\mathrm{RngNum}(x,t)\neq -1$.  
%By Lemma \ref{rngnum-is_-1},   $\mathrm{RngNum}(a^{\star},t)=\mathrm{RngNum}(c^{\star},t)=-1$.  
Therefore, either $(a^{\star},b^{\star})$ or $(c^{\star},d^{\star})$ is not adjacent to $(x,y)$, by definition.   
\end{enumerate}
\end{proof}

Lemma \ref{B_k-has-no-k-clique} implies that $\widetilde{\mathfrak{B}}_{k,m}$ has no $k$-clique. 

Note that, in the proof we show that $\mathfrak{B}_{k,m}^*$ contains no $k$-clique even if we do not consider the missing of edges defined in  2) c) of Definition \ref{B_km-star}, which is important for the following observation.%\\[-15pt]

\begin{corollary}
For any $t$ in $[1,m]$ and any $\gimel_t^b$, $\mathfrak{B}_{k,m}^{*\gimel_t^b}$ has no $k$-clique.
\end{corollary}

\section{$k$-Clique needs $k$ variables in $\fo$: virtual games and associated games over changing board} \label{winning-strategy}

In this section we introduce our main result, i.e. Duplicator has a winning strategy in the $(k-1)$-pebble  game over the game board $(\widetilde{\mathfrak{A}}_{k,m},\widetilde{\mathfrak{B}}_{k,m})$ (cf. Lemma \ref{main-lemma}). As have explained, most of the time we can study the game board $(\mathfrak{A}_{k,m},\mathfrak{B}_{k,m})$ instead of $(\widetilde{\mathfrak{A}}_{k,m},\widetilde{\mathfrak{B}}_{k,m})$.

We are able to prove the following lemma, which will be used shortly to prove the next crucial observation, i.e. Lemma \ref{no-missing-edges_xi-1}.
%, whose essence has been introduced in section \ref{existential-case-section}. 

%Let $\mathcal{O}_1^*$ be a coloured linear order where $|\mathcal{O}_1^*|=\alpha$ and all its elements have white colour except that the middle one, i.e. the $(2^m+1)$-th element, is a black element. Let $\mathcal{O}_2^*$ be a coloured linear order where $|\mathcal{O}_2^*|=\alpha$ and  all its elements are white. 
%\begin{lemma}
%The duplicator has a winning strategy in a restricted $m$-round $k$-pebble game over the game board $(\mathcal{O}_1^*,\mathcal{O}_2^*)$ if in such game Spoiler is restricted to put pebbles in  $\mathcal{O}_2^*$.
%\end{lemma}

%Indeed, it is because of this Lemma that we let the linear order factor $\alpha_i$ be $2^i$.

%============================================================

\begin{lemma}\label{flexibility-in-same-abstraction-1}
Assume that $P\subset \mathbb{X}_r^*-\mathbb{X}_{r+1}^*$ and $l=|P|\leq k-2$, and for any $(u_i,v_i),(u_j,v_j)$, $v_i\neq v_j$ if $(u_i,v_i)\in P$ and $(u_j,v_j)\in P$. Then for any string $w_1w_2\cdots w_l\in\{0,1\}^l$ and any $y\in [1,k-2]$, where $y\neq v_i$ for any $(u_i,v_i)\in P$, there is  $(x,y)\in\mathbb{X}_r^*-\mathbb{X}_{r+1}^*$ such that for any $(u_i,v_i)\in P$,  
\begin{equation}\label{SW-arbitrary} %\\[-20pt]
%(\mathbf{cc}([x]_{t-1},y)-\mathbf{cc}([u_i]_{t-1},v_i))\times (y-v_i)\times\\ \indent\hspace{25pt} 
\mathbf{BIT}(\mathrm{SW}((x,y),(u_i,v_i)),\hat{q}(y,v_i))=w_i.  
\end{equation}
\end{lemma}

The following observation, together with Lemma \ref{universal-simulator}, gives Duplicator the freedom in the scenario when she uses Strategy \ref{xi-1} (cf. the following main Lemma \ref{main-lemma}). Recall that, in the proof of Theorem \ref{existential-case}, Duplicator is always able to pick a vertex that is adjacent to all the pebbled vertices in $\mathcal{B}_k$. The following lemma (in particularly \textit{(3)}) roughly says the similar thing. It allows Duplicator to first choose a vertex that is adjacent to all the pebbled vertices. Afterwards, Duplicator can adjust and make her pick by looking for a proper one around this vertex.  
%\\[-17pt] 
\begin{lemma}\label{no-missing-edges_xi-1}
Let $1<t\leq m$. For any multiset $H$ of $k-2$ vertices $(x_1,y_1),\ldots,\\(x_{k-2},y_{k-2})\!\in\!\mathbb{X}_{t}^*$ and $y\!\in\! \{0,\cdots\!,k\!-\!1\}-\{y_i \mid 1\leq i\leq k-2\}$, 
the following hold. 
\begin{enumerate}[(1)]
\item  If there is $0\leq c\leq k-1$ such that $c\neq y$ mod $k-1$ and $c\neq y_i$ mod $k-1$ for any $1\leq i\leq k-2$, then 
for any $(x^{\star},y)\in \mathbb{X}_{t-1}^*$, there is a vertex $(x^\sharp,y)\in (\llbracket x^{\star}\rrbracket_{t-1},y)$ such that $((x^\sharp,y),(x_i,y_i))\!\in\! E_*$ for any $(x_i,y_i)\in H$;

\item  For any $(x^{\prime},y)\in \mathbb{X}_{t-1}^*$ where $(x^{\prime},y)\!\!\restriction\!\! S=\emptyset$,  
if $(\llparenthesis x^{\prime}\rrparenthesis_t,y)$ is adjacent to every vertex in $H$, then there is a vertex $(x^{\prime\prime},y)\in (\llbracket x^{\prime}\rrbracket_{t-1},y)$ s.t. $\mathrm{idx}(x^{\prime\prime},y)=t-1$, and $((x^{\prime\prime},y),(x_i,y_i))\in E_*$ for any $(x_i,y_i)\in H$; 

\item On condition that there are $(x_i,y_i),(x_j,y_j)\in H$ s.t. $x_i\neq x_j$ and $y_i=y_j$,  there is $(x,y)\!\in\! \mathbb{X}_{t-2}^*-\mathbb{X}_{t-1}^*$ s.t.   
$((x,y),(x_i,y_i))\!\in\! E_*$ for any $(x_i,y_i)\in H$, $[x]_{t-2}\equiv 0\hspace{3pt} (\mbox{mod }k-1)$, $g(x)=0$  and $\mathrm{RngNum}(x,t-2)=-1$. 

\item 
On condition that $y_i\neq y_j$ for any $(x_i,y_j)\neq (x_j,y_j)$, there is $(x,y)\!\in\! \mathbb{X}_{t-1}^*-\mathbb{X}_{t}^*$ such that    
$((x,y),(x_i,y_i))\!\in\! E_*$ for any $(x_i,y_i)\in H$, $[x]_{t-1}\equiv 0\hspace{3pt} (\mbox{mod }k-1)$, $g(x)=0$  and $\mathrm{RngNum}(x,t-1)=-1$. 
\end{enumerate}
\end{lemma}

%Lemma \ref{no-missing-edges_xi-1} \textit{(1)}, \textit{(2)} are two special cases of Lemma \ref{no-missing-edges_xi-1} \textit{(3)}.
%Obviously, this lemma holds if ``$E^B$'' is replaced by  ``$E^A$''. 

%Note that, the proof of Lemma \ref{no-missing-edges_xi-1} is important for understanding the proof of the following main lemma (cf. Strategy 2). 

Now we introduce our main lemma, which asserts that Duplicator has a winning strategy in any $m$-round $(k-1)$-pebble game over the game board $(\widetilde{\mathfrak{A}}_{k,m}, \widetilde{\mathfrak{B}}_{k,m})$. %Recall that here  $\mathfrak{A}_{k,m}$ refers to $\mathfrak{A}_{k,m}^+$. 
\begin{lemma}\label{main-lemma}
$\widetilde{\mathfrak{A}}_{k,m}\equiv_m^{k-1} \widetilde{\mathfrak{B}}_{k,m}$, for  $4\leq k$ and $(k-1)(k-2)<m$. 
\end{lemma}
At each round of the game, Duplicator's strategy first 
works in some specific abstraction of the associated 
structures $\widetilde{\mathfrak{A}}_{k,m}^*$ and $\widetilde{\mathfrak{B}}_{k,m}^*$, which will be explained soon in the proof. 
Suppose that in the current round the players are playing in the $\xi$-th abstraction of the structures. That is, for any pebbled vertex $(u,v)$, Duplicator regards $(\llparenthesis u\rrparenthesis_\xi,v)$, instead of $(u,v)$, as been pebbled.  
As in the proof of Theorem \ref{existential-case}, an indispensable component of a  strategy of Duplicator is to ensure that the players pick a pair of pebbles in the same row of the structures in each round. For each $i$, the $i$-th row of the structures consists of several intervals delimited by pebbled vertices. 
When we talk about intervals, they are not overlapped. 
Note that, at the beginning of the game, there is only one interval $[(0,i),(\gamma_{m-\xi}^*-1,i)]$ for the $i$-th row.\footnote{Such ``interval'' does not really exists since there is no pebble or delimiter that marks its boundary. It is an imaginary interval that exists in Duplicator's mind.} In each round of the game, Duplicator ensures that $\widetilde{\mathfrak{A}}_{k,m}^{*(\xi)}$ and $\widetilde{\mathfrak{B}}_{k,m}^{*(\xi)}$ have the same number of intervals in the same row. And if Spoiler puts a pebble in the $j^\star$-th interval of a row, so does Duplicator in the other structure in the same row, for any $j^\star$.  

In the following we introduce the basic ideas that will be used to deal with the linear orders. 
By a folklore knowledge (cf. Remark \ref{remark-linear-orders-1}), we know that it is impossible for Spoiler to find the difference between two linear orders if their lengths are large enough. %There are several slightly different versions of such folklore knowledge. Here we adopt a moderate version.

\begin{fact}\label{linear-orders-1}
For any $m\geq m^\prime\geq 0$, if $\mathcal{O}_a,\mathcal{O}_b$ are linear orders of length greater than or equal to $2^m-1$, then $\mathcal{O}_a\equiv_m^{m^\prime} \mathcal{O}_b$.
\end{fact}

Let $\ell_c$ be some number in $[1,m]$. 
In the $\ell_c$-th round of the  game over abstractions where $\xi> m-\ell_c$,  
recall that all the picked vertices are ``projected'' in $\mathbb{X}_\xi^*$, and 
assume that Spoiler ``picks'' vertex $(c,y)$ in the interval $[(a,y),(b,y)]$, where $c=\llparenthesis x\rrparenthesis_\xi$ and $(x,y)$ is the actually picked vertex in the associated game, thus splitting the  interval $[(a,y),(b,y)]$ into two smaller intervals $[(a,y),(c,y)]$ and $[(c,y),(b,y)]$. And assume that the corresponding interval in the other structure is $[(a^{\prime},y),(b^{\prime},y)]$. Note that all these vertices mentioned, e.g. $(a,y)$, are in $\mathbb{X}_{\xi}^*$. 
Let   $l_\xi:=\mathpzc{U}_\xi^*\cdot\beta_{m-\xi}^{m-1}$. If we regard every $l_\xi$ successive vertices as \textit{one object}, \label{def-main-object}
we get a linear order $\preceq^\xi$ induced from the original linear orders of the structures: $u\preceq^\xi u^{\prime}$ if and only if $\lfloor u/l_\xi\rfloor\leq \lfloor u^{\prime}/l_\xi\rfloor$, for any $u,u^{\prime}$ in $\mathbb{X}_1^*$. In the $\ell_c$-th round, Duplicator picks $(x^\prime,y)$ to respond Spoiler. Let $c^{\prime}:=\llparenthesis x^\prime\rrparenthesis_\xi$. Duplicator needs to ensure that $(c^\prime,y)$ is in the interval $[(a^{\prime},y),(b^{\prime},y)]$ such that the following condition, called \textit{abstraction-order-condition}, holds:  
 for any $1\leq i\leq \xi$, 
\begin{enumerate}[a)]%\label{xi-order-requirement}
\item if $0<\lfloor c/l_i\rfloor\!-\!\lfloor a/l_i\rfloor <2^{m-\ell_c}-1 \!\mbox{ then }\\ \indent \hspace{10pt}\lfloor c^{\prime}/l_i\rfloor-\lfloor a^{\prime}/l_i\rfloor=\lfloor c/l_i\rfloor-\lfloor a/l_i\rfloor$; otherwise,
\item if $0<\lfloor b/l_i\rfloor-\lfloor c/l_i\rfloor < 2^{m-\ell_c}-1$ then\\  $\indent\hspace{10pt}\lfloor b^{\prime}/l_i\rfloor-\lfloor c^{\prime}/l_i\rfloor=\lfloor b/l_i\rfloor-\lfloor c/l_i\rfloor$; otherwise,

\item if $\lfloor c/l_i\rfloor-\lfloor a/l_i\rfloor=0$ or $\lfloor b/l_i\rfloor-\lfloor c/l_i\rfloor=0$, then $[c]_i-[a]_i=[c^\prime]_i-[a^\prime]_i$ or $[b]_i-[c]_i=[b^\prime]_i-[c^\prime]_i$ respectively; otherwise,

\item %there are two cases:
%\begin{itemize}
%\item $[a^\prime]_m\neq [b^\prime]_m$ or $[a]_m\neq [b]_m$:\\
$\lfloor c^{\prime}/l_i\rfloor-\lfloor a^{\prime}/l_i\rfloor\geq 2^{m-\ell_c}-1$ and $\lfloor b^{\prime}/l_i\rfloor-\lfloor c^{\prime}/l_i\rfloor\geq 2^{m-\ell_c}-1$. 
%\item $[a^\prime]_m=[b^\prime]_m$ and $[a]_m=[b]_m$:\\
%if $\lfloor c/l_\xi\rfloor-\lfloor a/l_\xi\rfloor\leq \lfloor b/l_\xi\rfloor-\lfloor c/l_\xi\rfloor$ then 
%$\lfloor c^{\prime}/l_\xi\rfloor-\lfloor a^{\prime}/l_\xi\rfloor=2^{m-j}$; otherwise, $\lfloor b^{\prime}/l_\xi\rfloor-\lfloor c^{\prime}/l_\xi\rfloor=2^{m-j}$.
%\end{itemize} 

\end{enumerate}
%Note that, d) holds for $1\leq i\leq \xi-1$ if $\lfloor c/l_\xi\rfloor-\lfloor a/l_xi\rfloor\neq 0$ and $\lfloor b/l_\xi\rfloor-\lfloor c/l_\xi\rfloor\neq 0$. 
Note that, this strategy implies that, if Spoiler puts a pebble on the vertex  that is already pebbled, so does Duplicator; and if Spoiler picks a new vertex, so does Duplicator in the other structure. We call  $[\lfloor c/l_\xi\rfloor,\lfloor a/l_\xi\rfloor]$ ``\textit{unabridged interval}'' in the $\xi$-th abstraction, for any $(c,y),(a,y)\in\mathbb{X}_\xi^*$. Note that, such concepts like induced linear orders and unabridged intervals allow Duplicator takes care of the linear orders first, meanwhile leave space for the considerations for the partial isomorphism issue with respect to edges. 

If the abstraction-order-condition can be preserved, 
then Duplicator can win the game over the pair of (pure) induced linear orders. 
%(cf. the proof for Lemma \ref{linear-orders-1}).
If  this requirement cannot be satisfied, then Duplicator looks for the $(\xi-1)$-th abstraction for a solution: assume that the length of any unabridged interval in the  $\mathbf{\xi}$-th abstractions, e.g. $[\lfloor c/l_\xi\rfloor,\lfloor a/l_\xi\rfloor]$, is at least 1, then the following always holds:\footnote{Recall that $(a,y),(c,y),(a^\prime,y),(c^\prime,y)\in \mathbb{X}_\xi^*$. Because $\lfloor c/l_\xi\rfloor\neq\lfloor a/l_\xi\rfloor$, by induction hypothesis we have $\lfloor c^\prime/l_\xi\rfloor\neq\lfloor a^\prime/l_\xi\rfloor$. Note that $\xi\geq\theta>m-\ell_c$. It implies that $\lfloor c^\prime/l_{\xi-1}\rfloor-\lfloor a^\prime/l_{\xi-1}\rfloor\geq \frac{\beta_{m-\xi}^{m-\xi+1}}{\mathpzc{U}_{\xi-1}^*}=2^{\xi-1}\geq 2^{m-\ell_c}$.} %\\[-10pt]
\begin{equation}\label{xi-order-requirement-ensured}
\lfloor c^{\prime}/l_{\xi-1}\rfloor-\lfloor a^{\prime}/l_{\xi-1}\rfloor\geq 2^{m-\ell_c}; \lfloor b^{\prime}/l_{\xi-1}\rfloor-\lfloor c^{\prime}/l_{\xi-1}\rfloor\geq 2^{m-\ell_c}.
\end{equation}
 By Fact \ref{linear-orders-1}, since the linear orders are large enough now, it allows Duplicator to respond Spoiler properly for one more round, at the price that $\xi$ decreases by 1. 
By Lemma \ref{HighOrder-is-LowOrder}, if distinct pebbles are put on distinct objects of the induced linear orders, then the orders are preserved in the lower abstractions. In Remark \ref{abstract-order-in-main-lemma}, we discuss the order issue in more detail along this line. 

However, 
%if there are more than one pebble in one object of the induced linear orders, then 
Spoiler still has a way to win the game via linear orders, if there are two pairs of vertices $(u_1,v)$, $(u_1^{\prime},v)$, $(u_2,v)$, and $(u_2^{\prime},v)$ in the $\xi$-th abstraction of the structures, such that \label{page-def-varkappa}
\begin{itemize}
\item $(u_1,v)\Vdash (u_1^{\prime},v)$ and $(u_2,v)\Vdash (u_2^{\prime},v)$;

\item $\lfloor u_1/l_\xi\rfloor=\lfloor u_2/l_\xi\rfloor$;

\item $u_1<u_2\Leftrightarrow u_2^{\prime}<u_1^{\prime}$. \hfill  $(\varkappa)$

\end{itemize}

We use ``\textbf{virtual games}'' to denote the kinds of pebble games wherein the players play the games in their mind (following the usual rules) without really putting pebbles on the board. 
Now we introduce a sort of imaginary \textbf{games over changing boards}. That is, the game board can be different in each round. Certainly, this is not a precise definition for it doesn't tell us how the game board changes. In the following we introduce a specific kind of such games. Moreover, they are a sort of virtual games. 
Recall that, in $\widetilde{\mathfrak{A}}_{k,m}$ and $\widetilde{\mathfrak{B}}_{k,m}$, every vertex is associated with a board history, which is supposed to reflect reasonable ``evolution logic''  of the game. If Spoiler picks $(x,y)$ in a structure, say $\widetilde{\mathfrak{A}}_{k,m}$, 
Duplicator uses virtual games to determine the board history of the vertex $(x^\prime,y)$ she is going to pick. Note that,  \textit{from here on} we assume that $(x,y),(x^\prime,y)\in\mathbb{X}_1$ (in this proof). For simplicity, here we assume that no vertex is pebbled before Spoiler picking $(x,y)$.\footnote{In case when some vertices are pebbled, the reader can cf. the proof of Claim \ref{board-history-evolutions} because we need to take account of the order of the board histories of pebbled vertices when playing the virtual games.} 
A virtual game in such a simple setting consists of $\mathrm{i}_{\mathrm{cur}}^{x,y}-1$ ``virtual rounds''. No vertex is pebbled at the beginning of this virtual game. Spoiler ``picks'' according to $\chi(x,y)\!\!\restriction\!\!\mathrm{BH}(j)$, for $j=1$ to $\mathrm{i}_{\mathrm{cur}}^{x,y}-1$, and Duplicator ``replies'' in the other structure, i.e.  in $\widetilde{\mathfrak{B}}_{k,m}$. \label{def-virtual-game} In the following we define the \textit{virtual game board} at the beginning of  the $j$-th round, for any $j$ in $[1,\mathrm{i}_{\mathrm{cur}}^{x,y}]$. Let  $\mathbb{Z}_{xy}^j:=\chi(x,y)\!\!\restriction\!\!\mathrm{BH}(j-1)$ and $\mathbb{Z}_{x^\prime y}^j:=\chi(x^\prime,y)\!\!\restriction\!\!\mathrm{BH}(j-1)$. \label{gimel-ell-current}
Firstly, for any vertex $(e,f)\in\mathbb{X}_1^*$, $\gimel_j^a(e,f)=\mathbb{Z}_{xy}^j$ if $(e,f)\notin \mathbb{Z}_{xy}^j-\{(*,*)\}$; otherwise, $\gimel_j^a(e,f)=\mathbb{Z}_{xy}^\ell$ where $\ell=max\{i\in [2,j]\mid (e,f)\in\mathbb{Z}_{xy}^{i}\land (e,f)\notin\mathbb{Z}_{xy}^{i-1}\}$. Similarly, $\gimel_j^b(e,f)=\mathbb{Z}_{x^\prime y}^j$ if $(e,f)\notin \mathbb{Z}_{x^\prime y}^{j}-\{(*,*)\}$;  otherwise, $\gimel_j^b(e,f)=\mathbb{Z}_{x^\prime y}^{\ell}$ where $\ell=max\{i\in [2,j]\mid (e,f)\in\mathbb{Z}_{x^\prime y}^{{i}}\land (e,f)\notin\mathbb{Z}_{x^\prime y}^{{i-1}}\}$. 
For example, $\gimel_1^a(e,f)=\gimel_1^b(e,f)=\mathrm{BC}_\emptyset$, for any $(e,f)\in\mathbb{X}_1^*$. 
At the start of the $j$-th round, the game board is $((\widetilde{\mathfrak{A}}_{k,m}^{*\gimel_j^a},\mathbb{Z}_{xy}^j),
(\widetilde{\mathfrak{B}}_{k,m}^{*\gimel_j^b},\mathbb{Z}_{x^\prime y}^{j}))$. 
   Duplicator's ``virtual responses'' determine the board history of the vertex that she should actually pick. 
%And Duplicator mimics her picks in the real game if the virtual game has the same scenario as the real game of some stage.   
Recall that the players do not really use pebbles in such virtual games.  The strategy Duplicator uses in such virtual games will be introduced soon, cf. Strategy \ref{play-in-xi-abs}\textapprox Strategy \ref{t<xi}.  

\begin{proof}
We prove that Duplicator has a winning strategy in an $m$-round $(k\!-\!1)$-pebble game.
%, which incorporates the $\xi$-order-requirement. 
 Let $\overline{c_A}$ be the set of pebbled vertices in  $\widetilde{\mathfrak{A}}_{k,m}$ at the start of the current round and $\overline{c_B}$ be the set of pebbled vertices in $\widetilde{\mathfrak{B}}_{k,m}$, both in the natural order. As explained in Fact \ref{star_k-2}, we  assume that the lengths of $\overline{c_A}$ and $\overline{c_B}$ are less than or equal to $k-2$.  

%Duplicator can view the structures as three linear orders: the top row, the bottom row, and all the other middle rows are merged into one linear order, i.e. we use a vector $((u,1),\ldots,(u,k-2))^{\!\top}$ to denote one ``vertex''. Furthermore, Duplicator regard every $\rho_i$ successive ``vertices''(or $\rho_i^*$ vertices in the flat structures) as one object in a linear order induced from the original one. Then Duplicator can use the winning strategy, which works for $\Game(\mathfrak{A}_{3,m},\mathfrak{B}_{3,m})$ (for any $1\leq m$), to determine $\mathrm{RngNum}(x^\prime,i)$ for the vertex $(x^\prime,y)$ picked by Duplicator in the game $\Game(\mathfrak{A}_{k,m}^+,\mathfrak{B}_{k,m}^+)$. It ensures that $\mathrm{sgn}((u^\flat,v),\!(x^\flat,y))=0 \Leftrightarrow \mathrm{sgn}((u^{\prime\flat},v),\!(x^{\prime\flat},y))=0$ for any $(u,v)\Vdash (u^\prime,v)$ (cf. (\ref{k-equal-3-eqn}), appendix). Therefore, in the sequel we  mention Duplicator's choice for $\mathrm{RngNum}(x^\prime,i)$ only when necessary. 

Soon, we shall define a condition (\ref{main-diamond-xi}$^\diamond$) in page \pageref{def-xi}. Say that \textit{the game (and the board) is over the $i$-th abstraction},  
if (\ref{main-diamond-xi}$^\diamond$) holds for any pair of pebbled vertices, e.g. $(x^\flat,y)\Vdash (x^{\prime\flat},y)$, and the projections of pebbled vertices in the $i$-th abstraction define a partial isomorphism. 
And $\xi$ is the maximum number in $[1,m]$ that makes (\ref{main-diamond-xi}$^\diamond$) hold at the end of the current round. 

We use $\xi$ to remind Duplicator  which abstraction of the structures she should take care of at the \textit{start} of the current round.  At the start of the game, let $\xi:=m$. 
 
And we use $\theta$ to record how many rounds are still available to Spoiler at the start of the current round. In other words, the current round is the $(m+1-\theta)$-th round. Let $\ell_c:=m+1-\theta$. 
 Note that $\theta:=\theta-1$ after each round by default.

%Let $\xi:=min\{\xi_a,\xi_b\}$ in any round. 

\textit{From now on, in all the cases}, we assume that Spoiler picks $(x,y)$ in $\widetilde{\mathfrak{A}}_{k,m}$ such that $\mathrm{idx}(x^\flat,y)=t$ for some $t\in [1,m]$. The case when he picks in $\widetilde{\mathfrak{B}}_{k,m}$ is similar. Correspondingly, \textit{in all the cases}, assume that Duplicator picks $(x^{\prime}\!,y)$ in $\widetilde{\mathfrak{B}}_{k,m}$. Moreover, assume that  $\mathrm{idx}(x^{\prime\flat}\!,y)\!=\!t^{\prime}$ for some $t^\prime\in [1,m]$. 

Firstly, we define some ordered sets that are used in defining Duplicator's strategy. 
%Assume that $y\neq v$. 
%Let \begin{align}
%\begin{split}
%\mathcal{H}_{xy}^A &:=\{(u^\flat,v)\mid (u,v)\in\mathbb{X}_1; (u,v)\twoheadrightarrow (x,y)\mbox{ or }(x,y)\twoheadrightarrow (u,v)\};\\
%\mathcal{H}_{x^\prime y}^B &:=\{(u^\flat,v) \mid (u,v)\in\mathbb{X}_1; (u,v)\twoheadrightarrow (x^\prime,y)\mbox{ or }(x^\prime,y)\twoheadrightarrow (u,v)\}.
%\end{split}
%\end{align} 
In the sequel we use $\overrightarrow{c_A}$ and $\overrightarrow{c_B}$ to denote the following (ordered) sets (in the natural picking order). 
\begin{equation}\label{def-overrightarrow-c_A}
\begin{split}
\overrightarrow{c_A}&:=\{(u,v)\mid (u,v)\in (x,y)[\mathrm{BC}];v\neq y\};\\
\overrightarrow{c_B}&:=\{(u,v) \mid (u,v)\in (x^\prime,y)[\mathrm{BC}];v\neq y\}. 
\end{split}
\end{equation}
%++++++++++++++++++
%\begin{split}
%\overrightarrow{c_A}&:=\{(u^\flat,v)\mid (u,v)\in\overline{c_A}; (u,v)\twoheadrightarrow (x,y)\mbox{ or }(x,y)\twoheadrightarrow (u,v)\};\\
%\overrightarrow{c_B}&:=\{(u^\flat,v) \mid (u,v)\in\overline{c_B}; (u,v)\twoheadrightarrow (x^\prime,y)\mbox{ or }(x^\prime,y)\twoheadrightarrow (u,v)\}. 
%\end{split}
%\end{equation}
%+++++++++++++++++
%\begin{equation}\label{def-overrightarrow-c_A}
%\begin{split}
%\overrightarrow{c_A}:=\{(u^\flat,v)\mid (u,v)\in \overline{c_A}; (u,v)\xrightarrow[\mathrm{BC}]{con.}(x,y)\lor (x,y)\xrightarrow[\mathrm{BC}]{con.}(u,v)\};\\
%\overrightarrow{c_B}:=\{(u^\flat,v)\mid (u,v)\in \overline{c_B}; (u,v)\xrightarrow[\mathrm{BC}]{con.}(x^\prime,y)\lor (x^\prime,y)\xrightarrow[\mathrm{BC}]{con.}(u,v)\}.
%\end{split}
%\end{equation}
 % Soon we shall see, based on Claim \ref{board-history-evolutions}, that $\overrightarrow{c_A}\Vdash \overrightarrow{c_B}$. 
%Note that $\overrightarrow{c_A}\Vdash \overrightarrow{c_B}$ in the corresponding virtual game. And, 
Note that $|\overrightarrow{c_A}|=|\overrightarrow{c_B}|\leq k-2$ because of Fact \ref{star_k-2}. 
%For any $(u,v)\in \mathbb{X}_1^*$, we use $(u,v)\in (x,y)[\mathrm{BC}]$ to denote that $(u,v)$ is the $i$-th element of the tuple $(x,y)[\mathrm{BC}]$, which is not $(*,*)$, for some $i$. %\\[-18pt]
Now, we define the following ordered sets, whose orders inherit $\overrightarrow{c_A}$ and $\overrightarrow{c_B}$. 
\begin{multline*}
Z:=\{(u,v)\!\in\! \overrightarrow{c_{A}} \1 \mathrm{idx}(u,v)\!=\!q\!\geq\! t; \mathbf{cc}([u]_t,v)\!=\!\mathbf{cc}([x^\flat]_t,y); \\\indent\hspace{13pt} \mbox{either } (x^\flat,y) \mbox{ or } (u,v) \mbox{ is not a critical point of } \widetilde{\mathfrak{A}}_{k,m}^*\}\cup\\ 
 \indent\hspace{10pt}\{(u,v)\!\in\! \overrightarrow{c_A}\1 \mathrm{idx}(u,v)\!=\!q\!<\!t;  \mathbf{cc}([u]_q,v)\!=\!\mathbf{cc}([x^\flat]_q,y)\}; 
\end{multline*}
$\indent\hspace{-6pt} Z^{\geq\xi}:=Z\cap\mathbb{X}_\xi^*; Z^{<\xi}:=Z-Z^{\geq\xi}$;\\[-20pt]
\begin{multline*} 
U:=\{(u,v)\!\in\! \overrightarrow{c_A} \1 y,v\in [1,k-2];
\mathrm{idx}(u,v)\!=\!t;(\mathbf{cc}([u]_t,v)\!-\!\mathbf{cc}([x^\flat]_t,y))\times\\ 
(v-y)\times(-1)^{\mathbf{BIT}(\mathrm{SW}((u,v),(x^\flat,y)),\hat{q}(v,y))}<0\} \end{multline*} 
%\\[-25pt]
$\indent\hspace{-5pt} R:=\{(u,v)\!\in\! \overrightarrow{c_{A}} \1  \mathbf{cl}(u,v)\in\chi(x,y)\!\!\restriction\!\! S\}$;\footnote{The definition is equivalent if we replace $\chi(x,y)\!\!\restriction\!\! S$ by $\chi(x^\flat,y)\!\!\restriction\!\! S$.}\\
$\indent\hspace{-7pt} R^{\geq\xi}:=R\cap\mathbb{X}_\xi^*$; $R^{<\xi}:=R-R^{\geq\xi}$. %Similarly, we can define $R^{\prime<\xi-1}$ etc. 
%+++++++++++++++
%$\indent\hspace{-5pt} R:=\{(u,v)\!\in\! \overrightarrow{c_{A}} \1  (u,v)\in  (x,y)[\mathrm{BC}]\land\mathbf{cl}(u,v)\in\chi(x,y)\!\!\restriction\!\! S\}$;$\indent\hspace{-5pt} R:=\{(u,v)\!\in\! \overrightarrow{c_{A}} \1  (u,v)\in  (x,y)[\mathrm{BC}]\land\mathbf{cl}(u,v)\in\chi(x,y)\!\!\restriction\!\! S\}$;
%+++++++++++++++
%$\indent\hspace{-5pt} R:=\{(u^\flat,v)\!\in\! \overrightarrow{c_{A}} \1 (u,v)\in\overline{c_A}; (u^\flat,v)\in (x,y)[\mathrm{BC}]\land\mathbf{cl}(u,v)\in\chi(x,y)\!\!\restriction\!\! S\\ \indent\hspace{20pt} \mbox{ or }(x^\flat,y)\in (u,v)[\mathrm{BC}]\land\mathbf{cl}(x,y)\in\chi(u,v)\!\!\restriction\!\! S\}$; 
%$\indent\hspace{-7pt} R^{\geq\xi}:=R\cap\mathbb{X}_\xi^*; R^{<\xi}:=R-R^{\geq\xi}$;
\begin{equation*}%\\[-20pt]
 \indent\hspace{-78pt} D:=\{(u,v)\!\in\! \overrightarrow{c_A} \1  (u,v)\!\in\!\chi(x^\flat,y)\!\!\restriction\! \Omega \mbox{ or } (x^\flat,y)\!\in\!\chi(u,v)\!\!\restriction\! \Omega\};\\[-5pt]
\end{equation*}
%$\indent D^{<\xi}:=\{(u,v)\!\in\! \overrightarrow{c_A}-\overrightarrow{c_A}\cap\mathbb{X}_\xi \1  (x,y)\!\in\!\chi(u,v)\!\!\restriction\! \Omega\}$;
%$\indent D:=D^{\geq\xi}\cup D^{<\xi}$.
$\indent\hspace{-7pt} D^{\geq\xi}:=D\cap\mathbb{X}_\xi^*; D^{<\xi}:=D-D^{\geq\xi}$;\\
$\indent\hspace{-7pt} T:=\{(u,v)\!\in\! \overrightarrow{c_A} \1 v \mbox{ or }y\!\in\!\{0,k-1\};v\neq y;%\mathbf{cc}([u^\flat]_q,v)\!\neq\!\mathbf{cc}([x^\flat]_q,y) \mbox{ where }\\\indent\hspace{25pt} q\!=\!min\{t,\mathrm{idx}(u^\flat,v)\};
 \mathrm{sgn}((u,v),(x^\flat,y))\!=\!1\}$.\\[-1pt]

Correspondingly, 
\begin{multline*}
Z^{\prime}:=\{(u,v)\in \overrightarrow{c_B}\1 \mathrm{idx}(u,v)\!=q<t^{\prime}; \mathbf{cc}([u]_q,v)=\mathbf{cc}([x^{\prime\flat}]_q,y)\}\cup\\
\{(u,v)\in \overrightarrow{c_B}\1 \mathrm{idx}(u,v)=q\geq t^{\prime};\mathbf{cc}([u]_{t^{\prime}},v)=\mathbf{cc}([x^{\prime\flat}]_{t^{\prime}},y)\};
\end{multline*}
$\indent Z^{\prime\geq\xi}:=Z^{\prime}\cap\mathbb{X}_\xi^*; Z^{\prime <\xi}:=Z^{\prime}-Z^{\prime\geq\xi}$.\\[-8pt]

Similarly, if we replace $\overrightarrow{c_A}$, $t$ and $x$ by $\overrightarrow{c_B}$, $t^{\prime}$ and $x^{\prime}$ respectively, we obtain $U^{\prime}$, $R^{\prime}$, $D^{\prime}$, $T^\prime$ and $R^{\prime\geq\xi}$ etc.

%Let 
%\begin{align*}\\[-25pt]
%\overrightarrow{c_A}^*:=\{(\llparenthesis u\rrparenthesis_\xi,v)\1 (u,v)\in\overrightarrow{c_A}\};\\\overrightarrow{c_B}^*:=\{(\llparenthesis u\rrparenthesis_\xi,v)\1 (u,v)\in\overrightarrow{c_B}\}.
%\end{align*}

For each set $X$ that is just defined, we define an associated set $X^{(\xi)}$ by substituting $u$ with  $\llparenthesis u\rrparenthesis_{\xi}$ and $x$ with $\llparenthesis x\rrparenthesis_{\xi}$.
%in the \textit{defining part} of the set. 
For example, we define $D^{(\xi)}$ as the following set.
\begin{equation*}
  D^{(\xi)}\!:=\!\{(\llparenthesis u\rrparenthesis_\xi,v) \1 (u,v)\in\overrightarrow{c_A};(\llparenthesis u\rrparenthesis_{\xi},\!v)\!\in\!\chi(\llparenthesis x^\flat\rrparenthesis_\xi,y)\!\!\restriction\! \Omega \mbox{ or } \\(\llparenthesis x^\flat\rrparenthesis_\xi,y)\!\in\!\chi(\llparenthesis u\rrparenthesis_{\xi},v)\!\!\restriction\! \Omega\}. 
\end{equation*}

Similarly, we can define $Z^{(\xi)}$, $U^{(\xi)}$, $R^{(\xi)}$, and $T^{(\xi)}$. And correspondingly we can define the due sets $Z^{\prime(\xi)}$ etc. 

Recall that $\overline{c_A}$ is the set of pebbled vertices in  $\widetilde{\mathfrak{A}}_{k,m}$ at the start of the current round. Let
\begin{align}\label{def-widetilde-c_A}
\begin{split}
\widetilde{c_A}&:=\{(u,v)\in\overline{c_A}\mid (u,v)\twoheadrightarrow (x,y)\mbox{ or }(x,y)\twoheadrightarrow (u,v);v\neq y\};\\
\widetilde{c_B}&:=\{(u,v)\in\overline{c_B}\mid (u,v)\twoheadrightarrow (x^\prime,y)\mbox{ or }(x^\prime,y)\twoheadrightarrow (u,v);v\neq y\}.
\end{split}
\end{align}

Note that, by definition, $(u,v)\twoheadrightarrow (x,y)$ implies that $(u^\flat,v)\in (x,y)[\mathrm{BC}]$. Therefore, The set $\{(u^\flat,v)\mid (u,v)\in \widetilde{c_A}\}$ is a subset of $\overrightarrow{c_A}$. Similarly, $\overrightarrow{c_B}$ subsumes $\{(u^\flat,v)\mid (u,v)\in \widetilde{c_B}\}$. 

For any set $X$ in $\{Z,U,\cdots,T\}-\{R\}$, we define a related set $\widetilde{X}$, by substituting $(u,v)\in\overrightarrow{c_A}$ with $(u,v)\in\widetilde{c_A}$, and substituting $(u,v)$ with $(u^\flat,v)$ in the defining part (the statements behind ``$\mid$''), and substituting $\mathbb{X}_\xi^*$ with $\mathbb{X}_\xi$. For example, we define $\widetilde{D}$ as the following set.
\begin{equation*}%\\[-20pt]
  \widetilde{D}:=\{(u,v)\!\in\! \widetilde{c_A} \1  (u^\flat,v)\!\in\!\chi(x^\flat,y)\!\!\restriction\! \Omega \mbox{ or } (x^\flat,y)\!\in\!\chi(u^\flat,v)\!\!\restriction\! \Omega\}.
\end{equation*}
And, $\widetilde{D}^{\geq\xi}:=\widetilde{D}\cap\mathbb{X}_\xi; \widetilde{D}^{<\xi}:=\widetilde{D}-\widetilde{D}^{\geq\xi}$. 

In addition, we define $\widetilde{R}$ as the following set.
\begin{equation}
\widetilde{R}:=\{(u,v)\in\widetilde{c_A}\mid (x,y)\rightsquigarrow (u,v)\mbox{ or }(u,v)\rightsquigarrow (x,y)\}. 
\end{equation}

Moreover, $\widetilde{R}^{\geq\xi}:=\widetilde{R}\cap\mathbb{X}_\xi;  \widetilde{R}^{<\xi}:=\widetilde{R}-\widetilde{R}^{\geq\xi}$. 

Correspondingly, we can define $\widetilde{Z}^{\prime}$ etc.  

We can also define $\widetilde{Z}^{(\xi)}$ in the way like $Z^{(\xi)}$.  

Note that $\overrightarrow{c_A}\Vdash \overrightarrow{c_B}$, where ``$\Vdash$'' is defined in page \pageref{def-Vdash}. Now we define a similar denotation ``$\Vdash_i$'' as the follows. For any vertex $(u,v)\in\mathbb{X}_1^*$ picked in a game,   $(\llparenthesis u\rrparenthesis_i,v)\Vdash_i (\llparenthesis u^\prime\rrparenthesis_i,v)$ if $(u,v)\Vdash (u^\prime,v)$. Similarly, for two sest $X$ and $X^\prime$ of vertices we can define $X \Vdash_i X^\prime$ in the usual way.\\[0pt] 

%Finally, similar to (\ref{def-overrightarrow-c_A}), we define the following sets.
%\begin{equation}\label{def-overrightarrow-c_A-pebble}
%\begin{split}
%\overrightarrow{c_A}_\mathbf{p}:=\{(u^\flat,v)\mid (u,v)\in\overline{c_A}\land ((u,v)\xrightarrow[\mathrm{BC}]{con.}(x,y)\mbox{ or }(x,y)\xrightarrow[\mathrm{BC}]{con.}(u,v))\};\\
%\overrightarrow{c_B}_\mathbf{p}:=\{(u^\flat,v)\mid (u,v)\in\overline{c_B}\land ((u,v)\xrightarrow[\mathrm{BC}]{con.}(x^\prime,y)\mbox{ or }(x^\prime,y)\xrightarrow[\mathrm{BC}]{con.}(u,v))\}.
%\end{split}
%\end{equation}

Similar to the proof of Lemma \ref{winning-strategy-in-k=3}, this proof is also by simultaneous induction, wherein we show that the follows are preserved after each round. Let $\mathrm{S}_i^A:=\{(u,v)\!\in\! \overrightarrow{c_{A}} \1  \mathbf{cl}(\llparenthesis u\rrparenthesis_i,v)\in\chi(\llparenthesis x^\flat\rrparenthesis_i,y)\!\!\restriction\!\! S\}$ and $\mathrm{S}_i^B:=\{(u,v)\!\in\! \overrightarrow{c_{B}} \1 \mathbf{cl}(\llparenthesis u\rrparenthesis_i,v)\in\chi(\llparenthesis x^{\prime\flat}\rrparenthesis_i,y)\!\!\restriction\!\! S\}$. 
\begin{enumerate}[(1$^\diamond$)]
\item \label{theta-lessthan-xi} $\theta<\xi$ (after the first round); 
\item \label{xi-order-condition-hold} The abstraction-order condition holds; moreover, Duplicator's choice can prevent $(\varkappa)$ from occurring (cf. page \pageref{page-def-varkappa} for ``$(\varkappa)$'');  %That is, the lengths of each pair of intervals in the $\xi$-th abstraction of the structures are either equivalent or both are greater than or equal to $2^\theta$, when regarding $(\llbracket [x]_\xi\rrbracket_{\mathpzc{U}_\xi},y)$ in the $\xi$-th abstraction as one object, for any $(x,y)\in \mathbb{X}_1$.

\item \label{condition-for-DuWin-*}
Duplicator can win this round in the corresponding \textit{associated games over the $\xi$-th abstractions}. That is, 
$$Z^{(\xi)}\cup U^{(\xi)}\cup R^{(\xi)}\cup D^{(\xi)}\cup T^{(\xi)}\Vdash_\xi Z^{\prime(\xi)}\cup U^{\prime(\xi)}\cup R^{\prime(\xi)}\cup D^{\prime(\xi)}\cup T^{\prime(\xi)};$$ 
%where ``$\Vdash_\xi$'' is similar to ``$\Vdash$'' (recall it in page \pageref{def-Vdash}), except that for every vertex in the set, say $(u,v)$, $(u^\flat,v)$ is replaced by its projection in the $\xi$-th abstraction, i.e. $(\llparenthesis u^\flat\rrparenthesis_\xi,v)$.  
%we take it that the players ``pick'' $(\llparenthesis u^\flat\rrparenthesis_\xi,v)$, though not really, when they pick $(u^\flat,v)$. 

%\item \label{No-varkappa} $(\varkappa)$ does not present.

%\item In each round,\\[-13pt] 
%\begin{equation}\label{indexs-close-enough}
%|\mathrm{idx}(x,y)-\mathrm{idx}(x^{\prime},y)|\leq 1. 
%\end{equation}

\item \label{main-diamond-boundary-strategy}
Duplicator's choice makes the ``boundary checkout strategy'' (cf. p.  \pageref{page-boundary-checkout-strategy}) ineffective in the next round. In other words, it means that the game board over the $(\xi-1)$-th abstraction is in partial isomophism w.r.t. edges even if we take it as if the boundaries of rows of the $(\xi-1)$-th abstraction were occupied with extra immovable pebbles.

\item \label{main-diamond-xi}
If $t,t^\prime<\xi-1$ then,\footnote{By definition of vertex index, we have $x^\flat-\llparenthesis x^\flat\rrparenthesis_\xi=x^{\prime\flat}-\llparenthesis x^{\prime\flat}\rrparenthesis_\xi=0$ if $t,t^\prime\geq \xi$; $x^\flat-\llparenthesis x^\flat\rrparenthesis_{\xi-1}=x^{\prime\flat}-\llparenthesis x^{\prime\flat}\rrparenthesis_{\xi-1}=0$ if $t=t^\prime=\xi-1$.} for $t\leq i<\xi$, on condition that $\llparenthesis x^\flat\rrparenthesis_i\neq \llparenthesis x^\flat\rrparenthesis_\xi$, \label{def-xi} 
\begin{enumerate}[(i)]
\item $\mathrm{idx}(\llparenthesis x^\flat\rrparenthesis_i,y)=\mathrm{idx}(\llparenthesis x^{\prime\flat}\rrparenthesis_i,y)$; %suppose $\mathrm{idx}(\llparenthesis x^\flat\rrparenthesis_i,y)=idx_i$;

\item $\mathrm{cc}([x^\flat]_i,y)=\mathrm{cc}([x^{\prime\flat}]_i,y)$;

\item 
$\llbracket \llparenthesis x^\flat\rrparenthesis_i\rrbracket_{i}^{min}-[\llparenthesis x^\flat\rrparenthesis_\xi]_{i}\equiv\llbracket \llparenthesis x^{\prime\flat}\rrparenthesis_i\rrbracket_{i}^{min}-[\llparenthesis x^{\prime\flat}\rrparenthesis_\xi]_{i}$ (mod $\beta_{m-i-1}^{m-i}$);\\ $\llbracket \llparenthesis x^\flat\rrparenthesis_i\rrbracket_{i}^{min}\leq [\llparenthesis x^\flat\rrparenthesis_\xi]_{i}$ iff $\llbracket \llparenthesis x^{\prime\flat}\rrparenthesis_i\rrbracket_{i}^{min}\leq [\llparenthesis x^{\prime\flat}\rrparenthesis_\xi]_{i}$;\footnote{Note that, it is equivalent to the condition that  $x^\flat\leq \llparenthesis x^\flat\rrparenthesis_\xi$ iff $x^{\prime\flat}\leq \llparenthesis x^{\prime\flat}\rrparenthesis_\xi$.}\\
$\llparenthesis x^\flat\rrparenthesis_i-\llbracket \llparenthesis x^\flat\rrparenthesis_i\rrbracket_{i}^{min}=\llparenthesis x^{\prime\flat}\rrparenthesis_i-\llbracket \llparenthesis x^{\prime\flat}\rrparenthesis_i\rrbracket_{i}^{min}$ if $\mathrm{S}_i^A=\mathrm{S}_i^B$;

\item $\mathrm{S}_i^A\Vdash \mathrm{S}_i^B$; 

\item $\mathrm{RngNum}(\llparenthesis x^\flat\rrparenthesis_i,t)=\mathrm{RngNum}(\llparenthesis x^{\prime\flat}\rrparenthesis_i,t)$; 

\item $g(\llparenthesis x^\flat\rrparenthesis_i)=g(\llparenthesis x^{\prime\flat}\rrparenthesis_i)$;

\item For any vertex $(u,y)\in\overline{c_A}$ and $(u^{\prime},y)\in\overline{c_B}$ where $(u,y)\!\Vdash\! (u^{\prime},y)$, if $\llbracket\llparenthesis u^\flat\rrparenthesis_i\rrbracket_{i}^{min}\!=\!\llbracket\llparenthesis x^\flat \rrparenthesis_i\rrbracket_{i}^{min}$, then  $\llparenthesis u^\flat\rrparenthesis_i\leq \llparenthesis x^\flat\rrparenthesis_i$ iff $\llparenthesis u^{\prime\flat}\rrparenthesis_i\leq\llparenthesis x^{\prime\flat}\rrparenthesis_i$; 

\item $(a^\prime,b)\leq (\llparenthesis x^{\prime\flat}\rrparenthesis_i,y)$ iff $(a,b)\leq (\llparenthesis  x^\flat\rrparenthesis_i,y)$, for any $(a,b)\in \{(e,f)\in \mathbb{X}_{1}^*\mid (u,v)\in\overline{c_A};\mathrm{idx}(u^\flat,v)<t;(e,f)\in\chi(\llparenthesis u^\flat\rrparenthesis_i,v)\!\!\restriction\!\! S\}$ and corresponding $(a^\prime,b)$.\footnote{That is, $(a^\prime,b)$ is the $j$-th item in the tuple encoding  $\chi(\llparenthesis u^{\prime\flat}\rrparenthesis_i,v)\!\!\restriction\!\! S$ if $(a,b)$ is the $j$-th item in the tuple encoding  $\chi(\llparenthesis u^{\flat}\rrparenthesis_i,v)\!\!\restriction\!\! S$. We shall see that it is possible because $|\chi(\llparenthesis u^{\flat}\rrparenthesis_i,v)\!\!\restriction\!\! S|=|\chi(\llparenthesis  u^{\prime\flat}\rrparenthesis_i,v)\!\!\restriction\!\! S|$ if $\mathrm{idx}(u^\flat,v)<t<\xi-1$.} 
\end{enumerate}

%$\llbracket x^\flat\rrbracket_{t}^{min}-[\llparenthesis x^\flat\rrparenthesis_\xi]_{t}=\llbracket x^{\prime\flat}\rrbracket_{t}^{min}-[\llparenthesis x^{\prime\flat}\rrparenthesis_\xi]_{t}$.

\item  \label{condition-for-DuWin}
The associated game board is still in partial isomorphism after picking $(x^\flat,y)$ and $(x^{\prime\flat},y)$. That is, 
\begin{equation}
Z\cup U\cup R\cup D\cup T\Vdash Z^{\prime}\cup U^{\prime}\cup R^{\prime}\cup D^{\prime}\cup T^\prime.
\end{equation}

\end{enumerate}  

Call (\ref{theta-lessthan-xi}$^\diamond$)\textapprox (\ref{condition-for-DuWin}$^\diamond$) a \textit{winning-condition-set} of Duplicator. (\ref{xi-order-condition-hold}$^\diamond$) and (\ref{condition-for-DuWin-*}$^\diamond$) ensure partial isomorphism in the games over the $\xi$-th abstraction (the former takes care of the linear order, while the latter  takes care of edges). And we use (\ref{main-diamond-xi}$^\diamond$) to ensure that, if the associated game board is in partial isomorphism over the $\xi$-th abstraction at the start of the current round, then it also holds over the $(\xi-1)$-th abstraction (on condition that Duplicator uses the auxiliary games \eqref{eqn-1-round-S-game_xi} in Strategy \ref{play-in-xi-abs}, and \eqref{eqn-1-round-S-game_xi-1} in Strategy \ref{xi-1}, and \eqref{eqn4-1-round-game-reduction}  \eqref{eqn3-1-round-game-reduction} in Strategy \ref{t<xi}; 
cf. Remark \ref{partial-isom-propagate}). Note that (iii) of (\ref{main-diamond-xi}$^\diamond$) implies that $x^\flat-\llparenthesis x^\flat\rrparenthesis_\xi$ roughly equals $x^{\prime\flat}-\llparenthesis x^{\prime\flat}\rrparenthesis_\xi$.\footnote{It is because a unit of difference in higher abstraction means a huge difference in lower abstractions. This observation is also used in the proof of Lemma \ref{corollary-approxi-copy-cat}.} 
Hence (\ref{main-diamond-xi}$^\diamond$) is also called the \textit{approximate hr-copycat condition}. 

Later on, when we describe and discuss Duplicator's strategy, we shall delay the discussion of (\ref{main-diamond-boundary-strategy}$^\diamond$) to the end. We treat it in this way to avoid unnecessary repetition of arguments. 

Moreover, \textit{most of the time} the readers can take $((\mathfrak{A}_{k,m},\overline{c_A}),(\mathfrak{B}_{k,m},\overline{c_B}))$ as the game board instead of $((\widetilde{\mathfrak{A}}_{k,m},\overline{c_A}),(\widetilde{\mathfrak{B}}_{k,m},\overline{c_B}))$. Only when we discuss (4$^\diamond$) should we switch back to the real game board  $((\widetilde{\mathfrak{A}}_{k,m},\overline{c_A}),(\widetilde{\mathfrak{B}}_{k,m},\overline{c_B}))$.

\vspace{6pt}
\noindent\textbf{Basis:} At the beginning of the game, $\overline{c_A}$ and $\overline{c_B}$ are empty. Hence $\overrightarrow{c_A}$ and $\overrightarrow{c_B}$ are empty. 
Recall that Spoiler picks $(x,y)$ and Duplicator replies $(x^\prime,y)$.
In the first round, suppose that $\chi(x,y)\!\!\restriction\!\!\mathrm{bc}\!=\!0$, Duplicator simply mimics Spoiler's picking. 
%In the second round, Duplicator continues to mimic unless that Spoiler picks a critical point, say $(x,y)$, and another critical point with different second coordinate is already pebbled. Assume w.l.o.g. that $y=k-1$. In this case, Duplicator will pick a vertex $(x^{\prime},k-1)\in\mathbb{X}_m^B$ such that  $\mathbf{cc}([x^{\prime}]_m,k-1)\neq \mathbf{cc}([u]_m,v)$ if $(u,v)$ is the first picked vertex in $\mathfrak{B}_{k,m}$. 
 Since the signature contains no unary relation symbol and the graphs contain no self-loop,  
 the board is in partial isomorphism and Duplicator wins the first  round. The value of   
 $\xi$ is still $m$ at the end of the first round. But $\theta\!:=\!\theta\!-\!1=m\!-\!1$. Hence the winning-condition-set can be ensured.
 
In this situation Duplicator is a copycat. Certainly she is also an exact hr-copycat (recall p. \pageref{def-hr-copycat} for the notion ``hr-copycat'').  Note that this simple copycat strategy also works in the (associated) games over changing boards wherein a game board  constitutes a pair of ``flat'' structures which have no pebbled vertices at the start of the first round. 
 
 However, Spoiler can ``cheat'' by picking the first vertex with an arbitrary board history at the beginning. In this case, Duplicator is no more an exact hr-copycat. 
 Instead, using Strategy \ref{play-in-xi-abs}\textapprox Strategy \ref{t<xi}, which will be introduced soon, Duplicator first plays a \textit{virtual game} that determines the board history of $(x^{\prime},y)$ (cf. page \pageref{gimel-ell-current}); then she replies Spoiler in the associated structures, using Strategy \ref{play-in-xi-abs}\textapprox Strategy \ref{t<xi}. 

%In the following, we show that Duplicator has a strategy to win any round only if such condition set is preserved throughout the game.

\vspace{4pt}
\noindent\textbf{Induction Step:}
Assume that Duplicator wins the first $m-\theta$ rounds. We prove that she can also win the $(m-\theta+1)$-th round, and the winning-condition-set is preserved.  
%There are several cases that Duplicator needs to consider when Spoiler picks a vertex $(x,y)$.   
In \textit{all the cases}, soon we shall see that Duplicator follows some supplementary basic strategies. Henceforth,  we always assume that $(u,v)\!\in\!\overline{c_A}$ and $(u,v)\!\Vdash\! (u^{\prime},v)$.
\begin{enumerate}[B-1]
\item Duplicator gives the abstraction-order-condition the highest priority.\footnote{Duplicator first finds the ``allowed'' positions for picking in accordance with the abstraction-order-condition. Usually such positions are assembled into intervals rather than isolated points. Generally such intervals have vertices of all the necessary type labels. Hence, Duplicator can first use it to determine the approximate position for her pick, then selects a vertex of appropriate type label using Strategy \ref{play-in-xi-abs}\textapprox Strategy \ref{t<xi}, which will be introduced soon.}

%\item Duplciator ensures that $(\varkappa)$ will not occur;

\item Duplicator always ensures that 
\begin{equation}\label{picking-in-cl}    
\mathbf{cl}(u,v)\neq\mathbf{cl}(x,y)\Leftrightarrow \mathbf{cl}(u^{\prime},v)\neq\mathbf{cl}(x^{\prime},y).
\end{equation}

\item Duplicator uses virtual games to determine the board histories of her picked vertices.  

%\item Duplicator always ensures that 
%\begin{multline}\label{picking-in-RngNum}
%\mathrm{RngNum}(u,\mathrm{idx}(u^\flat,v))\neq \mathrm{RngNum}(x,\mathrm{idx}(x^\flat,y))\Leftrightarrow\\ 
%\mathrm{RngNum}(u^\prime,\mathrm{idx}(u^{\prime\flat},v))\neq \mathrm{RngNum}(x^\prime,\mathrm{idx}(x^{\prime\flat},y)).
%\end{multline}

\end{enumerate}

Note that Duplicator can ensure B-1 because of (\ref{xi-order-requirement-ensured}). Based on it, Duplicator can know the approximate position for her pick. 

Confer Remark \ref{remark-why-pick-in-cl} for the reason that Duplicator should follow B-2. 

%Confer Remark \ref{remark-why-pick-in-RngNum} for the reason that Duplicator should follow B-4. 

We are able to prove the following claim: if an edge is forbidden in one structure due to discontinuities, so is the corresponding edge in the other structure. 
\begin{claim}\label{board-history-evolutions}
If Duplicator stick to B-3, she has a way to ensure that, for any $(u,v), (u^\prime,v)$ where $(u,v)\!\Vdash\! (u^{\prime},v)$,   
\begin{enumerate}[(i)]
\item 
\begin{itemize}
\item  $(u,v)\twoheadrightarrow (x,y)$ if and only if $(u^{\prime},v)\twoheadrightarrow (x^{\prime},y)$.

\item $(x,y)\twoheadrightarrow (u,v)$ if and only if $(x^{\prime},y)\twoheadrightarrow (u^{\prime},v)$. 

\end{itemize}

\item If $y=v$, then $\lfloor x/(\gamma_{m-1}^*\times k)\rfloor$ mod $bh^\#$ $\leq \lfloor u/(\gamma_{m-1}^*\times k)\rfloor$ mod $bh^\#$ if and only if $\lfloor x^\prime/(\gamma_{m-1}^*\times k)\rfloor$ mod $bh^\#$ $\leq \lfloor u^\prime/(\gamma_{m-1}^*\times k)\rfloor$ mod $bh^\#$. 
\end{enumerate}
\end{claim}
%It is for this reason that we can put the continuity in (\ref{def-overrightarrow-c_A}) and $\overrightarrow{c_A}\Vdash\overrightarrow{c_B}$ still holds. 

%, and by Corollary \ref{linear-orders-1}, Duplicator is able to win the  game $\Game_m^k\!\!\left(\mathfrak{A}_{k,m}^{+(\xi)}|\langle \leq \rangle,\mathfrak{B}_{k,m}^{+(\xi)}|\langle \leq \rangle\right)$.  
%What she should do is to also ensure that she can win the game $\mathcal{G}_m(\mathfrak{A}_{k,m}^{+(\xi)}|\langle E\rangle,\mathfrak{B}_{k,m}^{+(\xi)}|\langle\ E\rangle)$. 

%Once Duplicator finds that  it is impossible to win the game $\Game_m^k\!\!\left(\mathfrak{A}_{k,m}^{+(\xi)},\mathfrak{B}_{k,m}^{+(\xi)}\right)$, Duplicator plays the game $\Game_m^k\!\!\left(\mathfrak{A}_{k,m}^{+(\xi-1)},\mathfrak{B}_{k,m}^{+(\xi-1)}\right)$. 

%(\ref{indexs-close-enough}) implies that $|\xi_a-\xi_b|\!\leq\! 1$.

Claim \ref{board-history-evolutions} \textit{(ii)} says that the board history of $(x,y)$ is less than (or equals to) that of $(u,v)$ (in the order explained in page \pageref{page-history-order}) if and only if the board history of $(x^\prime,y)$ is less than (or equals to) that of $(u^\prime,v)$. Note that in the case when equal holds, the order of the picked vertices is taken care of by the abstraction-order-condition introduced before. Moreover, from the proof of Claim \ref{board-history-evolutions}, we can see that board histories of pebbled vertices satisfy(\mbox{apx-}1) (cf. Remark \ref{remark-linear-orders-1}) if we regard board histories as objects of the linear orders.
%\footnote{Different board histories have different lengths. Hence we need to ``truncate'' the longer histories to compare them, as demonstrated in the proof of Claim \ref{board-history-evolutions} \textit{(ii)}.}
 Hence, for example, we know that $\lfloor x^\prime/(\gamma_{m-1}^*\times k)\rfloor$ mod $bh^\#=\lfloor u^\prime/(\gamma_{m-1}^*\times k)\rfloor+1$ mod $bh^\#$ if and only if $\lfloor x/(\gamma_{m-1}^*\times k)\rfloor$ mod $bh^\#=\lfloor u/(\gamma_{m-1}^*\times k)\rfloor+1$ mod $bh^\#$. 
 
Because of Claim \ref{board-history-evolutions} \textit{(i)}, we can safely assume that $(x,y)[\mathrm{BC}]$ is a valid board  configuration, for otherwise $(x,y)$ and $(x^\prime,y)$ are isolated vertices in respective structures; and we need only focus on the set of vertices that are in continuity with $(x,y)$ and $(x^\prime,y)$, in the following case by case discussion of Duplicator's strategy. 
Note that, by this claim, $\widetilde{c_A}\Vdash\widetilde{c_B}$. Moreover, $\widetilde{c_A}$ has $n$ distinctive vertices if and only if $\widetilde{c_B}$ does, where $0\leq n\leq k-1$. 

We describe Duplicator's strategy using simultaneous induction, in addition to B-1\textapprox B-3. But instead of studying the game directly, it is more convenient to study the game by a specific \textbf{associated game over  changing board} wherein the game board is $(\widetilde{\mathfrak{A}}_{k,m}^{*\gimel_1^a},\widetilde{\mathfrak{B}}_{k,m}^{*\gimel_1^b})$ at the start. Recall that $\gimel_1^a(e,f)=\gimel_1^b(e,f)=\mathrm{BC}_\emptyset$, for any $(e,f)\in\mathbb{X}_1^*$. Cf. page \pageref{gimel-ell-current} for the definition of $\gimel_i^a$ and $\gimel_i^b$. 
The game board at the start of the $\ell_c$-th round consists of a pair of structures $\widetilde{\mathfrak{A}}_{k,m}^{*\gimel_{\ell_c}^a}$ and $\widetilde{\mathfrak{B}}_{k,m}^{*\gimel_{\ell_c}^b}$;
%\footnote{It means that the board will be changed after Spoiler picking $(x,y)$ and Duplicator getting $(x^\prime,y)[\mathrm{BC}]$ using the virtual games (cf. page \pageref{def-virtual-game}).} 
in this round the players pick  $(x^{\flat},y)\in\mathbb{X}_1^*$ in the associated game $\Game_{m-\ell_c+1}^{k-1}((\widetilde{\mathfrak{A}}_{k,m}^{*\gimel_{\ell_c}^a},(x,y)[\mathrm{BC}]),
(\widetilde{\mathfrak{B}}_{k,m}^{*\gimel_{\ell_c}^b},(x^\prime,y)[\mathrm{BC}]))$  if they pick $(x,y)$ in the original game  $\Game_{m-\ell_c+1}^{k-1}((\widetilde{\mathfrak{A}}_{k,m},\overline{c_A}),(\widetilde{\mathfrak{B}}_{k,m},\overline{c_B}))$.

Claim \ref{board-history-evolutions} \textit{(i)} tells us that we don't have to consider the board histories when we study the original game restricted to edges; and Claim \ref{board-history-evolutions} \textit{(ii)} tells us that we don't have to be worry about the order issue if we ignore the board histories when we study the original game restricted to linear orders, which is relatively clear. 
It implies that it is possible for Duplicator to have a winning strategy in the original game $\Game_m^{k-1}((\widetilde{\mathfrak{A}}_{k,m},\overline{c_A}),(\widetilde{\mathfrak{B}}_{k,m},\overline{c_B}))$ if she has a winning strategy in the associated game over changing board. Soon we shall see the justification in  Strategy \ref{play-in-xi-abs}. 
Therefore, 
%together with 2) a) of Definition \ref{iterative-expansion}, 
we can talk about something like ``\textit{Spoiler picks $(x^\flat,y)$ in $\widetilde{\mathfrak{A}}_{k,m}^{*\gimel_{\ell_c}^a}$}'' instead of ``\textit{Spoiler picks $(x,y)$ in $\widetilde{\mathfrak{A}}_{k,m}$}''. 
%Note that, for any $(u,v)\in\mathbb{X}_1$, $(x^\flat,y)\rightsquigarrow (u^\flat,v)$ only if either $(u^\flat,v)$ is pebbled in the associated game board or $(u,v)$ is not adjacent to $(x,y)$ due to discontinuity. In the latter case, $(u^\prime,v)$ is also not adjacent to $(x^\prime,y)$, by Claim \ref{board-history-evolutions}.  As a consequence, we need only consider the former case and safely assume that $\mathbb{Z}$ consist of exactly the set of pebbled vertices.  
And we shall see that the following strategy is a winning strategy  for Duplicator in the specific associated game over changing board. Note that, games over changing boards are a sort of auxiliary games. That is, we do not play them alone. But they are the basis of the original game and the virtual games. 
We reduce the original game to the corresponding associated game and virtual games over changing board. And in both of these two kinds of imaginary games, Duplicator relies on the following strategy, i.e. Strategy \ref{play-in-xi-abs}\textapprox Strategy \ref{t<xi}, case by case. Note that, in the virtual games for bord histories, the imaginary ``pebbled'' vertices (not the vertices really picked) form the game configuration associated with the vertex to be ``picked'' in the virtual round.   
%Note that we have defined the sets, e.g. $Z$, over the game board $((\widetilde{\mathfrak{A}}_{k,m}^*,\overrightarrow{c_A}),(\widetilde{\mathfrak{B}}_{k,m}^*,\overrightarrow{c_B}))$. Abuse of notations, we can also define the same sets over the corresponding changing boards. 
  
%Then, by Claim \ref{dudal-games-transparency}, there is a corresponding winning strategy for Duplicator in the original game $\Game_m^{k-1}(\mathfrak{A}_{k,m},\mathfrak{B}_{k,m})$. 

Note that, by induction hypothesis, the game board is in partial isomorphism over the $\xi$-th abstraction at the start of the current round. We can show that this also holds over the $(\xi-1)$-th abstraction.\footnote{Cf. Remark \ref{partial-isom-propagate}. The readers are suggested to read the arguments a bit later, i.e. after reading Strategy \ref{t<xi}. It implies that (\ref{condition-for-DuWin}$^\diamond$) is entailed by the other conditions, in particular (\ref{condition-for-DuWin-*}$^\diamond$).} 

The readers can choose to read a brief sketch (Remark \ref{strategy-sketch}, appendix) of the following strategy before diving directly into the details. 

In the following, ``strategy $i$'' stands for the shorthand of ``strategy for the case $i$''. 
As we have discussed the order issue in details in Remark \ref{abstract-order-in-main-lemma}, which will be used by Duplicator to stick to B-1, we use Strategy \ref{play-in-xi-abs}\textapprox Strategy \ref{t<xi}, explained in the following, to mainly determine the type label of $(x^\prime,y)$, i.e. how this vertex is adjacent to other vertices.  

\begin{enumerate}
\item \label{play-in-xi-abs} 
$\langle$ Strategy 1 $\rangle$
Suppose that Spoiler picks a vertex $(x,y)\in\mathbb{X}_\xi^A$ (recall that $\mathrm{idx}(x^\flat,y)=t$; in other words, we assume $t\geq \xi$ in this case). And suppose that  
Duplicator can pick $(x^{\prime},y)\in \mathbb{X}_{\xi}^B$ such that (\ref{xi-order-condition-hold}$^\diamond$), (\ref{condition-for-DuWin-*}$^\diamond$),  (\ref{main-diamond-boundary-strategy}$^\diamond$), and (\ref{condition-for-DuWin}$^\diamond$) hold. 

Note that, the set $\widetilde{Z}\cup \widetilde{U}\cup \widetilde{R}\cup \widetilde{D}\cup \widetilde{T}$ is precisely the set of vertices that are not adjacent to $(x,y)$ in the current round, if do not consider the missing of edges due to 2) b) of Definition \ref{iterative-expansion}.   
Indeed, by Claim \ref{board-history-evolutions}, we don't have to consider it. 

In the following we explain why (\ref{condition-for-DuWin-*}$^\diamond$) and (\ref{condition-for-DuWin}$^\diamond$) matter. 

If $(x,y)\twoheadrightarrow (u,v)$ for some pebbled vertex $(u,v)$, then Duplicator simply let $(x^\prime,y)$ be $(u^\prime,v)\!\!\restriction\!\!_\mathrm{H}^{\mathrm{i}_{\mathrm{cur}}^{x,y}}$, and 
we are able to show that (\ref{condition-for-DuWin}$^\diamond$) and (\ref{condition-for-DuWin-*}$^\diamond$) hold (cf. Remark \ref{remark-strategy2}). Henceforth, we assume that $\lnot((x,y)\twoheadrightarrow(u,v))$ holds for any pebbled $(u,v)$. By Claim \ref{board-history-evolutions}, $\lnot((x^\prime,y)\twoheadrightarrow(u^\prime,v))$ holds for any pebbled $(u^\prime,v)$. 
That is,  one of the following two cases holds: 
\begin{enumerate}[(i)]
\item $(u,v)\twoheadrightarrow (x,y)$: by Claim \ref{board-history-evolutions}, $(u^\prime,v)\twoheadrightarrow (x^\prime,y)$;  

\item $\lnot((u,v)\twoheadrightarrow (x,y))$: by Claim \ref{board-history-evolutions},  $\lnot((u^\prime,v)\twoheadrightarrow (x^\prime,y))$, which implies that $((x,y),(u,v))\notin E^{A}\land ((x^\prime,y),(u^\prime,v))
\notin E^{B}$, due to 2 (b) of Definition \ref{iterative-expansion}. 
\end{enumerate} 
Therefore, we need only consider (i). As a consequence, henceforth we can assume that 
$\lnot((u,v)\rightsquigarrow (x,y))\land
\lnot((u^\prime,v)\rightsquigarrow
 (x^\prime,y))$ holds. That is, in this case %\\[-7pt]     
\begin{equation}\label{xi-1-simuluation-2}
\widetilde{R}^{<\xi}=\widetilde{R}^{\prime<\xi}=\emptyset. %\\[-4pt] 
\end{equation}
It implies that we need only consider the case that\footnote{Since $\mathrm{idx}(x,y)\geq \xi$, it is easy to see that $\{(u,v)\in\widetilde{c_A}\mid (x,y)\rightsquigarrow (u,v)\}\subseteq \widetilde{R}^{\geq\xi}$ because $\mathbf{cl}(u,v)\in\chi(x,y)\!\!\restriction\!\! S$ implies that $\mathrm{idx}(u,v)>\mathrm{idx}(x,y)\geq \xi$. On the other hand, for any $(u,v)\in\widetilde{c_A}$, $(u,v)\rightsquigarrow (x,y)$ would imply that $(x,y)\xrightarrow[\mathrm{BH}]{con.}(u,v)$.} 
\begin{equation}
\widetilde{R}=\widetilde{R}^{\geq\xi}=\{(u,v)\in\widetilde{c_A}\mid (x,y)\rightsquigarrow (u,v)\}.
\end{equation}
That is, in this case (i), $\widetilde{R}=\{(u,v)\in \overline{c_A}\mid (u,v)\twoheadrightarrow (x,y)\land \mathbf{cl}(u,v)\in \chi(x,y)\!\!\restriction\!\! S\}$.  

Therefore, by Claim \ref{board-history-evolutions}, Duplicator need only consider (\ref{condition-for-DuWin}$^\diamond$) in the associated game to ensure the following in the original game, in the case (i): 
%that   
%$(u,v)\xrightarrow[\mathrm{BH}]{con.}(x,y)\land (u^\prime,v)\xrightarrow[\mathrm{BH}]{con.} (x^\prime,y)$ for any pebbled vertex $(u,v)$:   
 \begin{equation}\label{condition-for-DuWin-1}
\widetilde{Z}\cup \widetilde{U}\cup \widetilde{R}\cup \widetilde{D}\cup \widetilde{T}\Vdash \widetilde{Z}^{\prime}\cup \widetilde{U}^{\prime}\cup \widetilde{R}^{\prime}\cup \widetilde{D}^{\prime}\cup \widetilde{T}^{\prime}.
\end{equation}

In other words,   
if Duplicator can ensure (\ref{condition-for-DuWin}$^\diamond$), then Spoiler cannot win this round in the original game. Similarly, if Duplicator can ensure (\ref{condition-for-DuWin-*}$^\diamond$), then Spoilver cannot win this round over the $\xi$-th abstraction in the original game, which in turn helps to ensure (\ref{condition-for-DuWin}$^\diamond$).\\[-1pt]  

By convention, we use $\circ$ to denote the concatenation of two tuples. 
Let $\overrightarrow{c_A}_s^\xi$  be an ordered set defined as the following. 
\begin{multline}\label{def-c_A_s-xi}
\overrightarrow{c_A}_s^\xi:=\{(\llparenthesis a\rrparenthesis_\xi,b)\mid (a,b)\in (x,y)[\mathrm{BC}]\}\ccirc\{(\llparenthesis e\rrparenthesis_{\xi},f)\mid (a,b)\in (x,y)[\mathrm{BC}];\\\mathrm{idx}(a,b)<\xi;(e,f)\in\chi(a,b)\!\!\restriction\!\! S\}.
\end{multline}
Likewise, $\overrightarrow{c_B}_s^\xi$ is defined in the similar way except that $(x,y)[\mathrm{BC}]$ is replaced by $(x^\prime,y)[\mathrm{BC}]$. The elements of $\overrightarrow{c_A}_s^\xi$ and $\overrightarrow{c_B}_s^\xi$ are in the natural order inheriting from $(x,y)[\mathrm{BC}]$ and $(x^\prime,y)[\mathrm{BC}]$. Note that $|\overrightarrow{c_A}_s^\xi|=|\overrightarrow{c_B}_s^\xi|<m$. Therefore, by Remark \ref{remark-linear-orders-1}, 
Duplicator can ensure that she wins the following one round game wherein she picks $(x^{\prime\flat},y)$ to respond to Spoiler's pick of $(x^\flat,y)$, in accordance with the abstraction-order-condition:  
\begin{equation}\label{eqn-1-round-S-game_xi}
\Game_1 ((\widetilde{\mathfrak{A}}_{k,m}^*[\mathbb{X}_\xi^*]|\langle \leq\rangle,\overrightarrow{c_A}_s^{\xi}),(\widetilde{\mathfrak{B}}_{k,m}^*[\mathbb{X}_\xi^*]|\langle \leq\rangle,\overrightarrow{c_B}_s^{\xi}).
\end{equation}  

Furthermore, $\theta:=\theta-1$ at the end of this round, and $\xi$ remains unchanged. That is, (\ref{theta-lessthan-xi}$^\diamond$) holds. %In addition, (\ref{main-diamond-xi}$^\diamond$) automatically hold because $(x^{\flat},y),(x^{\prime\flat},y)\in\mathbb{X}_\xi^*$.  

If Duplicator's pick can ensure  (\ref{xi-order-condition-hold}$^\diamond$), (\ref{condition-for-DuWin-*}$^\diamond$), (\ref{main-diamond-boundary-strategy}$^\diamond$), and (\ref{condition-for-DuWin}$^\diamond$), then she not only wins this round, but also wins it over the $\xi$-th abstraction. Otherwise, Duplicator uses Strategy \ref{xi-1} if she cannot ensure these conditions. 

%We shall see that (\ref{lower-index-points-are-innocent}) can be easily satisfied due to strategy \ref{t<xi}. Hence we ommit such discussion in strategy \ref{xi-1}. \\[2pt]
 
\item \label{xi-1} $\langle$ Strategy 2 $\rangle$ Spoiler picks a vertex in $\mathbb{X}_{\xi}^A$ and  Duplicator cannot pick an appropriate vertex in the $\xi$-th abstraction that satisfies  (\ref{xi-order-condition-hold}$^\diamond$), (\ref{condition-for-DuWin-*}$^\diamond$) and (\ref{condition-for-DuWin}$^\diamond$). It implies that $\mathbf{cl}(x,y)\neq \mathbf{cl}(u,v)$ for any pebbled $(u,v)$ in $\widetilde{\mathfrak{A}}_{k,m}$, for otherwise Duplicator can simply follow B-2, and the conditions (\ref{xi-order-condition-hold}$^\diamond$), (\ref{condition-for-DuWin-*}$^\diamond$) and (\ref{condition-for-DuWin}$^\diamond$) can be ensured. 
 In such a case, Duplicator picks a vertex $(x^{\prime},y)\in\mathbb{X}_{\xi-1}\!-\!\mathbb{X}_{\xi}$ s.t. $[x^{\prime\flat}]_{\xi-1}\equiv 0$ (mod $k-1$) (hence B-2 is followed), and resorts to the $(\xi-1)$-th abstraction for a solution.% (cf. p. \pageref{def-resort-to-closest-abs}). 
%Note that (5$^\diamond$) obviously holds.    

%Note that $R$ is a set depending on $(x,y)[\mathrm{BC}]$, which is different from other sets like $Z, D$ etc. 
 Note that, by definition, 
 \begin{equation}
 R^{<\xi-1}=R^{\prime<\xi-1}=\emptyset.
 \end{equation} 

In the following we show that Duplicator can ensure  (\ref{condition-for-DuWin-*}$^\diamond$) and (\ref{condition-for-DuWin}$^\diamond$). Let $M_A:=Z^{(\xi-1)}\cup R^{(\xi-1)}\cup D^{(\xi-1)}\cup U^{(\xi-1)}\cup T^{(\xi-1)}$ and $M_B$ be such a set that $M_A\Vdash_{\xi-1} M_B$.\footnote{In other words, $M_B=\{(\llparenthesis a^\prime\rrparenthesis_{\xi-1},b)\mid (a,b)\in\overrightarrow{c_A};(a,b)\Vdash (a^\prime,b);(\llparenthesis a\rrparenthesis_{\xi-1},b)\in M_A\}$.}
 Note that, in this case where $\mathrm{idx}(x^\flat,y)\geq \xi$, $M_A=Z^{(\xi-1)}\cup R^{(\xi-1)}\cup D^{(\xi-1)}\cup T^{(\xi-1)}\cup U$, 
because $\mathrm{idx}(\llparenthesis u^\flat\rrparenthesis_{\xi-1},v)=\xi$ implies that $\mathrm{idx}(u^\flat,v)=\xi$.\footnote{By Lemma \ref{lattice-point-high-is-lower}, if $\mathrm{idx}(u^\flat,v)=\xi-1$, then $\mathrm{idx}(\llparenthesis u^\flat\rrparenthesis_{\xi-1},v)=\xi-1$; if $\mathrm{idx}(u^\flat,v)>\xi$, then $\mathrm{idx}(\llparenthesis u^\flat\rrparenthesis_{\xi-1},v)=\mathrm{idx}(u^\flat,v)>\xi$.}  
%Note that $M_A$ and $M_B$ can be combined into some board configuration, denoted $M_A^B$. Clearly, $M_A^B$ is a tuple of up to $k-2$ pairs, and for any pair $((a,b),(c,d))$, $(a,b)\in M_A$ if and only if $(c,d)\in M_B$.
\begin{enumerate}[(\ref{xi-1}-1)]
\item By Lemma \ref{flexibility-in-same-abstraction-1}, Duplicator can ensure that $U^{\prime}=\emptyset$. 
%Because of $(\dagger)$, no vertex of $\widetilde{\mathfrak{B}}_{k,m}$ other than $(x^{\prime\flat},y)$ has index $\xi-1$. Therefore, $U^{\prime}=\emptyset$. 

\item  
By virtue of Lemma \ref{no-missing-edges_xi-1}, 
Duplicator can ensure that\footnote{\label{page-pebbled-in-same-row} Note that the above holds only when all the pebbled vertices are in different rows. If some of them are in the same row, then Duplicator need to resort to the $(\xi-2)$-th abstraction due to \textit{(3)} of Lemma \ref{no-missing-edges_xi-1}, which seems that $2m$ abstractions are needed for a structure instead of $m$ abstractions. However, even in this case, we can show that $m$ abstractions suffice for our purpose. The simplest way for Duplicator is to ensure that all the pebbled vertices in the same row have distinct indices. This is possible if she always resorts to the $(\xi-1)$-th abstraction whenever the row, in which she is going to put a pebble, already has a pebble. Note that in such case \textit{(3)} of Lemma \ref{no-missing-edges_xi-1} can be revised s.t. ``$t-2$'' (``$t-1$'' resp.) is replaced by $t-1$ ($t$ resp.), and the argument for it is similar to \textit{(4)} of Lemma \ref{no-missing-edges_xi-1}.}
\begin{itemize}
%\begin{equation}
% \left\{\begin{array}{l}
\item   $Z^{\prime(\xi)}\cup R^{\prime(\xi)}\cup D^{\prime(\xi)}=\emptyset$,
\item   $[x^{\prime\flat}]_{\xi-1}\equiv 0$ (mod  $k-1$),
\item   $\mathrm{RngNum}(x^{\prime\flat},\xi-1)=-1$. 
%   \end{array}\right.   
%\end{equation}
\end{itemize}

Soon we shall see in Strategy \ref{t<xi} that (\ref{main-diamond-xi}$^\diamond$) is ensured. Then by Lemma \ref{approxi-copy-cat} and Lemma \ref{corollary-approxi-copy-cat}, we have   
\begin{equation}\label{no-missing-edge-before-simuluation-1}
Z^{\prime(\xi-1)}\cup R^{\prime(\xi-1)}\cup D^{\prime(\xi-1)}=\emptyset.%\\[-4pt]
\end{equation} 
%Because of $(\dagger)$, for any pebbled $(u^\prime,v)$, $\mathrm{idx}(\llparenthesis u^\prime\rrparenthesis_{\xi-1},v)>\xi-1$. 
By Lemma \ref{proj-greater-index}, $\mathrm{idx}(\llparenthesis u^{\prime\flat}\rrparenthesis_{\xi-1},v)\geq \xi-1$. 
If $\mathrm{idx}(\llparenthesis u^{\prime\flat}\rrparenthesis_{\xi-1},v)=\xi-1$, then $T^\prime=\emptyset$, by definition. 
Now assume that $\mathrm{idx}(\llparenthesis u^{\prime\flat}\rrparenthesis_{\xi-1},v)>\xi-1$. 
By Lemma \ref{rngnum-is_-1}, we have $\mathrm{RngNum}(\llparenthesis u^{\prime\flat}\rrparenthesis_{\xi-1},\xi-1)=-1$. Therefore, by Lemma \ref{sgn-equal-0}, $T^\prime=\emptyset$ since $\mathrm{RngNum}(x^{\prime\flat},\xi-1)=-1$. 
Therefore, we have 
\begin{equation}\label{strate2-eqn-before-adaption}
Z^{\prime\geq\xi-1}\cup R^{\prime\geq\xi-1}\cup D^{\prime\geq\xi-1}\cup T^{\prime\geq\xi-1}=\emptyset.
\end{equation}     

\item  So far, the vertex Duplicator selected usually does not satisfy  (\ref{condition-for-DuWin-*}$^\diamond$). Therefore, Duplicator need fine-tune  $(x^{\prime},y)$ a little bit: rename the currently selected vertex as $(x^{\star},y)$, and find a new value for $x^{\prime}$ such that (\ref{condition-for-DuWin-*}$^\diamond$) holds. 
 By Lemma \ref{universal-simulator}, for any $S^{\prime}\!\in\! \wp(\mathbf{Cl}_{\xi})$ and any $(x^{\star},y)\!\in\! \mathbb{X}_{\xi-1}^B$, the sequence of $\mathpzc{U}_{\xi-1}^*$ successive vertices $(\llbracket x^{\star\flat}\rrbracket_{\xi-1},y)$ contains at least one vertex $(x^{\dagger},y)$ where  $\mathrm{idx}(x^{\dagger\flat},y)=\xi-1$ and $\chi(x^{\dagger},y)\!\!\restriction\!\! S\cap\{\mathbf{cl}(\llparenthesis a\rrparenthesis_{\xi-1},v)\mid (a,b)\in\overrightarrow{c_B}\}=S^{\prime}$. Hence Duplicator can simply pick $(x^{\prime},y)=(x^{\dagger},y)$ to ensure that%\\[-7pt] 
 \begin{equation}\label{xi-1-simuluation-1}
 S^{\prime}=\{\mathbf{cl}(e,f) \1 (e,f)\!\in\! M_B %\\[-2pt]
\}. 
\end{equation}

%And by Claim \ref{board-history-evolutions}, for any $(a,b)\in M_A$ and $(a,b)\Vdash(c,d)$,  $(a,b)\xrightarrow[\mathrm{BC}]{con.} (x,y)$ iff $(c,d)\xrightarrow[\mathrm{BC}]{con.}(x^{\prime},y)$. 
%Since we assume that $\overrightarrow{c_A}$ ($\overrightarrow{c_B}$ resp.) be the set of pebbled vertices that are in continuity with $(x,y)$ ($(x^\prime,y)$ resp.), $(a,b)\xrightarrow[\mathrm{BC}]{con.}(x,y)$. 
%Recall (\ref{def-overrightarrow-c_A}), in which continuity is presumed.   
By (\ref{xi-1-simuluation-1}),  for any $(a,b)\in M_A$ and $(a,b)\Vdash_{\xi-1}(c,d)$, $((a,b),(x^\flat,y))\\\notin E^A_*$ if and only if $((c,d),(x^{\prime\flat},y))\notin E^B_*$. In other words, (\ref{condition-for-DuWin-*}$^\diamond$) holds if we replace all the $(\xi)$ by $(\xi-1)$ in the superscripts of the sets. 
Therefore, 
\begin{multline}\label{strategy3_abstraction_xi-1_iso}
Z^{\geq\xi-1}\cup U\cup R^{\geq\xi-1}\cup D^{\geq\xi-1}\cup T^{\geq\xi-1}\Vdash\\ Z^{\prime\geq\xi-1}\cup U^{\prime}\cup  R^{\prime\geq\xi-1}\cup D^{\prime\geq\xi-1}\cup T^{\prime\geq\xi-1}.
\end{multline} 

 %Note that, due to Strategy 3, (\ref{condition-for-DuWin-*}$^\diamond$) holds if we replace all the $(\xi)$ by $(\xi-1)$ in the superscripts of the sets.  

\item Because  of Lemma \ref{cm=depth} and (\ref{main-diamond-xi}$^\diamond$) (ii),%\\[-7pt]
\begin{equation}\label{no-missing-edge-before-simuluation-2}
 Z^{<\xi-1}=Z^{\prime<\xi-1}%\\[-4pt] 
\end{equation} 

\item 
Because of Strategy \ref{t<xi}, Duplicator is approximately a hr-copycat (cf. (\ref{main-diamond-xi}$^\diamond$) (iii)). By Lemma \ref{approxi-copy-cat-1}, for any $(u,\!v)\!\in\! \overline{c_A}$ where $\mathrm{idx}(u^\flat\!,v)\!\!=\!\!i^{\star}\!\!<\!\!\xi-1$, 
 $\llbracket u^\flat\rrbracket_{i^{\star}}^{min}\!-\![\llparenthesis u^\flat\rrparenthesis_{\xi-1}]_{i^{\star}}\!=\!\llbracket u^{\prime\flat}\rrbracket_{i^{\star}}^{min}\!-\![\!\llparenthesis u^{\prime\flat}\rrparenthesis_{\xi-1}\!]_{i^{\star}}$. 
 By  
Strategy 3, $\mathrm{idx}(u^{\prime\flat}\!,v)\!=\!i^{\star}$ and  $\mathbf{cc}([u^\flat]_{i^{\star}}\!,v)\!=\!\mathbf{cc}([u^{\prime\flat}]_{i^{\star}}\!,v)$;  
by Lemma \ref{cm=depth}, $[x^\flat]_i\equiv[x^{\prime\flat}]_i\equiv 0$ (mod $k-1$) for any $i<\xi-1$; 
then by Lemma \ref{corollary-approxi-copy-cat} and the definition of ``$\restriction\!\!\Omega$'',  missing of an edge between $(x^\flat,y)$ and $(\llparenthesis u^\flat\rrparenthesis_{\xi-1},v)$, if there is one, would propagate downward to lower abstractions coincidently with missing of an edge between $(x^{\prime\flat},y)$ and $(\llparenthesis u^{\prime\flat}\rrparenthesis_{\xi-1},v)$ that behaves alike.    
Therefore,  we have%\\[-8pt]
\begin{equation}\label{no-missing-edge-before-simuluation-3}
D^{<\xi-1}=D^{\prime<\xi-1}.%\\[-6pt]
\end{equation}

\item Because of (\ref{main-diamond-xi}$^\diamond$) (v), for any pair of pebbled vertices $(u,v)\Vdash (u^\prime,v)$, where  $i^\star=\mathrm{idx}(u,v)=\mathrm{idx}(u^\prime,v)<\xi-1$, we have  $\mathrm{RngNum}(u^\flat,i^\star)=\mathrm{RngNum}(u^{\prime\flat},i^\star)$. 
%On the other hand, $\mathrm{RngNum}(x^\flat,i^\star)=\mathrm{RngNum}(x^{\prime\flat},i^\star)=-1$, by Lemma \ref{rngnum-is_-1}. 
Recall that $\mathrm{idx}(x^\flat,y)\geq \xi>\mathrm{idx}(x^{\prime\flat},y)=\xi-1$. 
Therefore, by definition, 
\begin{equation}\label{equiv-T_xi-1}
T^{<\xi-1}\Vdash T^{\prime<\xi-1}. 
\end{equation}  
\end{enumerate}
   
Putting the observations, i.e. (\ref{strategy3_abstraction_xi-1_iso}), (\ref{xi-1-simuluation-2}), (\ref{no-missing-edge-before-simuluation-2}),  (\ref{no-missing-edge-before-simuluation-3}), and (\ref{equiv-T_xi-1}), all together, it  implies that (\ref{condition-for-DuWin}$^\diamond$) holds for this newly selected vertex $(x^{\prime},y)$. 
After this round, $\xi:=\xi-1$ and $\theta:=\theta-1$. Hence (\ref{theta-lessthan-xi}$^\diamond$) still holds. 
%Since $(x^\flat,y),(x^{\prime\flat},y)\in\mathbb{X}_{\xi-1}^*$, (\ref{main-diamond-xi}$^\diamond$) holds automatically. 
Because Duplicator resorts to lower abstraction, thereby the first part of (\ref{xi-order-condition-hold}$^\diamond$) holds, due to (\ref{xi-order-requirement-ensured}).  And $(\varkappa)$ will not occur. Suppose on the contrary that it occurs, and $(x_0,y)\Vdash (x_0^\prime,y)$ are the pair of vertices that make it happen. Note that $\mathrm{idx}(x_0,y),\mathrm{idx}(x_0^\prime,y)<\xi$. If $x^\flat\neq \llparenthesis x_0^\flat\rrparenthesis_\xi$, then, by induction hypothesis, Duplicator can choose to pick $(x^\prime,y)$ such that $\llparenthesis x^{\prime\flat}\rrparenthesis_\xi\neq \llparenthesis x_0^\flat\rrparenthesis_\xi$. Hence $\lfloor x^{\prime\flat}/l_{\xi-1}\rfloor\neq \lfloor x_0^{\prime\flat}/l_{\xi-1}\rfloor$. A contradiction occurs. If $x^\flat=\llparenthesis x_0^\flat\rrparenthesis_\xi$, then Duplicator simply pick $(x^\prime,y)$ such that $x^{\prime\flat}=\llparenthesis x_0^\flat\rrparenthesis_\xi$ and she wins this round because she wins the last round by induction hypothesis \footnote{It is similar to the situation where she picks a ``pebbled'' vertex if Spoiler picks the corresponding ``pebbled'' vertex in the game over the $\xi$-th abstraction. Here ``pebbled'' vertex can be the projection of a realy pebbled vertex  in the $\xi$-th abstraction.} But in this case $(x^{\prime\flat},y)$ is a vertex in $\mathbb{X}_\xi^*$. That is, she doesn't have to resort to the $(\xi-1)$-th abstraction for a solution, thereby no need to use Strategy \ref{xi-1}. We arrive at a contradiction again.  In short, Duplicator can ensure that (\ref{theta-lessthan-xi}$^\diamond$)\textapprox (\ref{condition-for-DuWin}$^\diamond$) hold in this case.  
%since all the pebbled vertices in $\mathbb{X}_{\xi\!-\!1}$ lie in $\mathbb{X}_{\xi}$ except for one, and their order is already taken care of, because of (\ref{xi-order-requirement-ensured}).%\\[-10pt]

In the arguments, we haven't taken the boundary vertices into account yet, which is mainly handled in the dissussion of (\ref{main-diamond-boundary-strategy}$^\diamond$). Recall that we delay such discussion to the end of this proof. Here we only mention one thing. The decision of $(x^{\prime\flat},y)\!\!\restriction\!\! S$ has not considered the boundary vertices, despite that it should. Nevertheless, it is not a big issue, because Duplicator can adapt her pick in the following simple way when necessary: if the projection of a boundary vertex in the $(\xi-1)$-th abstraction is in $(x^{\flat},y)\!\!\restriction\!\! S$, then she add it in $(x^{\prime\flat},y)\!\!\restriction\!\! S$ too. Note that Duplicator has the freedom to do it. \label{Stra-2-boundary-S} 

Last but not least, the following is easy to observe: Duplicator can win the following one round game wherein she picks $(x^{\prime\flat},y)$ to reply the pick of $(x^\flat,y)$, in accordance with the abstraction-order-condition: 
\begin{equation}\label{eqn-1-round-S-game_xi-1}
\Game_1 ((\widetilde{\mathfrak{A}}_{k,m}^*[\mathbb{X}_{\xi-1}^*]|\langle \leq\rangle,\overrightarrow{c_A}_s^{\xi-1}),(\widetilde{\mathfrak{B}}_{k,m}^*[\mathbb{X}_{\xi-1}^*]|\langle \leq\rangle,\overrightarrow{c_B}_s^{\xi-1}). 
\end{equation} 
Here, $\overrightarrow{c_A}_s^{\xi-1}$ is defined similar to \eqref{def-c_A_s-xi}, except that $\xi$ is replaced by $\xi-1$. $\overrightarrow{c_B}_s^{\xi-1}$ is defined likewise. While playing the game \eqref{eqn-1-round-S-game_xi-1}, we can first regard each $\mathpzc{U}_{\xi-1}^*$-tuple as a unit. 

We've shown that Duplcator can win this round if she resorts to the $(\xi-1)$-th abstraction to respond the picking of $(x^\flat,y)\in\mathbb{X}_\xi^*$. In fact, this strategy also works if $\mathrm{idx}(x^\flat,y)=\xi-1$. The argment is \textit{very} similar to the one just introduced. Duplicator picks $(x^\prime,y)$ such that $\mathrm{idx}(x^{\prime\flat},y)=\xi-1$. In addition, she ensures that $\mathbf{cc}([x^\flat]_{\xi-1},y)=\mathbf{cc}([x^{\prime\flat}]_{\xi-1},y)$,  $g(x^\flat)=g(x^{\prime\flat})$ and $\mathrm{RngNum}(x^\flat,\xi-1)=\mathrm{RngNum}(x^{\prime\flat},\xi-1)$.\footnote{Note that, to this end, we need to adapt Lemma \ref{no-missing-edges_xi-1} a little bit, which is very easy. The point is that, by Lemma \ref{cm=depth}, $\mathbf{cc}(\llparenthesis u^\flat\rrparenthesis_\xi,v)=\mathbf{cc}(x^\flat,y)$ iff $\mathbf{cc}(\llparenthesis u^{\prime\flat}\rrparenthesis_\xi,v)=\mathbf{cc}(x^{\prime\flat},y)$. Therefore, the lemma can be adapted to take care of the situation where $\mathbf{cc}(\llparenthesis u^\flat\rrparenthesis_\xi,v)=\mathbf{cc}(x^\flat,y)$ and  $\mathbf{cc}(\llparenthesis u^{\prime\flat}\rrparenthesis_\xi,v)=\mathbf{cc}(x^{\prime\flat},y)$. Likewise, it is easy to see that the values of $\mathrm{RngNum}(\cdot,\cdot)$ and $g(\cdot)$ will not cause a problem to  \eqref{strategy3_abstraction_xi-1_iso}.}  %$\mathrm{RngNum}(x^\flat,\xi-1)=\mathrm{RngNum}(x^{\prime\flat},\xi-1)$. 
Also cf. the corresponding case (the third case) introduced in the proof of Lemma \ref{winning-strategy-in-k=3}, in page \pageref{page-third-case-k=3}. 

\item \label{t<xi} $\langle$ Strategy 3 $\rangle$ Spoiler picks a vertex $(x,y)$ in $\widetilde{\mathfrak{A}}_{k,m}$ where $\mathrm{idx}(x^\flat,y)=t<\xi-1$.  
%If $t=\xi-1$, then $\xi:=\xi-1$ and Duplicator resorts to the strategy \ref{play-in-xi-abs}. In such case, $\theta<\xi$ still holds since $\theta:=\theta-1$ after each round. Therefore, in the following we only consider the case when $t<\xi-1$.
   Recall that, in the associated game, Spoiler also picks a vertex $(x^\flat,y)$ in $\widetilde{\mathfrak{A}}_{k,m}^*$. Duplicator regards it as if $(\llparenthesis x^\flat\rrparenthesis_\xi,y)$ (or $(\llparenthesis x^\flat\rrparenthesis_{\xi-1},y)$ resp.) is also picked at the same time, and
picks a vertex $(x^{\prime\flat},y)$ whose index is also $t$ in the other structure such that $(\llparenthesis x^{\prime\flat}\rrparenthesis_{\xi},y)$ (or $(\llparenthesis x^{\prime\flat}\rrparenthesis_{\xi-1},y)$ resp.) is the vertex Duplicator will pick to respond the picking of $(\llparenthesis x^\flat\rrparenthesis_\xi,y)$ (or $(\llparenthesis x^{\prime\flat}\rrparenthesis_{\xi-1},y)$ resp.) using  strategy \ref{play-in-xi-abs} or strategy \ref{xi-1} (cf. the last paragraph). 
 It means that, if strategy \ref{play-in-xi-abs} and, in particular,  strategy \ref{xi-1} work well as claimed,  for any pair of pebbles $(u^\flat,v)$ and $(u^{\prime\flat},v)$ on the board, the following holds for $s=\xi$ or $\xi-1$ depending on which strategy is adopted.%\\[-13pt]
\begin{eqnarray}\label{strategy-3-eqn1}
((\llparenthesis x^\flat\rrparenthesis_s,y),\!(\llparenthesis u^\flat\rrparenthesis_s,v))\!\in\! E_*^{A} 
\!& \Leftrightarrow &\!
((\llparenthesis x^{\prime\flat}\rrparenthesis_s,y),\!(\llparenthesis u^{\prime\flat}\rrparenthesis_s,v))\!\in\! E_*^{B} %\nonumber \\[-18pt] 
\end{eqnarray}
That is, (\ref{condition-for-DuWin-*}$^\diamond$) holds. 

We can assume that $s=\xi$. It is similar when $s=\xi-1$. By Remark \ref{special-locally-isom}, the neighbourhood of $(\llparenthesis x^\flat\rrparenthesis_\xi,y)$ is the same as that of $(\llparenthesis x^{\prime\flat}\rrparenthesis_\xi,y)$: they contain a lot of vertices of the same indices and the same coordinate congruence numbers in the same abstractions; moreover, a unit of difference in higher abstraction means a huge difference in lower abstractions w.r.t. distance of first coordinates. 

Duplicator uses the following process to pick $(x^{\prime\flat},y)$ to meet (\ref{main-diamond-xi}$^\diamond$). She first finds a vertex, say $(z_{\xi-1},y)\in\mathbb{X}_1^*$, 
such that $(\llparenthesis z_{\xi-1}\rrparenthesis_{\xi},y)$ is the vertex she will pick to reply the picking of $(\llparenthesis x^\flat\rrparenthesis_{\xi},y)$, and $\mathrm{idx}(z_{\xi-1},y)=\mathrm{idx}(\llparenthesis x^\flat\rrparenthesis_{\xi-1},y)$ if $\mathrm{idx}(\llparenthesis x^\flat\rrparenthesis_{\xi-1},y)<\xi$, and a special variation of (\ref{main-diamond-xi}$^\diamond$) is met where $i$ is fixed to $\xi-1$ and ``$x^{\prime\flat}$'' is replaced by ``$z_{\xi-1}$''. Moreover, $z_{\xi-1}=\llparenthesis z_{\xi-1}\rrparenthesis_{\xi}$ 
if $\llparenthesis x^\flat\rrparenthesis_{\xi-1}=\llparenthesis x^\flat\rrparenthesis_{\xi}$. 
We are able to show that she can find such a vertex in a $\mathpzc{U}_{\xi-1}^*$-tuple. Afterwards she finds a vertex, say  $(z_{\xi-2},y)\in\mathbb{X}_1^*$, 
such that $\llparenthesis z_{\xi-2}\rrparenthesis_{\xi-1}=z_{\xi-1}$,  
$\mathrm{idx}(z_{\xi-2},y)=\mathrm{idx}(\llparenthesis x^\flat\rrparenthesis_{\xi-2},y)$ if $\mathrm{idx}(\llparenthesis x^\flat\rrparenthesis_{\xi-2},y)<\xi$, and a special variation of (\ref{main-diamond-xi}$^\diamond$) is met where $i$ is fixed to $\xi-2$ and ``$x^{\prime\flat}$'' is replaced by ``$z_{\xi-2}$''. Moreover, $z_{\xi-2}=\llparenthesis z_{\xi-2}\rrparenthesis_{\xi-1}$ 
if $\llparenthesis x^\flat\rrparenthesis_{\xi-2}=\llparenthesis x^\flat\rrparenthesis_{\xi-1}$. And so on. 
Note that, once $(z_{\xi-2},y)$ is chosen, the special variation of (\ref{main-diamond-xi}$^\diamond$) is also met where $i$ is fixed to $\xi-1$ and ``$x^{\prime\flat}$'' is replaced by ``$\llparenthesis z_{\xi-2}\rrparenthesis_{\xi-1}$'' (i.e. ``$z_{\xi-1}$''). 
Finally, she picks the vertex $(x^{\prime\flat},y)$ such that $\llparenthesis x^{\prime\flat}\rrparenthesis_{t+1}=z_{t+1}$, 
$\mathrm{idx}(x^{\prime\flat},y)=\mathrm{idx}(x^\flat,y)$, and a special variation of (\ref{main-diamond-xi}$^\diamond$) is met where $i$ is fixed to $t$. 
%Note that, it implies that (\ref{main-diamond-xi}$^\diamond$) is met, since $\llparenthesis x^{\prime\flat}\rrparenthesis_{t+1}=z_{t+1}$.  
Moreover, $x^{\prime\flat}=z_{t+1}$ if $x^\flat=\llparenthesis x^\flat\rrparenthesis_{t+1}$. In short, Duplicator can  use this process to pick the vertex $(x^{\prime\flat},y)$ such that $\llparenthesis x^{\prime\flat}\rrparenthesis_i=z_i$, which implies that (\ref{main-diamond-xi}$^\diamond$) is met.\footnote{Note that, if $\mathrm{S}_i^A\Vdash \mathrm{S}_i^B$ and $\mathrm{S}_i^A=\mathrm{S}_i^B$ for any $t\leq i<\xi$, then (\ref{main-diamond-xi}$^\diamond$)(iii) implies that 
$\llbracket x^\flat\rrbracket_{t}^{min}-[\llparenthesis x^\flat\rrparenthesis_\xi]_{t}\equiv\llbracket x^{\prime\flat}\rrbracket_{t}^{min}-[\llparenthesis x^{\prime\flat}\rrparenthesis_\xi]_{t}$ (mod $\beta_{m-\xi}^{m-r}$). By Remark \ref{special-locally-isom}, this in turn implies that (\ref{main-diamond-xi}$^\diamond$) (i), (ii) hold. With a little more thought, we know that  (\ref{main-diamond-xi}$^\diamond$) (v) (vi) also hold, and (iv) also holds provided that $i>t$. Therefore, by Remark  \ref{remark-ommit-mod}, in this special case Duplicator can also simply pick $(x^{\prime\flat},y)$ s.t. 
$\llbracket x^\flat\rrbracket_{t}^{min}-[\llparenthesis x^\flat\rrparenthesis_\xi]_{t}=\llbracket x^{\prime\flat}\rrbracket_{t}^{min}-[\llparenthesis x^{\prime\flat}\rrparenthesis_\xi]_{t}$.}
   
In the following, we use ``$i=t$'' as an example to explain how to find a vertex satisfies the special variation of (\ref{main-diamond-xi}$^\diamond$) where $i$ is fixed to $t$, provided that $(z_{t+1},y)$ is already determined. Henceforth, when we mention (\ref{main-diamond-xi}$^\diamond$), we mean this special variation unless otherwise specified. The arguments for  other variations  are very similar. At the same time, we show that (\ref{xi-order-condition-hold}$^\diamond$) and (\ref{condition-for-DuWin}$^\diamond$) can be met.  
%show that Duplicator does can ensure (\ref{main-diamond-xi}$^\diamond$) (i)\textapprox (\ref{main-diamond-xi}$^\diamond$) (viii) \textit{simultaneously}. 
Observe that, Duplicator has the freedom to pick a vertex to ensure that  (\ref{main-diamond-xi}$^\diamond$) (ii),(\ref{main-diamond-xi}$^\diamond$) (iii), (\ref{main-diamond-xi}$^\diamond$) (v)  and (\ref{main-diamond-xi}$^\diamond$) (vi) hold simultaneously. (\ref{main-diamond-xi}$^\diamond$) (i) is already met since $t^\prime=t$. Note that, $\llparenthesis x^\flat\rrparenthesis_t=x^\flat$ and $\llparenthesis x^{\prime\flat}\rrparenthesis_t=x^{\prime\flat}$ since $\mathrm{idx}(x^\flat,y)=\mathrm{idx}(x^{\prime\flat},y)=t$.\\[-11pt] 

(\ref{main-diamond-xi}$^\diamond$) (v) implies that ${\mathrm{sgn}((x^\flat\!,y),(u^\flat\!,v))}\!=\!\mathrm{sgn}((x^{\prime\flat},y),(u^{\prime\flat},v))$ for any pebbled pair of vertices $(u,v)\!\Vdash\! (u^{\prime},v)$.
%, since $\mathrm{RngNum}(u^\flat,t)\!=\!\mathrm{RngNum}(u^{\prime\flat},t)$ $=-1$ if $\mathrm{idx}(u^\flat,v)>t$ and $\mathrm{RngNum}(u^\flat,\mathrm{idx}(u^\flat,v))=\mathrm{RngNum}(u^{\prime\flat},(u^\flat,v))$ if $\mathrm{idx}(u^\flat,v)\leq t$. 
It means that $T\Vdash T^\prime$. 
Note that 
(\ref{main-diamond-xi}$^\diamond$) (ii) implies that
 $[x^\flat]_{t}\equiv[x^{\prime\flat}]_{t}$ (mod $k-1$). Therefore,  Duplicator can ensure that  
 $Z\Vdash Z^{\prime}$ (cf. Lemma \ref{cm=depth}) and $U\Vdash U^{\prime}$ (cf. (\ref{main-diamond-xi}$^\diamond$) (vi)). 

%(\ref{main-diamond-xi}$^\diamond$) (iii) is the approximate hr-copycat condition 
% for Lemma \ref{approxi-copy-cat} (cf. Remark \ref{remark-ommit-mod}). This lemma implies that $[x^\flat]_j\equiv [x^{\prime\flat}]_j$ mod $\mathpzc{U}_j^*$ for $t<j<\xi$.  
%($\because  \beta_{m-\xi}^{m-j}=\beta_{m-\xi}^{m-j-1}\beta_{m-j-1}^{m-j}=2^j\mathpzc{U}_j^*\beta_{m-\xi}^{m-j-1}$ and $\beta_{m-\xi}^{m-j-1}\in \mathbf{N}^+$), which in turn means that $\mathbf{cc}([x^\flat]_{j},y)=\mathbf{cc}([x^{\prime\flat}]_{j},y)$ ($\because$ $\mathpzc{U}_j^*$ is divisible by $k-1$). By Lemma \ref{corollary-approxi-copy-cat}, $\mathrm{idx}(\llparenthesis x^\flat\rrparenthesis_j,y)=\mathrm{idx}(\llparenthesis x^{\prime\flat}\rrparenthesis_j,y)$.  
Lemma \ref{cm=depth} and (\ref{main-diamond-xi}$^\diamond$) (ii) imply that any pair of pebbled vertices in $\mathbb{X}_t$ in respective structures have the same coordinate congruence number in the $t$-th abstraction. 
%Moreover, each pair of vertices whose indices are less than $\xi$ are roughly in the same position modulo some amount.  
%All of these, 
Together with  (\ref{strategy-3-eqn1}) and (\ref{main-diamond-xi}$^\diamond$) (i) (the full version), as well as the full version of (\ref{main-diamond-xi}$^\diamond$) (ii),  
it implies that $D^{\geq t}\Vdash D^{\prime\geq t}$  holds, since missing of edges in higher abstractions propagates to lower abstractions coincidently in these two structures. For the similar reason, $D^{<t}\Vdash D^{\prime <t}$ also holds. Hence, we have $D\Vdash\! D^{\prime}$. \label{page-D-Vadash-D-prime}

 In the following we give the justification that (\ref{main-diamond-xi}$^\diamond$) (viii) can be ensured.
   In short, it is because that Duplicator has a winning strategy in the games over sufficiently large pure linear orders. Let $(a,b)$ be defined as in page \pageref{main-diamond-xi}. 
 First, she can make it by ensuring that  $(\llbracket\llparenthesis a^\prime\rrparenthesis_{t}\rrbracket_{t}^{min},b)<(\llbracket x^{\prime\flat}\rrbracket_{t}^{min},y)$ if $(\llbracket\llparenthesis a\rrparenthesis_{t}\rrbracket_{t}^{min},b)<(\llbracket x^\flat\rrbracket_{t}^{min},y)$, or $(\llbracket\llparenthesis a^\prime\rrparenthesis_{t}\rrbracket_{t}^{min},b)>(\llbracket x^{\prime\flat}\rrbracket_{t}^{min},y)$ if $(\llbracket\llparenthesis  a\rrparenthesis_{t}\rrbracket_{t}^{min},b)>(\llbracket x^\flat\rrbracket_{t}^{min},y)$. 
Duplicator can achieve this by an auxiliary game over pure linear orders to determine the value for $\llbracket x^{\prime\flat}\rrbracket_{t}^{min}$. %It is a sort of \textit{game reductions} from a game board of one signature (here is $\langle E,\leq\rangle$) to a board of another signature (here is $\langle \leq \rangle$). 
%Let $\overrightarrow{a_+}$ be a tuple obtained by substituting those pebbled vertices $(u,v)$, whose indices are less than $t$, with $\{(\llparenthesis e\rrparenthesis_t,f)\mid (e,f)\in\chi(u^\flat,v)\!\!\restriction\!\! S\}$ (elements are in the natural order). Similarly, we can define the corresponding $\overrightarrow{b_+}$. Note that $|\overrightarrow{a_+}|=|\overrightarrow{b_+}|\leq (k-2)^2<m$.   
 It means that Duplicator is able to win the following game, wherein Spoiler picks $(x^\flat,y)$ and she replies with $(x^{\prime\flat},y)$. 
\begin{equation}\label{eqn4-1-round-game-reduction} 
\Game_1((\widetilde{\mathfrak{A}}_{k,m}^*[\mathbb{X}_t^*]|\langle \leq \rangle,\overrightarrow{c_A}_s^t),(\widetilde{\mathfrak{B}}_{k,m}^*[\mathbb{X}_t^*]|\langle \leq \rangle,\overrightarrow{c_B}_s^t)).
\end{equation} 
Here, $\overrightarrow{c_A}_s^{t}$ is defined similar to \eqref{def-c_A_s-xi}, except that $\xi$ is replaced by $t$. $\overrightarrow{c_B}_s^{t}$ is defined likewise. Note that,  $|\overrightarrow{c_A}_s^{t}|=|\overrightarrow{c_B}_s^{t}|<m$, and that $\gamma_t^*$, as well as $\gamma_t^*/\mathpzc{U}_{t}^*$, is much greater than $2^m$. Therefore, by Remark \ref{remark-linear-orders-1}, Duplicator has a winning strategy. 
Afterwards, Duplicator resorts to the (virtual) game \eqref{eqn3-1-round-game-reduction} to determine the type label of $(x^{\prime\flat},y)$, which will be introduced soon.  
Second,  assume that $(\llbracket\llparenthesis a\rrparenthesis_{t}\rrbracket_{t}^{min},b)=(\llbracket x^\flat\rrbracket_{t}^{min},y)$ and $(\llbracket\llparenthesis a^\prime\rrparenthesis_{t}\rrbracket_{t}^{min},b)=(\llbracket x^{\prime\flat}\rrbracket_{t}^{min},y)$. It implies that $b=y$, as well as $\llbracket\llparenthesis a\rrparenthesis_{t}\rrbracket_{t}^{min}=\llbracket x^\flat\rrbracket_{t}^{min}$ and $\llbracket\llparenthesis a^\prime\rrparenthesis_{t}\rrbracket_{t}^{min}=\llbracket x^{\prime\flat}\rrbracket_{t}^{min}$. Moreover, Duplicator can make it that  $|\chi(\llparenthesis a
\rrparenthesis_{t},b)\!\!\restriction\!\! S|=|\chi(\llparenthesis a^\prime
\rrparenthesis_{t},b)\!\!\restriction\!\! S|$ and $|\chi(x^\flat,y)\!\!\restriction\!\! S|=|\chi(x^{\prime\flat},y)\!\!\restriction\!\! S|$, provided that $\mathrm{idx}(u^{\flat},v)<\xi-1$ and $\mathrm{idx}(x^{\flat},y)<\xi-1$.  
Clearly, (\ref{main-diamond-xi}$^\diamond$) (viii) holds if $|\chi(x^\flat,y)\!\!\restriction\!\! S| \neq |\chi(\llparenthesis a
\rrparenthesis_{t},b)\!\!\restriction\!\! S|\neq 0$.\footnote{By Lemma \ref{proj-greater-index},  $\mathrm{idx}(\llparenthesis a\rrparenthesis_{t},b)\geq t$. If $\mathrm{idx}(\llparenthesis a\rrparenthesis_{t},b)>t$, it implies that $|\chi(\llparenthesis a
\rrparenthesis_{t},b)\!\!\restriction\!\! S|=|\chi(\llparenthesis a^\prime
\rrparenthesis_{t},b)\!\!\restriction\!\! S|=0$, which means that 
$(a,b)<(x^\flat,y)$ and $(a^\prime,b)<(x^{\prime\flat},y)$. If $\mathrm{idx}(\llparenthesis a\rrparenthesis_{t},b)=t$, then, by definition, $(a,b)<(x^\flat,y)$ if $0<|\chi(\llparenthesis a
\rrparenthesis_{t},b)\!\!\restriction\!\! S|<|\chi(x^\flat,y)\!\!\restriction\!\! S|$; and $(a,b)>(x^\flat,y)$ if $|\chi(\llparenthesis a
\rrparenthesis_{t},b)\!\!\restriction\!\! S|>|\chi(x^\flat,y)\!\!\restriction\!\! S|>0$.} If $|\chi(x^\flat,y)\!\!\restriction\!\! S|=|\chi(\llparenthesis a
\rrparenthesis_{t},b)\!\!\restriction\!\! S|\neq 0$, Duplicator resorts to the (virtual) game \eqref{eqn3-1-round-game-reduction}, which not only ensures (\ref{main-diamond-xi}$^\diamond$) (viii) but also determines the type label of $(x^{\prime\flat},y)$. Finally, note that (\ref{main-diamond-xi}$^\diamond$) (viii) is easy to ensure if at least one of $|\chi(\llparenthesis a
\rrparenthesis_{t},b)\!\!\restriction\!\! S|$ and $|\chi(x^\flat,y)\!\!\restriction\!\! S|$ is $0$.

%Now we show that Duplicator can recursively reduce a game, wherein Spoiler picks a vertex of index $t$, to a game, wherein Spoiler picks a vertex in $\mathbb{X}_{t+1}$, to ensure (\ref{main-diamond-xi}$^\diamond$) (vii) and (\ref{main-diamond-xi}$^\diamond$) (iv) simultaneously.  

(\ref{main-diamond-xi}$^\diamond$) (vii) prevents ($\varkappa$) from occurring. Itself is not difficult to ensure. The problem is whether Duplicator can ensure it without voilating (\ref{main-diamond-xi}$^\diamond$) (iv). In other words, Duplicator need find a way to satisfy (\ref{main-diamond-xi}$^\diamond$) (vii) and (\ref{main-diamond-xi}$^\diamond$) (iv) simultaneously. 
%Recall that the game is over changing board. That is, as explained before (cf. the argument for (\ref{xi-1-simuluation-2})), we need only consider the case wherein $R=\{(e,f)\!\in\! \overrightarrow{c_A} \1 \mathbf{cl}(e,f)\in \chi(x,y)\!\!\restriction\!\!S\}$. Similarly, we have  $R^\prime=\{(e^{\prime},f)\!\in\! \overrightarrow{c_B} \1 \mathbf{cl}(e^\prime,f)\in \chi(x^\prime,y)\!\!\restriction\!\!S\}$. In fact, 
Recall that, we don't have to consider the cases where there is $(a,b)\in\overline{c_A}$ and the length of the board history of $(a,b)$ is greater than that of $(x,y)$, because in such cases Duplicator can resort to the board history of $(a^\prime,b)$, where $(a,b)\Vdash (a^\prime,b)$, to determine $(x^\prime,y)$. Therefore, in the following we assume that $(a,b)\twoheadrightarrow (x,y)$ for any $(a,b)\in\overline{c_A}$.  

Now we show that Duplicator can ensure (\ref{main-diamond-xi}$^\diamond$) (vii) and (\ref{main-diamond-xi}$^\diamond$) (iv) simultaneously. Note that (\ref{main-diamond-xi}$^\diamond$) (vii) will not occur if $\llbracket u^\flat\rrbracket_{t}^{min}\!\neq\!\llbracket x^\flat \rrbracket_{t}^{min}$, or $\llbracket u^\flat\rrbracket_{t}^{min}\!=\!\llbracket x^\flat \rrbracket_{t}^{min}$ but $|\chi(x^\flat,y)\!\!\restriction\!\! S|\neq |\chi(u^\flat,v)\!\!\restriction\!\! S|$. Therefore, we assume that $\llbracket u^\flat\rrbracket_{t}^{min}\!=\!\llbracket x^\flat \rrbracket_{t}^{min}$ and $|\chi(x^\flat,y)\!\!\restriction\!\! S|=|\chi(u^\flat,v)\!\!\restriction\!\! S|$. 

Assume that $|\chi(x^\flat,y)\!\!\restriction\!\! S|=p$. If $p=0$, Duplicator
simply plays the game over abstractions (resort to strategy 1 or 2) so that  $x^\flat-\llparenthesis x^\flat\rrparenthesis_{t+1}=x^{\prime\flat}-\llparenthesis x^{\prime\flat}\rrparenthesis_{t+1}$ (assume that the game is over the $\xi$-th abstraction; it is similar if the game is over the $(\xi-1)$-th abstraction). So (\ref{main-diamond-xi}$^\diamond$) (vii) holds and $\left|\chi(x^{\prime\flat},y)\!\!\restriction\!\! S\right|=0$. The latter, i.e. $\left|\chi(x^{\prime\flat},y)\!\!\restriction\!\! S\right|=|\chi(x^{\flat},y)\!\!\restriction\!\! S|=0$,  implies that (\ref{main-diamond-xi}$^\diamond$) (iv) holds.  Henceforth assume that $p\geq 1$. 
It means that the value of 
$j^\star=\lfloor [x^\flat]_t/(k-1)\rfloor \mbox{ mod } cl_{t+1}^*$ (recall that $cl_{t+1}^*=2|\mathbb{X}_{t+1}^*|+\Sigma_{i=1}^{k-2} |\mathbb{X}_{t+1}^*|^i$) falls in the range $[(|\mathbb{X}_{t+1}^*|+\Sigma_{i=1}^{p-1} |\mathbb{X}_{t+1}^*|^i), (|\mathbb{X}_{t+1}^*|+\Sigma_{i=1}^{p} |\mathbb{X}_{t+1}^*|^i)-1]$. In other words, $j^\star-(|\mathbb{X}_{t+1}^*|+\Sigma_{i=1}^{p-1} |\mathbb{X}_{t+1}^*|^i)$ encodes a $p$-tuple $((x_1,y_1),\ldots,(x_p,y_p))\in |\mathbb{X}_{t+1}^*|^p$. Duplicator resorts to virtual game to determine her pick. She will pick $(x^\prime,y)$ such that 
$(\lfloor [x^{\prime\flat}]_t/(k-1)\rfloor \mbox{ mod } cl_{t+1}^*)-(|\mathbb{X}_{t+1}^*|+\Sigma_{i=1}^{p-1} |\mathbb{X}_{t+1}^*|^i)$ encodes a $p$-tuple $((x_1^\prime,y_1),\ldots,(x_p^\prime,y_p))\in |\mathbb{X}_{t+1}^*|^p$. 
Suppose that $(u_1,y)\ldots (u_r,y)$ are those pebbled vertices s.t. $\llbracket u_i^\flat\rrbracket_t^{min}\!=\!\llbracket x^\flat\rrbracket_t^{min}$ and $|\chi(u_i,y)\!\!\restriction\!\! S|=p$. Similarly, $(u_1^\prime,y)\ldots (u_r^\prime,y)$ are the corresponding vertices where $(u_i,y)\!\Vdash\! (u_i^\prime,y)$ in the original game, which implies that $|\chi(u_i^\prime,y)\!\!\restriction\!\! S|=p$ due to Strategy \ref{t<xi}. 
Clearly, $0\leq r\leq k-2$. Note that $(\varkappa)$ occurs only if $r>0$.

Here we assume that $r>0$. The case $r=0$ is similar.  
Suppose that $\mathrm{idx}(x_j,y_j)=t_j$. 
For any $l\in [1,r]$, let $\lfloor [u_l^\flat]_t/(k-1)\rfloor \mbox{ mod } cl_{t+1}^*-(|\mathbb{X}_{t+1}^*|+\Sigma_{i=1}^{p-1} |\mathbb{X}_{t+1}^*|^i)$ encodes a $p$-tuple $((u_{l1},v_{l1}),\ldots,(u_{lp},v_{lp}))\in |\mathbb{X}_{t+1}^*|^p$. 
Let $$H_{xy}^S:=\{i\mid (x_i,y_i)\in (x,y)[\mathrm{BC}] \mbox{ and }(x_i,y_i) \mbox{ is in the }p\mbox{-tuple}\}.\footnote{Here the $p$-tuple is $((x_1,y_1),\ldots,(x_p,y_p))$.}$$ 
 We can play the following pairs of games. The first one in a pair is an Ehrenfeucht-Fra\"iss\' e\xspace game over pure linear orders, which is used to determine the possible positions that is in accordance with the abstraction-order-condition; the second one is a $1$-round $(k-1)$-pebble game wherein the order is ``ignored'' temporarily and  Spoiler picks $(x_j,y_j)$ and she replies with  $(x_j^{\prime},y_j)$ in the $j$-th round, for $j\in [1,p]-H_{xy}^S$. Note that the second game is used to determine the type label of $(x^{\prime\flat},y)$. 
 Moreover, for the games over pure linear orders, they are played \textit{successively}\footnote{By default, once a vertex is picked, a pebble is on it unless the players lift the pebble. 
 As a consequence, the picking of $(x_{j_1},y_{j_1})$ will influence the picking of $(x_{j_2},y_{j_2})$, if $j_1,j_2\in [1,p]-H_{xy}^S$ and $(x_{j_2},y_{j_2})$ is picked later. Note that, the purpose of playing the games successively is that Duplictor need to ensure that $(u_{ij_1}^\prime,v_{ij_1})\leq (u_{ij_2}^\prime,v_{ij_2})$ iff $(x_{j_1}^\prime,y_{j_1})\leq (x_{j_2}^\prime,y_{j_2})$ where $j_1,j_2\in [1,p]$ and $(x_{j_1}^\prime,y_{j_1}),(x_{j_2}^\prime,y_{j_2})$ are in the $p$-tuple associated to $\chi(x^\prime,y)\!\!\restriction\!\! S$. Therefore, the order with which these games are played successively is not important.}; by contrast, for the pebble games, they are played \textit{independently}. By combining these two games Duplicator can decide what should  $(x_j^\prime,y_j)$ be. More precisely, in the first game, Duplicator uses the following strategy.    
 \begin{itemize}
 \item If $t_j\geq \xi$, then Duplicator resorts to \eqref{eqn-1-round-S-game_xi} or \eqref{eqn-1-round-S-game_xi-1} to determine the unabridged interval where $(x_j^\prime,y_j)$ should reside. 
 
 \item If $t_j<\xi$, then Duplicator first uses the game \ref{eqn4-1-round-game-reduction} to determine $(\llbracket x_j^\prime\rrbracket_{t_j}^{min},y_j)$.  

\end{itemize} 
In the second game, Duplicator follows the following strategy.
 \begin{itemize}
 \item If $t_j\geq \xi$, then Duplicator resorts to Strategy \ref{play-in-xi-abs}  and Strategy \ref{xi-1} to determine the type label of  $(x_j^\prime,y_j)$. Cf., in particular, \eqref{xi-1-simuluation-1}.  
 
 \item If $t_j<\xi$, then Duplicator uses the following one round (virtual) game \ref{eqn3-1-round-game-reduction} to determine the type label of $(x_j^\prime,y_j)$. 
\begin{equation}\label{eqn3-1-round-game-reduction}
\Game_1^{k-1}((\widetilde{\mathfrak{A}}_{k,m}^*,(x,y)[\mathrm{BC}]),(\widetilde{\mathfrak{B}}_{k,m}^*,(x^\prime,y)[\mathrm{BC}])) 
\end{equation}  
Note that in this case it may incur recursive calls. In the recursion a pair of games will be played to decide one pick, as just described. 
It is a sort of \textit{game reductions} from lower abstraction to higher abstraction because $\mathrm{idx}(x_j^\prime,y_j)>t$, and will finally return a valid type label for $(x_j^{\prime},y_j)$ since $\mathrm{idx}(x_j^\prime,y_j)$ cannot be greater than $m$. Note that, in the recursions, ``$(x,y)[\mathrm{BC}]$'', as well as ``$(x^\prime,y)[\mathrm{BC}]$'', is fixed in \eqref{eqn3-1-round-game-reduction}.      

\end{itemize}

 Moreover, because  the game boards are in partial isomorphism at the end of the games, it means that the following holds when $j\notin H_{xy}^S$.  
 \begin{equation}\label{eqn2-1-round-game-reduction}
 u_{ij}\leq x_j \mbox{ if and only if } u_{ij}^\prime\leq x_j^\prime. 
\end{equation}   
 Note that $(x_j,y_j)\notin R$ and $(x_j^\prime,y_j)\notin R^\prime$, since $j\notin H_{xy}^S$.   

 For any $j\in H_{xy}^S$, Duplicator can determine $(x_j^\prime,y_j)$ such that $(x_j,y_j)\Vdash (x_j^\prime,y_j)$ in the virtual game that she uses to determine $(x^\prime,y)[\mathrm{BC}]$. It implies that (\ref{main-diamond-xi}$^\diamond$) (iv) holds, since $(x_j,y_j)\in R$ iff $(x_j^\prime,y_j)\in R^\prime$ for $j\in [1,p]$. To see that (\ref{main-diamond-xi}$^\diamond$) (vii) also holds, we need only show that \eqref{eqn2-1-round-game-reduction} also holds when $j\in H_{xy}^S$. 
If both of $(u_{ij},v_{ij})$ and $(x_j,y_j)$ are in $(x,y)[\mathrm{BC}]$, then it is determined by the fact that Duplicator wins the virtual game that determines the board history. Now suppose that $(u_{ij},v_{ij})\notin (x,y)[\mathrm{BC}]$ and $(x_j,y_j)\in (x,y)[\mathrm{BC}]$. There are two cases. First, assume that $(x_j,y_j)$ is picked before $(u_i,v_i)$ in the virtual game that determines the board history of $(x^\prime,y)$. Obviously, \eqref{eqn2-1-round-game-reduction} holds when $(u_i,v_i)$ is picked in the original game: 
%Duplicator wins the game \ref{eqn3-1-round-game-reduction} in the case $r=0$, which serves the base case in an inductive argument; 
the point is that $(x_j,y_j)$ belongs to $(u_i,v_i)[\mathrm{BC}]$.   
 Second,  assume that $(x_j,y_j)$ is picked after $(u_i,v_i)$ in the virtual game. Recall that  $j\in H_{xy}^S$ and that $\mathrm{idx}(x_j,y_j)=t_j$.  If $t_j\geq \xi-1$, then \eqref{eqn2-1-round-game-reduction} holds, due to \eqref{eqn-1-round-S-game_xi} or \eqref{eqn-1-round-S-game_xi-1}. Suppose that $t_j<\xi-1$. Then $\mathrm{idx}(x_j^\prime,y_j)$ should be $t_j$, due to Strategy \ref{t<xi}. In such case \eqref{eqn2-1-round-game-reduction} still holds because (\ref{main-diamond-xi}$^\diamond$) (viii) can be ensured.\footnote{Note that, in the applying of (\ref{main-diamond-xi}$^\diamond$) (viii) and its argument, ``$(x^\flat,y)$'' and $(x^{\prime\flat},y)$ in (\ref{main-diamond-xi}$^\diamond$) (viii) should be replaced by $(x_j,y_j)$ and $(x_j^\prime,y_j)$ respectively; and ``$(u,v)$'' should be replaced by $(u_i,y)$ for some $i\in [1,r]$.}    

With this somewhat sophisticated argument we have shown that (\ref{xi-order-condition-hold}$^\diamond$) can be ensured, in accordance with (\ref{main-diamond-xi}$^\diamond$), in particular (iv), and with  (\ref{condition-for-DuWin}$^\diamond$). 

%In summary, (\ref{condition-for-DuWin}$^\diamond$) holds, in accordance with (\ref{xi-order-condition-hold}$^\diamond$). 

Note that, after this round, $\xi$ is either unchanged or decreased by one.
% and Duplicator continues to use the strategy that works in the current abstraction to  respond to Spoiler.  
And, as usual, $\theta:=\theta-1$. Hence (\ref{theta-lessthan-xi}$^\diamond$)  holds.  %\\[-25pt]
\end{enumerate}

Now we need to check whether (\ref{main-diamond-boundary-strategy}$^\diamond$) can be ensured. Recall that we need switch back to the game board  $((\widetilde{\mathfrak{A}}_{k,m},\overline{c_A}),(\widetilde{\mathfrak{B}}_{k,m},\overline{c_B}))$ for the discussion. 
To simplify the following discussion, we can safely assume that Duplicator's picks ensure that $\mathrm{idx}(x^{\prime\flat},y)\leq \mathrm{idx}(x^\flat,y)$, and $\mathrm{idx}(x^{\prime\flat},y)=\xi-1$ if $\mathrm{idx}(x^{\prime\flat},y)<\mathrm{idx}(x^\flat,y)$. Recall that $\mathrm{idx}(x^\flat,y)=t$ and $\mathrm{idx}(x^{\prime\flat},y)=t^\prime$. 
Note that  
$[x^{\prime\flat}]_{\xi-1}\equiv 0$ (mod $k-1$) if $t\geq \xi$.    
Firstly, assume that $t>t^\prime=\xi-1$. The case where $\mathrm{idx}(x^\prime,y)=\mathrm{idx}(x,y)\geq \xi$ and $\mathrm{idx}(x^\prime,y)=\mathrm{idx}(x,y)=\xi-1$ are similar.\footnote{In the latter case, just observe that the values of $\mathrm{RngNum}(\cdot,\cdot)$, $g(\cdot)$ and $\restriction\!\! S$ will not cause a problem to (\ref{main-diamond-boundary-strategy}$^\diamond$), provided Duplicator uses Strategy \ref{xi-1} (cf. the last paragraph of it).} There are several cases need to discuss. 
Suppose that $0<y<k-1$. 
Then, for any boundary vertex $(a,b)\in\mathbb{X}_1$ where $0<b<k-1$,  
\begin{equation}\label{strategy-for-boundary-attack}
(\llparenthesis a^\flat\rrparenthesis_{\xi-1},b)\notin\chi(x^{\prime\flat},y)\!\!\restriction \!\Omega  \mbox{ and }(\llparenthesis a^\flat\rrparenthesis_{\xi-1},b)\notin\chi(x^\flat,y)\!\!\restriction \!\Omega,
\end{equation} 
simply because $[\llparenthesis a^\flat\rrparenthesis_{\xi-1}]_{\xi-1}\not\equiv 0$ (mod $k-1$) (recall that $a^\flat=\gamma_{m-1}^*-tr(b)$ and $[\llparenthesis a^\flat\rrparenthesis_{\xi-1}]_{\xi-1}=[a^\flat]_{\xi-1}\not\equiv 0$ (mod $k-1$); cf. \eqref{boundary-nequiv-0-mod-2}). If $b=0$ or $b=k-1$, again by definition, \eqref{strategy-for-boundary-attack} holds simply because $\mathrm{idx}(\llparenthesis a^\flat\rrparenthesis_{\xi-1},b)=\xi-1=t^\prime<t$ (cf. Lemma \ref{boundary-index-over-abstractions}). In addition, by definition, $\mathrm{sgn}((\llparenthesis a^\flat\rrparenthesis_{\xi-1},b),(x^\flat,y))=1$ and $\mathrm{sgn}((\llparenthesis a^\flat\rrparenthesis_{\xi-1},b),(x^{\prime\flat},y))=1$.   
By Lemma \ref{cm=depth}, $\mathbf{cc}([x^\flat]_{\xi-1},y)=y$ mod ($k-1$), which equals  $\mathbf{cc}([x^{\prime\flat}]_{\xi-1},y)$. According to Strategy \ref{xi-1}  (cf. p. \pageref{Stra-2-boundary-S}) $(\llparenthesis a^\flat\rrparenthesis_{\xi-1},b)\in \chi(x^{\prime\flat},y)\!\!\restriction \! S$ iff $(\llparenthesis a^\flat\rrparenthesis_{\xi-1},b)\in \chi(x^\flat,y)\!\!\restriction \! S$. 
Furthermore, $g(\llparenthesis a^\flat\rrparenthesis_{\xi-1})=0$ and $g(x^{\prime\flat})$ can be chosen to $0$ (cf. Lemma \ref{no-missing-edges_xi-1}), which implies that $\mathrm{SW}((\llparenthesis a^\flat\rrparenthesis_{\xi-1},b),(x^{\prime\flat},y))=0$. It means that $(\mathbf{cc}([\llparenthesis a^\flat\rrparenthesis_{\xi-1}]_{\xi-1},b)-\mathbf{cc}([x^{\prime\flat}]_{\xi-1},y))\times  
(b-y)\times(-1)^{\mathbf{BIT}(\mathrm{SW}((\llparenthesis a^\flat\rrparenthesis_{\xi-1},b),(x^{\prime\flat},y)),\hat{q}(b,y))}>0$.  
Therefore, $(\llparenthesis a^\flat\rrparenthesis_{\xi-1},b)$ is adjacent to $(x^{\prime\flat},y)$ if and only if $(\llparenthesis a^\flat\rrparenthesis_{\xi-1},b)$ is adjacent to $(x^\flat,y)$. 
Now suppose that $y=0$ or $k-1$. Because $\mathbf{cc}([x^{\prime\flat}]_{\xi-1},y)=\mathbf{cc}([x^\flat]_{\xi-1},y)=\mathbf{cc}([\llparenthesis a^\flat\rrparenthesis_{\xi-1}]_{\xi-1},b)=0$, $(\llparenthesis a^\flat\rrparenthesis_{\xi-1},b)$ is not adjacent to $(x^{\prime\flat},y)$ and $(x^\flat,y)$. \label{discussion-boundary-checkout} 

Secondly, assume that $t=t^\prime<\xi$ (cf. Strategy \ref{t<xi}). The arguments are similar except one point. By the last arguments, now we can ensure that either (i) $(\llparenthesis a\rrparenthesis_\xi,b)$ is adjacent to $(\llparenthesis x^\flat\rrparenthesis_\xi,y)$ if and only if $(\llparenthesis a\rrparenthesis_\xi,b)$, in the other structure, is adjacent to $(\llparenthesis x^{\prime\flat}\rrparenthesis_\xi,y)$, or (ii) similar to (i) but $\xi$ is substituted with $\xi-1$. Suppose that (i) holds. The other case is similar. The point is that, if $b=0$ or $k-1$, we need to show that $(\llparenthesis a\rrparenthesis_\xi,b)\in\chi(x^\flat,y)\!\!\restriction\!\!\Omega$ if and only if $(\llparenthesis a\rrparenthesis_\xi,b)\in\chi(x^{\prime\flat},y)\!\!\restriction\!\!\Omega$. But this is clear now, the argument needed is similar to the one that we show that $D\Vdash D^\prime$ (cf. page \pageref{page-D-Vadash-D-prime}). 

In summary, the boundary checkout strategy is not effective for Spoiler over the ``flat'' game board $((\widetilde{\mathfrak{A}}_{k,m}^*,(x,y)[\mathrm{BC}]),(\widetilde{\mathfrak{B}}_{k,m}^*,(x^\prime,y)[\mathrm{BC}]))$ or over the ``flat'' changing boards. It implies that the boundary checkout strategy is not effective for Spoiler over the game board $((\widetilde{\mathfrak{A}}_{k,m},\overline{c_A}),(\widetilde{\mathfrak{B}}_{k,m},\overline{c_B}))$ as well. 
%Therefore, we shall not mention the boundary checkout strategy anymore in the following arguments (with an exception, i.e. in the discussion of Strategy \ref{t<xi}).  

All in all, Duplicator can ensure that (\ref{theta-lessthan-xi}$^\diamond$)\textapprox (\ref{condition-for-DuWin}$^\diamond$) hold throughout the game. 
\end{proof}

Note that the game board may change for each round in a virtual game. Indeed, if game board change is allowed in a normal game, then there is no need to introduce $(\widetilde{\mathfrak{A}}_{k,m},\widetilde{\mathfrak{B}}_{k,m})$ at all. The strategy described in the proof of Lemma \ref{main-lemma} work over the ``flat'' associated board  $(\widetilde{\mathfrak{A}}_{k,m}^*,\widetilde{\mathfrak{B}}_{k,m}^*)$ directly if board change is allowed. In other words, games over changing bord is the basis of all the other games. In such games, we use some auxiliary games to help Duplicator make her decision. In particular, in Strategy 3, we have used a technique, called game reductions, to prevent Spoiler from winning the associated game simply by picking continuously inside a $\mathpzc{U}_i^*$-tuple for some $i$ (cf. \eqref{eqn4-1-round-game-reduction} and \eqref{eqn3-1-round-game-reduction}). 
We also use them to determine Duplicator's picks in the virtual games wherein Strategy \ref{t<xi} applies.  

Fig. \ref{Fig-reduction-between-abstractions} gives some hints on a way Spoiler can use to detect the difference between structures $\widetilde{\mathfrak{A}}_{k,m}^*$ and $\widetilde{\mathfrak{B}}_{k,m}^*$, wherein Duplicator has to resort to game reductions to reply properly. In the figure, a node is either black or red representing un-pebbled (in black) and pebbled (in red) vertices respectively. We assume that the vertex ``$\mathrm{a}$'' is picked in the last round and the vertex ``$\mathrm{b}$'' is picked in the current round. 
The dotted (in blue) and dashed (in red) arrows together indicate the set ``$\restriction\!\! S$'' associated with a vertex.  Moreover, 
we use dashed arrows to indicate the set $R$. Here we regard the vertex b as the vertex ``$(x^\flat,y)$''. Note that, although 
dotted arrows \textit{and} dashed arrows are used to indicate the edges forbidden in $\widetilde{\mathfrak{A}}_{k,m}^*$ or $\widetilde{\mathfrak{B}}_{k,m}^*$, they are not necessary forbidden in the changed boards, wherein only the dashed arrow have to be forbidden.  For instance, from the figure we know that both $\mathbf{cl}(\mathrm{c})$ and $\mathbf{cl}(\mathrm{d})$ are in $\chi(\mathrm{b})\!\!\restriction\!\!S$, or $\chi(x^\flat,y)\!\!\restriction\!\!S$; and  $(x,y)\rightsquigarrow \mathrm{c}^\star$ where $\mathrm{c}^\star$ is in $\overline{c_A}$ such that $\mathrm{c}$ is the corresponding vertex in $\overrightarrow{c_A}$.  
After playing the auxiliary virtual games over the linear orders, Duplicator is able to determine $b^\prime$, i.e. $(x^{\prime\flat},y)$, which she should pick to respond the picking of $b$ by Spoiler. That is, we use the auxiliary games to determine the picks of Duplicator in the game over changing board that is associated to the original game. 

Spoiler can make it that $\mathrm{a}$ and $\mathrm{b}$ be very closed. 
For example, $\mathrm{b}=\mathrm{a}-1$. In such case we have $\mathrm{b}^\prime=\mathrm{a}^\prime-1$, provided that Duplicator responds properly, i.e. preserving the abstraction-order-condition. 
Note that, in the case $\mathrm{b}=\mathrm{a}-1$, the result of virtual games should witness that $\mathrm{b}^\prime=\mathrm{a}^\prime-1$. 

%\begin{comment}
\begin{figure}[]
\hspace*{11mm}
%\centering
\includegraphics[trim = 0mm 0mm 0mm 0mm, clip, scale=0.43]{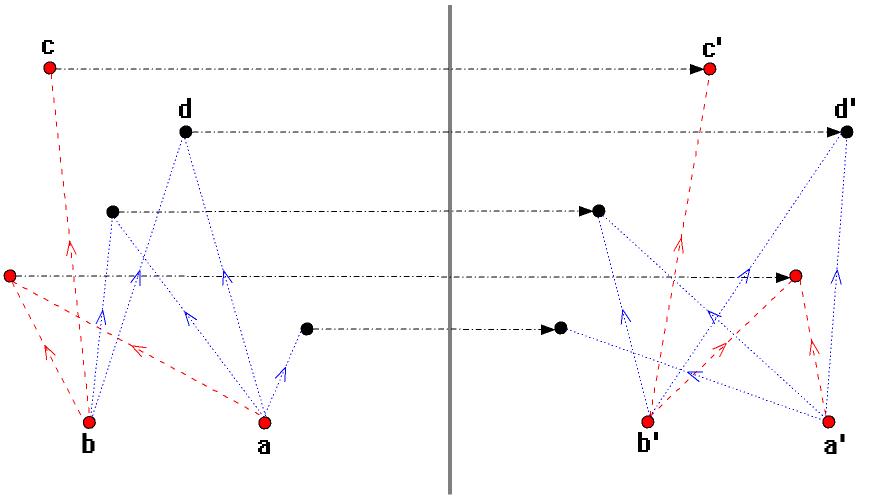}
%\scalebox{10}{}
\caption{The left part is in $\widetilde{\mathfrak{A}}_{k,m}^*$; the right part is in $\widetilde{\mathfrak{B}}_{k,m}^*$. Suppose that $\mathrm{idx}(\mathrm{a})=\mathrm{idx}(\mathrm{b})=t<\xi$, and that $\mathrm{a}$ is picked in the last round and $\mathrm{b}$ is picked in the current round. The ponit is that $\mathrm{a}$ and $\mathrm{b}$ can be very closed.}
\label{Fig-reduction-between-abstractions}
\end{figure}
%\end{comment}  

Now we are able to prove our main result of this paper, using Lemma \ref{main-lemma}. 
\begin{theorem}\label{main-theorem}
For any $k$, $k$ variables are necessary and sufficient to describe $k$-Clique  in $\fo$ on finite ordered graphs. 
\end{theorem}
\begin{proof}
Suppose that there is a $\mathcal{L}^{k^{\prime}}$ formula, where $k^{\prime}<k$, to describe $k$-Clique, and assume that its quantifier rank is $m$. We can safely assume that $k^{\prime}<k\leq m$ and $(k-2)^2<m$ if $k>3$. If it is not true, we can define another logically equivalent $\mathcal{L}^{k-1}$ formula by artificially increasing the quantifier rank of the formula. Consequently, Spoiler has a winning strategy in the game $\Game_m^{k-1}(\widetilde{\mathfrak{A}}_{k,m},\widetilde{\mathfrak{B}}_{k,m})$, which is in contradiction to Lemma \ref{winning-strategy-in-k=3} and 
 Lemma \ref{main-lemma}.\footnote{Recall that, the simple case where $k=2$ is already proved at the start of section \ref{special-cases}.} Therefore, $k$ variables are needed to define $k$-Clique over finite ordered graphs. 

On the other hand, the following first-order formula describes $k$-Clique:
$\exists x_1\cdots\exists x_k \bigwedge_{i\neq j} (\lnot(x_i=x_j)\land E(x_i,x_j))$. Here we use $x_i$ to denote a vertex. Therefore, $k$ variables are sufficient to define $k$-Clique.
\end{proof}

As a direct corollary of Theorem \ref{main-theorem}, we have the following well-known result, which was first proved by Rossman \cite{RossmanStoc}. 
\begin{corollary}
The bounded variable hierarchy in $\fo$ is strict.
\end{corollary}
 That is, for any $k$, over the finite ordered graphs there is a property that is expressible by $k$ variables, but not expressible by $k-1$ variables in $\fo$.\footnote{When infinite ordered structures are concerned, it was proved by Venema \cite{Venema1990Infinite}.} In other words, first order logic needs infinite many variables. 
 Hence we have given an alternative proof for this important result in finite model theory, based on pure finite model-theoretic tools.

\section{Worst-case lower bound of $k$-Clique on constant-depth circuits}
\label{chapter-circuit-bound}

Recall that an ordered graph is a graph with a linear order in the background. In this section, we fist show that precisely $k$ variables are needed to define $k$-Clique on the class of graphs with arbitrary arithmetic background relations, a result akin to Theorem \ref{main-theorem} except that the graphs are equipped with $\mathbf{BIT}$ in the background. It is another well-known challenge to play pebble games on such kind of structures, which has its root in circuit complexity \cite{Immerman1999Book}. Note that $\mathbf{BIT}$ predicate can be used to define arbitrary arithmetic predicates including linear orders (for a survey, cf. \cite{Schweikardt05Arithmetic}). 
Supprisingly, due to a work of Schweikardt and Schwentick \cite{SchweikardtLinearOrder}, it turns out that this challenge is very similar to the challenge caused by linear orders. 
%That is, in terms of $\fo$, $\mathbf{BIT}$ is equivalent to two special linear orders. 
In section \ref{lowerbound-in-FO(BIT)} we embeds the main structures in section \ref{section-structures} in the pure arithmetic structure in \cite{SchweikardtLinearOrder} to show that $k$ variables are needed to define $k$-Clique in $\fo(\mathbf{BIT})$. 
Afterwards, in section \ref{circuit-lowerbound} we show that this result implies a worst-case lower bound of $k$-Clique on constant-depth unbounded fan-in circuits.  

Here we assume that the readers are familiar with the ideas and notations in the paper \cite{SchweikardtLinearOrder}, wherein a sort of clever construction is presented to show that $\mathbf{BIT}$ can be replaced by two special linear orders. 
Because of its constructive nature, it offers a tool that can bridge the gap between pure linear orders and pure arithmetic. 
Abuse of notations, we also use $\prec$ to denote one of the linear orders, as in \cite{SchweikardtLinearOrder}. The readers should not confuse it with the induced linear order $\preceq^{\xi}$ introduced in the main Lemma \ref{main-lemma}. Another linear order is $<$.   

In the following, we briefly sketch related basic ideas of \cite{SchweikardtLinearOrder}. The paper \cite{SchweikardtLinearOrder} introduced a pure arithmetic structure whose elements are ordered in a specific way. We can also regard it as a set of isolated vertices ordered and organized as an isosceles right triangle (cf. Fig 1 of \cite{SchweikardtLinearOrder}) in a two dimension coordinate plane. Note that, there is a bijection between these two  universes.  
For any vertices $(x_1,y_1)$ and $(x_2,y_2)$, $(x_1,y_1)<(x_2,y_2)$ if $x_1<x_2$ or $(x_1=x_2 \mbox { and } y_1<y_2)$. This is called the bottom-to-top, left-to-right, column major order. On the other hand, we define $\prec_0$ as the left-to-right, bottom-to-top, row major order: $(x_1,y_1)<(x_2,y_2)$ if $y_1<y_2$ or $(y_1=y_2 \mbox { and } x_1<x_2)$. Confer p3 of \cite{SchweikardtLinearOrder}. Furthermore, we introduce two unary relations $C$ and $Q$. We use $C$ to encode a binary representation of $x+1$, and use $Q$ to encode a binary representation of $\binom{x+2}{2}$.   
Schweikardt and Schwentick showed that $\fo(<,\prec_0, C,Q)$ has the same expressive power as $\fo(\mathbf{BIT})$.  

Then it is shown that $\fo(<,\prec_0, C,Q)$ is equivalent to $\fo(<,\prec)$ in terms of expressive power, where    
$\prec$ is a special linear order that can be used to encode the two unary relations. More precisely, if $\ell$ is sufficiently large, we can use the order $\prec$ on every \textit{complete interval} (cf. p10 of \cite{SchweikardtLinearOrder} for this important concept) $\{(x,y+1),\ldots,(x,y+\ell-1)\}$ to encode $C, Q$ on the elements $(x,y),(x,y+1),\ldots,(x,y+3\ell-1)$. Note that an order corresponds to a permutation, say $\pi_i$, where $i$ is represented by a $0$-$1$ string of length $6\ell$. This string can be used to encode a unary relation on the $6\ell$ elements. Hence, for any $(x,y)$, we are able to know whether $(x,y)\in C$ or $(x,y)\in Q$. 
  
\subsection{$k$ variables are needed to define $k$-Clique in $\fo(\mathbf{BIT})$}
\label{lowerbound-in-FO(BIT)}
We first briefly outline the ideas. Suppose that we have a pair of sufficiently large isomorphic arithmetic structures as described in \cite{SchweikardtLinearOrder}. Here, by ``sufficiently large'' we mean that $\ell$ is such a big number that we can embed our main structure $\widetilde{\mathfrak{A}}_{k,m}$ or $\widetilde{\mathfrak{B}}_{k,m}$ (cf. section \ref{section-structures}) into a piece of a complete interval where those two linear orders, i.e. $<$ and $\prec$, coincide: we embed a copy of the structure into it in the similar way we make the list $\mathfrak{L}_1$ (cf. p. \pageref{page-def-L-lists-for-cl}) with a difference: instead of fixing the second coordinates, now we fix the first coordinates to a constant.   

Due to Stirling's formula, we assume that $n!=e^{c_{n}\cdot n\ln n}=2^{(\log_2 e) c_n\cdot n\ln n}$, for some $c_{n}\in\mathbf{R}^+$. 
Note that, $c_{n}\rightarrow 1$ if $n\rightarrow \infty$. In fact,  $c_{n}\approx \frac{9}{10}$ when $n\geq e^{10}$.  
We choose a sufficiently large natural number $\ell$ such that $\ell> max\left\{e^{\frac{5}{c_{\ell-1}}}+1, \gamma_{m-1}\right\}$. 
Recall that $\gamma_{m-1}$ is a number depending only on $k$ and $m$. 
As a consequence,  
$(\ell-1)!>2^{(\log_2 e)c_{\ell-1} (\ell-1) \ln e^{\frac{5}{c_{\ell-1}}}}=2^{5(\log_2 e)(\ell-1)}> 2^{6\ell}$, and $\ell>\gamma_{m-1}$. 

Therefore, these two structures are roughly a (huge) set of isolated vertices in the form of an isosceles right triangle except that, in a pair of corresponding complete intervals, there lies the pair of twisted structures isomorphic to $\widetilde{\mathfrak{A}}_{k,m}$ and $\widetilde{\mathfrak{B}}_{k,m}$ respectively. To make it easier, we assign the permutaion of this pair of complete intervals,  say $\pi_i$ for some $i\in\{0,1\}^{6\ell}$, to an order that is isomorphic to the natural order of an initial segment of $\mathbf{N}_0$. Note that, the linear orders defined in $\widetilde{\mathfrak{A}}_{k,m}$ and $\widetilde{\mathfrak{B}}_{k,m}$ are also isomorphic to the natural order of an initial segment of $\mathbf{N}_0$, thereby isomorphic to this order. It allows us to embed $\widetilde{\mathfrak{A}}_{k,m}$ and $\widetilde{\mathfrak{B}}_{k,m}$ into this pair of complete intervals without rearranging the vertices.   

It remains to show that Duplicator has a winning strategy over this pair of sparse structures. To this end, we refer to a simple strategy composition as follows. Observe that the embedded structures  $\widetilde{\mathfrak{A}}_{k,m}$ and $\widetilde{\mathfrak{B}}_{k,m}$ are disjoined with other parts of the structures. As a consequence, when Spoiler picks in an embedded structure, Duplicator resorts to the strategy introduced in Section \ref{winning-strategy}, since $<$ and $\prec$ coincide (hence reduced to one order); when he picks in other parts of the structure, Duplicator is simply a copycat. Obviously, such a composed strategy works for Duplicator in the $m$-round $k$-pebble games. Afterwards, using an argument similar to that of Theorem \ref{main-theorem}, we arrive at our claim. 
\begin{theorem}\label{thm-k-variable-BIT}
$k$ variables are needed for $k$-Clique in $\fo(\mathbf{BIT})$.
\end{theorem}
 Note that this lower bound is optimal, as the one given by Theorem \ref{main-theorem}.\\     

\noindent\textit{\textbf{Further Remark}}
 
Could we reduce the size of the astonishing huge structures just introduced? Or what is the tight lower bound on such a size?  
Although we have proved our claim constructively, we believe that there exist smaller constructions. To have such a construction coincide with the ideas of \cite{SchweikardtLinearOrder}, in our viewpoint, it would require an undertanding of the patterns of binary encodings of $C$ and $Q$, which is not trivial. Obviously, it is preferable if the patterns are periodical. 
To have an impression of how mysterious would such patterns be, the readers could read, for example, a paper by Rowland \cite{Rowland2009BinaryRegularity}. 
%It is perhaps related to something about $2$-adic number \cite{Rowland2009BinaryRegularity}. 
%To have an impression of how mysterious would such a pattern be, the readers could have a look at the book (\cite{Wolfram2002NewScience}, Chapter 4).   

\subsection{Size lower bound of $k$-Clique on constant-depth circuits}
\label{circuit-lowerbound}

We assume the readers know standard concepts and notations in circuit complexity.  
It is well-known that $\fo$ is closely related to constant-depth circuits \cite{Barringtion1990Uniformity,DawarHowmany,Denenberg1986circuits,Gurevich1984Circuit2Logic,Immerman1983languageThat,Immerman1989Parallel,Immerman1999Book}. In particular, the introducing of $\mathbf{BIT}$ as a background relation allows a suitable form of uniformity for circuit families, thereby establishing the equivalence between $\fo(\mathbf{BIT})$ and \textit{DLOGTIME-uniform} $\mathrm{AC}^0$ (or \textit{FO-uniform} $\mathrm{AC}^0$). In this section we always assume the presence of $\mathbf{BIT}$. Moreover, we always assume that $k\geq 5$ and that $n$ is the cardinality of the vertex set of input graph.   
Since we resort frequently to \cite{Denenberg1986circuits,Barringtion1990Uniformity} for inspiration and insights,  we asume that the readers are familiar with the notions introduced in these papers.
In particular, we will use a key notion called \textit{regular circuit} (cf. \cite{Denenberg1986circuits}, p237), which is defined as follows. By DeMorgan's law, we can assume that the negations only appear in the input level of circuits. It will not influence the size of a circuit significantly. 

A circuit takes an (ordinary) encoding of a structure as its input. Here we only talk about graphs, and by ``ordinary'' we mean the usual binary encodings of graphs.  The order of a graph is the cardinality of its vertex set. The \textit{order} of a circuit is the order of the input graph. Usually we use $C_n$ to denote a circuit of order $n$. A circuit $C_n$ is formatted w.r.t. $n$ and $\langle E\rangle$. That is, in the context of our concern, there is a surjection from the inputs to atoms $E(a,b)$, $a=b$, $\mathbf{BIT}(a,b)$ or their negations where $a, b$ are vertices of the input graph.  %for simplicity we may assume safely that there are exactly $n^2$ (or $\binom{n}{2}$, since the graph is undirected and free of self-loops) such inputs;
Note that, those arithmetic literals, e.g. $a=b$ and $\mathbf{BIT}(a,b)$, can be replaced by two constant inputs $0$ and $1$, because their values are independent of the inputs.  

A circuit $C_n$ is regular if the following hold. 
\begin{enumerate}[(1)]
\item Its structure is symmetric (satisfying some conditions such that the circuit structure completely respects the syntactic structure of some first-order sentence as well as its evaluations on assignments, cf. \cite{Denenberg1986circuits}, p.236\textapprox p.237); hence   
 its wires can be labeled in a way that reflects the syntactic structure of the sentence (cf. (2)). It implies that the following hold.  
\begin{enumerate}[(a)]
\item Its gates (without considering inputs) induce a tree where the output gate is the root of this tree;\\ 
\item  Each of its inner nodes of the tree, if we do not regard the inputs as leaves, has either $n$ children or two children, depending on whether it corresponds to a quantifier or a logical operator, i.e. $\land $ or $\lor$;\\

\end{enumerate}

\item  The wires are labelled as the follows (cf. \cite{Denenberg1986circuits} page 236).   
Let $\Gamma_n=\{0,1,\ldots,n-1,\#_L,\#_R\}$. Let $f_w$ be a permutation of $\{0,1,\ldots,n-1\}$, which can be extended to a permutation (aslo called $f_w$) of $\Gamma_n^*$ by setting $f_w(\#_L)=\#_L$, $f_w(\#_R)=\#_R$, and $f_w(c_1c_2\ldots c_\ell)=c_1^\prime c_2^\prime\ldots c_\ell^\prime$, where $f_w(c_i)=c_i^\prime$ for each $i$. 
Let $\bar x,\bar y\in \Gamma_n^*$ and $\bar z\in [n]^*$. Assume that $\bar x $ is a wire of $C_n$. 
Then for any $f_w$, the following hold.
\begin{enumerate}[(i)]
\item $f_w(\bar x)$ is a wire of $C_n$;

\item $\bar x$ and $f(\bar x)$ are outputs from gates of the same type or are both input wires;

\item $f_w(\bar y)$ is a child of $f_w(\bar x)$ if $\bar y$ is a child of $\bar x$;

\item\hspace{-5pt}* if $\bar x$ is an input wire whose formula label is an atomic formula $P\bar z$ then the formula lable of $f_w(\bar x)$ is $Pf_w(\bar z)$, provided that the predicate $P$ is neither $\leq$ nor $\mathbf{BIT}$.     
\end{enumerate}

\end{enumerate}

A regular circuit has order $n$ if the cardinality of the universe of the structure for the input is $n$. 

Note that we have given a slightly different symmetric in the labelling, i.e. (2) (iv)*, from the original definition. The original definition of regular circuits are the circuits that completely respect the syntactic structures of  formulas and semantic requirement of logical queries to ensure closure under isomorphism, while we require that the closure under isomorphism holds when ignoring arithmetic predicates.
%\footnote{ By contrast, in \cite{Dawar2014SymCirc,Otto97SymCirc} another form of symmetric circuits are studied, which do not have to respect syntatic structure of formulas that define the properties. By definition, any unordered graph properties (closed under isomorphism) or order-invariant graph properties are recognizable by such circuits,  the automorphisms of whose structures witness such invariant. } 
By \cite{Denenberg1986circuits}, a family of circuits, if there is any, recognize a first-order graph property invariant to the permutations of vertices only if there is a family of regular circuits, wherein every circuit is symmetric to ensure this property to be closed under isomorphism. Nevertheless, since we are talking about ordered graphs, requirement of symmetry should be tailored because ordered graphs are isomorphic if and only if they are the same \cite{Immerman1999Book}. 
%\footnote{If in this case we still insist on symmetry on the circuit structures, we are in fact interested in order-invariant graph properties. Note that $k$-Clique is order-invariant.} 
In this case it corresponds to ``general expression'' in \cite{Barringtion1990Uniformity}.  
%By \cite{Denenberg1986circuits}, if there exists a family of circuits that recognizes a graph property, we can find an equivalent one that is \textit{regular}, which is formatted w.r.t. $n$ and $\langle E\rangle$.  
Fortunately, $k$-Clique is order-invariant, which allows a variant of the symmetry in the labelling, i.e. (2) (iv).  

Note that regular circuits, or general expressions, are not ``space efficient'' (for it is essentially a tree) so that it does not rule out the possibility that a family of much more succinct cricuits can recognize the same property. 
Therefore, we relax the requirement that such a circuit should be a tree (without considering inputs).   
That is, for any such circuit $C_n$ (\textit{succinct regular circuit}, for short) in this family,  albeit still retaining (2) (cf. the definition of ``regular circuit''), its wires may be succintly arranged, i.e. the output of a gate can be many. For, and only for, convenience, we also require that all the children of a gate are ``uniform'' in  case that these children corresponding to a quantified variable: the structures of the subcircuits (ignore the inputs), whose outputs are these children, are isomorphic. It makes the succinct regular circuit look more like a regular circuit. Call this \textit{uniform-children condition}. Note that such condition is implicitly a part of the definition of regular circuits.

Succinct regular circuits reduce logically equivalent subformulas, thereby giving more succinct representations. Obviously regular circuits are special kind of succinct regular circuits. 
%Therefore any first-order graph property that is invariant to permutations of vertices can be recognized by a succinct regular circuit. 
Note that, for any succinct circuits $C$, a regular circuit can be obtained from $C$ by straightforward ``unraveling'', a process of producing copies of subcircuits. Although the size of the regular circuit may explode, the number of variables it needs for translating a circuit into a formula will not be changed.

Recall a well-known result of Barrington et al. \cite{Barringtion1990Uniformity}, which connects first-order definable uniformity to \textit{DLOGTIME}-uniform. 
\begin{fact}\label{fo-equal-AC0}
The following are equivalent.
\begin{enumerate}
\item $\mathcal{L}$ is first-order definable.
\item $\mathcal{L}$ is recognized by a $\textit{DLOGTIME-uniform}$ family of constant-depth, unbounded fan-in, polynomial-size circuits. 
\item $\mathcal{L}$ is recognized by a first-order definable  family of such circuits.
\item $\mathcal{L}$ is recognized by a \textit{DLOGTIME-uniform} family of constant-depth, polynomial-size general expression. 
\item $\mathcal{L}$ is recognized by a first-order definable family of such expressions. 
\end{enumerate}
\end{fact}

Fig. \ref{circuit2formula} summarizes the relationships between graph properties, general expressions, ordinary circuits and regular circuits. Soon we shall see how to convert a $\fo$-uniform family of constant-depth circuits into a $\fo$-uniform family of succinct regular circuits, which is indicated by ``$\mapsto$''. 
Note that, we haven't given a formal definition for first-order definable family of (succinct) regular circuits yet. But it is clear and similar to those in $\fo$-uniform $\mathrm{AC}^0$. Here we adopt a slightly different notion of regular circuits wherein no such labels are explicitly put in the circuits: a circuit is regular if \textit{there exists} a labelling such that those conditions are met.     

%\begin{comment}
\begin{figure}[]
%\hspace*{-1mm}
%\centering
\includegraphics[trim = 22mm 0mm 0mm 0mm, scale=0.48]{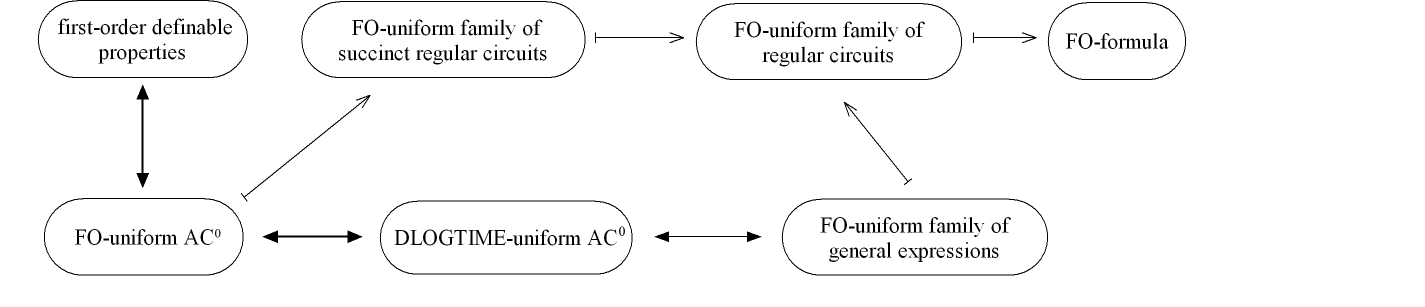}
%\scalebox{10}{}
\caption{Relationships between graph properties, general expressions, ordinary circuits and regular circuits.}
\label{circuit2formula}
\end{figure}
%\end{comment} 

In the following we prove that $n^{k-1}$ gates not suffice to compute $k$-Clique for $k\geq 5$, based on Theorem \ref{thm-k-variable-BIT} and an assumption. To show it, the following is needed and sufficient. 
\begin{proposition}\label{circuits2variables}
For any  graph property, it is describable by a formula in $\fo(\mathbf{BIT})$  using at most $k$ variables if it is recognizable by a first-order definable family of constant-depth, unbounded fan-in circuits of size $O(n^{k})$, provided that $k\geq 5$.  
\end{proposition}
Here a graph property refers to a set of graphs, defined in the usual way.   
We do not talk about other variations such as hypergraphs.  Assume that the circuit depth is bounded by $m$ where $m\geq k\geq 5$. 
We make a \textit{first try} using a natural idea from \cite{Barringtion1990Uniformity,Immerman1999Book}. 
The basic idea is that, using a first-order sentence, we can simulate the function of a circuit by describing the  structure of this circuit, which is defined by the first-order query we used to produce this family of circuits. Hence the sentence is true in the input graph if and only if the circuit accepts the encoding of the graph. 
%Therefore, we can use first-order queries instead of DLOGTIME Turing machines to describe a uniform family of circuits. 
Assume that there is a first-order query that maps a string $0^n$ to a  circuit $C_n$ of size $O(n^{k})$.  
Then the query may use $2(k+1)$ variables to describe the structure of circuits. In addition, it needs to describe the  relations that associate a node with a label in $\{\land, \lor, \lnot, input, output\}$, which also needs $2(k+1)$ variables. Moreover, we are able to show that the bounded variable hierarchy collapses to $\fo^3$ on pure arithmetic structures (cf. Remark \ref{variable-hierarchy-collapse-arithm-struc}).   By Corollary \ref{collapse-in-BIT-0}, it implies that the query can use at most $2k+5$ variables to do the work, by reusing the variables.  Corollary \ref{collapse-in-BIT-1} gives us a slightly sharper lower bound, i.e. $2(k+1)$ variables.

As usual, we can associate a node of circuit with a unique $(k+1)$-tuple, and translate a circuit into a formula. It means that it needs at most $2(k+1)$ variables to simulate the circuit when we use the method introduced in \cite{Immerman1999Book}.

As have been seen, we still have a big gap between $2(k+1)$ and $k$ variables. It only shows that $O(n^{0.5k-1.5})$ gates not suffice to compute $k$-Clique if $k\geq 5$. 
 Indeed, the first-order sentence, obtained from describing circuit structures via the first-order query defining this circuit family, will have to use more than $k$ variables.  
Hence we need alternative ideas, which allow us to handle with more ``regular'' circuits, or circuits in some ``normal form'', whose structures helps us to know something about the property.    
It is for this reason we resort to succinct regular circuits. 
We introduce a new technique called regularization to reduce the gap between an ordinary circuit and a succinct regular circuit. Nevertheless, we should note that the following proof is not constructive and is based on a reasonable but not well-known assumption. 
Recall that all the circuits in discourse are formatted w.r.t. $n$ and $\langle E\rangle$, and we do not take the inputs as gates. The assumption essentially says that {\em any first-order query, which defines a family of constant-depth unbounded fan-in circuits, implicitly defines the syntactical structure of some first-order sentence}. Here we adopt a variant that considers the uniform-children condition. 
\begin{assumption}\label{assump-circuitDefinesFormula}
For any first-order query $I_0$, which defines a family of constant-depth unbounded fan-in circuits $\{C_n\}$ of size $O(n^k)$, there exists a first-order  query $I_1$, which defines a uniform-children family of constant-depth circuits of size $O(n^k)$ defining the same property, 
implicitly defines the syntactical structure of \textit{some}  first-order sentence: for every $C_n$, its gates are either those whose children correspond to the assigment of a block of (or several block of) relativized quantified variables, or those whose children correspond to distinct subformulas.    
\end{assumption} 

Here we briefly explain why this assumption is reasonable. An \textit{instance of a first-order formula} is a propositional logic sentence where every variable of the first-order formula is replaced by a value assigned to it. 
An ordinary circuit can be divided into several pieces, each of which computes a first-order formula or an instance of it. However, the children of a gate $x$ (for quantifiers), may stand for instances of distinct but logically equivalent formulas. In such case, we can select the simplest one w.r.t. the size of the (sub)circuit that compute it, and replace all the other (sub)circuits computing the same function with this subcircuit. Therefore, the number of gates will not increase. Such a process can make the circuits much more ``regular'' and easier to define. 
At the end, we piece up all the instances of formulas to obtain the sentence that defines the property. 
The case is a little more complicated when the children of a gate (for quantifiers) stand for instances of distinct formulas that are not equivalent. In such case, we should only note that the number of distinct formulas are finite, for otherwise it is not first-order definable, and that the query essentially express that some sort of children (determined by some first-order formula) compute some function, and some sort of children compute another function, and so on. Therefore, we can replace the subcircuit, whose output gate is the gate $x$, with a slightly different but equivalent one: we ``split'' the gate $x$ with several gates (depend on the number of distinct sorts of children), each of which computes a distinct function, and take these gates as the children of a new gate. Then, for each of these subcircuits, we ``regularize'' it using the method introduced in the last case. 
Note that this replacement will not increase the size of circuit significantly. 
So far we show that a family of $\fo$-uniform circuits can be converted into a first-order uniform-children family of circuit without significantly increasing w.r.t. circuit size. The reason we think that the assumption is resonable also relies on the fact that a gate of $n$-children (corresponding to an expression in propositional logic) can be ``interpreted'' as a quantifier in first-order logic. Although the number of children may vary, it should be definable by a first-order formula, for, otherwise, the circuit structure is not first-order definable. Similarly, somehow we can also ``interpret'' the functions of some sorts of gates as first-order formulas (e.g. cf. Example \ref{example-query2circuit}).      

It seems that our argument is constructive. But it isn't. According to Trakhtenbrot's Theorem \cite{Trakhtenbrot1950Decidability}, whether two first-order formulas are equivalent or not is not decidable. It implies that the argument involving repalcements in the last paragraph is not constructive.  Similarly, the following proof of Proposition \ref{circuits2variables} is also not constructive.

\begin{proof}

Recall that an inner gate of a succinct regular circuit has either two children or $n$ children, standing for either a logical operator $\land$ (or $\lor$) or a quantifier. In the latter case, say that a wire is quantifiered by $x$ if it is among one of the $n$ children that correspond to distinct values of the quantified variable $x$ in the sentence.

For any first-order sentence $\phi$, let $f_v(\phi)$ be the number of distinct variables in $\phi$. 
If $\phi=\psi_1\land\psi_2$, then clearly $f_v(\phi)=max\{f_v(\psi_1), f_v(\psi_2)\}$. Similarly,  $f_v(\phi)=max\{f_v(\psi_1), f_v(\psi_2)\}$ if $\phi=\psi_1\lor\psi_2$. 
Note that, once we have a family of succinct regular circuits, it is easy to read their structures and get the corresponding first-order sentence. 
Suppose that $\phi$ is the sentence that correspond to this family of circuits. 
 It implies that any directed path in a succinct regular circuit (of this family) contains wires that are quantified by  at most $k$ distinct variables if the circuits has at most $O(n^k)$ gates.  
 Therefore, $f_v(\phi)\leq k$.      

The difference between ordinary circuits and succinct regular circuits are the follows.
\begin{enumerate}[(1$^\star$)]
\item Ordinary circuits can be more ``succinct'' than the so called succinct regular circuits in the representation of quantifier structures. That is, in an ordinary circuit a gate can use $n^i$ children to denote a block or even several blocks of quantifiers.

\item Even if the fan-in corresponds to one quantified variable, the number of children of a node can be varied, i.e. doesn't have to be $n$. Therefore, in general, a gate can have $\ell$ children where $n^i<\ell<n^{i+1}$ for some $i$. 

For example, consider a very simple $1$-ary first-order query that defines a family of small cicuits. 
Assume that $y$ is an $\land$-gate. Suppose that $x$ is a child of $y$ if $x$ and $y$ satisfy some first-order definable property, say defined by $\eta(z_1,z_2)$. That is, the following is a part of the first-order formula that defines the structure of the circuit. 
\begin{equation}\label{formula-for-local-structure}
\forall x(\eta(x,y)\rightarrow E(x,y))  
\end{equation}

In this case the number of children of $y$ is determined by $\eta(z_1,z_2)$. The question is that $f_v(\eta(z_1,z_2))$ could be arbitrary. It may be much greater than $k$. 

\item In an ordinary circuit, we can represent the conjunction or disjunction of several subformulas succinctly, whereas in a succinct regular circuit an $\land$-gate only has two children in such case. For example, given a quantifier-free subformula $\varphi_1\land \varphi_2\land \varphi_3$, in an ordinary circuit we can use one $\land$-gate with three children to describe it, whereas in a succinct regular circuit we need two $\land$-gates connected in the obvious way, with the subcircuits computing $\varphi_i$ as their children.   

\item Ordinary circuits can use the arithmetic literals in an implicit way, whereas regular circuits have to use them explicitly. For instance, in an ordinary circuit we don't have to represent explicitly that $x_i\neq x_j$. It can be  encoded in the way the children of a gate are distributed (when the number of children is not exactly $n$).    

\item There are no labels in ordinary circuits. In contrast, regular circuits are defined based on proper labelling of wires. 

\item In an ordinary circuit, the children of a gate (for quantifiers) can stand for instances of distinct subformulas (subcircuits). But in a regular one, all the children should compute the same subformula of distinct instances.   
\end{enumerate}

We shall see that succinct regular circuits are not very different from ordinary circuits, and we can convert an ordinary one to a regular one without significantly increasing the number of gates. 

A \textit{regularization} of an ordinary circuit is a process that makes the quantifier structure more explicit in the circuit and a schemetic labelling of its wires consistantly w.r.t. \textit{some} first-order sentence as follows. It replaces a gate that has $n^i$ children with $\sum_{j=0}^{i-1} n^j=\frac{n^i-1}{n-1}$ gates, i.e. replacing the gate by a perfect $n$-ary tree of gates. Note that such a process will not increase the number of gates significantly compared with the size of the circuit, which is very important. In addition, when we regularize a circuit, we first add two constants to its inputs, i.e. $0$ and $1$. We can use these two constant inputs to represent all the arithmetic atoms, i.e. $x=y$, $x\leq y$ and $\mathbf{BIT}(x,y)$. Recall that we don't have to label the wires explicitly in succinct regular circuits. We need only show the existence of such a valid, or consistent, labelling. 
As a consequence, (1$^\star$), (4$^\star$) and (5$^\star$) can be handled easily. Nevertheless, to obtain a slightly better lower bound, here we adapt the definition of succinct regular circuits such that the succinct representation of a block of quantifiers is allowed. That is, in the new definition, we allow the number of children of a gate to be $n^i$ for any $i$.  
 
(6$^\star$) will not make a big difference. We have briefly explained it in the discourse of Assumption \ref{assump-circuitDefinesFormula}.   

(2$^\star$) will not be a problem if $\eta(z_1,z_2)$ is logically equivalent to some formula using bounded number of variables.  
To this end, we show that $\eta(z_1,z_2)$ is indeed equivalent to a formula in $\fo^5$.  
Note that $\eta(z_1,z_2)$ is a Boolean combinations of constant number of formulas that describe intervals or size of intervals, because the atom formulas consist only of $=$, $\leq$ and $\mathbf{BIT}$, each of which can be described by three variables. Cf. Remark \ref{variable-hierarchy-collapse-arithm-struc} (Corollary \ref{collapse-in-BIT-0}) to see why this can be justified. Note that
\eqref{formula-for-local-structure} is equivalent to 
\begin{equation}\label{ex-regularization-eqn-0}
\forall x(\lnot\eta(x,y)\lor E(x,y)).
\end{equation} 
Recall that, all the arithmetic atoms are evaluated directly: their values are represented by those two constant inputs $0$ and $1$. Therefore, the number of gates will not increase significantly.

In general, we need to handle the following situation to ensure that not too much new variables are introduced. 
\begin{equation}\label{ex-regularization-eqn-quant-block-0}
\forall x_1(\eta_1(x_0,x_1)\rightarrow\exists x_2(\eta_2(x_0,x_1,x_2)\land \forall x_3(\eta_3(x_0,x_1,x_2,x_3)\rightarrow \cdots)\cdots))
\end{equation}
But, obviously, our method used to deal with the last simpler case can be applied here. Furthermore, \eqref{ex-regularization-eqn-quant-block-0} can be rewritten in the following form, provided that there are $k$ quantifiers and the last quantifier is $\forall x_k$ (the following formula is similar when the last quantifier is an existential one). 
\begin{equation}\label{ex-regularization-eqn-quant-block-1}
\forall x_1\exists x_2\forall x_3\cdots,\forall x_k(\eta(x_0,x_1,x_2,\cdots,x_k)\rightarrow \cdots)\cdots)
\end{equation}
Recall that $k\geq 5$, by Corollary \ref{collapse-in-BIT-1}, $\eta(x_0,x_1,x_2,\cdots,x_k)$ doesn't add new variables to the number of variables needed in the query. Therefore, the number of gates will not increase. It is for this reason we allow the number of children of a gate (in a succinct regular circuit) to be $n^i$ for any $i\geq 1$ in the subcircuit that computes the quantifier structure. It saves three variables. 

Finally, (3$^\star$) is very easy to handle. We can replace a succinctly represented conjuction or disjunction by an equivalent subcircuit whose gates has two children. Just note that we only need to add constant number of new gates to deal with it, where the number is independent of $n$.  

Suppose that we are given a first-order definable family of ordinary constant-depth unbounded fan-in circuits $\{C_n\mid n\geq 5\}$ where $C_n$ has $O(n^k)$ gates. From a valid regularization of $\{C_n\}$ we can obtain a family of succinct regular circuits $\{C^\prime_n\mid n\geq 5\}$ that recognize the same first-order graph property, and along any path there are at most $k$ wires that are quantified by distinct variables because there are at most $O(n^{k})$ gates in $C^\prime_n$. Then the graph property defined by this family of circuits can be uniformly defined by one first-order sentence with at most $k$ variables. The uniformity of this family of succinct regular circuits  comes from the fact that we use the same process to regularize the ordinary circuits based on the first-order query that defines this family of ordinary circuits. And if one scheme of regularization works for $C_n$, it also works for $C_{n+1}$ because these circuits are very similar except for the value of fan-in:  $n$ for the former and $n+1$ for the latter. To summarize, in principle, we can regularize a first-order definable family of ordinary circuits, provided that Assumption \ref{assump-circuitDefinesFormula} holds.  \end{proof}

In the following we give an example to illustrate how we can obtain a regularization from a first-order query. Different from the last argument, it is constructive, based on the first-order query defined in the following example. 

\begin{example}\label{example-query2circuit}
It was mentioned in \cite{Lynch86Circuit}, and presented in \cite{RossmanStoc}, a simple, and maybe optimal, circuit algorithm that computes $k$-Clique using $n^{k}$ gates: it simply enumerates all the sets (of vertices) of size $k$. Such a circuit family is first-order definable. Here is a $(k+1)$-ary first-order query that defines this circuit family. It is important that such a query is independent of $n$, the order of a circuit. It is for this reason that the following regularization can be achieved: on the one hand, it implies that the size of the circuit after regularization will not increase significantly because the length of the sentence defining the property is independent of $n$; on the other hand, it implies that the circuits after regularization are very similar except for the values of fan-in (e.g. $n$ for $C_n$ and $n+1$ for $C_{n+1}$).  

Note that two variables suffice to define a constant  number, in the presence of a linear order \cite{Dawar96Number}. Hence, the numbers $0$, $1$ and $2$ are first-order definable and we shall use them for free. 
 Let $\bar x:=x_0,x_1,\ldots,x_k$. 

Let $\varphi_{_\lor}(\bar x)$ define the unique $\lor$-gate. 
\begin{equation}
\varphi_{_\lor}(\bar x):=x_0=2\land \bigwedge_{i\in [1,k]} x_i=0
\end{equation} 

Let $\varphi_{_r}(\bar x)$ define the output gate as $\varphi_{_\lor}(\bar x)$. 

Let $\varphi_{_\land}(\bar x)$ define those $\land$-gates.  
\begin{equation}
\varphi_{_\land}(\bar x):=x_0=1\land \bigwedge_{i,j\in [1,k];i\neq j} x_i\neq x_j
\end{equation}
By convention, here we use $x_i\neq x_j$ to denote $\lnot (x_i=x_j)$. 

Let $\varphi_{_\lnot}(\bar x)$ define the $\lnot$-gates.
\begin{equation}
\varphi_{_\lnot}(\bar x):=\mathrm{FALSE}
\end{equation}  

Let $\varphi_{\mathrm{in}}(\bar x)$ define the inputs.
\begin{equation}
\varphi_{_\mathrm{in}}(\bar x):=x_0=0\land x_1\neq x_2\land \bigwedge_{i\in [3,k]}x_i=0
\end{equation}

Let $\varphi_0$ define the universe of a circuit structure. 
\begin{equation} 
\varphi_0(\bar x):=\varphi_{_\lor}(\bar x)\lor \varphi_{_\land}(\bar x)\lor \varphi_{\mathrm{in}}(\bar x)
\end{equation}

Finally, we give the heart of the definition, i.e. those arrows that exhibit the inputs and outputs of the gates (or the edge relation between the gates). Let $\bar y:=y_0,y_1,\ldots,y_k$. 

$\varphi_{_{R}}(\bar x,\bar y):=\psi_0(\bar x,\bar y)\lor \psi_1(\bar x,\bar y)$ where
\begin{align}
\psi_0(\bar x,\bar y) &:=\varphi_{_\lor}(\bar x)\land \varphi_{_\land}(\bar y)\\
\psi_1(\bar x,\bar y) &:=\varphi_{_\land}(\bar x)\land \varphi_{_\mathrm{in}}(\bar y)\land \bigvee_{i,j\in [1,k],i\neq j}(x_i=y_1\land x_j=y_2)
\end{align}
Note that, from $\psi_0(\bar x,\bar y)$ we can obtain the quantifier structure needed in the regularization. 
Together with $\varphi_{_r}(\bar x)$ and $\varphi_{_\lor}(\bar x)$, it tells us that the children $\bar y$ of the output gate is any string of length $k$ where every element is distinct. So clearly these children correspond to a block of $k$ relativized existential quantified variables.  That is, 
$\psi_0(\bar x,\bar y)$ tells us that the sentence is begin with an existential quantifier block in the form:
\begin{equation}\label{ex-regularization-eqn-1}
\exists x_1\exists x_2\cdots\exists x_k \left(\left(\bigwedge_{i,j\in[1,k];i\neq j} x_i\neq x_j\right)\land \xi\right) 
\end{equation} 
Here, for the sake of simplicity, we have used ``$x_1,x_2,\ldots,x_k$'' as the names of the quantified variables, inheriting from the definition of the first-order query.
Note that, along any path, which assigns values to the variables $x_1x_2\cdots x_k$, the subformula $x_1\neq x_2\neq\cdots\neq x_k$ can be evaluated immediately: every path leads to an $\land$-gate; one of its children is either $0$ or $1$ depending on the trueth value of $x_1\neq x_2\neq\cdots\neq x_k$; the other child of this $\land$-gate is the output gate of the subcircuit that computes the subformula $\xi$. Therefore, we do not need to adding new gates to compute those subformulas $x_1\neq x_2\neq \cdots\neq x_k$.    

From $\varphi_{_{R}}(\bar x,\bar y)$ we know that $\psi_1(\bar x,\bar y)$ defines the quantifier-free subformulas of the sentence. There is only one quantifier-free subformula in this example. The formula $\varphi_{_\land}(\bar x)$ tells us that the atom formulas are collected in a conjunction.  The formula  $\psi_1(\bar x,\bar y)$ tells us how to bind the variables with inputs, a process that give names to the variables in a formula.\footnote{Without this process, we don't know the names of variables. Hence the atoms are in the form like $E(?,?)$. That is, we could only know the structure of a sentence, not the complete description of the sentence. Using this process, we can know the precise description, up to renaming.} It is possible to work out the binding directly from the qurey. In particular, the subformula $\bigvee_{i,j\in [1,k],i\neq j}(x_i=y_1\land x_j=y_2)$ tells us that the conjunction is in the form ``$\bigwedge_{i,j\in [1,k],i\neq j}E(x_i,x_j)$'' if we give the quantified variables the name ``$x_1,x_2,\ldots,x_k$'' inheriting from the first-order query. 
We can also know the binding from a technique called ``consistent reverse assignment''. Assume that $\ell$ is a constant greater than the length of the query. 
Having a picture of $C_\ell^\prime$ (i.e. the regularized $C_\ell$) in our mind,  the binding follows the following process. Along a directed path of the quantifier structure, we have a subcircuit standing for the quantifier-free subformula. We first guess a labelling of the wires in the quantifier structure. It is a map from a value in $[n]$ to the name of a variable. Call such a map \textit{reverse assignment}.\footnote{It is so named because usually an assignment maps variables to values. Such kind of map has been used by Denenberg et al., cf. \cite{Denenberg1986circuits}, p.238. But we handle it differently. That is, in general we do not map a value $b$ to $v_b$.} 
Then for each subcircuit we guess a consistent reverse assignment. Because the first-query defines a circuit family that express a first-order property, there exists a \textit{consistent} guess that matches both the guess of the reverse assignment of the subcircuit that computes the quantifier structure\footnote{In this simple example, there is only one subcircuit that computes the unique quantifier structure. In more complicated cases, it could have several such subcircuits that compute different quantifier structures.} and the guesses of the reverse assignment for all the subcircuits that compute the instances of the quantifier-free subformula. 
Such a consistent guess must exist, for otherwise it is in contradiction with Assumption \ref{assump-circuitDefinesFormula}.  
Note that this process is still constructive since we can enumerate all the possible guesses, for up to $k$ distinct  variables. The reason we need just consider $k$ distinct variables is because the quantifier structure is regularized such that every gate in the circuit, except those whose children are inputs, has either $n^i$ (corresponding to $i$ quantifiers) or $2$ children (corresponding to ``$\land$'' or ``$\lor$''), and  there are at most $O(n^{k})$ gates can be reached from any gate following the arrows.  

Last but not least, the family of regularized circuits are first-order definable. The uniformity is inherited from the uniformity of the original family of gates.   

In summary, suppose that we are given a first-order definable family of constant-depth circuits of size $O(n^{k})$, we can convert it to a first-order definable family of succinct regular circuits of size $O(n^{k})$, whose wires are not necessary explicitly labelled, that computes the same first-order property. Then we show that such property can be defined using $k$ distinct variables because the number of children of a gate in a succinct regular circuit, which corresponds to \textit{one} quantifier, is $n$.\footnote{If it is a block or several block of quantifiers, the number of children could be $n^i$ for some $i$.} Moreover, in this example we can even obtain the sentence defining $k$-Clique straightforward from reading the structure of $C_\ell^\prime$, using consistent reverse assignments. Hence, it is completely constructive w.r.t. this first-order query. 
\hfill\ensuremath{\divideontimes}
\end{example}

The following is straightforward, due to Theorem \ref{thm-k-variable-BIT} and Proposition \ref{circuits2variables} (with Assumption \ref{assump-circuitDefinesFormula}). 
\begin{corollary}\label{circuit-lower-bound-1}
On condition that $k\geq 5$ and Assumption \ref{assump-circuitDefinesFormula} holds, $k$-Clique cannot be computed by any first-order definable family of constant-depth unbounded fan-in circuits of size $O(n^{k-1})$. 
\end{corollary}
It conditionally answers a question raised in \cite{RossmanThesis} (cf. p.71), \cite{RossmanStoc} (cf. p.10), and also in \cite{DawarHowmany} (cf. p.25), based on Assumption \ref{assump-circuitDefinesFormula}. Note that, by our method, $O(n^2)$ not suffice to compute $k$-Clique if $k=4$, because $k+1$  variables are needed instead of $k$ in Proposition \ref{circuits2variables}, and  $O(1)$ not suffice to compute $k$-Clique if $k=3$, because $k+2$ variables are needed. Obviously, this trivial lower bound for the case $k=3$ tells us nothing. It is not clear what are the tight lower bounds in these two special cases.

\section{Conclusions}\label{conclusions}

The study of finite model theory has been mostly motivated by questions in computational complexity theory and database theory. 
It has been one of the major challenges in finite model theory for a long time to establish lower bounds on finite  \textit{ordered} graphs, which stands for a well-known barrier that is in the way when we try to solve open problems in computational complexity using finite model-theoretic toolkit. We developed novel concepts and techniques to this end, for worst-case lower bounds.  

Based on a kind of pebble games (with changing boards) over abstractions and a notion called board history, we provide an alternative proof for the strictness of bounded variable hierarchy in $\fo$, which was first proved by Rossman \cite{RossmanStoc} using tools from circuit complexity.     
Note that our proof is purely constructive. 
That is, 
we construct a pair of extraordinary huge graphs explicitly, use them as the game board and demonstrate the winning strategies for Duplicator with full details. 

Moreover, we use the explicit constructions to prove an optimal lower bound of $k$-Clique, which fully answers a question \cite{DawarHowmany,RossmanStoc,RossmanThesis} that goes back to an early paper of Immerman \cite{Immerman1982Conj}, which represents the challenge. 
Contrary to popular opinion, we find that big size is not a big issue in the constructions.  On the contrary, big size even helps: intuitively, logics tend to exhibit their difference on large enough structures.

This work is extended by introducing other arbitrary arithmetic predicates to $\fo$, other than linear orders. 
Recall that $\mathbf{BIT}$ predicate can be used to define arbitrary arithmetic predicates. 
We show that precisely $k$ variables are necessary and sufficient to describe $k$-Clique in $\fo$ on the class of finite graphs with built-in $\mathbf{BIT}$. 

Afterwards, we apply this result to circuit complexity, which is motivated by the question on the worst-case lower bounds of constant-depth circuits raised in \cite{Lynch86Circuit,DawarHowmany,RossmanStoc,RossmanThesis}. 
It is also related to the question of Immerman since there is a well-known connection between first-order logic and first-order definable families of constant-depth unbounded fan-in circuits. Recall that Rossman's tight lower-bound in average-case is also a unconditional worst-case lower bound, which says $O(n^{\frac{k}{4}})$ gates not suffice to compute $k$-Clique on constant-depth circuits. We improve the state of the art by a unconditional worst-case lower bound $O(n^{\frac{k-3}{2}})$. Then, based on a not well-known but still reasonable assumption (i.e. Assumption \ref{assump-circuitDefinesFormula}), we give the tight worst-case lower bound $O(n^{k-1})$. Certainly, it is not completely satisfactory for it is not unconditional. A proof from circuit complexity using its own, maybe very novel, techniques is expected to show the same tight lower bound without any assumption.

We can also study the $k$-Clique problem or strictness of bounded variable hierarchy in other more expressive logics. 
Moreover, it is interesting to know whether the notions and techniques introduced here can help to improve the state of affairs of other questions in circuit complexity. For instance, could we get better lower bounds for more general problems, say subgraph isomorphism problem, on the $\fo$-uniform $\mathrm{AC}^0$ model?  
Possibly, it is within reach of the techniques introduced in this paper.

%######################### APPENDIX ######################################

\newpage

\appendixpage
%\begin{appendix}
We put some remarks and proofs in this appendix to help the readers understand and evaluate the ideas.

\begin{remark}\label{remark-expansion}
As it turns out, $\mathfrak{B}_{3,m}^\prime$ is quite large even for moderate $m$. Hence, to deliver some essence of the notion ``structural expansion'', we use the small structure $\mathcal{B}_3$ as the start point of an expansion.      

\begin{figure}[htbp]
\centering
\includegraphics[trim= 0mm 0mm 0mm 0mm, scale=0.5]{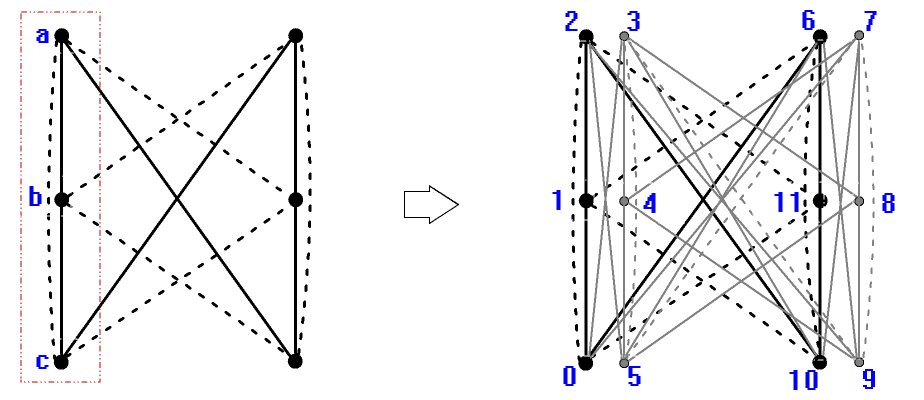}
%\scalebox{10}{}
\caption{ }
\label{expansion-example}
\end{figure}  

In Fig. \ref{expansion-example}, the graph on the left side is $\mathcal{B}_3$. The graph  on the right side is an expansion of it. The ``bricks'' are akin to $\mathcal{B}_3$ except that the adjacency between the three vertices in the first column  can be different. For example, the brick formed by the vertices $0, 1, \ldots,5$ is isomorphic to $\mathcal{B}_3$; whereas the brick formed by $0$, $1$, $4$, $5$, $6$ and $7$ is akin to $\mathcal{B}_3$ except that the subgraph induced by $a$, $b$ and $c$ in $\mathcal{B}_3$ is different from that induced by $0$, $1$ and $6$ in the other. In short, the bricks are bound in such a way  that their first columns  respect the structure (or adjacency) of the graph $\mathcal{B}_3$. Hence, in some sense we can call $\mathcal{B}_3$  an ``abstraction'' (or a skeleton, or a blueprint, whatsoever) of its expansion. Moreover, we can call the latter the first abstraction, and $\mathcal{B}_3$ the second abstraction. 
Note that the universes of $\mathcal{B}_3$ and its (first) expansion, as well as all the bricks, are isomorphic to upright square lattices. The width of the ``bricks'' is $2$. It ``corresponds'' to $\beta_0^1$ in \eqref{beta-func}. 

Obviously, the definition of $\mathfrak{B}_{3,m}^\prime$ is more complicated. But the ideas are similar. We intend to use $\mathbb{X}_i^*$ to denote the $i$-th abstraction. And we can regard $\mathbb{X}_i^*$ as an ``abstraction'' of $\mathbb{X}_{i-1}^*$. Hence $\mathbb{X}_m^*$ is the highest abstraction, on the top of the hierarchy of abstractions; whereas $\mathbb{X}_1^*$ is the lowest one, on the bottom of this hierarchy. The index of a vertex tells us at which stage it is created in the structual expansion. The value $\beta_{m-j}^{m-i}$ tells us the width of the ``bricks'' we will use to build the $i$-th abstraction based on the ``skeleton'' $\mathbb{X}_j^*$, or the number of vertices that are around a vertex in $\mathbb{X}_j^*$. Suppose $(x,y)$ is a vertex of the $p$-th abstraction. The construction should ensure that it is also a vertex in the $q$-th abstraction for any $q<p$. The vertex $(\llparenthesis x\rrparenthesis_i,y)$ is the projection of $(x,y)$ in the $i$-th abstraction. If $x=\llparenthesis x\rrparenthesis_i$, this means that $(x,y)$ is already in the $i$-th abstraction, which in turn implies that the index of $(x,y)$ is at least $i$. The value $[x]_i$ tells us where $(x,y)$ is in the $y$-th row of the $i$-th abstraction, provided that $(x,y)$ is a vertex in the first abstraction.     
\end{remark}

\begin{example}\label{example-congr-label}
We give an example on the concept ``congruence label'', which is easier to illustrate in the ``flat'' structures $\mathfrak{A}_{k,m}^*$ and $\mathfrak{B}_{k,m}^*$ that ``forget'' board histories of vertices. 

%Note that  we can use the ``height'' of congruence superscript of $(x,y)$ to infer the index of $(x,y)$. 
From the definition, if $$\mathbf{cl}(x,y)=\underline{0,1;m-2;1;\left\{\underline{0,3;m;2;\emptyset}, \underline{1,0;m-1;0;\{2,2;m;-1;\emptyset\}}\right\}},$$ it means that $\mathbf{cc}([x]_{m-2},y)=0$;  $y=1$; $\mathrm{idx}(x,y)\!=\!m\!-\!2$ (i.e. \!\!$(x,y)\in\mathbb{X}_{m-2}^*-\mathbb{X}_{m-1}^*$); $\mathrm{RngNum}(x,m-2)=1$; 
and $(x,y)$ is not adjacent to any vertex whose congruence label is in  $\left\{\underline{0,3;m;2;\emptyset},  \underline{1,0;m-1;0;\{2,2;m;-1;\emptyset\}}\right\}$.  In other words, $(x,y)$ is not adjacent to any vertex $(u,3)$ whose index is $m$, $\mathbf{cc}([u]_m,3)\!=\!0$ and $\mathrm{RngNum}(u,m)=2$; and  $(x,y)$ is not adjacent to any vertex $(e,0)$ whose index is $m\!-\!1$, $\mathbf{cc}([e]_{m\!-\!1},0)\!=\!1$, $\mathrm{RngNum}(e,m-1)=0$, \textbf{and} $(e,0)$ is not adjacent to  any vertex $(e^{\prime},2)$ whose index is $m$,  $\mathbf{cc}([e^{\prime}]_m,2)\!=\!2$ and $\mathrm{RngNum}(e^\prime,m)=-1$. 
\end{example}

\begin{remark}
Our structures are defined in a natural way. The only seemingly artificial bits are the introducing of $\mathrm{RngNum}$ and $\mathrm{SW}$ functions, which deserve more explanation. Note that 
 $\mathcal{B}_k$ is very symmetric (cf. section \ref{existential-case-section}). Indeed, it is so symmetric that the automorphism group of a $k$-clique  is isomorphic to a subgroup of the automorphism group of  $\mathcal{B}_k$ (when ``forgetting'' the order), which is a cyclic group. We can also use this fact to prove Theorem \ref{existential-case}. 
Note that $\mathfrak{B}_{k,m}^*[\mathbb{X}_i^*]$  resembles $\mathcal{B}_k$ to some extent. 
However, without  $\mathrm{SW}$ function $\mathfrak{B}_{k,m}^*[\mathbb{X}_i^*]$ is not as symmetric as $\mathcal{B}_k$ because of missing of some edges: all the edges satisfying $(\mathbf{cc}([x_i]_{t-1},y_i)-\mathbf{cc}([x_j]_{t-1},y_j))\times(y_i-y_j)<0$ would be missing. Spoiler can use such asymmetry to win the pebble games. But with $\mathrm{SW}$ function in the definition, the structures are sufficiently symmetric and Spoiler cannot exploit the asymmetry anymore. Note that  
$\mathrm{SW}$ function alone will cause a problem: without $\mathrm{RngNum}$ function the structure $\mathfrak{B}_{k,m}^*$, as well as $\mathfrak{B}_{k,m}$, will have $k$-cliques. 
% $\mathrm{RngNum}$ is also closely related to the construction of the game board for the game $\Game(\mathfrak{A}_{3,m},\mathfrak{B}_{3,m})$. 
To fully understand $\mathrm{SW}$ function, cf. Lemma \ref{no-missing-edges_xi-1}. And to fully undertand $\mathrm{RngNum}$ function, cf. case (3) in the proof of Lemma \ref{B_k-has-no-k-clique}.
%\hfill\ensuremath{\divideontimes}
\end{remark}

\begin{remark}\label{remark-linear-orders-1}
Fact \ref{linear-orders-1} comes directly from  the following folklore knowledge \cite{EbbinghausFlum99FiniteM}.\footnote{Cf. Example 2.3.6 on page 22 of the second edition of \cite{EbbinghausFlum99FiniteM}. 
The optimal lower bound used here is from a \href{http://www.cl.cam.ac.uk/teaching/1213/L15/notes.pdf}{course note} of Dawar.} 
 
%\begin{lemma} \label{linear-orders-0}
\textit{For any $m>0$, if $\mathcal{O}_a,\mathcal{O}_b$ be linear orders of length greater than or equal to $2^m-1$, then $\mathcal{O}_a\equiv_m \mathcal{O}_b$.}
%\end{lemma}

The notion ``$\equiv_m$'' is related to the standard Ehrenfeucht-Fra\"iss\' e\xspace games, wherein the players can use arbitrary number of pebbles. 
At the beginning, there is one interval for a linear order, i.e. the linear order itself. 
Recall that, in this paper, whenever we talk about an interval, it is an empty interval (i.e. no pebble is inside the interval), except that it may contain the newly picked vertex. 
That is, all the intervals are not overlapped. 
Duplicator's strategy in such games likes the following.
 In the $i$-th round, where $i\leq m$, assume that an interval $[a,b]$ in $\mathcal{O}_a$ is split into two parts, say $[a,x]$ and $[x,b]$. Duplicator tries to ensure that, the corresponding interval in $\mathcal{O}_b$ is also split into two parts , say $[a^\prime,x^\prime]$ and $[x^\prime,b^\prime]$, such that  
\begin{itemize}
\item if $x-a<2^{m-i}-1$ then $x-a=x^\prime-a^\prime$;
\item if $b-x<2^{m-i}-1$ then $b-x=b^\prime-x^\prime$;
\item if $x-a\geq 2^{m-i}-1$ and $b-x\geq 2^{m-i}-1$ then\\ $x^\prime-a^\prime\geq 2^{m-i}-1$ and $b^\prime-x^\prime\geq 2^{m-i}-1$.
\hfill\ensuremath{(\mbox{apx-}1)} 
\end{itemize}

By a simple induction, we can see that it is a winning strategy of Duplicator in such games. That is, she can ensure (apx-$1$) throughout the game, which implies that $a_i\leq a_j$ if and only if $b_i\leq b_j$, for any $a_i, a_j$ in $\mathcal{O}_a$ and $b_i,b_j$ in $\mathcal{O}_b$, where $a_i\Vdash b_i$ and $a_j\Vdash b_j$.
%Now let us consider a new scenario, where 
%\hfill\ensuremath{\divideontimes}
\end{remark}

\begin{remark}\label{abstract-order-in-main-lemma}
In this remark, we assume that the game board consists of $\mathfrak{A}_{k,m}^*$ and $\mathfrak{B}_{k,m}^*$.  Recall that the players are playing in the associated structures $\mathfrak{A}_{k,m}^*$ and $\mathfrak{B}_{k,m}^*$.
We show that Duplicator can preserve the abstraction-order-condition throughout the game, at the possible price that the value of $\xi$ is decreased by one in a round. Moreover, we show that Duplicator is able to avoid picking any object\footnote{Here, the notion ``object'' (cf. page  \pageref{def-main-object}) is slightly  different from the one introduced before (cf. page  \pageref{def-special-object}, Remark \ref{ExplanationOfAbstraction-specalcase}).} 
that is in the $(\xi-1)$-th abstractoin and that contains a critical point  except for the case wherein no object of this size containing a crtitical point is already ``picked'' (and the exceptions due to her basic strategy B-2., cf. Lemma \ref{main-lemma}), and the case where this object is already ``picked'' (in the $\xi$-th abstraction or below).  

Assume that Spoiler picks $(x,y)$ whose index is $t$, and Duplicator replies with $(x^\prime,y)$ whose index is $t^\prime$.   Both of the vertices are in $\mathbb{X}_1^*$. Note that, in the main part of the proof of Lemm \ref{main-lemma} (from page \pageref{def-virtual-game}), $(x,y)$ and $(x^\prime,y)$ are vertices in $\mathbb{X}_1$. 

When we talk about ``picking $(x^\prime,y)$'', we are interested in $(\llparenthesis x^\prime\rrparenthesis_\xi,y)$ instead. 
Once $(\llparenthesis x^\prime\rrparenthesis_\xi,y)$ is determined, then we can determine $(x^\prime,y)$ based on the approximate hr-copycat condition (cf. (\ref{main-diamond-xi}$^\diamond$), page \pageref{def-xi}). 
In the game over abstractions, assume that $[(a,y),(b,y)]$ is the interval that contains $(\llparenthesis x\rrparenthesis_\xi,y)$  and that $[(a^\prime,y),(b^\prime,y)]$ is the corresponding 
interval in the other structure. Note that $(a,y),(b,y)\in \mathbb{X}_{\xi}^*$. 
Duplicator will pick a vertex $(x^\prime,y)$ such that  $a^\prime\leq \llparenthesis x^\prime\rrparenthesis_\xi\leq b^\prime$. 
In the first round, Duplicator simply mimics Spoiler. 
%In the next round, Duplicator mimics Spoiler's pick except when Spoiler picks a critical point (or an object containing a critical point) while another critical point (or another object containing a critical point) with different second coordinate is already pebbled (or ``picked''). Note that $\xi$ is still $m$ after this round. 
In the following rounds, if the abstraction-order-condition is always preserved in the $m$-th abstraction, then Duplicator is happy. 
In the sequel, we assume that Duplicator has to resort to lower abstractions and currently they are playing the $\ell_c$-th round where $\ell_c>1$. Therefore, $$m>\xi\geq m-\ell_c+2.$$

Now, we summarize the ideas, and explain briefly how they work. 

Firstly, note that all the unabridged intervals formed in the game either are sufficiently big, i.e. greater than or equal to $2^{m-\ell_c}-1$, or are isomorphic. It ensures the basic requirement for pure (induced) linear orders. 

To force Duplicator to pick a critical point, in the $j$-th round Spoiler should pick $(x,y)$ in $\mathfrak{B}_{k,m}^*$. Assume that the interval $[(a^\prime,y),(b^\prime,y)]$, where $(x^\prime,y)$ should be settled, contains a critical point. 
Consider the following cases.
\begin{enumerate}[(i)]
\item Assume that  $\lfloor \llparenthesis x\rrparenthesis_\xi/l_\xi\rfloor-\lfloor a/l_\xi\rfloor\geq 2^{m-\ell_c}-1$ and $\lfloor b/l_\xi\rfloor-\lfloor \llparenthesis x\rrparenthesis_\xi/l_\xi\rfloor\geq 2^{m-\ell_c}-1$. By induction hypothesis, we know that Duplicator can make it that $\lfloor \llparenthesis x^\prime\rrparenthesis_\xi/l_\xi\rfloor-\lfloor a^\prime/l_\xi\rfloor\geq 2^{m-\ell_c}-1$ and $\lfloor b^\prime/l_\xi\rfloor-\lfloor \llparenthesis x^\prime\rrparenthesis_\xi/l_\xi\rfloor\geq 2^{m-\ell_c}-1$. Then by (\ref{xi-order-requirement-ensured}), the abstraction-order-condition holds (for any $i$-th abstraction where $1\leq i\leq \xi$). Moreover, 
if $\lfloor b^\prime/l_\xi\rfloor-\lfloor a^\prime/l_\xi\rfloor\geq 2^{m-\ell_c+1}-1$, then by Fact \ref{linear-orders-1}, Duplicator can avoid picking any object containing a critical point in this round, since the (unabridged) interval is sufficiently big and more than enough. If $\lfloor b^\prime/l_\xi\rfloor-\lfloor a^\prime/l_\xi\rfloor=2^{m-\ell_c+1}-2$, then Duplicator has to let $\lfloor \llparenthesis x^\prime\rrparenthesis_\xi/l_\xi\rfloor-\lfloor a^\prime/l_\xi\rfloor= 2^{m-\ell_c}-1$. In this case, she is still able to pick an object, but possibly in the $(\xi-1)$th abstraction, that does not contain a critical point, by (\ref{xi-order-requirement-ensured}).  

\item Assume that either $\lfloor \llparenthesis x\rrparenthesis_\xi/l_\xi\rfloor-\lfloor a/l_\xi\rfloor\!<\! 2^{m-\ell_c}\!-\!1$ or $\lfloor b/l_\xi\rfloor-\lfloor \llparenthesis x\rrparenthesis_\xi/l_\xi\rfloor< 2^{m-\ell_c}-1$. Then due to her strategy, Duplicator will pick $(x^\prime,y)$ such that $\lfloor \llparenthesis x^\prime\rrparenthesis_\xi/l_\xi\rfloor\!-\!\lfloor a^\prime/l_\xi\rfloor=\lfloor \llparenthesis x\rrparenthesis_\xi/l_\xi\rfloor\!-\!\lfloor a/l_\xi\rfloor$ or $\lfloor b^\prime/l_\xi\rfloor-\lfloor \llparenthesis x^\prime\rrparenthesis_\xi/l_\xi\rfloor=\lfloor b/l_\xi\rfloor-\lfloor \llparenthesis x\rrparenthesis_\xi/l_\xi\rfloor$ respectively. Consider the following cases. 
\begin{enumerate}
\item If $\lfloor \llparenthesis x^\prime\rrparenthesis_\xi/l_\xi\rfloor-\lfloor a^\prime/l_\xi\rfloor>1$ and $\lfloor b^\prime/l_\xi\rfloor-\lfloor \llparenthesis x^\prime\rrparenthesis_\xi/l_\xi\rfloor>1$, then by (\ref{xi-order-requirement-ensured}), Spoiler cannot force Duplicator to pick an object in the $(\xi-1)$-th abstraction that  contains a critical point, and the abstraction-order-condition holds for the lower abstractions.  

\item Assume that $\lfloor b^\prime/l_\xi\rfloor-\lfloor \llparenthesis x^\prime\rrparenthesis_\xi/l_\xi\rfloor=1$ and $\lfloor \llparenthesis x^\prime\rrparenthesis_\xi/l_\xi\rfloor-\lfloor a^\prime/l_\xi\rfloor\geq 1$. The case when $\lfloor b^\prime/l_\xi\rfloor-\lfloor \llparenthesis x^\prime\rrparenthesis_\xi/l_\xi\rfloor\geq 1$ and $\lfloor \llparenthesis x^\prime\rrparenthesis_\xi/l_\xi\rfloor-\lfloor a^\prime/l_\xi\rfloor=1$ is similar. 
Note that, $(c^\star,y)\in\mathbb{X}_m^*$ for any critical point $(c^\star,y)$. By Lemma \ref{i=0theni-1=0}, $(c^\star,y)\in\mathbb{X}_\xi^*$. 
If $\lfloor c^\star/l_\xi\rfloor\neq\lfloor x^\prime/l_\xi\rfloor$, then obviously Duplicator avoids picking the critical point. Hence, 
assume that $\lfloor c^\star/l_\xi\rfloor=\lfloor x^\prime/l_\xi\rfloor$. Recall that $m>\xi$, it implies that $x^\prime\geq c^\star$, because $[c^\star]_\xi=\llbracket c^\star\rrbracket_\xi^{min}$, which is easy to show (cf. p. \pageref{llbracket-min} for the definition of $\llbracket c^\star\rrbracket_\xi^{min}$).  
Recall that $\xi\geq m-\ell_c+2$. Therefore, for any $(u,y)$ where $\lfloor u/l_\xi\rfloor\neq \lfloor x^\prime/l_\xi\rfloor$ and $u>x^\prime$,
$\lfloor u/l_{\xi-1}\rfloor-\lfloor c^\star/l_{\xi-1}\rfloor\geq \lfloor\frac{([c^\star]_\xi+1)\beta_{m-\xi}^{m-1}-c^\star}{l_{\xi-1}}\rfloor=
\lfloor\frac{\frac{1}{2}\beta_{m-\xi}^{m-1}-\frac{1}{2}\sum_{1<i\leq \xi-1}\beta_{m-i}^{m-1}}{l_{\xi-1}}\rfloor>
\lfloor\frac{\frac{1}{2}\beta_{m-\xi}^{m-1}-\beta_{m-\xi+1}^{m-1}}{l_{\xi-1}}\rfloor\\
=\lfloor\frac{\beta_{m-\xi+1}^{m-1}(\frac{1}{2}\beta_{m-\xi}^{m-\xi+1}-1)}{\beta_{m-\xi+1}^{m-1}\mathpzc{U}_{\xi-1}^*}\rfloor
=\lfloor\frac{2^{\xi-2}\mathpzc{U}_{\xi-1}^*-1}{\mathpzc{U}_{\xi-1}^*}\rfloor>
2^{\xi-2}-1\geq 2^{m-\ell_c}-1$. See Fig. \ref{order-remark-ii-b-1}. Note that black vertices in the figure (the $2$nd, $5$th, $8$th, and $11$th vertices) are those in $\mathbb{X}_{\xi}^*$ and grey vertices (the $1$st, $3$rd, $4$th, $6$th, $9$th, $10$th and $12$th vertices) are those whose indices are $1$. 
Therefore,  $\lfloor b^\prime/l_{\xi-1}\rfloor-\lfloor  c^\star/l_{\xi-1}\rfloor\geq 2^{m-\ell_c}-1$. 

%\begin{comment}
\begin{figure}[]
\hspace*{-3mm}
%\centering
\includegraphics[trim = -35mm 0mm 0mm 0mm, clip, scale=0.39]{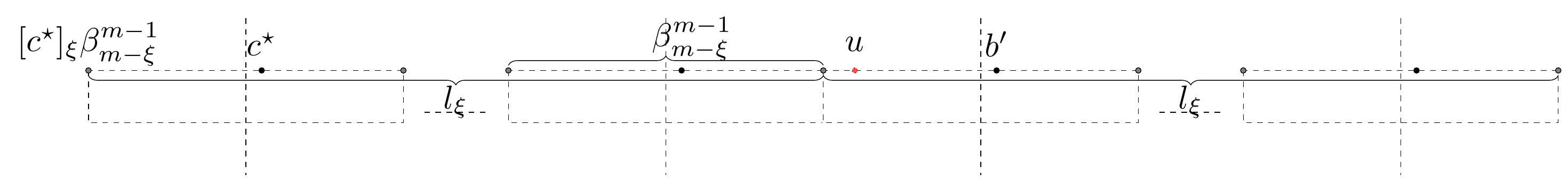}
%\scalebox{10}{}
\caption{  The case (ii) (b): $\lfloor \frac{b^\prime}{l_{\xi-1}}\rfloor-\lfloor\frac{c^\star}{l_{\xi-1}}\rfloor\geq 2^{m-\ell_c}-1$.}
\label{order-remark-ii-b-1}
\end{figure}
%\end{comment}

On the other hand,  
 $\lfloor c^\star/l_{\xi-1}\rfloor-\lfloor a^\prime/l_{\xi-1}\rfloor>\lfloor \frac{\beta_{m-\xi}^{m-1}}{2l_{\xi-1}}\rfloor=\lfloor\frac{\beta_{m-\xi}^{m-\xi+1}}{2\mathpzc{U}_{\xi-1}^*} \rfloor=2^{\xi-2}\geq 2^{m-\ell_c}$. See Fig. \ref{order-remark-ii-b-2}. 
As a consequence, $\lfloor b^\prime/l_{\xi-1}\rfloor-\lfloor a^\prime/l_{\xi-1}\rfloor\geq 2^{m-\ell_c+1}-1$. 
Therefore, in this case Duplicator has the freedom to avoid picking a critical point, since the (unabridged) interval is big enough.  
%\begin{comment}
\begin{figure}[htbp]
\hspace*{-3mm}
\centering
\includegraphics[trim = -35mm 0mm 0mm 0mm, clip, scale=0.38]{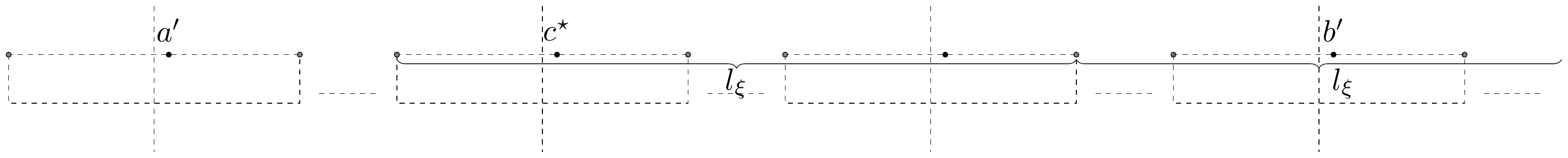}
%\scalebox{10}{}
\caption{  The case (ii) (b): $\lfloor \frac{c^\star}{l_{\xi-1}}\rfloor-\lfloor \frac{a^\prime}{l_{\xi-1}}\rfloor>2^{m-\ell_c}$.}
\label{order-remark-ii-b-2}
\end{figure}
%\end{comment}

Similarly, $\lfloor b/l_{\xi-1}\rfloor-\lfloor \llparenthesis x\rrparenthesis_\xi/l_{\xi-1}\rfloor\geq 2^{m-\ell_c}-1$ and $\lfloor \llparenthesis x\rrparenthesis_\xi/l_{\xi-1}\rfloor-\lfloor a/l_{\xi-1}\rfloor\geq  2^{m-\ell_c}$.
In other words, in such case, the abstraction-order-condition holds (for any $1\leq i\leq \xi-1$).  

\item Assume that either $\lfloor \llparenthesis x^\prime\rrparenthesis_\xi/l_\xi\rfloor=\lfloor a^\prime/l_\xi\rfloor$ or $\lfloor \llparenthesis x^\prime\rrparenthesis_\xi/l_\xi\rfloor=\lfloor b^\prime/l_\xi\rfloor$. Let $(c^\star,y)$ be any critical point that is settled between $(a^\prime,y)$ and $(b^\prime,y)$. 
If $[a^\prime]_\xi=[c^\star]_\xi$, it means that the pebbled vertex $(x_{a^\prime},y)$, where $\llparenthesis x_{a^\prime}\rrparenthesis_\xi=a^\prime$, is projected to $(c^\star,y)$ in the $\xi$-th abstractions, i.e. $\llparenthesis x_{a^\prime}\rrparenthesis_\xi=c^\star$, because of Lemma \ref{projection} and $[\llparenthesis x_{a^\prime}\rrparenthesis_\xi]_\xi=[x_{a^\prime}]_\xi=[a^\prime]_\xi=[c^\star]_\xi$. 
See Fig. \ref{order-remark-ii-b-3}. 
%\begin{comment}
\begin{figure}[htbp]
\hspace*{1mm}
\centering
\includegraphics[trim = -35mm 0mm 0mm 0mm, clip, scale=0.39]{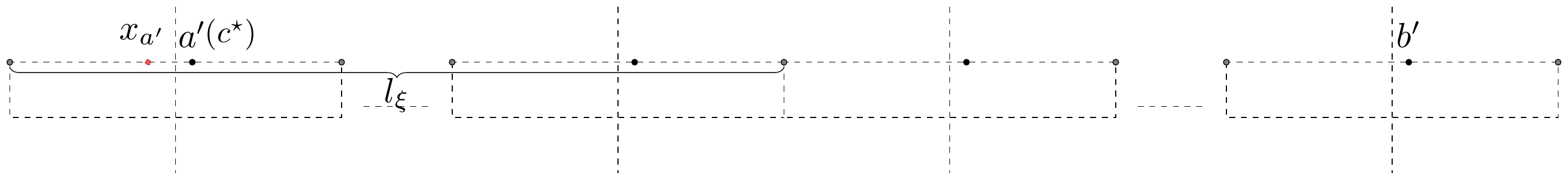}
%\scalebox{10}{}
\caption{  The case (ii) (c): $a^\prime=c^\star$.}
\label{order-remark-ii-b-3}
\end{figure}
%\end{comment} 
If Duplicator is forced to pick $(c^\star,y)$ due to the abstraction-order-condition, i.e. 
$(x^\prime,y)$ is $(c^\star,y)$, then $a^\prime=x^\prime$. Therefore, by the abstraction-order-condition, $x=a$. But, Duplicator wins this round since she wins the last round. 
Similarly, if $[b^\prime]_\xi=[c^\star]_\xi$, Duplicator can also win this round. 
Therefore, we need only consider the case where $[a^\prime]_\xi\neq [c^\star]_\xi\neq [b^\prime]_\xi$. See Fig. \ref{order-remark-ii-b-4}. 
%In this case, $b^\prime-a^\prime\geq \beta_{m-1}^{m-\xi}=2^{\xi-1} l_{\xi-1}\geq 2^{m-\ell_c+1}l_{\xi-1}$. 
%\begin{comment}
\begin{figure}[htbp]
\hspace*{1mm}
\centering
\includegraphics[trim = -35mm 0mm 0mm 0mm, clip, scale=0.39]{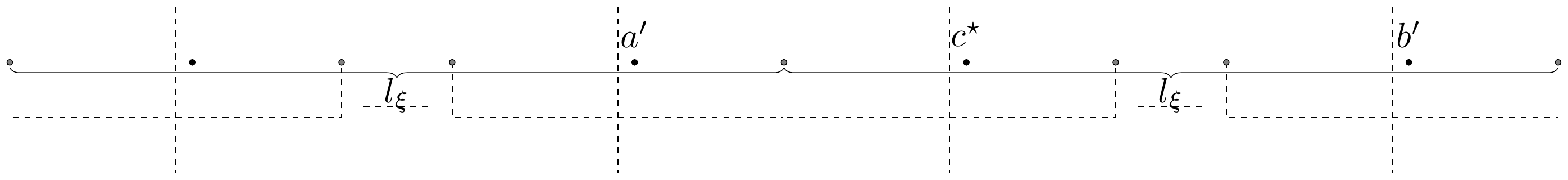}
%\scalebox{10}{}
\caption{  The case (ii) (c): $[a^\prime]_\xi\neq [c^\star]_\xi\neq [b^\prime]_\xi$.}
\label{order-remark-ii-b-4}
\end{figure}
%\end{comment}
Because $(a^\prime,y),(b^\prime,y)\in\mathbb{X}_\xi^*$, we have  $\lfloor\frac{b^\prime}{l_{\xi-1}}\rfloor-\lfloor\frac{a^\prime}{l_{\xi-1}}\rfloor=\lfloor\frac{[b^\prime]_\xi\beta_{m-\xi}^{m-1}+\frac{1}{2}\sum_{1<i\leq \xi}\beta_{m-i}^{m-1}}{l_{\xi-1}}\rfloor-\lfloor\frac{[a^\prime]_\xi\beta_{m-\xi}^{m-1}+\frac{1}{2}\sum_{1<i\leq \xi}\beta_{m-i}^{m-1}}{l_{\xi-1}}\rfloor=2^{\xi-1}([b^\prime]_\xi-[a^\prime]_\xi)\geq 2^{\xi}>2^{m-\ell_c+1}$. Therefore, Duplicator has the freedom to avod picking any object in the $(\xi-1)$-th abstraction that contains a critical point in such case.

\end{enumerate}

\end{enumerate}

In summary, in all the cases, the abstraction-order-condition holds for any abstractions below the $\xi$-th abstractions. 
Therefore, the abstraction-order-condition works.  
Moreover, Duplicator can also avoid picking critical points in most cases. 

Another issue, which is directly related to linear orders, needs to be mentioned briefly. Recall Strategy \ref{xi-1} in the proof of the main lemma \ref{main-lemma}, we defaultly assume that Lemma \ref{no-missing-edges_xi-1} applies in accordance with the abstraction-order-condition. This is, however, quite obvious, according to the same intuition as presented in the last arguments. That is, a unit of difference in higher abstraction equals huge difference in lower abstractions. Hence, once the game enters lower abstractions, abstraction-order-condition will be ensured automatically. 
\end{remark}

\begin{remark}\label{remark-ommit-mod}
$[\llparenthesis x\rrparenthesis_\xi]_{r}-[x]_r=[x]_\xi\beta_{m-\xi}^{m-r}+\frac{1}{2}\sum_{r<j\leq \xi}\beta_{m-j}^{m-r}-[x]_r<\frac{x}{\beta_{m-\xi}^{m-1}}\beta_{m-\xi}^{m-r}+\frac{1}{2}\sum_{r<j\leq \xi}\beta_{m-j}^{m-r}-\left(\frac{x}{\beta_{m-r}^{m-1}}-1\right)=
\frac{1}{2}\sum_{r<j\leq \xi}\beta_{m-j}^{m-r}+1<\frac{1}{2}\beta_{m-\xi}^{m-r}+\frac{\xi-r-1}{2}\beta_{m-\xi+1}^{m-r}\\<\frac{1}{2}\beta_{m-\xi}^{m-r}+2^{\xi-2}\mathpzc{U}_{\xi-1}^*\beta_{m-\xi+1}^{m-r}=\beta_{m-\xi}^{m-r}$. 

On the other hand, $[\llparenthesis x\rrparenthesis_\xi]_{r}-[x]_r>\left(\frac{x}{\beta_{m-\xi}^{m-1}}-1\right)\beta_{m-\xi}^{m-r}+\frac{1}{2}\sum_{r<j\leq \xi}\beta_{m-j}^{m-r}-\frac{x}{\beta_{m-r}^{m-1}}>-\beta_{m-\xi}^{m-r}$. 

In other words, $\left|[\llparenthesis x\rrparenthesis_\xi]_{r}-[x]_r\right|<\beta_{m-\xi}^{m-r}$. It implies that, we can usually ommit ``mod $\beta_{m-\xi}^{m-r}$''. 
For example, it is not difficult to see that  
\begin{multline*}
\llbracket x\rrbracket_{r}^{min}-[\llparenthesis x\rrparenthesis_\xi]_{r}\equiv\llbracket x^{\prime}\rrbracket_{r}^{min}-
[\llparenthesis x^{\prime}\rrparenthesis_\xi]_{r} \hspace{5pt}(\mbox{mod }\beta_{m-\xi}^{m-r})\Leftrightarrow\\ \llbracket x\rrbracket_{r}^{min}-[\llparenthesis x\rrparenthesis_\xi]_{r}=\llbracket x^{\prime}\rrbracket_{r}^{min}-
[\llparenthesis x^{\prime}\rrparenthesis_\xi]_{r}.
\end{multline*} 
\end{remark}

\begin{remark}
Note that $\mathfrak{A}_{k,m}^*$ is an ordinary ``flat'' graph. By contrast,  $\mathfrak{A}_{k,m}$ has a logical structure (or temporal structure) that reflects evolution of a game, thereby enforce the game evolves \textit{reasonably}. If Spoiler try to breakout current evolution by picking a vertex associated with board configuration that is not in the progress, the game will split into two, each progresses in isolation, with its own evolution of board configurations, since any pair of vertices whose board histories are not in a succession (or one history cannot continue the other) are not adjacent. 
Two evolutions may converge,  if their board configurations are almost the same except for one vertex. However, two board histories that once diverge will never converge again since the initial segments of the histories will never match again, which is formalized by the following Lemma.  
\begin{fact}\label{boardHistory-never-converge-1} 
For any $(x_1,y_1),(x_2,y_2),(x_3,y_3)\in\mathbb{X}_1$, if $(x_1,y_1)\xrightarrow[\mathrm{BC}]{*}(x_3,y_3)$,  $(x_2,y_2)\xrightarrow[\mathrm{BC}]{*}(x_3,y_3)$ and  $\chi(x_1,y_1)\!\!\restriction\mathrm{bc}\neq\chi(x_2,y_2)\!\!\restriction\mathrm{bc}$, then either $(x_1,y_1)\xrightarrow[\mathrm{BC}]{*}(x_2,y_2)$ or $(x_2,y_2)\xrightarrow[\mathrm{BC}]{*}(x_1,y_1)$, but not both.
\end{fact}

%\begin{proof}
It is obvious by definition. Here we give a summary. 

$(x_1,y_1)\xrightarrow[\mathrm{BC}]{*}(x_3,y_3)$ and $(x_2,y_2)\xrightarrow[\mathrm{BC}]{*}(x_3,y_3)$ say that the initial parts of the board histories of $(x_1,y_1)$ and $(x_2,y_2)$ must be the same, thereby one of the history can evlove to the other. Also note that, by definition, evolution has a direction, akin to time. That is, if board history $H_1$ can evolve to history $H_2$, then $H_2$ cannot evolve to $H_1$.   
\end{remark}

\begin{remark}\label{strategy-sketch}
The proof is involved. We sketch the ideas as the following. Assume that Spoiler picks  $(x,y)$ in $\widetilde{\mathfrak{A}}_{k,m}$, and Duplicator replies with $(x^\prime,y)$ in $\widetilde{\mathfrak{B}}_{k,m}$.  We use $(\widetilde{\mathfrak{A}}_{k,m}^*,\widetilde{\mathfrak{B}}_{k,m}^*)$ to denote the \textit{associated game board} of $(\widetilde{\mathfrak{A}}_{k,m},\widetilde{\mathfrak{B}}_{k,m})$. Hence, if Spoler picks $(x,y)$ in $\widetilde{\mathfrak{A}}_{k,m}$, we take it that he also picks $(x^\flat,y)$ in $\widetilde{\mathfrak{A}}_{k,m}^*$ in an \textit{associated game}.  

First of all, we introduce a sort of imaginary games called\textit{ pebble games over changing board}. It is similar to the usual pebble games except that the game board can be different in each round, i.e. the pair of graphs (``flat'' structures) are continuously changing during the game. Such games are the basis for the so called virtual games. 
We use ``\textit{virtual games}'' to denote the kinds of pebble games over changing board  wherein the players play the games in their mind without really putting pebbles on the board. For simplicity, here we assume that no vertex is pebbled before Spoiler picks $(x,y)$.
A virtual game in such a simple setting consists of $\mathrm{i}_{\mathrm{cur}}^{x,y}-1$ ``virtual rounds''. No vertex is pebbled at the beginning of this virtual game. Spoiler ``picks'' in a structure akin to $\widetilde{\mathfrak{A}}_{k,m}^*$ (varying in each round)  according to $\chi(x,y)\!\!\restriction\!\!\mathrm{BH}(j)$, for $j=1$ to $\mathrm{i}_{\mathrm{cur}}^{x,y}-1$, and Duplicator ``replies'' in the other (changing) structure, i.e.  in a structure akin to $\widetilde{\mathfrak{B}}_{k,m}$. The point is that, in a changing board, for any vertex $(u,v)\in\mathbb{X}_1^*$, if $(u,v)$ is not adjacent to $(x^\flat,y)$ simply because it is in $\chi(x^\flat,y)\!\!\restriction\!\! S$, then it is already pebbled. Similarly for any vertex $(u^\prime,v)$ in the other structure w.r.t. the adjacency to $(x^\prime,y)$. 
Duplicator uses virtual games to determine the board history of the vertex $(x^\prime,y)$ she is going to pick.  

We are able to prove the following claim: Duplicator has a strategy such that, if an edge is forbidden in one structure due to discontinuities, so is the corresponding edge in the other structure; moreover, the orders of the board histories of the pebbled vertices can be properly taken care of by Duplicator.

Therefore, via the virtual games, we can reduce the original game to the associated game over changing board, wherein Duplicator uses strategy over abstractions. 
 Note that the $\xi$-th abstraction is the abstraction that Duplicator would care about at the start of the current round. Suppose we have a version of abstraction-order-condition that is similar to 1$^\diamond$\textapprox 6$^\diamond$ introduced in  page  \pageref{special-winning-condtion-set}. 
The point is that, if Duplicator can win this round in the game over the $\xi$-th abstraction, she can also win this round in lower abstraction. Particularly, it means she can win this round in the first abstraction, i.e. over the game board $(\widetilde{\mathfrak{A}}_{k,m}^*,\widetilde{\mathfrak{B}}_{k,m}^*)$.   
At the beginning of the game, $\xi=m$. Suppose Spoiler picks $(x^\flat,y)$ and Duplcator replies $(x^{\prime\flat},y)$ in the current round. While playing the game, in each round Duplicator first finds the candidate positions that are in accordance with the abstraction-order-condition, and that usually form intervals containing vertices of necessary type labels (using an auxiliary game over linear orders);  afterwards, she determines the type label of $(x^{\prime\flat},y)$ and makes the pick. The following strategy will help her decide the type label of $(x^{\prime\flat},y)$. 
%However, we omit the involved details here. 
It helps  Duplicator keep the game board in partial isomorphism at the end of a round, not only in the original associated game, but also in the associated game over abstractions. The strategy (i.e. Strategy \ref{play-in-xi-abs}\textapprox Strategy \ref{t<xi}) of Duplicator can be sketched briefly and roughly as  follows.  In short, she needs to ensure that the winning-condition-set is preserved throughout the game. 
\begin{enumerate}
\item $\langle$ Strategy 1 $\rangle$ Assume that Spoiler picks $(x^\flat,y)$ in $\mathbb{X}_\xi^*$. Duplicator first \textit{tries to} use the strategy that works for the $\xi$-th abstraction of the structures. That is, she plays the game over the $\xi$-th abstraction and  pick a vertex in $\mathbb{X}_\xi^*$ s.t. the winning-condition-set holds.   

\item $\langle$ Strategy 2 $\rangle$ Assume that Spoiler picks $(x^\flat,y)$ in $\mathbb{X}_\xi^*$. If Strategy \ref{play-in-xi-abs} does not work, i.e. Duplicator cannot pick a vertex in $\mathbb{X}_\xi^*$ satisfying the winning-condition-set, then Duplicator resorts to the strategy that works for the $(\xi-1)$-th abstraction of the structures and pick a vertex in $\mathbb{X}_{\xi-1}^*-\mathbb{X}_{\xi}^*$. Note that, she can always find such a vertex that ensures a win for her in this round. More precisely, she first find a vertex of index $\xi-1$ such that it is adjacent to the projection of all the pebbled vertices in the $(\xi-1)$-th abstraction (cf. Lemma \ref{no-missing-edges_xi-1} \textit{(4)}); afterwards, she adjust her pick such that it satisfies the winning-condition-set.   

This strategy also works if $\mathrm{idx}(x^\flat,y)=\xi-1$. 
Duplicator will pick a vertex of index $\xi-1$.  
In addition, she ensures that $\mathbf{cc}([x^\flat]_{\xi-1},y)=\mathbf{cc}([x^{\prime\flat}]_{\xi-1},y)$,  $g(x^\flat)=g(x^{\prime\flat})$ and $\mathrm{RngNum}(x^\flat,\xi-1)=\mathrm{RngNum}(x^{\prime\flat},\xi-1)$. 
 
\item $\langle$ Strategy 3 $\rangle$ If Spoiler picks $(x^\flat,y)$ in $\mathbb{X}_t^*-\mathbb{X}_{t+1}^*$, where $t<\xi-1$, then Duplicator regards it as if $(\llparenthesis x^\flat\rrparenthesis_\xi, y)$ is picked, and replies with $(x^{\prime\flat},y)$ such that $(\llparenthesis x^{\prime\flat}\rrparenthesis_\xi, y)$ is the vertex she will pick to respond $(\llparenthesis x^\flat\rrparenthesis_\xi, y)$ using her strategy that works in the $\xi$-th abstraction (or responds with $(\llparenthesis x^{\prime\flat}\rrparenthesis_{\xi-1}, y)$ using her strategy that works in the $(\xi-1)$-th abstraction, if she cannot respond properly in the $\xi$-th abstraction). At the same time, Duplicator ensures that in the original game  (\ref{main-diamond-xi}$^\diamond$) holds, i.e. $x^{\prime\flat}-\llparenthesis x^{\prime\flat}\rrparenthesis_\xi$ is \textit{roughly}\footnote{Just note that a unit of difference in higher abstraction means a huge difference in lower abstractions  w.r.t. distance of first coordinates.} the same as $x^\flat-\llparenthesis x^\flat\rrparenthesis_\xi$ (or $x^{\prime\flat}-\llparenthesis x^{\prime\flat}\rrparenthesis_{\xi-1}$ is roughly the same as $x^\flat-\llparenthesis x^\flat\rrparenthesis_{\xi-1}$ if she cannot respond properly in the $\xi$-th abstraction).  Hence Duplicator is an approximate hr-copycat, which can ensure that lower abstractions are in partial isomorphism if so is some higher abstration. 
Moreover, Duplicator resorts to a sort of game reduction from lower abstraction to higher abstraction to prevent Spoiler from finding difference via linear order (the auxiliary game over linear can only tell her the approximate (candidate) positions she should consider; it cannot avoid $(\varkappa)$). It also helps Duplicator decide the type label of $(x^{\prime\flat},y)$.  In sum, Duplicator can ensure that, without exploring the difference in higher abstraction of the structures, Spoiler is \textit{not} able to find difference between the structures by exploring the lower abstractions at some specific stage of the game. 

\end{enumerate}
%\hfill\ensuremath{\divideontimes}
\end{remark}

\begin{comment}
\begin{remark}
An alternative way to prove our main lemma is to introduce constants. 

\begin{definition}
Let $\mathfrak{A}_{k,m}^+$ and $\mathfrak{B}_{k,m}^+$ be built from $\mathfrak{A}_{k,m}$ and $\mathfrak{B}_{k,m}$ respectively  by adding a set of constants $\{(a,b) \1 a=0\mbox{ or } a=\gamma_{m-1}-1; b\in[0,k-1]\}$. 
% 
\end{definition}

Note that we can regard constants as extra immovable ``pebbles''. By definitions, the following lemma is self-evident.
\begin{lemma}
$\mathfrak{A}_{k,m}^+\equiv_m^{k-1} \mathfrak{B}_{k,m}^+$ implies $\mathfrak{A}_{k,m}\equiv_m^{k-1} \mathfrak{B}_{k,m}$.
\end{lemma}

Although this way has some advantage, we do not adopt it simply to avoid introducing more complicated denotations, such as $\mathfrak{A}_{k,m}^{*\mathbb{Z}_i+}$.  
\end{remark}
\end{comment}

\begin{remark}\label{remark-why-pick-in-cl}
It is obvious that the strategy introduced in the proof of Lemma \ref{main-lemma} preservs (\ref{picking-in-cl}) throughout the games. The reason we need (\ref{picking-in-cl}) is that we depend on (\ref{xi-1-simuluation-1}) to ensure that (\ref{condition-for-DuWin}$^\diamond$) holds in Strategy 2. Moreover, in Strategy 3 we also use it to ensure that (\ref{main-diamond-xi}$^\diamond$) (vii) and (\ref{main-diamond-xi}$^\diamond$) (iv) hold simultaneously, cf. the corresponding remark. 
We haven't introduced the standard concept ``types'' yet. But, $\mathbf{cl}(x^\flat,y)=\mathbf{cl}(u^\flat,v)$ means that the vertices $(x^\flat,y)$ and $(u^\flat,v)$ are \textit{roughly} the same, although they are usually different objects in the linear order, and they do not necessary have the same type. But all the critical points with the same second coordinate have the same type. We introduce the concept ``type label'' as an alternative of ``type'', which gives us some flexibility in the constructions and proofs.   
\end{remark}

%\end{remark}

\begin{remark}\label{remark-strategy2}
In Strategy \ref{play-in-xi-abs} (cf. the proof of Lemma \ref{main-lemma}), we claim that ``If $(x,y)\twoheadrightarrow (u,v)$ for some pebbled vertex $(u,v)$, then Duplicator simply let $(x^\prime,y)$ be $(u^\prime,v)\!\!\restriction\!\!_\mathrm{H}^{\mathrm{i}_{\mathrm{cur}}^{x,y}}$, and 
we are able to show that (\ref{condition-for-DuWin}$^\diamond$) and (\ref{condition-for-DuWin-*}$^\diamond$) hold''. 
We first show that (\ref{condition-for-DuWin}$^\diamond$) holds. 
 In other words, if $(x,y)\xrightarrow[\mathrm{BH}]{con.}(u,v)\land (x^\flat,y)\in (u,v)[\mathrm{BC}]$ for some pebbled vertex $(u,v)$, then $(x^\flat,y)$ is adjacent to a vertex in $\overrightarrow{c_A}$ if and only if  $(x^{\prime\flat},y)$ is adjacent to the corresponding vertex in $\overrightarrow{c_B}$.
Here we give more explanation. 

We can divide the pebbled vertices in a structure into two sets according to whether their associated board histories are in continuity with the board history of $(x,y)$. By  2 (b) of Definition \ref{iterative-expansion}, we know that, a vertex is not adjacent to $(x,y)$ if their associated histories are not in continuity. 
Recall that $(x,y)\Vdash (x^\prime,y)$ and $(u,v)\Vdash (u^\prime,v)$. 
By Claim \ref{board-history-evolutions}, we know that if $(u,v)$ is not adjacent to $(x,y)$ because of this reason, so is $(u^\prime,v)$ to $(x^\prime,y)$.   
Recall that the set of pebbled vertices, whose histories are in continuity with the history of $(x,y)$, is $\widetilde{c_A}$ and  $\widetilde{c_A}\Vdash\widetilde{c_B}$. By assumption, $(u,v)\in \widetilde{c_A}$. 
    Recall that Duplicator has a winning strategy in the virtual games over changing board that determine the board history of a vertex (this strategy is akin to Strategy \ref{play-in-xi-abs}\textapprox Strategy \ref{t<xi}, but the pebbled vertices at the start of a round may be different).  
Duplicator's strategy only depends on the pair of board configurations at the start of the current virtual round.\footnote{\label{footnote-free-will} It means that, in the virtual rounds of a board history wherein Spoiler picks one vertex, e.g. $(x,y)$, in the pair $(x,y)\Vdash(x^\prime,y)$, Duplicator \textit{always} responds with the other vertex, i.e. $(x^\prime,y)$, in the pair, if the game board is in the same state, i.e. the pair of board configurations is the same. } 
By inductive hypothesis, Duplicator wins the virtual round wherein the pair of board configurations at the start are $(u,v)[\mathrm{BC}]$ and $(u^\prime,v)[\mathrm{BC}]$, and the players pick $(u,v)$, $(u^\prime,v)$ in this virtual round. Therefore, $(x,y)$ is adjacent to $(u,v)$ if and only if $(x^\prime,y)$ is adjacent to $(u^\prime,v)$, according to the definition of virtual games, B-3 and the premise that $(x^\flat,y)\in (u,v)[\mathrm{BC}]$. 
Moreover, for \textit{any} $(u^{\star},v^{\star})\in \widetilde{c_A}$ where $(u^{\star},v^{\star})\Vdash (u^{\star\prime},v^{\star})$, 
if $(u^{\star},v^{\star})\!\twoheadrightarrow\!(x,y)$, then by Claim \ref{board-history-evolutions},  $(u^{\star\prime},v^{\star})\!\twoheadrightarrow\!(x^\prime,y)$, and by the transitivity of $\twoheadrightarrow$,  $(u^{\star},v^{\star})\!\twoheadrightarrow\!(u,v)$. So is  $(u^{\star\prime},v^{\star})\!\twoheadrightarrow\!(u^\prime,v)$. 
By definition, $(x,y)\xrightarrow[\mathrm{BH}]{con.}(u,v)$ and $(x,y)[\mathrm{BC}]\sqsubseteq(u,v)[\mathrm{BC}]$. 
Therefore, 
$(u^{\star\flat},v^\star)\in (x,y)[\mathrm{BC}]\sqsubseteq(u,v)[\mathrm{BC}]$. 
Hence, $(u^\star,v^\star)[\mathrm{BC}]\ccirc (u^{\star\flat},v^\star)\sqsubseteq(u,v)[\mathrm{BC}]$.  
In other words, 
both $(x^\flat,y)$ and $(u^{\star\flat},v^{\star})$ are in $(u,v)[\mathrm{BC}]$. Similarly,  $(x^{\prime\flat},y), (u^{\star\prime\flat},v^{\star})\in (u^\prime,v)[\mathrm{BC}]$. It implies that  $(x^\flat,y)$ is adjacent to $(u^{\star\flat},v^\star)$ if and only if $(x^{\prime\flat},y)$ is adjacent to $(u^{\star\prime\flat},v^\star)$, since Duplicator wins the virtual round wherein 
the pair of board configurations is made of $(u,v)[\mathrm{BC}]$ and  $(u^\prime,v)[\mathrm{BC}]$ at the start of the round, and 
the players pick the pair of vertices $(u^\flat,v)$ and $(u^{\prime\flat},v)$.  
For any vertex that is not in $\widetilde{c_A}$, say $(a,b)$ where $(a,b)\Vdash (a^\prime,b)$, we know that $((a,b),(x,y))\notin E^{A}$ and $((a^{\prime},b),(x^\prime,y))\notin E^{B}$, by Claim \ref{board-history-evolutions} and  Definition \ref{iterative-expansion}. 
 All in all, we have shown that the game board is in partial isomorphism after the players pick $(x,y)$ and $(x^\prime,y)$, if $(x,y)\xrightarrow[\mathrm{BH}]{con.}(u,v)\land (x^\flat,y)\in (u,v)[\mathrm{BC}]$ for some pebbled vertex $(u,v)$. 
 
In the above argument, we show that (\ref{condition-for-DuWin}$^\diamond$) holds, based on the assumption that Duplicator has a winning strategy (cf. Strategy \ref{play-in-xi-abs}\textapprox Strategy \ref{t<xi}) in the virtual games. Note that this strategy is also a strategy over abstractions. In other words, (\ref{condition-for-DuWin-*}$^\diamond$) can also be ensured, using similar argument. 
\end{remark}

\begin{remark}\label{partial-isom-propagate}
If the game board is in partial isomorphism over the $\xi$-th abstraction at the start of the current round, then we can show that this also holds over the $(\xi-1)$-th abstraction. 
Here we give a brief explanation, since many ideas have already been explaned in Strategy \ref{play-in-xi-abs}\textapprox Strategy \ref{t<xi}. 
Firstly, for any pebbled vertex $(u,v)$, we have $(\llparenthesis u^\flat\rrparenthesis_{\xi-1},v)\!\!\restriction\!\!\Omega\Vdash (\llparenthesis u^{\prime\flat}\rrparenthesis_{\xi-1},v)\!\!\restriction\!\!\Omega$. 
The readers can cf. Strategy \ref{t<xi} for the arguments needed (i.e. the arguments for $D\Vdash D^\prime$). The main point is that, if $\mathrm{idx}(\llparenthesis u^\flat\rrparenthesis_{\xi-1},v)<\xi$, it implies that Duplicator has used Strategy \ref{t<xi} in the round where $(u,v)$ is picked. Then by (\ref{main-diamond-xi}$^\diamond$) (i) (ii), we have $\mathrm{idx}(\llparenthesis u^\flat\rrparenthesis_{\xi-1},v)=\mathrm{idx}(\llparenthesis u^{\prime\flat}\rrparenthesis_{\xi-1},v)$, say equal to $i$, and  
$\mathbf{cc}([\llparenthesis u^\flat\rrparenthesis_{\xi-1}]_{i},v)=\mathbf{cc}([\llparenthesis u^{\prime\flat}\rrparenthesis_{\xi-1}]_{i},v)$ if $i\neq \xi$. 
Second, for similar reason, for pebbled pairs $(u,v)\!\Vdash\! (u^\prime,v)$ and $(e,f)\!\Vdash\! (e^\prime,f)$, if either $\mathrm{idx}(\llparenthesis u^\flat\rrparenthesis_{\xi-1},v)<\xi$ or $\mathrm{idx}(\llparenthesis e^\flat\rrparenthesis_{\xi-1},f)<\xi$, then we can show that 
\begin{align}\label{eqn-partial-isom-propagate}
\mathbf{cc}([\llparenthesis u^\flat\rrparenthesis_{\xi-1}]_\ell,v) &=\mathbf{cc}([\llparenthesis u^{\prime\flat}\rrparenthesis_{\xi-1}]_\ell,v) \nonumber\\
\mathbf{cc}([\llparenthesis e^\flat\rrparenthesis_{\xi-1}]_\ell,f) &=\mathbf{cc}([\llparenthesis e^{\prime\flat}\rrparenthesis_{\xi-1}]_\ell,f),
\end{align}
 where $\ell=\min\{\mathrm{idx}(\llparenthesis u^\flat\rrparenthesis_{\xi-1},v),\mathrm{idx}(\llparenthesis u^{\prime\flat}\rrparenthesis_{\xi-1},v)\}$.   
Third, , due to (\ref{main-diamond-xi}$^\diamond$) (iv), $(\llparenthesis u^\flat\rrparenthesis_{\xi-1},v)\!\!\restriction\!\! S\cap (x,y)[\mathrm{BC}]=(\llparenthesis u^{\prime\flat}\rrparenthesis_{\xi-1},v)\!\!\restriction\!\! S\cap (x^\prime,y)[\mathrm{BC}]$, if $v\neq y$. 
Fourth, $\mathbf{BIT}(\mathrm{SW}((\llparenthesis u\rrparenthesis_i,v),(\llparenthesis  x^\flat\rrparenthesis_i,y)),\hat{q}(v,y))$ iff $\mathbf{BIT}(\mathrm{SW}((\llparenthesis u^\prime\rrparenthesis_i,v),(\llparenthesis  x^{\prime\flat}\rrparenthesis_i,y)),\hat{q}(v,y))$ if $\mathrm{idx}(\llparenthesis u^\flat\rrparenthesis_{\xi-1},v)<\xi$, due to (\ref{main-diamond-xi}$^\diamond$) (vi). 
Finally, we need explain one more thing: the adjacency determined by $\mathrm{sgn}((\llparenthesis u\rrparenthesis_i,v),(\llparenthesis e^\flat\rrparenthesis_i,f))$ will not cause a problem when $\xi$ is decreased by $1$, because of (\ref{main-diamond-xi}$^\diamond$) (v). Note that, so far we have only considered the case when $\mathrm{idx}(\llparenthesis u^\flat\rrparenthesis_{\xi-1},v)<\xi$ or $\mathrm{idx}(\llparenthesis e^\flat\rrparenthesis_{\xi-1},f)<\xi$. If both $\mathrm{idx}(\llparenthesis u^\flat\rrparenthesis_{\xi-1},v)\geq \xi$ and $\mathrm{idx}(\llparenthesis e^\flat\rrparenthesis_{\xi-1},f)\geq\xi$, then obviously $(\llparenthesis u^\flat\rrparenthesis_{\xi-1},v)$ is adjacent to $\llparenthesis e^\flat\rrparenthesis_{\xi-1},f)$ iff $(\llparenthesis u^{\prime\flat}\rrparenthesis_{\xi-1},v)$ is adjacent to $\llparenthesis e^{\prime\flat}\rrparenthesis_{\xi-1},f)$, because in such case $(\llparenthesis u^\flat\rrparenthesis_{\xi-1},v)=(\llparenthesis u^\flat\rrparenthesis_{\xi},v)$ and $(\llparenthesis u^{\prime\flat}\rrparenthesis_{\xi-1},v)=(\llparenthesis u^{\prime\flat}\rrparenthesis_{\xi},v)$. 
In short, the game bord is still in partial isomorphism w.r.t. the edges when $\xi$ is decreased by $1$. 

Moreover, we can also show that the game bord is still in partial isomorphism w.r.t. the orders when $\xi$ is decreased by $1$. By Lemma \ref{i=0theni-1=0}, we have $u^\flat=\llparenthesis u^\flat\rrparenthesis_i$ for $1\leq i\leq \mathrm{idx}(u^\flat,v)$. Therefore, for pebbled pairs $(u,v)\!\Vdash\! (u^\prime,v)$ and $(e,v)\!\Vdash\! (e^\prime,v)$, if both $\mathrm{idx}(\llparenthesis u^\flat\rrparenthesis_{\xi-1},v)\geq \xi$ and $\mathrm{idx}(\llparenthesis e^\flat\rrparenthesis_{\xi-1},v)\geq\xi$, then $\llparenthesis u^\flat\rrparenthesis_{\xi-1} \leq \llparenthesis e^\flat\rrparenthesis_{\xi-1}$ if $\llparenthesis u^\flat\rrparenthesis_{\xi} \leq \llparenthesis e^\flat\rrparenthesis_{\xi}$. That is, the following holds:
\begin{equation}\label{eqn-remark-partial-isom-propagate}
(\llparenthesis u^\flat\rrparenthesis_{\xi-1},v)\leq (\llparenthesis e^\flat\rrparenthesis_{\xi-1},v)\Leftrightarrow (\llparenthesis u^{\prime\flat}\rrparenthesis_{\xi-1},v)\leq (\llparenthesis e^{\prime\flat}\rrparenthesis_{\xi-1},v)
\end{equation}
Otherwise, if  either  $\mathrm{idx}(\llparenthesis u^\flat\rrparenthesis_{\xi-1},v)\geq \xi$ or $\mathrm{idx}(\llparenthesis e^\flat\rrparenthesis_{\xi-1},v)\geq\xi$ but not both, then clearly either $\llparenthesis u^\flat\rrparenthesis_{\xi-1} < \llparenthesis e^\flat\rrparenthesis_{\xi-1}$ and $\llparenthesis u^{\prime\flat}\rrparenthesis_{\xi-1} < \llparenthesis e^{\prime\flat}\rrparenthesis_{\xi-1}$, or, $\llparenthesis u^\flat\rrparenthesis_{\xi-1} > \llparenthesis e^\flat\rrparenthesis_{\xi-1}$ and $\llparenthesis u^{\prime\flat}\rrparenthesis_{\xi-1} > \llparenthesis e^{\prime\flat}\rrparenthesis_{\xi-1}$. If $\mathrm{idx}(\llparenthesis u^\flat\rrparenthesis_{\xi-1},v)<\xi$, then it is because $\llparenthesis u^\flat\rrparenthesis_{\xi}-\llparenthesis u^\flat\rrparenthesis_{\xi-1}$ \textit{roughly} equals $\llparenthesis u^{\prime\flat}\rrparenthesis_{\xi}-\llparenthesis u^{\prime\flat}\rrparenthesis_{\xi-1}$ (modulo $\beta_{m-\xi}^{m-\xi+1}$).\footnote{It is equivalent if we ignore a difference in distance up to $\mathpzc{U}_{\xi-1}^*$.} Now suppose that both $\mathrm{idx}(\llparenthesis u^\flat\rrparenthesis_{\xi-1},v)< \xi$ and $\mathrm{idx}(\llparenthesis e^\flat\rrparenthesis_{\xi-1},v)<\xi$. In such case  $\mathrm{idx}(\llparenthesis u^{\prime\flat}\rrparenthesis_{\xi-1},v)=\mathrm{idx}(\llparenthesis u^\flat\rrparenthesis_{\xi-1},v)$ and $\mathrm{idx}(\llparenthesis e^{\prime\flat}\rrparenthesis_{\xi-1},v)=\mathrm{idx}(\llparenthesis e^\flat\rrparenthesis_{\xi-1},v)$. 
If $\mathrm{idx}(\llparenthesis u^\flat\rrparenthesis_{\xi-1},v)\neq \mathrm{idx}(\llparenthesis e^\flat\rrparenthesis_{\xi-1},v)$, then $\llparenthesis u^\flat\rrparenthesis_{\xi-1}\neq\llparenthesis e^\flat\rrparenthesis_{\xi-1}$. If $\llparenthesis u^\flat\rrparenthesis_{\xi}=\llparenthesis e^\flat\rrparenthesis_{\xi}$, then, similar to the last case, \eqref{eqn-remark-partial-isom-propagate} clearly holds. If $\llparenthesis u^\flat\rrparenthesis_{\xi}\neq\llparenthesis e^\flat\rrparenthesis_{\xi}$, then \eqref{eqn-remark-partial-isom-propagate} holds because a ``unit'' of difference in higher abstraction is huge in lower abstraction, i.e. a vertex in higher abstraction corresponds to a very big interval in lower abstractions.\footnote{More precisely, $\beta_{m-\xi}^{m-1}$ is greater than both $|\llparenthesis u^\flat\rrparenthesis_{\xi}-\llparenthesis u^\flat\rrparenthesis_{\xi-1}|$ and $|\llparenthesis u^{\prime\flat}\rrparenthesis_{\xi}-\llparenthesis u^{\prime\flat}\rrparenthesis_{\xi-1}|$.} 
Suppose that $\mathrm{idx}(\llparenthesis u^\flat\rrparenthesis_{\xi-1},v)=\mathrm{idx}(\llparenthesis e^\flat\rrparenthesis_{\xi-1},v)$. Similar to the last case, \eqref{eqn-remark-partial-isom-propagate} clearly holds if $\llbracket\llparenthesis u^\flat\rrparenthesis_{\xi-1} \rrbracket_{\xi-1}\neq \llbracket\llparenthesis e^\flat\rrparenthesis_{\xi-1} \rrbracket_{\xi-1}$. Assume that $\llbracket\llparenthesis u^\flat\rrparenthesis_{\xi-1} \rrbracket_{\xi-1}=\llbracket\llparenthesis e^\flat\rrparenthesis_{\xi-1} \rrbracket_{\xi-1}$,  
\eqref{eqn-remark-partial-isom-propagate} is easy to prove because of  \eqref{eqn2-1-round-game-reduction}.

In summary, the game board is in partial isomorphism over the $(\xi-1)$-th abstraction if it is in partial isomorphism over the $\xi$-th abstraction, at the start of the current round.   
\end{remark}

\begin{remark}\label{variable-hierarchy-collapse-arithm-struc}
It is well-known that the bounded variable hierarchy collapses to $\fo^3$ on coloured linear orders  \cite{Poizat1982ColorOrder}. 
Similarly, we can  prove  that it also collapses to $\fo^3$ on pure arithmetic structures, using similar pebble game type argument (cf., e.g., \cite{DawarHowmany}, p.9\textapprox p.10, or \cite{Immerman1999Book} p.105\textapprox p.107).

For clarity and proofreading, we put this proof here. Note that the proof of Lemma 3 introduced in \cite{DawarHowmany} implicitly relies on transitivity of linear orders. But it is not true for $\mathbf{BIT}$. Therefore, we need adapt the lemma as well as the proof a little bit to ensure that the partial isomorphisms over pairs of small pieces of structures can be merged consistently into one partial isomorphism over a pair of bigger piece. That is, we need to show that Duplicator's strategies in $3$-pebble games can be merged to ensure one partial isomorphism that extends all the partial isomorphisms in the $3$-pebble games. Recall that we assume the structures in discourse to be $\langle \leq, \mathbf{BIT}\rangle$-structures. That is, here we only consider pure arithmetic structures. 

\begin{lemma}\label{lemma-collapse-in-BIT}
Let $s=(a_1,\ldots,a_\ell)$ and $t=(b_1,\ldots,b_\ell)$ be $\ell$-tuples where  $a_i\in |\mathfrak{A}|$, $b_i\in |\mathfrak{B}|$ and $a_i\leq a_{i+1}$, $b_i\leq b_{i+1}$ for any $i$. If $(\mathfrak{A},a_i,a_j)\equiv_m^3 (\mathfrak{B},b_i,b_j)$ for any $1\leq i,j\leq \ell$, then $(\mathfrak{A},s)\equiv_m (\mathfrak{B},t)$.    
\end{lemma}
\begin{proof}
The proof is by induction on $m$. We only focus on $\mathbf{BIT}$. For the argument that takes care of linear orders, the readers can cf. e.g. \cite{DawarHowmany} or \cite{Immerman1999Book}. 
 
\textit{Basis:}

Let $f(a_i)=b_i$. If $m=0$, then it is easy to verify that the map $f$ defines a partial isomorphism from $\mathfrak{A}$ to $\mathfrak{B}$. 

We can take it that $(a_i,b_i)$ be the pair of elements that are pebbled in the same round.   

\textit{Induction step:} 

Assume that the claim holds for $m=d$ and that $(\mathfrak{A},a_i,a_j)\equiv_{d+1}^3 (\mathfrak{B},b_i,b_j)$ for any $1\leq i,j\leq \ell$. We need to show that  $(\mathfrak{A},s)\equiv_{d+1} (\mathfrak{B},t)$. In the $(d+1)$-th round, if Spoiler picks a pebbled element $u$, then Duplicator simply picks the other pebbled element in the pair containing $u$ and by induction hypothesis she wins this round. Hence we assume that Duplicator picks a new element in this round. Suppose w.l.o.g. that, in the first round of the Ehrenfeucht-Fra\"iss\' e\xspace game  $\Game_{d+1}((\mathfrak{A},s),(\mathfrak{B},t))$, Spoiler picks $a^\star$ in $\mathfrak{A}$. Moreover, assume that $a_i\leq a^\star\leq a_{i+1}$ for some $i$ (The cases when $a_{\ell}\leq a^\star$ and $a^\star\leq a_1$ are similar).
Duplicator can resort to the strategy that works over $\Game_{d+1}^3((\mathfrak{A},a_i,a_{i+1}),(\mathfrak{B},b_i,b_{i+1}))$. 
The point is that three variables are necessary and sufficient to simulate the Ehrenfeucht-Fra\"iss\' e\xspace game in the $3$-pebble game over a piece of the structures s.t. its strategy can be extended to give an isomorphism on a bigger piece, where $a_i$ and $a_{i+1}$ can be regarded as either constants or pebbled vertices (the pebbles are from the three pairs of pebbles in the pebble game). 
Observe that she really has \textit{a family of} strategies work well in this round, varying on $\mathbf{BIT}(a^\star,a_j)$ (if $a_j\leq a^\star$) or $\mathbf{BIT}(a_j,a^\star)$ (if $a^\star\leq a_j$) for any $a_j$ different from $a_i$ and $a_{i+1}$, only if one of them works well. For example, suppose that $\mathbf{BIT}(a^\star,a_i)$ is true and $\mathbf{BIT}(a_{i+1},a^\star)$ is false. By induction hypothesis, Duplicator can find $b^\star$ such that $\mathbf{BIT}(b^\star,b_i)$ is true and $\mathbf{BIT}(b_{i+1},b^\star)$ is false. Then she has a family of strategies that are in accordance with this condition, but different in other aspects, e.g. she can choose the one that sets $\mathbf{BIT}(b^\star,b_{i-1})$ to true or she can choose the one that sets $\mathbf{BIT}(b^\star,b_{i-1})$ to false. Note that either way leads to  a valid strategy, which is crucial for the following arguments.\footnote{It explains why this proof does not work if the signature of the structures contains one binary relation that is not fixed as $\mathbf{BIT}$ does. Usually we also call $\mathbf{BIT}$  a background relation.} Similarly, for any $j^\prime,j$, we can find a family of strategies work for Duplicator over the pebble game       
$\Game_{d+1}^3((\mathfrak{A},a_j,a_{j^\prime}),(\mathfrak{B},b_j,b_{j^\prime}))$. The point is that the intersection of these family of strategies is not empty only if one of them works. Therefore any strategy in the intersection works over any $3$-pebble games mentioned. Hence, in the first round Duplicator need only choose a strategy in the intersection to respond Spoiler, and this strategy ensures that $(\mathfrak{A},a^\star,a_p)\equiv_d^3 (\mathfrak{B},b^\star,b_p)$ for any $p$.
Let $s^\prime$ be the ordered list that extends $s$ by inserting the element $a^\star$ in appropriate position and similarly for $t^\prime$ by inserting $b^\star$.  
 Then by induction hypothesis, we have $(\mathfrak{A},s^\prime)\equiv_{d} (\mathfrak{B},t^\prime)$. Therefore, $(\mathfrak{A},s)\equiv_{d+1} (\mathfrak{B},t)$.   
\end{proof}    

%\begin{comment}
\begin{figure}[]
\hspace*{-6mm}
%\centering
\includegraphics[trim = -30mm 0mm 0mm 0mm, scale=0.35]{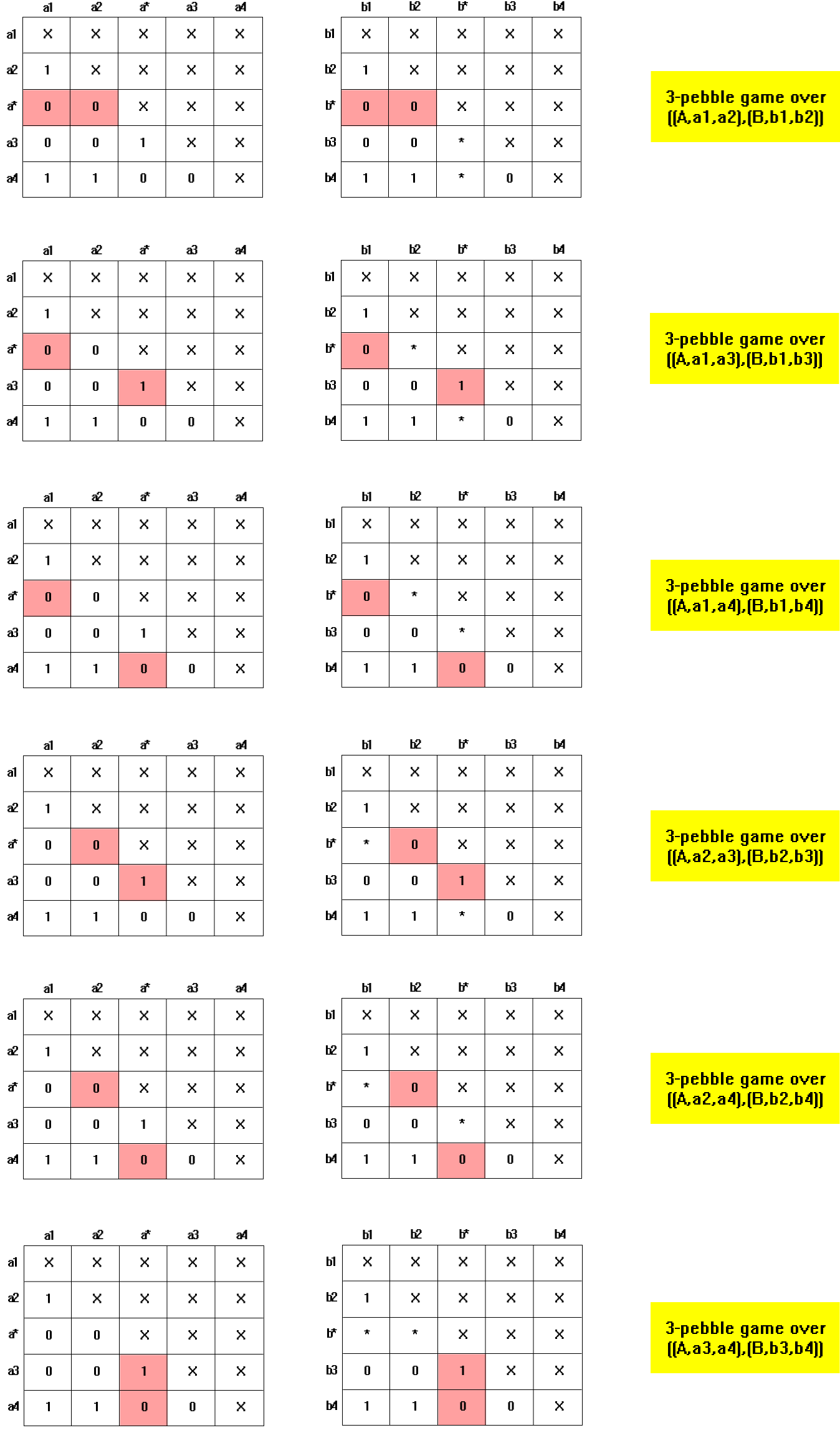}
%\scalebox{10}{}
\caption{This example illustrates that the intersection of the families of strategies is not empty, and Duplcator can compute the intersection to obtain a strategy that works for the game $\Game_m((\mathfrak{A},s),(\mathfrak{B},t))$.}
\label{pic-collapse-in-BIT}
\end{figure}
%\end{comment}

A diagram may help us to understand the computation of intersection of families of strategies. Here we give a small example to illustrate it. See Fig. \ref{pic-collapse-in-BIT}. Here $s=(a_1,a_2,a_3,a_4)$ and $t=(b_1,b_2,b_3,b_4)$. 
In the $3$-pebble games, We assume w.l.o.g. that Spoiler picks $a^\star$ and Duplicator responds with $b^\star$. The elements are listed according to the linear order. Hence $a^\star$ is the third element in the order. For each entry $(a_i,a_j)$ that has a value $0$ or $1$,  we can see that $a_j\leq a_i$. And this entry tells us whether $\mathbf{BIT}(a_i,a_j)$ is $0$ or $1$. Note that   
``\textasteriskcentered'' in the entries stands for either $0$ or $1$, , i.e. both values are allowed. Pink cells stand for the partial isomorphisms that should be fixed in the corresponding game (indicated on the right side). 

Once Lemma \ref{lemma-collapse-in-BIT} is proved, by the Theorem 4 in \cite{DawarHowmany}, we know that the following holds. 
\begin{corollary}\label{collapse-in-BIT-0}
The bounded variable hierarchy collapses to $\fo^3$ on pure arithmetic structures.
\end{corollary}

With a careful analysis of the proof of Lemma \ref{lemma-collapse-in-BIT}, we get a variant of Corollary \ref{collapse-in-BIT-0} as follows. 
\begin{corollary}\label{collapse-in-BIT-1}
For any $k\geq 5$, on pure arithmetic structures, any sentence $Q_1 x_1 Q_2 x_2\cdots Q_k x_k \varphi(x_1,x_2,\cdots,x_k)$ is equivalent to a sentence in $\fo^k$, where $Q_i\in \{\exists,\forall\}$ and $\varphi$ is any first-order formula. 
\end{corollary}

\end{remark}

%+++++++++++++++++++ proofs of lemmas

\textbf{Proof of Lemma \ref{projection}}.
\begin{proof}

Since $(x^{\prime},y)\in\mathbb{X}_i^*$, by definition, we have
\begin{equation*}
\begin{split}
x^{\prime} &=\llparenthesis x^\prime\rrparenthesis_i\\  
           &=[x^{\prime}]_i\beta_{m-i}^{m-1}+\displaystyle\frac{1}{2}\sum_{1<p\leq 
             i}\beta_{m-p}^{m-1}
\end{split}
\end{equation*}

By definition, we also have 
\begin{equation*}
\llparenthesis x\rrparenthesis_i= [x]_i\beta_{m-i}^{m-1}+\displaystyle\frac{1}{2}\sum_{1<p\leq i}\beta_{m-p}^{m-1}
\end{equation*}

Therefore, $x^\prime=\llparenthesis x\rrparenthesis_i$, insomuch as $[x^\prime]_i=[x]_i$.
\end{proof}

\textbf{Proof of Lemma \ref{i=0theni-1=0}}.
\begin{proof}
Since $(x,y)\in\mathbb{X}_i^*$, by definition $x=\llparenthesis x\rrparenthesis_i$.
It is trivial when $i=1, 2$.
For any $2<i\leq m$, we show that $x=\llparenthesis x\rrparenthesis_{i-1}$  
if $x=\llparenthesis x\rrparenthesis_i$. 

By definition, and $\beta_{m-i}^{m-1}/\beta_{m-i+1}^{m-1}=\gamma_{m-i+1}/\gamma_{m-i}>i-2$, we have   
\begin{equation}\label{app-i=0theni-1=0-eqn0} 
\displaystyle\sum_{1<j\leq i-1}\beta_{m-j}^{m-1}\!<\!(i-2) \beta_{m-i+1}^{m-1}\!<\!\beta_{m-i}^{m-1}. 
\end{equation}

By definition, 
\begin{equation}\label{app-i=0theni-1=0-eqn3}
\frac{1}{2}\beta_{m-i}^{m-1} \mbox{ is divisible by } \beta_{m-i+1}^{m-1}. 
\end{equation}

In the following we show that 
\begin{equation}\label{app-i=0theni-1=0-eqn4}
[x]_i\beta_{m-i}^{m-1}+\frac{1}{2}\beta_{m-i}^{m-1}=[x]_{i-1}\beta_{m-i+1}^{m-1}.
\end{equation} 
First, suppose for a contradiction that $[x]_i\beta_{m-i}^{m-1}+\frac{1}{2}\beta_{m-i}^{m-1}>[x]_{i-1}\beta_{m-i+1}^{m-1}$. 

Let $\psi_1:=[x]_i\beta_{m-i}^{m-1}+\frac{1}{2}\beta_{m-i}^{m-1}$.

By $x=\llparenthesis x\rrparenthesis_i$ and $i>2$, we have  
\begin{equation*}
x>\psi_1.
\end{equation*}

Let $\psi_2:=([x]_{i-1}+1)\beta_{m-i+1}^{m-1}$.

Then by the assumption and (\ref{app-i=0theni-1=0-eqn3}), we know that 
\begin{equation*}
x>\psi_1\geq \psi_2. 
\end{equation*}

Note that 
\begin{equation*}
\left\lfloor\frac{x}{\beta_{m-i+1}^{m-1}}\right\rfloor \beta_{m-i+1}^{m-1}\leq x<\left(\!\left\lfloor\frac{x}{\beta_{m-i+1}^{m-1}}\right\rfloor+1\!\!\right) \beta_{m-i+1}^{m-1}=\psi_2.
\end{equation*}

Therefore, 
\begin{equation*}
x>\psi_2>x. 
\end{equation*}
A contradiction occurs. 

Second, suppose that $[x]_i\beta_{m-i}^{m-1}+\frac{1}{2}\beta_{m-i}^{m-1}<[x]_{i-1}\beta_{m-i+1}^{m-1}$. Therefore, 
\begin{eqnarray*}
x&\geq& [x]_{i-1}\beta_{m-i+1}^{m-1}\\
&\geq&  [x]_i\beta_{m-i}^{m-1}+\frac{1}{2}\beta_{m-i}^{m-1}+\beta_{m-i+1}^{m-1} \hspace{67pt} [\mbox{by (\ref{app-i=0theni-1=0-eqn3})}]\\
&\geq& x+\frac{1}{2}\beta_{m-i+1}^{m-1}-\frac{1}{2}\sum_{1<j\leq i-2}\beta_{m-j}^{m-1}.  \hspace{35pt} [\because x=\llparenthesis x\rrparenthesis_i]
\end{eqnarray*}

By (\ref{app-i=0theni-1=0-eqn0}), we have 
\begin{equation*}
\displaystyle\sum_{1<j\leq i-2}\beta_{m-j}^{m-1}<\beta_{m-i+1}^{m-1}. 
\end{equation*}

As a consequence, we have $x>x$. 
A contradiction occurs again. Therefore, (\ref{app-i=0theni-1=0-eqn4}) holds.

Therefore, 
\begin{eqnarray*}
\displaystyle\llparenthesis x\rrparenthesis_{i-1}&=&[x]_{i-1}\beta_{m-i+1}^{m-1}+\frac{1}{2}\sum_{1<j<i-1} \beta_{m-j}^{m-1} \quad [\mbox{by definition}]\\ &=&
[x]_i\beta_{m-i}^{m-1}+\frac{1}{2}\beta_{m-i}^{m-1}+\frac{1}{2}\sum_{1<j<i-1} \beta_{m-j}^{m-1} \quad [\mbox{by (\ref{app-i=0theni-1=0-eqn4})}]\\ &=&
[x]_i\beta_{m-i}^{m-1}+\frac{1}{2}\sum_{1<j\leq i} \beta_{m-j}^{m-1}\\&=&\llparenthesis x\rrparenthesis_i. \hspace{140pt} [\mbox{by definition}]
\end{eqnarray*} 
As a consequence, the claim holds. 
\end{proof}

\begin{comment}
\textbf{Proof of Lemma \ref{lattice-point-high-is-lower}}.
\begin{proof}
By definition, we have $x=\llparenthesis x\rrparenthesis_i$. Therefore, we have $$\llparenthesis x\rrparenthesis_j=\left\lfloor \frac{[x]_i\beta_{m-i}^{m-1}+\frac{1}{2}\sum_{1<\ell\leq i}\beta_{m-\ell}^{m-1}}{\beta_{m-j}^{m-1}}\right\rfloor\times \beta_{m-j}^{m-1}+\frac{1}{2}\sum_{1<\ell\leq j}\beta_{m-\ell}^{m-1}.$$ 
Similar to \eqref{app-i=0theni-1=0-eqn0}, we have 
$\sum_{1<\ell\leq j-1}\beta_{m-\ell}^{m-1}\!<\!(j-2) \beta_{m-j+1}^{m-1}\!<\!\beta_{m-j}^{m-1}$.  
Hence 
\begin{equation}\label{eqn-lattice-point-high-is-lower}
\frac{1}{2}\sum_{1<\ell\leq j}\beta_{m-\ell}^{m-1}<\beta_{m-j}^{m-1}.
\end{equation} 
On the other hand, $\beta_{m-p}^{m-1}$ is greater than $\beta_{m-j}^{m-1}$ and is divisible by $\beta_{m-j}^{m-1}$, for any $p>j$. 
Therefore, we have $\llparenthesis x\rrparenthesis_j=([x]_i\beta_{m-i}^{m-j}+\frac{1}{2}\sum_{j<\ell\leq i}\beta_{m-\ell}^{m-j})\beta_{m-j}^{m-1}+\frac{1}{2}\sum_{1<\ell\leq j}\beta_{m-\ell}^{m-1}=[x]_i\beta_{m-i}^{m-1}+\frac{1}{2}\sum_{1<\ell\leq i}\beta_{m-\ell}^{m-1}=x$. 
\end{proof}
\end{comment}

\textbf{Proof of Lemma \ref{proj-abs}}.
\begin{proof}
By definition, $[\llparenthesis x\rrparenthesis_i]_j=\left\lfloor \frac{[x]_i\beta_{m-i}^{m-1}+\frac{1}{2}\sum_{1<\ell\leq i}\beta_{m-\ell}^{m-1}}{\beta_{m-j}^{m-1}}\right\rfloor$. 

Recall that $\beta_{m-j}^{m-1}/\beta_{m-i}^{m-1}=\beta_{m-j}^{m-i}\in\mathbf{N}^+$. By  \eqref{app-i=0theni-1=0-eqn0}, we have 
\begin{equation}\label{eqn-lattice-point-high-is-lower}
\frac{1}{2}\sum_{1<\ell\leq i}\beta_{m-\ell}^{m-1}<\beta_{m-i}^{m-1}. 
\end{equation}  
Therefore, we have 
\begin{equation}\label{eqn-abs-proj-proof}
[\llparenthesis x\rrparenthesis_i]_j=\left\lfloor \frac{[x]_i\beta_{m-i}^{m-1}}{\beta_{m-j}^{m-1}}\right\rfloor.
\end{equation} 

If $x$ is divisible by $\beta_{m-i}^{m-1}$, then clearly $[\llparenthesis x\rrparenthesis_i]_j=[x]_j$. Assume that $x=c\cdot\beta_{m-i}^{m-1}+\Delta$ where $0<\Delta<\beta_{m-i}^{m-1}$. Hence, $[\llparenthesis x\rrparenthesis_i]_j=\left\lfloor \frac{c\cdot\beta_{m-i}^{m-1}}{\beta_{m-j}^{m-1}}\right\rfloor$. 
By the similar reason for \eqref{eqn-abs-proj-proof}, we have 
$[x]_j=\left\lfloor \frac{c\cdot\beta_{m-i}^{m-1}}{\beta_{m-j}^{m-1}}\right\rfloor$. Therefore, $[\llparenthesis x\rrparenthesis_i]_j=[x]_j$.  

Now we prove \textit{(2)}.  
Just note that, by definition, $\llparenthesis\llparenthesis x\rrparenthesis_i\rrparenthesis_j=[\llparenthesis x\rrparenthesis_i]_j\beta_{m-j}^{m-1}+\frac{1}{2}\sum_{1<\ell\leq j}\beta_{m-\ell}^{m-1}$. 
By \textit{(1)}, we have $\llparenthesis\llparenthesis x\rrparenthesis_i\rrparenthesis_j=[x]_j\beta_{m-j}^{m-1}+\frac{1}{2}\sum_{1<\ell\leq j}\beta_{m-\ell}^{m-1}=\llparenthesis x\rrparenthesis_j$.   
\end{proof}

\textbf{Proof of Lemma \ref{abstraction-strategy-premier}}. 

\begin{proof}
It comes from the intuition that one unit of difference in higher abstraction is huge in lower abstractions. Note that 
\begin{equation}\label{eqn-lattice-distance-0}
\beta_{m-\xi+1}^{m-1}>|a-\llparenthesis a\rrparenthesis_{\xi-1}| \mbox{ and } \beta_{m-\xi+1}^{m-1}>|a^\prime-\llparenthesis a^\prime\rrparenthesis_{\xi-1}|.
\end{equation} 
We shall see that, $a-\llparenthesis a\rrparenthesis_{\xi}\neq a^\prime-\llparenthesis a^\prime\rrparenthesis_{\xi}$ if $\llparenthesis a\rrparenthesis_{\xi}-\llparenthesis a\rrparenthesis_{\xi-1}\neq\llparenthesis a^\prime\rrparenthesis_{\xi}-\llparenthesis a^\prime\rrparenthesis_{\xi-1}$, hence a contradiction occurs.  

For example, assume that $\llparenthesis a\rrparenthesis_{\xi}-\llparenthesis a\rrparenthesis_{\xi-1}>\llparenthesis a^\prime\rrparenthesis_{\xi}-\llparenthesis a^\prime\rrparenthesis_{\xi-1}$ and $a-\llparenthesis a\rrparenthesis_{\xi}<0$. 
The other cases are similar. 
Firstly, $a-\llparenthesis a\rrparenthesis_{\xi}=(a-\llparenthesis a\rrparenthesis_{\xi-1})-(\llparenthesis a\rrparenthesis_{\xi}-\llparenthesis a\rrparenthesis_{\xi-1})$. Note that $a^\prime-\llparenthesis a^\prime\rrparenthesis_{\xi}<0$, for $a-\llparenthesis a\rrparenthesis_{\xi}=a^\prime-\llparenthesis a^\prime\rrparenthesis_{\xi}$. 
By Lemma \ref{proj-greater-index}, both $\llparenthesis a\rrparenthesis_{\xi}$ and $\llparenthesis a\rrparenthesis_{\xi-1}$ are vertices of index greater than or equal to $\xi-1$, which means that both of them are in $\mathbb{X}_{\xi-1}^*$. By Fact \ref{unit-distance}, we have  
\begin{equation}\label{eqn-lattice-distance}
|\llparenthesis a\rrparenthesis_{\xi}-\llparenthesis a\rrparenthesis_{\xi-1})|\geq \beta_{m-\xi+1}^{m-1}>|a-\llparenthesis a\rrparenthesis_{\xi-1}|.
\end{equation} 
Therefore, $\llparenthesis a\rrparenthesis_{\xi}-\llparenthesis a\rrparenthesis_{\xi-1}>0$ 
since $(a-\llparenthesis a\rrparenthesis_{\xi-1})-(\llparenthesis a\rrparenthesis_{\xi}-\llparenthesis a\rrparenthesis_{\xi-1})<0$, and so is $\llparenthesis a^\prime\rrparenthesis_{\xi}-\llparenthesis a^\prime\rrparenthesis_{\xi-1}$. 
In this case observe that either $a\leq \llparenthesis a\rrparenthesis_{\xi-1}<\llparenthesis a\rrparenthesis_\xi$ and  $a^\prime\leq \llparenthesis a^\prime\rrparenthesis_{\xi-1}<\llparenthesis a^\prime\rrparenthesis_\xi$ or
 $\llparenthesis a\rrparenthesis_{\xi-1}\leq a<\llparenthesis a\rrparenthesis_\xi$ and  $\llparenthesis a^\prime\rrparenthesis_{\xi-1}\leq a^\prime<\llparenthesis a^\prime\rrparenthesis_\xi$. Suppose that the former holds. The other case is similar.  
Then, by \eqref{eqn-lattice-distance-0}, $\beta_{m-\xi+1}^{m-1}>||a-\llparenthesis a\rrparenthesis_{\xi-1}|-|a^\prime-\llparenthesis a^\prime\rrparenthesis_{\xi-1}||=|(a-\llparenthesis a\rrparenthesis_{\xi-1})-(a^\prime-\llparenthesis a^\prime\rrparenthesis_{\xi-1})|$. 
Similar to \eqref{eqn-lattice-distance}, by Fact \ref{unit-distance}, $(\llparenthesis a\rrparenthesis_{\xi}-\llparenthesis a\rrparenthesis_{\xi-1})-(\llparenthesis a^\prime\rrparenthesis_{\xi}-\llparenthesis a^\prime\rrparenthesis_{\xi-1})\geq \beta_{m-\xi+1}^{m-1}$, since $\llparenthesis a\rrparenthesis_{\xi}-\llparenthesis a\rrparenthesis_{\xi-1}>\llparenthesis a^\prime\rrparenthesis_{\xi}-\llparenthesis a^\prime\rrparenthesis_{\xi-1}$. Therefore, $(a-\llparenthesis a\rrparenthesis_{\xi})-(a^\prime-\llparenthesis a^\prime\rrparenthesis_{\xi})=((a-\llparenthesis a\rrparenthesis_{\xi-1})-(a^\prime-\llparenthesis a^\prime\rrparenthesis_{\xi-1}))-((\llparenthesis a\rrparenthesis_{\xi}-\llparenthesis a\rrparenthesis_{\xi-1})-(\llparenthesis a^\prime\rrparenthesis_{\xi}-\llparenthesis a^\prime\rrparenthesis_{\xi-1}))<0$. We arrive at a contradiction. 
This shows that (1) holds, which immediately implies that (2) holds. 
\end{proof}

\textbf{Proof of Lemma \ref{cm=depth}}.
\begin{proof}
Since $(x,y)\in\mathbb{X}_i^*$, by definition, $x=\llparenthesis x\rrparenthesis_i=[x]_i\beta_{m-i}^{m-1}+\displaystyle\frac{1}{2}\sum_{1<p\leq i} \beta_{m-p}^{m-1}$. Then for any $1\leq j<i$, 
\begin{eqnarray*}
[x]_{j}&=&\displaystyle\left\lfloor \frac{x}{\beta_{m-j}^{m-1}}\right\rfloor\\ &=&\displaystyle\left\lfloor \frac{[x]_i\beta_{m-i}^{m-1}+\displaystyle\frac{1}{2}\sum_{1<p\leq i} \beta_{m-p}^{m-1}}{\beta_{m-j}^{m-1}}\right\rfloor
\\&=&\displaystyle\left\lfloor \frac{[x]_i\beta_{m-j}^{m-1}\beta_{m-i}^{m-j}+\displaystyle\frac{1}{2}\sum_{1<p\leq i} \beta_{m-j}^{m-1}\beta_{m-p}^{m-j}}{\beta_{m-j}^{m-1}}\right\rfloor\\
&=& [x]_i\beta_{m-i}^{m-j}+\displaystyle\frac{1}{2}\sum_{j<p\leq i}\beta_{m-p}^{m-j}
\end{eqnarray*}

By definition,  both $\beta_{m-i}^{m-j}$ and $\frac{1}{2}\beta_{m-p}^{m-j}$  are divisible by $k-1$ 
for any $j<p\leq i$. Therefore, $[x]_{j}$ is divisible by $k-1$, and $\mathbf{cc}([x]_j,y)=y$ mod $k-1$.
\end{proof}

\textbf{Proof of Lemma \ref{conquer-boundary-strategy}}. 

\begin{proof}

We first show that, for any $i$ where $1\leq i\leq q$, 
\begin{equation}\label{eqn-boundary-strategy} 
[\llparenthesis x\rrparenthesis_p]_i\equiv [\llparenthesis x^\prime\rrparenthesis_p]_i \hspace{3pt}(\mbox{mod }k-1).
\end{equation} 

If $q<min\{\mathrm{idx}(\llparenthesis x\rrparenthesis_p,y),\mathrm{idx}(\llparenthesis x^\prime\rrparenthesis_p,y)\}$, then \ref{eqn-boundary-strategy} holds due to Lemma \ref{cm=depth}. Henceforth we assume that $q=min\{\mathrm{idx}(\llparenthesis x\rrparenthesis_p,y),\mathrm{idx}(\llparenthesis x^\prime\rrparenthesis_p,y)\}$. 

Let $\mathrm{idx}(\llparenthesis x\rrparenthesis_p,y)=\ell$ and  
$\mathrm{idx}(\llparenthesis x^\prime\rrparenthesis_p,y)=\ell^\prime$. 
W.l.o.g we assume that $\ell^\prime\leq \ell$. 
By Lemma \ref{proj-greater-index}, $p\leq \ell^\prime\leq\ell$. 
If $[\llparenthesis x\rrparenthesis_p]_{\ell^\prime}$ mod $k-1\neq 0$, it implies that  $\ell=\ell^\prime$, because, if $\ell^\prime<\ell$ then, by Lemma \ref{cm=depth}, $[\llparenthesis x\rrparenthesis_p]_{\ell^\prime}$ mod $k-1=0$ and $[\llparenthesis x\rrparenthesis_p]_i\equiv [\llparenthesis x^\prime\rrparenthesis_p]_i \equiv 0$ (mod $k-1$) for $1\leq i<\ell^\prime$. Hence \eqref{eqn-boundary-strategy} holds. Suppose that  $[\llparenthesis x\rrparenthesis_p]_{\ell^\prime}$ mod $k-1=0$. By Lemma \ref{cm=depth}, $[\llparenthesis x\rrparenthesis_p]_i\equiv [\llparenthesis x^\prime\rrparenthesis_p]_i\equiv 0$ (mod $k-1$), for $1\leq i\leq \ell^\prime$. 

It is easy to observe that \eqref{eqn-boundary-strategy} implies \eqref{eqn-boundary-strategy-1}, provided that (1) holds throughout the game. 
%Similar to Remark \ref{ExplanationOfAbstraction-specalcase}, a strict argument is verbose. 
Briefly speaking, it relies on an observation that the neighbourhoods of vertices of the same index are isomorphic. 
%For $1\leq j\leq q$, if $i=q$, by Lemma \ref{proj-abs}, $[\llparenthesis x^\prime\rrparenthesis_j]_i=[\llparenthesis x^\prime\rrparenthesis_p]_i=[x^\prime]_i$. 
By definition, $\mathrm{idx}(\llparenthesis x\rrparenthesis_p,y)\geq q$. 
If $\mathrm{idx}(\llparenthesis x\rrparenthesis_p,y)>q$, then by Lemma \ref{cm=depth} and \textit{(2)},  $[\llparenthesis x\rrparenthesis_p]_i\equiv [\llparenthesis x^\prime\rrparenthesis_p]_i\equiv 0 \hspace{3pt}(\mbox{mod }k-1)$ for $1\leq i\leq q$. Now suppose that $\mathrm{idx}(\llparenthesis x\rrparenthesis_p,y)=q$.  
By Lemma \ref{abstraction-strategy-premier}, $\llparenthesis x\rrparenthesis_j-\llparenthesis x\rrparenthesis_p=\llparenthesis x^\prime\rrparenthesis_j-\llparenthesis x^\prime\rrparenthesis_p$. 
%Note that both $(\llparenthesis x\rrparenthesis_p,y)$ and $(\llparenthesis x^\prime\rrparenthesis_p,y)$ are in $\mathbb{X}_q^*$. 
Hence, 
$$\frac{\llparenthesis x\rrparenthesis_j-\llparenthesis x\rrparenthesis_p}{\beta_{m-i}^{m-1}}=\frac{\llparenthesis x^\prime\rrparenthesis_j-\llparenthesis x^\prime\rrparenthesis_p}{\beta_{m-i}^{m-1}}.
$$
Therefore, \eqref{eqn-boundary-strategy-1} holds because of \eqref{eqn-boundary-strategy}. 
\end{proof}

\textbf{Proof of Lemma \ref{HighOrder-is-LowOrder}}.
\begin{proof}
Let $a:=[x_1]_i$ and $b:=[x_2]_i$. By the assumption $[x_1]_i<[x_2]_i$, hence $a+1\leq b$. 
Note that, by the definition of the floor functions, $ x_1<(a+1)\beta_{m-i}^{m-1}\leq b\beta_{m-i}^{m-1}\leq x_2$. 

Therefore, 
\begin{equation*}
\begin{split}
\frac{x_1}{\beta_{m-i+1}^{m-1}}&<\frac{(a+1)\beta_{m-i}^{m-1}}{\beta_{m-i+1}^{m-1}}\\
&=(a+1)\beta_{m-i}^{m-i+1}\\
&\leq b\beta_{m-i}^{m-i+1}\\
\end{split}
\end{equation*}

Note that, $b\in\mathbf{N}^+$ and $\beta_{m-i}^{m-i+1}\in\mathbf{N}^+$. 

Therefore, 
\begin{equation*}
\begin{split}
[x_1]_{i-1}
&\leq\frac{x_1}{\beta_{m-i+1}^{m-1}}\\
&< \left\lfloor\frac{b\beta_{m-i}^{m-1}}{\beta_{m-i}^{m-1}}\beta_{m-i}^{m-i+1}\right\rfloor\\
&\leq \left\lfloor\frac{x_2}{\beta_{m-i}^{m-1}}\beta_{m-i}^{m-i+1}\right\rfloor\\
\end{split}
\end{equation*}

That is, 
\begin{equation*}
\begin{split}
[x_1]_{i-1}&< \left\lfloor\frac{x_2}{\beta_{m-i+1}^{m-1}} \right\rfloor\\
&=[x_2]_{i-1}.
\end{split}
\end{equation*}
\end{proof}

%\textbf{Proof of Lemma \ref{boardHistory-never-converge}}.

%===================================
\begin{comment}
\textbf{Proof of Lemma \ref{rngnum-equiv}}.
\begin{proof}
The reason is simple and intuitive. %Cf. Fact \ref{specialcase-fact-surrounding}. 
 If $l<j$, then $l<\mathrm{idx}(\llparenthesis x\rrparenthesis_i,y)$. It means that $\mathrm{RngNum}(x,l)=\mathrm{RngNum}(\llparenthesis x\rrparenthesis_i,l)=-1$. If $i=l$, then it clearly holds since, by definition, $[x]_l=[\llparenthesis x\rrparenthesis_l]_l$. 

Now suppose that $i<l$. By definition, $[\llparenthesis x\rrparenthesis_i]_l=\left\lfloor \frac{[x]_i\beta_{m-i}^{m-1}+\frac{1}{2}\sum_{1<\ell\leq i}\beta_{m-\ell}^{m-1}}{\beta_{m-l}^{m-1}}\right\rfloor=\left\lfloor \frac{[x]_i\beta_{m-i}^{m-1}}{\beta_{m-l}^{m-1}}\right\rfloor$. Because $i<l$, we have   $$\lfloor x/\beta_{m-l}^{m-1}\rfloor-1<\left\lfloor\frac{\left(\frac{x}{\beta_{m-i}^{m-1}}-1\right)\beta_{m-i}^{m-1}}{\beta_{m-l}^{m-1}} \right\rfloor\leq [\llparenthesis x\rrparenthesis_i]_l\leq \lfloor x/\beta_{m-l}^{m-1}\rfloor.$$
As a consequence, the claim holds since $[x]_l=[\llparenthesis x\rrparenthesis_i]_l$. 
\end{proof}

\end{comment}
%===========================

\textbf{Proof of Lemma \ref{universal-simulator}}.
\begin{proof}
By the modular arithmetic, we immediately have the following observation: for any $\mathfrak{a}\in [k-1]$, $(e^{\prime},f)$ can be such a vertex that $[e^{\prime}]_r+f\equiv \mathfrak{a}$ mod $k-1$. By the definition of $\mathbb{X}_i^*$ and Definition \ref{vertex-index}, there is at most one vertex in $(\llbracket e\rrbracket_{r},f)$ whose index is greater than $r$. And all the other vertices with index $r$ encodes all the vertices in $\mathbb{X}_{r+1}^*$ via ``$\restriction\!\! S$''.  
It explains why there must be such a vertex $(e^{\prime},f)$ that $\mathbf{cl}(e^{\prime},f)=\underline{(\underline{\mathfrak{a}},f);r;\ell;w}$, by Definition \ref{type-label} and Definition \ref{iterative-expansion}.   
Indeed, the vaule of $[e^{\prime}]_r$ mod $\eta_r$ determines the values of $\mathfrak{a}$, $r$ and $w$. By definition, $\ell=\left\lfloor \frac{[e^\prime]_{r}\mbox{ mod }\mathpzc{U}_{r}^*}{\frac{1}{3}\mathpzc{U}_{r}^*} \right\rfloor-1$. We can find a vertex $(e^{\prime\prime},f)$ where $|e^{\prime\prime}-e^\prime|\equiv 0$ (mod $\frac{1}{3}\mathpzc{U}_{r}^*$). Then it is clear that $\mathbf{cl}(e^{\prime\prime},f)$ is similar to $\mathbf{cl}(e^{\prime},f)$ except that they may have different value for $\mathrm{RngNum}(\cdot,\cdot)$. In other words, we can choose $(e^{\prime},f)$ properly such that $\ell$ can be any element in $\{-1,0,1\}$. 
Also note that there are many vertices satisfy the requirements other than $(e^{\prime},f)$.
That is, Lemma \ref{universal-simulator} can be ensured.
\end{proof}

%$\indent\vspace{50pt}$

%\textbf{Proof of Lemma \ref{B_k-has-no-k-clique}.}

\textbf{Poof of Lemma \ref{flexibility-in-same-abstraction-1}.}
\begin{proof}
%For any $x\in\mathbb{X}_1^*$, let $g(x):=g(x)$. 
We use a binary string $s\in \{0,1\}^{\binom{k-2}{2}}$ to encode $g(x)$ mod $2^{\binom{k-2}{2}}$. We use $(s)_{10}$ to denote the value encoded by $s$. On the other hand, recall that for a natural number $n$, we use $(n)_{2;\binom{k-2}{2}}$ to denote the binary representation of $n$, a $0\textnormal{-}1$ string of length $\binom{k-2}{2}$.   
%We use $s[i,j]$ to denote the substring of $s$ formed by the $i$-th bit to the $j$-th bit, and 
We use $s\!\downarrow\![i,j]$ to denote the string adjusted from $s$ by turning every bit to $0$ except for the $i$-th bit and the $j$-th bit, as well as the bits between them, which are unchanged. 

Because $v_i\neq v_j$, $\hat{q}(y,v_i)\neq \hat{q}(y,v_j)$. We can give an order $\lessdot$ to the element $(u_i,v_i)$ of $P$ based on $\hat{q}(y,v_i)$ such that $(u_i,v_i)\lessdot (u_j,v_j)$ if and only if $\hat{q}(y,v_i)<\hat{q}(y,v_j)$. Assume that $(u_1^{\lessdot},v_1^{\lessdot}),\ldots (u_l^{\lessdot},v_l^{\lessdot})$ are these elements of $P$ that are in the given order, i.e. $(u_i^{\lessdot},v_i^{\lessdot})\lessdot(u_j^{\lessdot},v_j^{\lessdot})$ if and only if  $i<j$. Note that this order is usually different from the linear orders of the structures. Let $f_\sigma(u_i^{\lessdot},v_i^{\lessdot})=j$ if $u_i^{\lessdot}=u_j$ and $v_i^{\lessdot}=v_j$. Let $P_j:=\{(u_i^{\lessdot},v_i^{\lessdot})\mid 1\leq i\leq j\}$. 
The main idea is that we can adjust the value of $x$ gradually to satisfy the lemma where $P=P_i$ for $i=1$ to $l$, step by step.   

For any $1\leq p\leq \binom{k-2}{2}$, let $trip(p)$ be a $0$-$1$ string of length $\binom{k-2}{2}$ such that all the elements in the string is $0$ except for the $p$-th element, called a ``trip point'', which is $1$.  Recall that the rightmost element of the string is the $0$-th element, i.e. the lowest order bit. 

 At the beginning,  we choose an $x$ such that $g(x)=g(u_1^{\lessdot})$.\footnote{There are many choices. Such freedom is necessary for us to apply the lemma. For the purpose of proving this lemma, we can simply let $x$ be the minimal one that makes $g(x)=g(u_1^{\lessdot})$ hold. In the following, we will talk about adjusting $g(x)$. Such adjusting certainly involves changing the value of $x$.}
 Then we adjust $g(x)$ such that   $g(x):=g(x)\!\!\downarrow\!\![0,\hat{q}(y,v_1^{\lessdot})]$. 
Afterwards, we adjust $g(x)$ if and only if  $$\mathbf{BIT}(\mathrm{SW}((x,y),(u_1^{\lessdot},v_1^{\lessdot})),\hat{q}(y,v_1^{\lessdot}))\neq w_{f_{\sigma}(u_1^{\lessdot},v_1^{\lessdot})}.$$ Assume that it is necessary to adjust $g(x)$, i.e. adjust $x$ when $w_{f_{\sigma}(u_1^{\lessdot},v_1^{\lessdot})}=1$.  
We adjust $x$ such that $g(x)$ is decreased by $(trip(\hat{q}(y,v_1^{\lessdot})))_{10}$ if $g(x)\!\geq\! (trip(\hat{q}(y,v_1^{\lessdot})))_{10}$; 
or $g(x)$ is increased by the same amount, otherwise.
 Now it is straightforward to verify that%\\[-10pt]
$$\mathbf{BIT}(\mathrm{SW}((x,y),(u_1^{\lessdot},v_1^{\lessdot})),\hat{q}(y,v_1^{\lessdot}))=
w_{f_{\sigma}(u_1^{\lessdot},v_1^{\lessdot})}.$$ 
Note that $g(x)<(trip(\hat{q}(y,v_1^{\lessdot})+1))_{10}$. 

Assume that for some $c\!\in\! [1,l-1]$ and for any $ 1\!\leq\! i\!\leq\! c$, 
\begin{equation}\label{SW-proof-eqn1}
\mathbf{BIT}(\mathrm{SW}((x,y),(u_i^{\lessdot},v_i^{\lessdot})),\hat{q}(y,v_i^{\lessdot}))=w_{f_{\sigma}(u_i^{\lessdot},v_i^{\lessdot})}. 
\end{equation}

Let $\mathrm{bit}_{c+1}^x:=\mathbf{BIT}(\mathrm{SW}((x,y),(u_{c+1}^{\lessdot},v_{c+1}^{\lessdot})),\hat{q}(y,v_{c+1}^{\lessdot}))$.

We adjust $x$  not only to preserve (\ref{SW-proof-eqn1})  but also to ensure that
\begin{equation}\label{SW-proof-eqn2}
\noindent \mathrm{bit}_{c+1}^x=w_{f_{\sigma}(u_{c+1}^{\lessdot},v_{c+1}^{\lessdot})}.
\end{equation} 
In other words, we adjust the vaule of $x$ to make one more vertex in $P$ satisfy (\ref{SW-proof-eqn1}), if necessary. 

If  
(\ref{SW-proof-eqn2}) holds,  
then $x$ is unchanged. Otherwise, we fine-tune $x$ for several rounds to satisfy (\ref{SW-proof-eqn2}), meanwhile still preserve (\ref{SW-proof-eqn1}).

First, assume that
$g(u_{c+1}^{\lessdot})\geq g(x)+(trip(\hat{q}(y,v_{c+1}^{\lessdot})))_{10}$. 

Let $\delta_{c+1}:=g(u_{c+1}^{\lessdot})-g(x)$,
and let $$\delta_{c+1}^{\pm}:=\delta_{c+1}+\left(\mathrm{bit}_{c+1}^x-w_{f_{\sigma}(u_{c+1}^{\lessdot},v_{c+1}^{\lessdot})}\right)\times(trip(\hat{q}(y,v_{c+1}^{\lessdot})))_{10}.
$$
Obviously, if $g(x)$ is increased by $\delta_{c+1}^{\pm}$, then (\ref{SW-proof-eqn2}) is ensured. 
Let $\Delta_{c+1}^{x\uparrow}\!:=\!\left((\delta_{c+1}^{\pm})_{2;\binom{k-2}{2}}\!\downarrow\!
[0,\hat{q}(y,v_{c+1}^{\lessdot})]\right)_{10}$. Let $g(x):=g(x)+\Delta_{c+1}^{x\uparrow}$. 
Note that (\ref{SW-proof-eqn2}) still holds because adding $(\delta_{c+1}^{\pm})_{2;\binom{k-2}{2}}\!\downarrow\![\hat{q}(y,v_{c+1}^{\lessdot})+1,\binom{k-2}{2}-1]$ to the sum doesn't influence the $0$-th bit to the $\hat{q}(y,v_{c+1}^{\lessdot})$-th bit of the sum.  

In the following we adjust $g(x)$ to satisfy (\ref{SW-proof-eqn1}) step by step. For $j=1$ to $c$, in the $j$-th round, 
if (\ref{SW-proof-eqn1}) still holds when $i=j$, then go to the next round. Otherwise, we add or minus  $(trip(\hat{q}(y,v_j^{\lessdot})))_{10}$ to $g(x)$ to enforce that (\ref{SW-proof-eqn1}) holds in case of $i=j$, depending on whether it will propagate to influence (\ref{SW-proof-eqn2}) or not. Note that one of the choices will not propagate to influence (\ref{SW-proof-eqn2}). Also note that, if $1<j$, this will not influence (\ref{SW-proof-eqn1}) for $1\leq i<j$. Therefore, (\ref{SW-proof-eqn1}) holds for $1\leq i\leq j$ at the end of the $j$-th round. 

In short, both (\ref{SW-proof-eqn1}) and (\ref{SW-proof-eqn2}) can be ensured when $g(u_{c+1}^{\lessdot})\geq g(x)+(trip(\hat{q}(y,v_{c+1}^{\lessdot})))_{10}$. This inductive argument shows that (\ref{SW-arbitrary}) holds in this case.  

Second, assume that
$$g(u_{c+1}^{\lessdot})< g(x)+(trip(\hat{q}(y,v_{c+1}^{\lessdot})))_{10}.$$ 

If (\ref{SW-proof-eqn2}) is already satisfied, then $x$ is unchanged. 
Otherwise, 
let $g(x)\!:=\!g(u_{c+1}^{\lessdot})+(trip(\hat{q}(y,v_{c+1}^{\lessdot})))_{10}$ if $g(u_{c+1}^{\lessdot})<\binom{k-2}{2}-(trip(\hat{q}(y,v_{c+1}^{\lessdot})))_{10}$ and let $g(x):=g(u_{c+1}^{\lessdot})-(trip(\hat{q}(y,v_{c+1}^{\lessdot})))_{10}$ otherwise. 
%So far we have only considered the case when $g(u_{c+1}^{\lessdot})>g(x)$. The situation is similar if  $g(u_{c+1}^{\lessdot})<g(x)$. At the beginning we decrease $g(x)$ by $\delta_{c+1}$, 
Then, similar to the last case, we adjust $g(x)$ gradually to satisfy  (\ref{SW-proof-eqn1}) step by step, 
minusing or adding $(trip(\hat{q}(y,v_{i}^{\lessdot})))_{10}$, depending on whether  (\ref{SW-proof-eqn2}) is still satisfied and $g(x)$ is still in the range $[0,\binom{k-2}{2}-1]$.

In summary, we can gradually adjust the value of $x$ to satisfy (\ref{SW-arbitrary}).
\end{proof}

\textbf{Poof of Lemma \ref{no-missing-edges_xi-1}.}
\begin{proof}
In the following arguments we choose such vertices $(x^\sharp,y)$, $(x^{\prime\prime},y)$ and  $(x,y)$  that  
\begin{equation}
\chi(x^\sharp,y)\!\!\restriction\!\! S=\chi(x^{\prime\prime},y)\!\!\restriction\!\! S=\chi(x,y)\!\!\restriction\!\! S=\emptyset.
\end{equation} 

 The claim \textit{(1)} is  obvious according to  Definition \ref{iterative-expansion}, since, by Lemma \ref{universal-simulator}, $(x^\sharp,y)$ can be such a vertex that 
\begin{itemize}
\item $\mathrm{idx}(x^\sharp,y)=t-1$;
\item $\mathbf{cc}([x^\sharp]_{t-1},y)=c$. By Lemma \ref{cm=depth}, $c$ is different from $\mathbf{cc}([x_i]_{t-1},y_i)=y_i$ mod $k-1$ for any $i$. Moreover, $[x^\sharp]_{t-1}$ mod $k-1\neq 0$ 
because $\mathbf{cc}([x^\sharp]_{t-1},y)=c\neq y$ mod $k-1$, which means that $\chi(x^\sharp,y)\!\!\restriction\!\! \Omega\cap H=\emptyset$. Note that $\chi(x^\sharp,y)\!\!\restriction S\cap H=\emptyset$ because $\chi(x^\sharp,y)\!\!\restriction\!\! S=\emptyset$.  And $(x^\sharp,y)\notin \chi(x_i,y_i)\!\!\restriction\!\! S$ since $\mathrm{idx}(x^\sharp,y)=t-1<\mathrm{idx}(x_i,y_i)$ for any $(x_i,y_i)\in H$.  

\end{itemize}

The claim \textit{(2)} is also obvious according to Definition \ref{iterative-expansion}:  we can choose  $(x^{\prime\prime},y)$ to be such a vertex that $\mathrm{idx}(x^{\prime\prime},y)=t-1$ and $[x^{\prime\prime}]_{t-1}\equiv 0$ (mod $k-1$). The latter means that  
$\mathbf{cc}([x^{\prime\prime}]_{t-1},y)=y$ mod $k-1$, which implies that $\mathbf{cc}([x^{\prime\prime}]_{t-1},y)\neq \mathbf{cc}([x_i]_{t-1},y_i)=y_i$ mod $k-1$ by Lemma \ref{cm=depth}. 
Because $(x^\prime,y)$ is adjacent to $(x_i,y_i)$ for any $(x_i,y_i)\in H$, $(x_i,y_i)\notin \chi(x^{\prime\prime},y)\!\!\restriction\!\! \Omega$. 
Moreover, we can choose $(x^{\prime\prime},y)$ to be such a vertex that $\chi(x^{\prime\prime},y)\!\!\restriction\!\! S=\emptyset$. Hence $(x^{\prime\prime},y)\notin \chi(x_i,y_i)\!\!\restriction\!\! S$ and $(x_i,y_i)\notin \chi(x^{\prime\prime},y)\!\!\restriction\!\! S$ for any $(x_i,y_i)\in H$. 

In the above arguments, we haven't considered $\mathrm{RngNum}$ yet. 
 Note that  $\mathrm{RngNum}(x_i,t-1)=-1$ for any $i$, by Lemma \ref{rngnum-is_-1}. Therefore, we can let $\mathrm{RngNum}(x^\sharp,t-1)=\mathrm{RngNum}(x^{\prime\prime},t-1)=-1$. 

To prove the claim \textit{(3)}, we first apply \textit{(1)} to find $(x^\prime,y)\in\mathbb{X}_{t-1}^*-\mathbb{X}_{t-2}^*$ such that $(x^\prime,y)$ is adjacent to every vertex in $H$; then we apply \textit{(2)} to find $(x,y)\in\mathbb{X}_{t-2}^*-\mathbb{X}_{t-3}^*$ such that $[x]_{t-2}\equiv 0$ (mod $k-1$), $g(x)=0$ and $\mathrm{RngNum}(x,t-2)=-1$. 

To prove the claim \textit{(4)}, we choose $(x,y)\in \mathbb{X}_{t-1}^*-\mathbb{X}_t^*$ to be such a vertex 
that $\mathbf{cc}([x]_{t-1},y)\\=y$ mod $k-1$, i.e. $[x]_{t-1}\equiv 0$ (mod $k-1$). By Lemma \ref{cm=depth}, 
\begin{equation}\label{app-no-missing-edges_xi-1-eqn1}
\mathbf{cc}([x]_{t-1},y)\neq \mathbf{cc}([x_i]_{t-1},y_i).
\end{equation}

Moreover, we can make $(x,y)$ be such a vertex that 
$(x_i,y_i)\notin\chi(x,y)\!\!\restriction\!\!\Omega$ for any  $(x_i,y_i)\in H$. Here is the process.  We first find a vertex $(x^{\dag},y)$ of index $t$ satisfying some conditions s.t. the vertex is adjacent to any vertex in $H$. Afterwards, we find such a vertex $(x,y)$ that 
\begin{enumerate}[(a)]
\item $\mathrm{idx}(x,y)=t-1$;

\item  $[x]_{t-1}\equiv 0$ (mod $k-1$), i.e. $\mathbf{cc}([x]_{t-1},y)=y$ mod $k-1$;

\item $[x^{\dag}]_t=[x]_t$. It means that $\llparenthesis x\rrparenthesis_t=x^{\dag}$, by Lemma \ref{projection};

\item $g(x)=0$;

\item $\mathrm{RngNum}(x,t-1)=-1$; 

\item $\chi(x, y)\!\!\restriction\!\! S=\emptyset$.

\end{enumerate}  

The crutial point is to find such a vertex $(x^{\dag},y)$. 

Let $H_\ell:=\{(\llparenthesis x_i\rrparenthesis_\ell,y_i)\mid (x_i,y_i)\in H\}$. Let $(x^{\dag},y)$ be  such a vertex that for any $t\leq j\leq m$ the following hold:
\begin{enumerate}[(i)]
\item $\mathrm{idx}(\llparenthesis x^{\dag}\rrparenthesis_j,y)=j$;

\item $\chi(\llparenthesis x^{\dag}\rrparenthesis_{j}, y)\!\!\restriction\!\! S=\emptyset$;

\item  $(\llparenthesis x^\dag \rrparenthesis_j, y)$ is adjacent to $(\llparenthesis x_i \rrparenthesis_j, y_i)$  for any $(x_i,y_i)\in H_j$;

\item $\mathrm{RngNum}(\llparenthesis x^{\dag}\rrparenthesis_j,j)=-1$.   

\end{enumerate}
Note that, (i) (ii) is easy to satisfy. So is (iv). 
 It is (iii) that needs some justification.

Firstly, we can ensure that 
\begin{equation}\label{eqn-adjacency-in-jth-abstracton}
\mathbf{cc}([x^{\dag}]_{j},y)\neq \mathbf{cc}([x_i]_{j},y_i)
\end{equation}
 because of $|H_j|\leq k-2$ and 
Lemma \ref{universal-simulator} (it says that we can choose a value for $\mathbf{cc}([x^{\dag}]_{j},y)$ freely). By Lemma \ref{proj-abs},  $\mathbf{cc}([x^{\dag}]_{j},y)=\mathbf{cc}([\llparenthesis x^\dag \rrparenthesis_j]_{j}, y)$ and $\mathbf{cc}([x_i]_{j},y_i)=\mathbf{cc}([\llparenthesis x_i\rrparenthesis_j]_{j},y_i)$. Therefore, 
\begin{equation}
\mathbf{cc}([\llparenthesis x^\dag \rrparenthesis_j]_{j}, y)\neq \mathbf{cc}([\llparenthesis x_i\rrparenthesis_j]_{j},y_i).
\end{equation} 
 Moreover, $\mathrm{idx}(\llparenthesis x_i\rrparenthesis_j,y_i)\geq j$, by Lemma \ref{proj-greater-index}. Consequently, by definition,  
\begin{equation}
\mathrm{sng}((\llparenthesis x^{\dag}\rrparenthesis_j,y),(\llparenthesis x_i\rrparenthesis_j,y_i))=0.
\end{equation}
 Therefore, 
 \begin{equation}
 (\llparenthesis x_i\rrparenthesis_j,y_i)\notin \chi(\llparenthesis x^{\dag}\rrparenthesis_j,y)\!\!\restriction\!\!\Omega.
 \end{equation}

Secondly, by Lemma \ref{cm=depth}, $y\in\{1,\cdots\!,k-2\}-\{\mathbf{cc}([x_i]_{t-1},y_i)\1 (x_i,y_i)\in H\}$ since $y\in\{1,\cdots\!,k-2\}-\{y_i\mid (x_i,y_i)\in H\}$.  By Lemma \ref{flexibility-in-same-abstraction-1}, we can choose $(x^{\dag},y)$ to be such a vertex that, for any $(u,v)\in H_j\cap (\mathbb{X}_j^*-\mathbb{X}_{j+1}^*)$ where $v\in [1,k-2]$,   
 \begin{multline}\label{app-no-missing-edges_xi-1-eqn2}
(\mathbf{cc}([u]_{j},v)-\mathbf{cc}([x^\dag]_{j},y))\times(v-y)\times\\ \hspace{15pt} (-1)^{\mathbf{BIT}(\mathrm{SW}((u,v),(\llparenthesis x^\dag \rrparenthesis_j, y)),\hat{q}(v,y))}\!>\!0.  
\end{multline}

Third, we have chosen $(x^{\dag},y)$ to be such a vertex that $\chi(\llparenthesis x^{\dag}\rrparenthesis_j,y)\!\!\restriction\!\! S=\emptyset$. Hence 
\begin{equation}
\mathbf{cl}(\llparenthesis x_i\rrparenthesis_j,y_i)\notin \chi(\llparenthesis x^{\dag}\rrparenthesis_j,y)\!\!\restriction\!\! S. 
\end{equation}

Note that  the cases can be satisfied simultaneously. Then by definition (iii) holds.  

Note that $\llparenthesis x^{\dag}\rrparenthesis_t=x^{\dag}$. 
Due to (iii), $(x_i,y_i)\notin\chi(x^{\dag},y)\!\!\restriction\!\!\Omega$ because $(\llparenthesis x^\dag \rrparenthesis_j, y)$ is adjacent to $(\llparenthesis x_i \rrparenthesis_j, y_i)$  for any $j$. That is, $\chi(x,y)\!\!\restriction\!\!\Omega\cap H=\emptyset$.

In summary, $(x^\dag,y)$ is adjacent to any vertex in $H$. 
Then by the proof of \textit{(2)}, the claim \textit{(4)} holds.  
\end{proof}

\textbf{Proof of Claim \ref{board-history-evolutions}}.
\begin{proof}
Firstly, we show that \textit{(i)} holds. 
Recall that Spoiler picks $(x,y)$ and Duplicator replies with $(x^\prime,y)$. 
If $\chi(x,y)\!\!\restriction\!\!\mathrm{BH}$ is not a valid board history, then for any $(u,v)$,  neither $(x,y)\xrightarrow[\mathrm{BC}]{*}(u,v)$ nor $(u,v)\xrightarrow[\mathrm{BC}]{*}(x,y)$. In other words, $(x,y)$ is an isolated vertex in $\mathfrak{A}_{k,m}$. 
Duplicator simply picks $(x^\prime,y)$ such that $\chi(x^\prime,y)\!\!\restriction\!\!\mathrm{BH}$ is not a valid board history. Clearly the claim holds. Futhermore, since there are sufficiently many such isolated vertices, she can do it in accordance with the usual condition for winning the Ehrenfeucht-Fra\"iss\' e\xspace games over pure linear orders, i.e. (\mbox{apx-}1) (cf. Remark \ref{remark-linear-orders-1}). Therefore, in the following argument we assume that both $\chi(x,y)\!\!\restriction\!\!\mathrm{BH}$ and $\chi(x^\prime,y)\!\!\restriction\!\!\mathrm{BH}$ are valid.   
 
We first prove the following by induction. That is, $(u^{\prime},v)\twoheadrightarrow(x^{\prime},y)$ if   $(u,v)\twoheadrightarrow(x,y)$, and  $(x^{\prime},y)\twoheadrightarrow(u^{\prime},v)$ if   $(x,y)\twoheadrightarrow(u,v)$. Note that if the players do not take off pebbles in a virtual game then $\twoheadrightarrow$ equals $\xrightarrow[\mathrm{BC}]{*}$. 
Recall that in the virtual games, for any $(a,b)\!\Vdash\! (a^\prime,b)$, if Spoiler takes off one of the pebble, e.g. $(a,b)$, then Duplicator will take off the other one in the pair, e.g. $(a^\prime,b)$. Consequently, We need only prove that $(u^{\prime},v)\xrightarrow[\mathrm{BC}]{*}(x^{\prime},y)$ if $(u,v)\xrightarrow[\mathrm{BC}]{*}(x,y)$, and  $(x^{\prime},y)\xrightarrow[\mathrm{BC}]{*}(u^{\prime},v)$ if $(x,y)\xrightarrow[\mathrm{BC}]{*}(u,v)$. 

\textbf{Basis:} Assume that there is only one pair of pebbled vertices $(u,v)$, $(u^{\prime},v)$ in the game board, and Spoiler picks $(x,y)$ in $\widetilde{\mathfrak{A}}_{k,m}$.  
Since Duplicator sticks to B-3, by the definitions, $(u^{\prime},v)\!\xrightarrow[\mathrm{BC}]{*}\!(x^{\prime},y)$ if   $(u,v)\!\xrightarrow[\mathrm{BC}]{*}\!(x,y)$, and  $(x^{\prime},y)\!\xrightarrow[\mathrm{BC}]{*}\!(u^{\prime},v)$ if   $(x,y)\!\xrightarrow[\mathrm{BC}]{*}\!(u,v)$.  

\textbf{Induction Step:}\label{app-history-order-induction}
There are up to $k-2$ pairs of pebbled vertices $(u_{i},v_{i})$, $(u_{i}^\prime,v_{i})$ in the game board, where the claim holds. 
Let $\ell:=max\{j\mid (u_i,v_i)_{[j]}\xrightarrow[\mathrm{BC}]{*}(x,y)\}$. 
Let $(u_c,v_c)$ be a vertex, not necessary pebbled,  such that $\mathrm{i}_{\mathrm{cur}}^{u_c,v_c}=\ell$ and $(u_c,v_c)\xrightarrow[\mathrm{BC}]{*}(x,y)$. 
And let $(u^\star,v^\star)$ be one of these pebbled vertices, if there is one, that $\mathrm{i}_{\mathrm{cur}}^{u^\star\!\!,v^\star}=min\{\mathrm{i}_{\mathrm{cur}}^{u_i\!,v_i}\mid (x,y)\xrightarrow[\mathrm{BC}]{*}(u_i,v_i)\}$. Since the binary relation $\xrightarrow[\mathrm{BC}]{*}$ is transitive, $(u_c,v_c)\xrightarrow[\mathrm{BC}]{*}(u^\star,v^\star)$. 
Let $(u^\star,v^\star)\Vdash\! (u^{\star\prime},v^\star)$. Moreover, let $(u_c^\prime,v_c)$ be such a vertex that  $(u_c^\prime,v_c)\xrightarrow[\mathrm{BC}]{*} (u^{\star\prime},v^\star)$ and $\mathrm{i}_{\mathrm{cur}}^{u_c^\prime,v_c}=\mathrm{i}_{\mathrm{cur}}^{u_c,v_c}$.
%By induction hypothesis, $(u_c^\prime,v_c)\xrightarrow[\mathrm{BC}]{con.}(u^{\star\prime},v^\star)$. 
Duplicator simply picks $(x^{\prime},y)$ such that its associated board history mimics that of $(u^{\star\prime},v^\star)$ in the first $\chi(x,y)\!\!\restriction\!\!\mathrm{bc}$ rounds. By definition, $(u_c^\prime,v_c)\xrightarrow[\mathrm{BC}]{*} (x^\prime,y)\xrightarrow[\mathrm{BC}]{*} (u^{\star\prime},v^\star)$. Because the relation $\xrightarrow[\mathrm{BC}]{*}$ is transitive, the claim holds for other pebbled vertices, e.g. $(u^\prime,v)$.\footnote{For example, assume that $(u,v)\Vdash (u^\prime,v)$ and $(u,v)\xrightarrow[\mathrm{BC}]{*} (x,y)$, and that $\chi(u,v)\!\!\restriction\!\!\mathrm{bc}<\chi(u_c,v_c)\!\!\restriction\!\!\mathrm{bc}$. By Fact \ref{boardHistory-never-converge-1},  $(u,v)\xrightarrow[\mathrm{BC}]{*} (u_c,v_c)$ holds. Then  $(u^\prime,v)\xrightarrow[\mathrm{BC}]{*} (u_c^\prime,v_c)$. Hence by the transitivity of $\xrightarrow[\mathrm{BC}]{*}$, we know that $(u^\prime,v)\xrightarrow[\mathrm{BC}]{*} (x^\prime,y)$. In short, $(u^\prime,v)\xrightarrow[\mathrm{BC}]{*} (x^\prime,y)$ if $(u,v)\xrightarrow[\mathrm{BC}]{*} (x,y)$. 
The case when $(u,v)\Vdash (u^\prime,v)$ and $(x,y)\xrightarrow[\mathrm{BC}]{*} (u,v)$ is similar.} 
If there is no such vertex $(u^\star,v^\star)$, Duplicator first search for the pebbled vertices in $\widetilde{\mathfrak{A}}_{k,m}$ (it is similar when Spoiler picks in $\widetilde{\mathfrak{B}}_{k,m}$ in this round) and if she can find a pebbled vertex $(a,b)$ such that $\chi(a,b)\!\!\restriction\!\!\mathrm{IBH}[i]$ is equivalent to $\chi(u_c,v_c)\!\!\restriction\!\!\mathrm{IBH}[i]$ for $i\leq \chi(u_c,v_c)\!\!\restriction\!\!\mathrm{bc}$, then Duplicator picks $(x^\prime,y)$ such that $\chi(x^\prime,y)\!\!\restriction\!\!\mathrm{IBH}[i]$ equals $\chi(a^\prime,b)\!\!\restriction\!\!\mathrm{IBH}[i]$ where $(a,b)\Vdash (a^\prime,b)$; afterwards, she uses the virtual game to determine $\chi(x^\prime,y)\!\!\restriction\!\!\mathrm{BH}(j)$ for $\mathrm{i}_{\mathrm{cur}}^{u_c,v_c}\leq j< \mathrm{i}_{\mathrm{cur}}^{x\!,y}$. 
In this case, let $(u_c^{\prime\flat},v_c)$ be the vertex that Duplicator will pick in the $\mathrm{i}_{\mathrm{cur}}^{u_c,v_c}$-th round of this virtual game. Thus we can get $(u_c^{\prime},v_c)$. 
By definition,  $(u_c^{\prime},v_c)\xrightarrow[\mathrm{BC}]{*}(x^{\prime},y)$ if   $(u_c,v_c)\xrightarrow[\mathrm{BC}]{*}(x,y)$. Then by the transitivity of $\xrightarrow[\mathrm{BC}]{*}$, the claim holds for pebbled vertices.

Now we prove that $(u^{\prime},v)$ and $(x^{\prime},y)$ are not in continuity if $(u,v)$ and $(x,y)$ are not in continuity. 
Note that, to avoid verbos case analysis, we discuss ``$\xrightarrow[\mathrm{BC}]{con.}$ (continuity)''  in the next paragraph instead of ``$\twoheadrightarrow$''. But it is easy to see that similar argument holds in the latter case.

Suppose that $(u,v)$ and $(x,y)$ are not in continuity. 
If $\chi(u,v)\!\!\restriction\!\!\mathrm{bc}=\chi(x,y)\!\!\restriction\!\!\mathrm{bc}$, then $\chi(u^\prime,v)\!\!\restriction\!\!\mathrm{bc}=\chi(x^\prime,y)\!\!\restriction\!\!\mathrm{bc}$. Hence $(u^{\prime},v)$ and $(x^{\prime},y)$ are not in continuity. In the following we assume that  $\chi(u,v)\!\!\restriction\!\!\mathrm{bc}\neq\chi(x,y)\!\!\restriction\!\!\mathrm{bc}$. There are two cases. 
First, suppose that, for some $0<i<min\{\mathrm{i}_{\mathrm{cur}}^{x\!,y},\mathrm{i}_{\mathrm{cur}}^{u\!,v}\}-1$,  $\chi(x,y)\!\!\restriction\!\!\mathrm{IBH}[j]=\chi(u,v)\!\!\restriction\!\!\mathrm{IBH}[j]$ if $0\leq j\leq i$, and $\chi(x,y)\!\!\restriction\!\!\mathrm{BH}(i+1)\neq\chi(u,v)\!\!\restriction\!\!\mathrm{BH}(i+1)$. 
%Assume that the type label of the vertex Spoiler ``picks'' in the $(i+1)$-th virtual round is $\mathrm{tp}_{xy}$, and Duplicator replies with a vertex of type label $\mathrm{tp}_{x^\prime y}$.  
Duplicator can make it that  $\chi(x^\prime,y)\!\!\restriction\!\!\mathrm{IBH}[j]=\chi(u^\prime,v)\!\!\restriction\!\!\mathrm{IBH}[j]$ for $0\leq j\leq i$. Moreover, Duplicator can ensure that $\chi(x^\prime,y)\!\!\restriction\!\!\mathrm{BH}(i+1)\neq\chi(u^\prime,v)\!\!\restriction\!\!\mathrm{BH}(i+1)$. She need only ensure that in the games with the same board configuration she reples with different vertices in a structure if Spoiler picks different vertices in the other structure. This would be made more precise in the following argument for \textit{(ii)} wherein games over pure linear orders are introduced. 
The second case is very similar. Assume w.l.o.g. that  $\chi(u,v)\!\!\restriction\!\!\mathrm{bc}<\chi(x,y)\!\!\restriction\!\!\mathrm{bc}$ and  $\chi(x,y)\!\!\restriction\!\!\mathrm{BH}(j)=\chi(u,v)\!\!\restriction\!\!\mathrm{BH}(j)$ for $0\leq j\leq \chi(u,v)\!\!\restriction\!\!\mathrm{bc}$. Then in the $\mathrm{i}_{\mathrm{cur}}^{u\!,v}$-th virtual round of the board history of $(x,y)$, a vertex that is different from $(u,v)$ is picked. Similarly, Duplicator need only ensure that in the $\mathrm{i}_{\mathrm{cur}}^{u\!,v}$-th virtual round of the board history of $(x^\prime,y)$, a vertex that is different from $(u^\prime,v)$ is picked. How to pick such a vertex will be introduced shortly. 
[Q.E.D]

A concern may probably rise. That is, in the above argument we don't know if this will cause problem w.r.t. the linear order. We shall see why it is not a problem in the following argument. 

Recall that $\overline{c_A}$ is the tuple of pebbled vertices in $\widetilde{\mathfrak{A}}_{k,m}$ and $\overline{c_A}(i)$ is the $i$-th element in the tuple. We can construct a tree of board configurations where the root is the empty board configuration and the leaves are the board configurations associated with vertices in $\overline{c_A}$; each branch of the tree corresponds to a board history associated with $\overline{c_A}(i)$ for some $i$.\footnote{For some other vertices in $\overline{c_A}$, their board histories maybe correspond to a path from the root to some inner node. By Fact \ref{boardHistory-never-converge-1}, such a tree is possible.} 
Hence this tree encodes the board histories of vertices in $\overline{c_A}$. Each arrow represents one step evolution of some board history. 
Such a representation tells us where two board histories diverge in the tree. Clearly, the number of children of a node is bounded by $k-2$ since there are up to $k-2$ leaves (recall that $|\overline{c_A}|\leq k-2$). Some of the children correspond to a board configuration that is obtained by adding a pebbled vertex to its farther board configuration. Some of them correspond to a board configuration that is obtained by removing a pebble from its farther (a board configuration). By the order defined in Definition \ref{iterative-expansion}, the children of the latter case are less than the children of the former case in the order. %To simply the argument, we may assume that all the children belong to the former case. 
We can also construct a similar tree that represents the board histories of vertices in $\overline{c_B}$. Then by \textit{(i)} and, in particular, by the strategy introduced in page \pageref{app-history-order-induction} (the induction step), we can observe that these two trees are isomorphic, without considering the order. In the following we give the justification for \textit{(ii)} in detail. That is, we show that the trees are isomorphic even in the presence of the order. 
As usual, we study it case by case. Recall that $v=y$.    

If $\chi(u,v)\!\!\restriction\!\! \mathrm{bc}<\chi(x,y)\!\!\restriction\!\! \mathrm{bc}$, then by the definition of the order of histories, $\lfloor u/(\gamma_{m-1}^*\times k)\rfloor$ mod $bh^\# <\lfloor x/(\gamma_{m-1}^*\times k)\rfloor$ mod $bh^\#$. Likewise, $\lfloor u^\prime/(\gamma_{m-1}^*\times k)\rfloor$ mod $bh^\# <\lfloor x^\prime/(\gamma_{m-1}^*\times k)\rfloor$ mod $bh^\#$ since $\chi(u^\prime,v)\!\!\restriction\!\! \mathrm{bc}<\chi(x^\prime,y)\!\!\restriction\!\! \mathrm{bc}$ (recall that $\chi(u^\prime,v)\!\!\restriction\!\! \mathrm{bc}=\chi(u,v)\!\!\restriction\!\! \mathrm{bc}$ and $\chi(x^\prime,y)\!\!\restriction\!\! \mathrm{bc}=\chi(x,y)\!\!\restriction\!\! \mathrm{bc}$). Hence the claim holds. Now 
assume that $\chi(u,v)\!\!\restriction\!\! \mathrm{bc}=\chi(x,y)\!\!\restriction\!\! \mathrm{bc}=\ell_u$.

Assume that $\lfloor x/(\gamma_{m-1}^*\times k)\rfloor$ mod $bh^\#$ $=\lfloor u/(\gamma_{m-1}^*\times k)\rfloor$ mod $bh^\#$. That is, $\chi(u,v)\!\!\restriction\!\! \mathrm{BH}=\chi(x,y)\!\!\restriction\!\! \mathrm{BH}$. 
In such case, Duplicator can ensure that $\lfloor x^\prime/(\gamma_{m-1}^*\times k)\rfloor$ mod $bh^\#$ $=\lfloor u^\prime/(\gamma_{m-1}^*\times k)\rfloor$ mod $bh^\#$, by the strategy that Duplicator used (cf. page \pageref{app-history-order-induction}, the induction step).

Suppose that $\lfloor x/(\gamma_{m-1}^*\times k)\rfloor$ mod $bh^\#$ $>\lfloor u/(\gamma_{m-1}^*\times k)\rfloor$ mod $bh^\#$. Another case where $\lfloor x/(\gamma_{m-1}^*\times k)\rfloor$ mod $bh^\#$ $<\lfloor u/(\gamma_{m-1}^*\times k)\rfloor$ mod $bh^\#$ is very similar.  
 Recall that $\chi(x,y)\!\!\restriction\!\!\mathrm{BH}$ and $\chi(u,v)\!\!\restriction\!\!\mathrm{BH}$ have the same ancestor, i.e. the empty board history (the root of the tree). These two histories must be diverged at some point, i.e. some node in the tree. Suppose that this node is $\mathrm{BC}_0$ (a board configuration) and that $(a,b)$ and $(e,b)$ are two vertices in $\mathrm{BC}_0$. Moreover, assume that in the initial segment of the board history from the beginning empty board configuration to $\mathrm{BC}_0$, the game is evolved into the $\xi$-th abstraction. Then, over the $\xi$-th abstraction, the length of the interval between $(a,b)$ and $(e,b)$ is $|\lfloor a/l_\xi\rfloor-\lfloor e/l_\xi\rfloor|$. 
By \eqref{xi-order-requirement-ensured} and Remark \ref{abstract-order-in-main-lemma}, we know that any interval over the $\xi$-th abstraction, whose length is greater than 1, is sufficiently large over the $(\xi-1)$-th abstraction. Hence it is easy for Duplicator to win the Ehrenfeucht-Fra\"iss\' e\xspace games over such pure linear orders (the intervals over the $(\xi-1)$-th abstraction) in up to $k-2$ rounds, because $m>k$ and the length of the orders are greater than $2^m$. 
The exceptions, where the length of an interval over the $\xi$-th abstraction is 0 or 1, can be handled easily. 
Note that in such cases the pair of intervals in respective structures have the same length, i.e. either 0 or 1. 
By her strategy that deals with order, Duplicator should mimic Spoiler's picking in such intervals. So far, the games are played over the $\xi$-th or the $(\xi-1)$-th abstraction. However, the vertices picked can be in very low abstractions, e.g. in $\mathbb{X}_1^*-\mathbb{X}_2^*$. If, in the $\xi$-th abstraction, a vertex  is picked again, then the corresponding vertex will also be picked again in this abstraction. Or more precisely, their projection in the $\xi$-th abstraction are picked more than once. It remains to show that $(\varkappa)$ will not occur in such case, or more precisely, (\ref{main-diamond-xi}$^\diamond$) (vii) and (\ref{main-diamond-xi}$^\diamond$) (iv) hold simultaneously (when $\mathrm{idx}(x^\flat,y)=\mathrm{idx}(x^{\prime\flat},y)<\xi$). But this is already explained in the proof of Lemma \ref{main-lemma} (cf. Strategy \ref{t<xi}).

Fig. \ref{pic-board-history-isom} illustrates two isomorphic trees that encode the board histories of the vertices in $\overline{c_A}$ and $\overline{c_B}$. The tree on the left is used to represent board histories associated with the vertices in $\overline{c_A}$, whereas  the tree on the right is for those board histories of the vertices in $\overline{c_B}$. Here the board configurations are succinctly represented. The node Root, as well as Root', is the empty board configuration. There is an arrow from Root to the board configuration $(c0,d0)$. It means that the latter can be evolved from the former in one step in the board history that goes through them. Assume that $v=y=b$.  We use $\mathrm{BC}_1$ to denote $(c0,d0)(c1,d1)(a,b)$ , which is associated with a vertex in $\overline{c_A}$. likewise, $\mathrm{BC}_5$ is for $(c0,d0)(c1,d1)(e,b)$ (associated with  a vertex in $\overline{c_A}$). 
 Similarly, $\mathrm{BC}_1^\prime$ stand for the board configuration $(c0^\prime,d0)(c1^\prime,d1)(a^\prime,b)$ and $\mathrm{BC}_5^\prime$ is for $(c0^\prime,d0)(c1^\prime,d1)(e^\prime,b)$. They are associated with a vertex in $\overline{c_B}$. We require that $(c0,d0)\Vdash\!(c0^\prime,d0),\ldots, (e,b)\Vdash\!(e^\prime,b)$. Here ``$\Vdash$'' is determined by a virtual game that combines two sorts of virtual games. The first sort is the imaginary pebble game played on the changing game board (cf. page \pageref{def-virtual-game}). The second sort is the Ehrenfeucht-Fra\"iss\' e\xspace games over the pure linear orders, e.g. the interval between $(a,b)$ and $(e,b)$ over the $\xi$-th abstraction or the $(\xi-1)$-th abstraction. 
 Note that $\mathrm{BC}_i$ is surrounded by a rounded rectangle. We use such a rectangle to denote an $l_\xi$-tuple of successive board configurations that include $\mathrm{BC}_i$ as a member, whereas no other board configuration in it is associated with a vertex in $\overline{c_A}$.   
 Note that the board configurations in an $l_\xi$-tuple can only be distinguished by the last vertex in their representation. 
 In this example, we assume that the third item of $\mathrm{BC}_i$ ($\mathrm{BC}_i^\prime$ resp.)  is $(h_i,b)$ ($(h_i^\prime,b)$ resp.), where $a=h_1<h_2=h_6<h_3=u<h_4<h_5=e$ ($a^\prime=h_1^\prime<h_2^\prime=h_6^\prime<h_3^\prime=u^\prime<h_4^\prime<h_5^\prime=e^\prime$, resp.). Here we use the dashed arrow issued from $\mathrm{BC}_6$ to mean that $\mathrm{BC}_6$ have some children.  
From the figure, we know that, w.r.t. the linear order, the board history indicated by the path from Root to $\mathrm{BC}_i$ is less than that indicated by the path from Root to $\mathrm{BC}_j$ if $i<j$. Moreover, all the board histories that go through $\mathrm{BC}_i$ share the initial segment of them (the first two virtual rounds, indicated by the blue arrows). In this example, if Spoiler picks $(x,y)$ in $\widetilde{\mathfrak{A}}_{k,m}$ s.t. $\chi(x,y)\!\!\restriction\!\! \mathrm{BH}$ goes through $\mathrm{BC}_5$, then Duplicator can pick $(x^\prime,y)$ in $\widetilde{\mathfrak{B}}_{k,m}$ s.t. $\chi(x^\prime,y)\!\!\restriction\!\! \mathrm{BH}$ goes through $\mathrm{BC}_5^\prime$.  

\begin{figure}[]
\hspace*{-12mm}
%\centering
\includegraphics[scale=0.59]{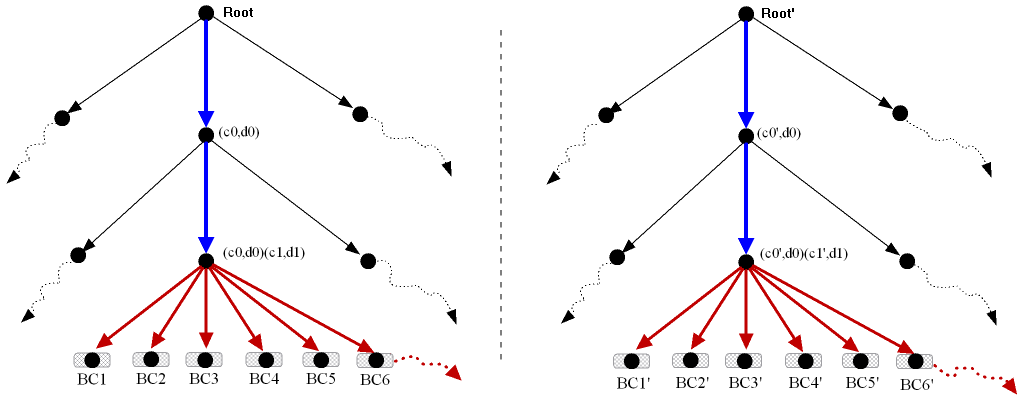}
%\scalebox{10}{}
\caption{Trees that encode the board histories of the vertices in $\overline{c_A}$ and $\overline{c_B}$. Each node is a board configuration.}
\label{pic-board-history-isom}
\end{figure}
 
 We've shown that Duplicator is able to ensure that the orders of an initial segment of histroies (up to $\ell_u$ rounds; associated with $\overline{c_A}$ and $\overline{c_B}$) are isomorphic at the point of divergence, by \eqref{xi-order-requirement-ensured} and Remark \ref{abstract-order-in-main-lemma}. More precisely, for some $i\leq\ell_u$, $\chi(x,y)\!\!\restriction\!\!\mathrm{IBH}[i]>\chi(u,v)\!\!\restriction\!\!\mathrm{IBH}[i]$ if $\chi(x^\prime,y)\!\!\restriction\!\!\mathrm{IBH}[i]>\chi(u^\prime,v)\!\!\restriction\!\!\mathrm{IBH}[i]$, on condition that $\chi(x^\prime,y)\!\!\restriction\!\!\mathrm{BH}(j)=\chi(u^\prime,v)\!\!\restriction\!\!\mathrm{BH}(j)$ for $0\leq j<i$ and $\chi(x^\prime,y)\!\!\restriction\!\!\mathrm{BH}(i)\neq\chi(u^\prime,v)\!\!\restriction\!\!\mathrm{BH}(i)$. 
  Recall that we order the board histories based on lexicographic ordering. Provided that the above holds, it implies that the orders of these board histories (of the same length; associated with the pebbled vertices) are preserved throughout the game, i.e. $\lfloor x^\prime/(\gamma_{m-1}^*\times k)\rfloor$ mod $bh^\#$ $>\lfloor u^\prime/(\gamma_{m-1}^*\times k)\rfloor$ mod $bh^\#$ if $\lfloor x/(\gamma_{m-1}^*\times k)\rfloor$ mod $bh^\#$ $>\lfloor u/(\gamma_{m-1}^*\times k)\rfloor$ mod $bh^\#$, since $\chi(u,v)\!\!\restriction\!\! \mathrm{bc}=\chi(x,y)\!\!\restriction\!\! \mathrm{bc}$.

Recall that Duplicator uses Strategy \ref{play-in-xi-abs}\textapprox Strategy \ref{t<xi} to choose the type label for $(x^\prime,y)$; the order issue, when ignoring board histories, is handled in Remark \ref{abstract-order-in-main-lemma}. Here we just introduced the idea that takes care of the orders of board histories (associated with the pebbled vertices) and that  make up these strategies: it ensures that \textit{(ii)} holds, in accordance with \textit{(i)}. 
\end{proof}

\begin{remark}\label{special-locally-isom}
If the readers have already been confirmed by the intuition stated in Remark \ref{ExplanationOfAbstraction-specalcase}, then there is no need to read the following involved arguments. But, if the readers are still doubt about the claim that ``the subgraph induced by $cex(x_0,y_0,t-1)$ and $cex(x_1,y_1,t-1)$ is isomorphic to the subgraph induced by  $cex(x_0^\prime,y_0,t-1)$ and $cex(x_1^\prime,y_1,t-1)$ if and only if the adjacency between $(x_0,y_0)$ and $(x_1,y_1)$ is the same as that between  $(x_0^\prime,y_0)$ and $(x_1^\prime,y_1)$'', then patience should be paid to getting through the following proofs, by which we show a more general result: it is true even for the structures $\mathfrak{A}_{k,m}^*$ and $\mathfrak{B}_{k,m}^*$ defined in section \ref{section-structures}. It justifies the notion ``abstraction'' in the more general cases ($k\geq 4$), just akin to the intuition explained in Remark \ref{ExplanationOfAbstraction-specalcase}.

Suppose that $\mathrm{idx}(x,y)=r$. 
$\llbracket x\rrbracket_{r}^{min}\!-\![\llparenthesis x\rrparenthesis_\xi]_{r}\!\equiv\llbracket x^{\prime}\rrbracket_{r}^{min}\!-
\![\llparenthesis x^{\prime}\rrparenthesis_\xi]_{r}\!$ 
roughly corresponds to the equation $x^\flat-\llparenthesis x^\flat\rrparenthesis_\xi=x^{\prime\flat}-\llparenthesis x^{\prime\flat}\rrparenthesis_\xi$. On the other hand, note that $x^\flat-\llparenthesis  x^\flat\rrparenthesis_\xi=\llbracket x^\flat\rrbracket_{1}^{min}-[\llparenthesis x^\flat\rrparenthesis_\xi]_{1}$ if $\mathrm{idx}(x^\flat,y)=1$.     
Note that $\llbracket x\rrbracket_{r}^{min}=[x]_r$ if $k=3$ since in this case $\mathpzc{U}_i^*=1$ for any $i$. %Hence, $\llbracket x\rrbracket_{r}^{min}-[\llparenthesis x\rrparenthesis_\xi]_{r}=x-\llparenthesis x\rrparenthesis_\xi$ if $r=1$ and $k=3$.
 Therefore, $x-\llparenthesis x\rrparenthesis_\xi=x^\prime-\llparenthesis x^\prime\rrparenthesis_\xi$ implies that 
$\llbracket x\rrbracket_{r}^{min}-[\llparenthesis x\rrparenthesis_\xi]_{r}=\llbracket x^{\prime}\rrbracket_{r}^{min}-[\llparenthesis x^{\prime}\rrparenthesis_\xi]_{r}$.   

The following Lemma roughly says that, modulo some amount, if two vertices are (approximately) in the same position in lower abstraction (i.e. finer scale), then so are they in higher abstraction (i.e. coarser scale). See Fig. \ref{app-copycat-lemmas-1}. Here in the figure vertex $a$, $b$, $a^\prime$ and $b^\prime$ stand for $(x,y)$, $(\llbracket x\rrbracket_r^{min},y)$,  $(x^\prime,y)$ and $(\llbracket x^\prime\rrbracket_r^{min},y)$ respectively. Vertices of different colour stand for vertices of different indices: black vertices have index $\xi$; the grey ones have index $r$; while the blue ones have index strictly between $\xi$ and $r$. Vertices under one black brace stand for a $\mathpzc{U}_r^*$-tuple. The number of vertices above a red brace stand for $\llbracket x\rrbracket_{r}^{min}-[\llparenthesis x\rrparenthesis_\xi]_{r}$. 
%\begin{comment}
\begin{figure}
\centering
%\hspace*{-10mm}
\begin{tikzpicture}[scale=0.35]
\tikzstyle{myedgestyle} = [decorate, decoration={brace, amplitude=5pt, raise=1mm}];
\tikzstyle{myedgestyle-mirror} = [decorate, decoration={brace, amplitude=10pt, raise=1mm,mirror}];

 \draw [black,xstep=1,ystep=4, dotted] (0,0) node (v1) {} grid (35,8);
 \foreach \i in {0,1,...,35} {
     \draw [fill=black!45] (\i,4) circle(2pt); 
  }
   
 \foreach \i in {0,1,...,8}{
     \draw [myedgestyle] (\i*4,4.1) --(\i*4+3,4.1);
}
 
\draw [ultra thick, dotted]  plot[smooth, tension=.7] coordinates {(0,8) (v1)};
\draw [blue, thick,dotted] plot[smooth, tension=.7] coordinates {(16,8) (16,0)};
\draw [blue, thick,dotted] plot[smooth, tension=.7] coordinates {(32,8) (32,0)};
\draw [fill=black] (0,0) circle (4pt); \draw [fill=black] (0,4) circle (4pt); \draw [fill=black] (0,8) circle (4pt); 
\draw [fill=blue] (16,0) circle (3pt); \draw [fill=blue] (16,4) circle (3pt); \draw [fill=blue] (16,8) circle (3pt); 
\draw [fill=blue] (32,0) circle (3pt); \draw [fill=blue] (32,4) circle (3pt); \draw [fill=blue] (32,8) circle (3pt); 

\draw [red] [myedgestyle-mirror] (0,4) --(32,4);
 %\draw [fill=red] (33,4) circle(2.5pt); 
\node at (33,3.4) {$a^\prime$};
\node at (32,3.4) {$b^\prime$};

%+++++++++++++++++++++++++++++++++++++++++++++++++++++++++++++++

\draw [black,xstep=1,ystep=4, dotted, shift={(0,10)}] (0,0) node (v2) {} grid (35,8);
 \foreach \i in {0,1,...,35} {
     \draw [fill=black!45] (\i,14) circle(2pt); 
  }
   
 \foreach \i in {0,1,...,8}{
     \draw [myedgestyle] (\i*4,14.1) --(\i*4+3,14.1);
}
 
\draw [ultra thick,dotted]  plot[smooth, tension=.7] coordinates {(0,18) (v2)};
\draw [blue, thick, dotted] plot[smooth, tension=.7] coordinates {(16,18) (16,10)};
\draw [blue, thick,dotted] plot[smooth, tension=.7] coordinates {(32,18) (32,10)};
\draw [fill=black] (0,10) circle (4pt); \draw [fill=black] (0,14) circle (4pt); \draw [fill=black] (0,18) circle (4pt); 
\draw [fill=blue] (16,10) circle (3pt); \draw [fill=blue] (16,14) circle (3pt); \draw [fill=blue] (16,18) circle (3pt); 
\draw [fill=blue] (32,10) circle (3pt); \draw [fill=blue] (32,14) circle (3pt); \draw [fill=blue] (32,18) circle (3pt); 

\draw [red] [myedgestyle-mirror] (0,14) --(32,14);
% \draw [fill=red] (35,14) circle(2.5pt); 
\node at (35,13.6) {$a$};
\node at (32,13.6) {$b$};

\end{tikzpicture}
\caption{\label{app-copycat-lemmas-1} A figure that illustrates Lemma \ref{approxi-copy-cat}.} 
\end{figure}
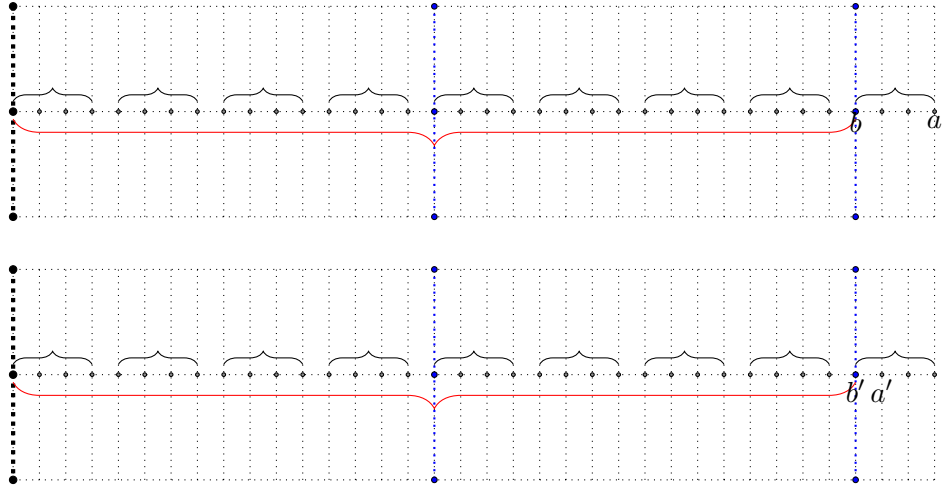
%\end{comment}

%\begin{comment}
\begin{figure}
\centering
\hspace*{-5mm}
\begin{tikzpicture}[scale=0.35]
\tikzstyle{myedgestyle-1} = [decorate, decoration={brace, amplitude=5pt, raise=1mm}];
\tikzstyle{myedgestyle-mirror-1} = [decorate, decoration={brace, amplitude=10pt, raise=3mm,mirror}];
\tikzstyle{myedgestyle-mirror-2} = [decorate, decoration={brace, amplitude=5pt, raise=2mm,mirror}];

 \draw [black,xstep=1,ystep=4, dotted] (0,0) node (v1) {} grid (35,8);
 \foreach \i in {0,1,...,35} {
     \draw [fill=black!45] (\i,4) circle(2pt); 
  }
   
 \foreach \i in {0,1,...,8}{
     \draw [myedgestyle-1] (\i*4,4.1) --(\i*4+3,4.1);
}
 
\draw [ultra thick, dotted]  plot[smooth, tension=.7] coordinates {(0,8) (v1)};
\draw [blue, thick,dotted] plot[smooth, tension=.7] coordinates {(16,8) (16,0)};
\draw [blue, thick,dotted] plot[smooth, tension=.7] coordinates {(32,8) (32,0)};
\draw [fill=black] (0,0) circle (4pt); \draw [fill=black] (0,4) circle (4pt); \draw [fill=black] (0,8) circle (4pt); 
\draw [fill=blue] (16,0) circle (3pt); \draw [fill=blue] (16,4) circle (3pt); \draw [fill=blue] (16,8) circle (3pt); 
\draw [fill=blue] (32,0) circle (3pt); \draw [fill=blue] (32,4) circle (3pt); \draw [fill=blue] (32,8) circle (3pt); 

\draw [red] [myedgestyle-mirror-1] (0,4) --(24,4);
\draw [yellow] [myedgestyle-mirror-2] (16,4) --(24,4);
 %\draw [fill=red] (27,4) circle(2.5pt); 
\node at (27,3.4) {$a^\prime$};
\node at (24,3.4) {$b^\prime$};
\node at (16,3.4) {$c^\prime$};
\node at (-0.5,3.4) {$d^\prime$};
%+++++++++++++++++++++++++++++++++++++++++++++++++++++++++++++++

\draw [black,xstep=1,ystep=4, dotted, shift={(0,10)}] (0,0) node (v2) {} grid (35,8);
 \foreach \i in {0,1,...,35} {
     \draw [fill=black!45] (\i,14) circle(2pt); 
  }
   
 \foreach \i in {0,1,...,8}{
     \draw [myedgestyle-1] (\i*4,14.1) --(\i*4+3,14.1);
}
 
\draw [ultra thick,dotted]  plot[smooth, tension=.7] coordinates {(0,18) (v2)};
\draw [blue, thick, dotted] plot[smooth, tension=.7] coordinates {(16,18) (16,10)};
\draw [blue, thick,dotted] plot[smooth, tension=.7] coordinates {(32,18) (32,10)};
\draw [fill=black] (0,10) circle (4pt); \draw [fill=black] (0,14) circle (4pt); \draw [fill=black] (0,18) circle (4pt); 
\draw [fill=blue] (16,10) circle (3pt); \draw [fill=blue] (16,14) circle (3pt); \draw [fill=blue] (16,18) circle (3pt); 
\draw [fill=blue] (32,10) circle (3pt); \draw [fill=blue] (32,14) circle (3pt); \draw [fill=blue] (32,18) circle (3pt); 

\draw [red] [myedgestyle-mirror-1] (0,14) --(24,14);
\draw [yellow] [myedgestyle-mirror-2] (16,14) --(24,14);
% \draw [fill=red] (26,14) circle(2.5pt); 
\node at (26,13.6) {$a$};
\node at (24,13.6) {$b$};
\node at (16,13.6) {$c$};
\node at (-0.5,13.6) {$d$};
\end{tikzpicture}
\caption{\label{app-copycat-lemmas-2} A figure that illustrates Lemma \ref{approxi-copy-cat-1} and Lemma \ref{corollary-approxi-copy-cat}.} 
\end{figure}
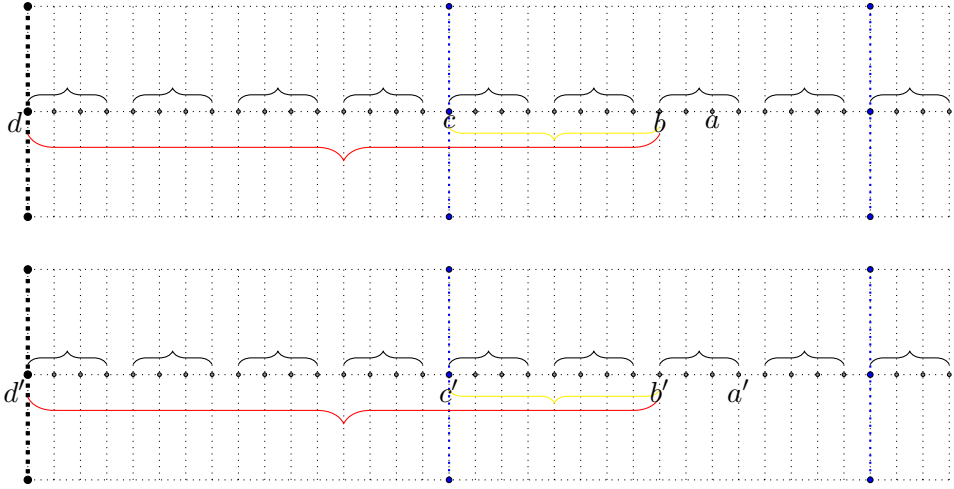
%\end{comment}

\begin{lemma}\label{approxi-copy-cat}
For any $1\!\leq\! r\!<\!\xi\!\leq\! m$ and  $(x,y),(x^{\prime},y)\!\in\! \mathbb{X}_r^*-\mathbb{X}_{r+1}^*$, if 
$\llbracket x\rrbracket_{r}^{min}\!-\![\llparenthesis x\rrparenthesis_\xi]_{r}\!\equiv\llbracket x^{\prime}\rrbracket_{r}^{min}\!-
\![\llparenthesis x^{\prime}\rrparenthesis_\xi]_{r}\!$ (mod $\beta_{m-\xi}^{m-r}$), then $[x]_i\equiv [x^{\prime}]_i$ (mod $\beta_{m\!-\!\xi}^{m\!-\!i}$) for any $r<i<\xi$.
\end{lemma}
Note that the condition ``$\llbracket x\rrbracket_{r}^{min}\!-\![\llparenthesis x\rrparenthesis_\xi]_{r}\!\equiv\llbracket x^{\prime}\rrbracket_{r}^{min}\!-
\![\llparenthesis x^{\prime}\rrparenthesis_\xi]_{r}\!$ (mod $\beta_{m-\xi}^{m-r}$)'' is similar to the condition 1$^\diamond$ in the proof of Lemma \ref{winning-strategy-in-k=3}. By remark \ref{remark-ommit-mod}, it is equivalent to ``$\llbracket x\rrbracket_{r}^{min}\!-\![\llparenthesis x\rrparenthesis_\xi]_{r}=\llbracket x^{\prime}\rrbracket_{r}^{min}\!-
\![\llparenthesis x^{\prime}\rrparenthesis_\xi]_{r}\!$''. 

The following lemma says something similar: two distances (between vertices) are equivalent w.r.t. coarser scale if they are equivalent w.r.t. finer scale. 
It is similar to Lemma \ref{abstraction-strategy-premier}. 
 See Fig. \ref{app-copycat-lemmas-2}. Similar to the last figure, vertices $a$, $b$, $c$, $d$, $a^\prime$, $b^\prime$, $c^\prime$ and $d^\prime$ stand for $([x]_r,y)$, $(\llbracket x\rrbracket_r^{min},y)$, $([\llparenthesis x\rrparenthesis_i]_{r},y)$, $([\llparenthesis x\rrparenthesis_\xi]_{r},y)$, $([x^\prime]_r,y)$, $(\llbracket x^\prime\rrbracket_r^{min},y)$,  $([\llparenthesis x^\prime\rrparenthesis_i]_{r},y)$ and $([\llparenthesis x^\prime\rrparenthesis_\xi]_{r},y)$ respectively. Vertices of different colour stand for vertices of different indices: black vertices have index $\xi$; the grey ones have index $r$; while the blue ones have index $i$ that is strictly between $\xi$ and $r$. Vertices of the same colour have the same index. The set of vertices under one black brace stand for a $\mathpzc{U}_r^*$-tuple. The number of vertices above a red brace stand for $\llbracket x\rrbracket_{r}^{min}\!-\![\llparenthesis x\rrparenthesis_\xi]_{r}$.
The number of vertices above a yellow brace is $\llbracket x\rrbracket_{r}^{min}\!-\![\llparenthesis x\rrparenthesis_{i}]_{r}$.

\begin{lemma}\label{approxi-copy-cat-1}
For any $1\!\leq\! r\!<\!\xi\!-1\!<\! m$ and $(x,y),(x^{\prime},y)\!\in\! \mathbb{X}_r^*-\mathbb{X}_{r+1}^*$, if 
$\llbracket x\rrbracket_{r}^{min}-[\llparenthesis x\rrparenthesis_\xi]_{r}\!\equiv\llbracket x^{\prime}\rrbracket_{r}^{min}-
[\llparenthesis x^{\prime}\rrparenthesis_\xi]_{r}\!$ (mod $\beta_{m-\xi}^{m-r}$), then 
$\llbracket x\rrbracket_{r}^{min}\!-\![\llparenthesis x\rrparenthesis_{i}]_{r}\!\equiv\!\llbracket x^{\prime}\rrbracket_{r}^{min}\!-
\![\llparenthesis x^{\prime}\rrparenthesis_{i}]_{r}\!\!$ (mod $\beta_{m-\xi}^{m-r}$) for any $r<i<\xi$.
\end{lemma}

%++++++++++++++++++++++++ begin comment ++++++++++++++++++++++++++++

The following lemma says that, modulo some amount, if two vertices are (approximately) in the same position in lower abstraction, then not only so are they in higher abstraction, but the indices of their projections in higher abstractions are precisely the same. See Fig. \ref{app-copycat-lemmas-2}. 
\begin{lemma}\label{corollary-approxi-copy-cat}
For any $1\!\leq\! r\!<\!\xi\!\leq\! m$ and  $(x,y),(x^{\prime},y)\!\in\! \mathbb{X}_r^*\!-\!\mathbb{X}_{r+1}^*$, if 
%\begin{enumerate}[(1)]
%\item 
$\llbracket x\rrbracket_{r}^{min}\!-\![\llparenthesis x\rrparenthesis_\xi]_{r}\!\equiv\llbracket x^{\prime}\rrbracket_{r}^{min}\!-
\![\llparenthesis x^{\prime}\rrparenthesis_\xi]_{r}\not\equiv 0\hspace{2pt}(\mbox{mod }\beta_{m-\xi}^{m-r})$ 
%\item  $\mathrm{idx}(\llparenthesis x\rrparenthesis_{r},y)=\mathrm{idx}(\llparenthesis x^{\prime}\rrparenthesis_{r},y)$,
%\item $\mathrm{idx}(\llparenthesis x\rrparenthesis_\xi,y)=\mathrm{idx}(\llparenthesis x^{\prime}\rrparenthesis_\xi,y)$,
%\end{enumerate}
 then for any $i$ where $r<i<\xi$, $$\mathrm{idx}(\llparenthesis x\rrparenthesis_i,y)=\mathrm{idx}(\llparenthesis x^{\prime}\rrparenthesis_i,y),$$ if $\mathrm{idx}(\llparenthesis x\rrparenthesis_i,y)<\xi$ and  $\mathrm{idx}(\llparenthesis x^\prime\rrparenthesis_i,y)<\xi$.
\end{lemma}

In the following we give the proofs of these lemmas. 

\textbf{Proof of Lemma \ref{approxi-copy-cat}}.
\begin{proof}
 
By definition,  
$\llbracket x\rrbracket_{r}^{min}\!-\![\llparenthesis x\rrparenthesis_\xi]_{r}\!\equiv\llbracket x^{\prime}\rrbracket_{r}^{min}\!-
\![\llparenthesis x^{\prime}\rrparenthesis_\xi]_{r}\!$ (mod $\beta_{m-\xi}^{m-r}$) implies that  $$\mid[x]_r \mbox{ mod } \beta_{m-\xi}^{m-r}-[x^{\prime}]_r \mbox{ mod } \beta_{m-\xi}^{m-r}\mid\, <\mathpzc{U}_{r}^*$$ Let $[x^{\prime}]_r=[x]_r+a\beta_{m-\xi}^{m-r}+b$ where $a,b\in \mathbf{Z}$ and $|b|< \mathpzc{U}_{r}^*$.

By definition, for any $x$ and $i$ (recall that $r<i<\xi$),\\[-8pt] 
$$[x^{\prime}]_i=\left\lfloor \frac{x^{\prime}}{\beta_{m-r}^{m-1}}\times\frac{\beta_{m-r}^{m-1}}{\beta_{m-i}^{m-1}}\right\rfloor=\left\lfloor \frac{x^{\prime}}{\beta_{m-r}^{m-1}}\times\frac{\gamma_{m-i}^*}{\gamma_{m-r}^*}\right\rfloor.$$ 

Note that $[x^{\prime}]_r\leq\frac{x^{\prime}}{\beta_{m-r}^{m-1}}< [x^{\prime}]_r+1$. 

Hence, $\left\lfloor[x^{\prime}]_r\frac{\gamma_{m-i}^*}{\gamma_{m-r}^*}\right\rfloor\leq[x^{\prime}]_i< \left\lfloor \left([x^{\prime}]_r+1\right)\frac{\gamma_{m-i}^*}{\gamma_{m-r}^*}\right\rfloor$. 

In other words,\\[-10pt]  
\begin{equation*}
\left\lfloor\left([x]_r+a\beta_{m-\xi}^{m-r}+b\right)\frac{\gamma_{m-i}^*}{\gamma_{m-r}^*}\right\rfloor\leq[x^{\prime}]_i< \left\lfloor \left([x]_r+a\beta_{m-\xi}^{m-r}+b+1\right)\frac{\gamma_{m-i}^*}{\gamma_{m-r}^*}\right\rfloor.
\end{equation*} 

Similarly, by definition,  
$$[x]_r=\left\lfloor \frac{x}{\beta_{m-i}^{m-1}}\times\frac{\beta_{m-i}^{m-1}}{\beta_{m-r}^{m-1}}\right\rfloor=\left\lfloor \frac{x}{\beta_{m-i}^{m-1}}\times\frac{\gamma_{m-r}^*}{\gamma_{m-i}^*}\right\rfloor.$$ 

Because of $[x]_i\leq\frac{x}{\beta_{m-i}^{m-1}}< [x]_i+1$ and $\frac{\gamma_{m-r}^*}{\gamma_{m-i}^*}\in\mathbf{N}^+$,

$\indent\qquad [x]_i\frac{\gamma_{m-r}^*}{\gamma_{m-i}^*}\leq [x]_r<([x]_i+1)\frac{\gamma_{m-r}^*}{\gamma_{m-i}^*}$. 

Hence 
\begin{multline*}
\left\lfloor\left([x]_i\frac{\gamma_{m-r}^*}{\gamma_{m-i}^*}+
a\beta_{m-\xi}^{m-r}+b\right)\frac{\gamma_{m-i}^*}{\gamma_{m-r}^*}\right\rfloor \leq  [x^{\prime}]_i \\
< \left\lfloor \left(\left([x]_i+1\right)\frac{\gamma_{m-r}^*}{\gamma_{m-i}^*}+a\beta_{m-\xi}^{m-r}+b+1\right)
\frac{\gamma_{m-i}^*}{\gamma_{m-r}^*}\right\rfloor.
\end{multline*} 

Observe that $|(b+1)\times \frac{\gamma_{m-i}^*}{\gamma_{m-r}^*}|<1$, and $\beta_{m-\xi}^{m-r}=\beta_{m-\xi}^{m-i}\times\frac{\gamma_{m-r}^*}{\gamma_{m-i}^*}$. 

Therefore, $[x]_i+a\beta_{m-\xi}^{m-i}\leq [x^{\prime}]_i<[x]_i+1+a\beta_{m-\xi}^{m-i}$. 

Therefore,  $[x]_i\equiv [x^{\prime}]_i$ (mod $\beta_{m\!-\!\xi}^{m\!-\!i}$), for any $r\!<\!i\!< \!\xi$.
\end{proof}

\textbf{Proof of Lemma \ref{approxi-copy-cat-1}}.
\begin{proof}
%It is a direct corollary of Lemma \ref{approxi-copy-cat}. 
Firstly, we show that 
\begin{equation}\label{app-corollary-copycat-eqn1}
[\llparenthesis x\rrparenthesis_{\xi-1}]_r\!-\![\llparenthesis x\rrparenthesis_{\xi}]_r\!\equiv\! [\llparenthesis x^\prime\rrparenthesis_{\xi-1}]_r\!-\![\llparenthesis x^\prime\rrparenthesis_{\xi}]_r \hspace{3pt}(\mbox{mod }\beta_{m-\xi}^{m-r})
\end{equation}

By Lemma \ref{approxi-copy-cat}, we have  
\begin{equation*}
[x]_{\xi-1}\equiv [x^\prime]_{\xi-1} \mbox{ mod }\beta_{m-\xi}^{m-\xi+1}.
\end{equation*}

Let $[x]_{\xi-1}:=[x^\prime]_{\xi-1}+a\beta_{m-\xi}^{m-\xi+1}$ for some $a\in\mathbf{N}^+$.

For any $r<i\leq \xi$, 
\begin{equation}\label{app-corollary-copycat-eqn3}
\begin{split}
[\llparenthesis x\rrparenthesis_{i}]_r &=
\left\lfloor\frac{[x]_i\beta_{m-i}^{m-1}+\frac{1}{2}\sum_{1<j\leq i}\beta_{m-j}^{m-1}}{\beta_{m-r}^{m-1}} \right\rfloor \\
&= [x]_i\beta_{m-i}^{m-r}+\frac{1}{2}\sum_{r<j\leq i}\beta_{m-j}^{m-r}
\end{split}
\end{equation}

Let $\psi_0:=([\llparenthesis x\rrparenthesis_{\xi-1}]_r-[\llparenthesis x\rrparenthesis_{\xi}]_r)-([\llparenthesis x^\prime\rrparenthesis_{\xi-1}]_r-[\llparenthesis x^\prime\rrparenthesis_{\xi}]_r)$.

We have 
\begin{equation*}
\begin{split}
\psi_0&=([x]_{\xi-1}\beta_{m-\xi+1}^{m-r}-[x^\prime]_{\xi-1}\beta_{m-\xi+1}^{m-r})-([x]_{\xi}\beta_{m-\xi}^{m-r}-[x^\prime]_{\xi}\beta_{m-\xi}^{m-r})\\
&=\beta_{m-\xi+1}^{m-r}\left([x]_{\xi-1}-[x^\prime]_{\xi-1}\right)- \beta_{m-\xi}^{m-r}\left([x]_{\xi}-[x^\prime]_{\xi}\right)\\
&=a\beta_{m-\xi}^{m-\xi+1}\beta_{m-\xi+1}^{m-r}-
  \beta_{m-\xi}^{m-r}\left([x]_{\xi}-[x^\prime]_{\xi}\right)\\
&=\beta_{m-\xi}^{m-r}\left(a-\left([x]_{\xi}-[x^\prime]_{\xi}\right)\right).
\end{split}
\end{equation*}

Therefore, the claim (\ref{app-corollary-copycat-eqn1}) holds. 

Secondly, we can prove that $x, \llparenthesis x\rrparenthesis_{\xi-1}, \llparenthesis x\rrparenthesis_{\xi}$ has the same order as $x^\prime$, $\llparenthesis x^\prime\rrparenthesis_{\xi-1}$, $\llparenthesis x^\prime\rrparenthesis_{\xi}$ does. For example, $x\geq \llparenthesis x\rrparenthesis_{\xi-1}\geq\llparenthesis x\rrparenthesis_{\xi}$ if and only if  $x^\prime\geq\llparenthesis x^\prime\rrparenthesis_{\xi-1}\geq\llparenthesis x^\prime\rrparenthesis_{\xi}$. Here, 
we only prove the special case $x\geq \llparenthesis x\rrparenthesis_{\xi-1}\geq\llparenthesis x\rrparenthesis_{\xi}$. The other cases are not very different.

By (\ref{app-corollary-copycat-eqn3}), 
\begin{equation}\label{app-corollary-copycat-eqn2}
\begin{split}
[\llparenthesis x\rrparenthesis_\xi]_r &= 
[x]_\xi\beta_{m-\xi}^{m-r}+\frac{1}{2}\sum_{r<i\leq \xi}
\beta_{m-i}^{m-r}\\
&= \beta_{m-r-1}^{m-r}\left(
[x]_\xi\beta_{m-\xi}^{m-r-1}+\frac{1}{2}\sum_{r<i\leq \xi}
\beta_{m-i}^{m-r-1}\right)\\
&= 2^r\mathpzc{U}_r^*\left(
[x]_\xi\beta_{m-\xi}^{m-r-1}+\frac{1}{2}\sum_{r<i\leq \xi}
\beta_{m-i}^{m-r-1}\right)\\
&= 2^r\mathpzc{U}_r^*\!\left(
[x]_\xi\beta_{m-\xi}^{m-r-1}\!+\!\frac{1}{2}\!\sum_{r+1<i\leq \xi}
\!\beta_{m-i}^{m-r-1}\right)+2^{r-1}\mathpzc{U}_r^*
\end{split}
\end{equation}

Note that $\beta_{m-\xi}^{m-r-1}$ is a  natural number, since $\xi>r+1$.  Therefore, $[\llparenthesis x\rrparenthesis_\xi]_r$ is divisible by $\mathpzc{U}_r^*$. Hence, $\llbracket \llparenthesis x\rrparenthesis_\xi\rrbracket_r^{min}=[\llparenthesis x\rrparenthesis_\xi]_r$. Together with the assumption that  
$\llbracket x\rrbracket_{r}^{min}\!-\![\llparenthesis x\rrparenthesis_\xi]_{r}\!\equiv\llbracket x^{\prime}\rrbracket_{r}^{min}\!-
\![\llparenthesis x^{\prime}\rrparenthesis_\xi]_{r}$ (mod $\beta_{m-\xi}^{m-r}$), it
 implies that $\llparenthesis x\rrparenthesis_\xi\leq x$ if and only if $\llparenthesis x^\prime\rrparenthesis_\xi\leq x^\prime$.

Assume for the purpose of a contradiction that $\llparenthesis x^\prime\rrparenthesis_{\xi-1}\!<\! \llparenthesis x^\prime\rrparenthesis_{\xi}$. 

By definition, we have  
\begin{equation*}
\begin{split}
[\llparenthesis x^\prime\rrparenthesis_{\xi-1}]_{\xi-1}&=\left\lfloor\frac{[x^\prime]_{\xi-1}\beta_{m-\xi+1}^{m-1}+\frac{1}{2}\sum_{1<j\leq \xi-1}\beta_{m-j}^{m-1}}{\beta_{m-\xi+1}^{m-1}} \right\rfloor\\
& =[x^\prime]_{\xi-1}.
\end{split}
\end{equation*}

Therefore, 
$[x^\prime]_{\xi-1}=
[\llparenthesis x^\prime\rrparenthesis_{\xi-1}]_{\xi-1}\leq
[\llparenthesis x^\prime\rrparenthesis_{\xi}]_{\xi-1}\leq [x^\prime]_{\xi-1}$, by $\llparenthesis x^\prime\rrparenthesis_{\xi-1}\!<\! \llparenthesis x^\prime\rrparenthesis_{\xi}\leq x^\prime$. 

By definition, for any $1\leq i\leq m$ and $(x^\prime,y)\in\mathbb{X}_1^*$, 
\begin{equation}
(\llparenthesis x^\prime\rrparenthesis_i,y)\in \mathbb{X}_i^*
\end{equation} 

In particular, $(\llparenthesis x^\prime\rrparenthesis_\xi,y)\!\in\! \mathbb{X}_{\xi}^*$.
Then by Lemma \ref{i=0theni-1=0}, $(\llparenthesis x^\prime\rrparenthesis_\xi,y)\!\in\! \mathbb{X}_{\xi-1}^*$. 

Therefore, by Lemma \ref{projection}, $\llparenthesis x^\prime\rrparenthesis_{\xi}=\llparenthesis x^\prime\rrparenthesis_{\xi-1}$. We arrive at a contradiction.

Now assume for the purpose of a contradiction that $x^\prime<\llparenthesis x^\prime\rrparenthesis_{\xi-1}$. Since both $x^\prime\in\mathbb{X}_r^*$ and $\llparenthesis x^\prime\rrparenthesis_{\xi-1}\in\mathbb{X}_r^*$ (by Lemma \ref{i=0theni-1=0}), we have  $[x^\prime]_r<[\llparenthesis x^\prime\rrparenthesis_{\xi-1}]_r$, by definition. 
Let $\varphi_0:=(\llbracket x\rrbracket_{r}^{min}\!-\![\llparenthesis x\rrparenthesis_\xi]_{r})\!\mbox{ mod }\beta_{m-\xi}^{m-r}$.
Then we have 
\begin{equation*} 
\begin{split}
\varphi_0&=(\llbracket x^{\prime}\rrbracket_{r}^{min}\!-
\![\llparenthesis x^{\prime}\rrparenthesis_\xi]_{r})\mbox{ mod }\beta_{m-\xi}^{m-r}\\
&\leq ([x^\prime]_r-[\llparenthesis x^{\prime}\rrparenthesis_\xi]_{r})\mbox{ mod }\beta_{m-\xi}^{m-r}\\
&<([\llparenthesis x^\prime\rrparenthesis_{\xi-1}]_r-[\llparenthesis x^\prime\rrparenthesis_{\xi}]_r)\mbox{ mod }\beta_{m-\xi}^{m-r}. 
\end{split}
\end{equation*}

On the other hand, by an argument akin to (\ref{app-corollary-copycat-eqn2}), we have    $\llbracket \llparenthesis x\rrparenthesis_{\xi-1}\rrbracket_r^{min}=[\llparenthesis x\rrparenthesis_{\xi-1}]_r$. Therefore, 
we also have that 
\begin{equation*}
\varphi_0\geq ([\llparenthesis x\rrparenthesis_{\xi-1}]_{r}-[\llparenthesis x\rrparenthesis_\xi]_{r})\mbox{ mod }\beta_{m-\xi}^{m-r}. 
\end{equation*}
By (\ref{app-corollary-copycat-eqn1}), $\varphi_0\geq [\llparenthesis x^\prime\rrparenthesis_{\xi-1}]_{r}\!-\![\llparenthesis x^\prime\rrparenthesis_\xi]_{r}\mbox{ mod }\beta_{m-\xi}^{m-r}$. A contradiction occurs. 

We have shown that 
\begin{equation}
x\!\geq\! \llparenthesis x\rrparenthesis_{\xi-1}\!\geq\!\llparenthesis x\rrparenthesis_{\xi}\Leftrightarrow x^\prime\!\geq\!\llparenthesis x^\prime\rrparenthesis_{\xi-1}\!\geq\!\llparenthesis x^\prime\rrparenthesis_{\xi}.
\end{equation}
Together with (\ref{app-corollary-copycat-eqn1}) and the assumption that $\llbracket x\rrbracket_{r}^{min}\!-\![\llparenthesis x\rrparenthesis_\xi]_{r}\equiv\llbracket x^\prime\rrbracket_{r}^{min}\!-\![\llparenthesis x^\prime\rrparenthesis_\xi]_{r}\!\hspace{3pt}(\mbox{mod }\beta_{m-\xi}^{m-r})$, the claim of the lemma holds. 
\end{proof}
 
\textbf{Proof of Lemma \ref{corollary-approxi-copy-cat}}.
\begin{proof}

Let $r^+:=r+1$. By definitions, $(\llparenthesis x\rrparenthesis_{r^+},y)\in\mathbb{X}_{r^+}^*$, thereby 
 $\mathrm{idx}(\llparenthesis x\rrparenthesis_{r^+},y)\geq r^+\geq 2$. 

The following proof is by contradiction. Suppose that for some $b$ where $r<b<\xi$, 
\begin{itemize}
\item $\mathrm{idx}(\llparenthesis x\rrparenthesis_{b},y), \mathrm{idx}(\llparenthesis x^{\prime}\rrparenthesis_{b},y)<\xi$;
\item $\mathrm{idx}(\llparenthesis x\rrparenthesis_{b},y)\neq\mathrm{idx}(\llparenthesis x^{\prime}\rrparenthesis_{b},y)$;
\item for any $t$ where $b<t<\xi$, $\mathrm{idx}(\llparenthesis x\rrparenthesis_t,y)=\mathrm{idx}(\llparenthesis x^{\prime}\rrparenthesis_t,y)$, if $\mathrm{idx}(\llparenthesis x\rrparenthesis_t,y)<\xi$ and $\mathrm{idx}(\llparenthesis x^{\prime}\rrparenthesis_t,y)<\xi$.

\end{itemize}
 
Recall that $\mathrm{idx}(\llparenthesis x\rrparenthesis_b,y)\geq b$. 
Suppose that $b=\xi-1$.  Because $\mathrm{idx}(\llparenthesis x\rrparenthesis_b,y)\neq\mathrm{idx}(\llparenthesis x^\prime\rrparenthesis_b,y)$, then either $\mathrm{idx}(\llparenthesis x\rrparenthesis_b,y)\geq \xi$ or $\mathrm{idx}(\llparenthesis x^\prime\rrparenthesis_b,y)\geq \xi$, a contradiction. 
 Therefore, $b<\xi-1$.

%Note that $s-1>r^+$ because of \textit{(2)}.

Let $\mathrm{idx}(\llparenthesis x\rrparenthesis_{b},y)\!=\!j$ and $\mathrm{idx}(\llparenthesis x^{\prime}\rrparenthesis_{b},y)\!=\!j^{\prime}$. Assume without loss of generality that $j\!>\!j^{\prime}$. The case wherein  $j\!<\!j^{\prime}$ is symmetric.

For any $(u,v)\in \mathbb{X}_r^*$ and any $r+2<s\leq \xi$, assume that  $\mathrm{idx}(\llparenthesis u\rrparenthesis_{s-1},v)\!=\! p$ and $\mathrm{idx}(\llparenthesis u\rrparenthesis_{s},v)\!=\!q\!>\! p$. Observe that, 
there is no $(u^{\prime},v)\in \mathbb{X}_r^*$ such that its index is $p+1$ and it is strictly between $\mathrm{idx}(\llparenthesis u\rrparenthesis_{s-1},v)$ and $\mathrm{idx}(\llparenthesis u\rrparenthesis_{s},v)$. Otherwise, $(u^{\prime},v)$ would be $\mathrm{idx}(\llparenthesis u\rrparenthesis_{s},v)$, a contradiction to the assumption that $u^{\prime}\neq\llparenthesis u\rrparenthesis_{s}$. 
Since $\llparenthesis u\rrparenthesis_{s}, \llparenthesis u\rrparenthesis_{s-1}\in \mathbb{X}_p^*$, 
by definition, for some $l\in [1,\frac{1}{2}\beta_{m-p-1}^{m-p}-\frac{1}{2}]$,  
\begin{equation*}\label{idx-distance-bound3}
\left|[\llparenthesis u\rrparenthesis_{s}]_{r^+}-[\llparenthesis u\rrparenthesis_{s-1}]_{r^+}\right|=l\cdot\beta_{m-p}^{m-r^+}.  
\end{equation*}

Hence, the following holds:                                                                                                                                                                                                                                                                                                                                                                                                                                                                                                                                                                                                                                                                                                                                                                                                                                                                                                                                                                                                                                                                                                                                                                                                                                                                                                                                                                                                                                                                                                                                                                                                                                                                                                                                                                                                                                                                                                                                                                                                                                                                                                                                                                                                                                                                                                       
\begin{eqnarray}
\beta_{m-p}^{m-r^+}&\leq& \left|[\llparenthesis u\rrparenthesis_{s}]_{r^+}-[\llparenthesis u\rrparenthesis_{s-1}]_{r^+}\right|\label{idx-distance-bound1}\\&\leq& \frac{1}{2}\beta_{m-p-1}^{m-r^+}-\frac{1}{2}\beta_{m-p}^{m-r^+} \label{idx-distance-bound2}
\end{eqnarray}

We use the observation that one unit of difference in higher abstraction means a tremendous difference in a lower one  (w.r.t. distance of first coordinates) to prove the following Claim. Before that,  we first prove some observations.

\begin{fact}\label{app-equal-index-2-equal-len-fac1}
If $m>i>i^\prime\geq r^+$, then \\
$\indent\hspace{40pt}\beta_{m-i-1}^{m-r^+}-
\beta_{m-i}^{m-r^+}>\beta_{m-i^\prime-1}^{m-r^+}-\beta_{m-i^\prime}^{m-r^+}$.
\end{fact}
\underline{\textit{Proof of Fact:}}\\[4pt]
$\left(\beta_{m-i-1}^{m-r^+}-\beta_{m-i}^{m-r^+}\right)-
\left(\beta_{m-i^\prime-1}^{m-r^+}-\beta_{m-i^\prime}^{m-r^+}\right)$
\begin{eqnarray*}
&>& \beta_{m-i-1}^{m-r^+}-2\beta_{m-i}^{m-r^+}\\
&=& \beta_{m-i}^{m-r^+}\left(\beta_{m-i-1}^{m-i}-2\right)\\
&=& \beta_{m-i}^{m-r^+}\left(2^i\mathpzc{U}_i^*-2\right)\\
&>& 0.
\end{eqnarray*}
 
\underline{\textit{Q.E.D. of Fact.}}

\begin{fact}\label{app-equal-index-2-equal-len-fac2}
If $[\llparenthesis x\rrparenthesis_{\xi}]_{r^+}<[x]_{r^+}$, then
\begin{equation*}
[x]_{r^+} \mbox{ mod } \beta_{m-\xi}^{m-r^+}=
[x]_{r^+}-[\llparenthesis x\rrparenthesis_{\xi}]_{r^+}+\frac{1}{2}\sum_{r^+< i\leq \xi}\beta_{m-i}^{m-r^+}.
\end{equation*}  
\end{fact}
\underline{\textit{Proof of Fact:}}\\[-4pt]

By definitions, 
\begin{eqnarray*}
[\llparenthesis x\rrparenthesis_{\xi}]_{r^+}&=& \left\lfloor \frac{[x]_\xi\beta_{m-\xi}^{m-1}+\frac{1}{2}\sum_{1<i\leq\xi}\beta_{m-i}^{m-1}}
{\beta_{m-r^+}^{m-1}}\right\rfloor\\
&=& [x]_\xi\beta_{m-\xi}^{m-r^+}+\displaystyle\frac{1}{2}\sum_{r^+< i\leq \xi}\beta_{m-i}^{m-r^+}.
\end{eqnarray*}

Note that $\llparenthesis x\rrparenthesis_\xi\neq x$, because $\llbracket x\rrbracket_{r}^{min}\!-\![\llparenthesis x\rrparenthesis_\xi]_{r}\!\equiv\llbracket x^{\prime}\rrbracket_{r}^{min}\!-
\![\llparenthesis x^{\prime}\rrparenthesis_\xi]_{r}\not\equiv 0\hspace{3pt}(\mbox{mod }\beta_{m-\xi}^{m-r})$ and $\llbracket\llparenthesis x\rrparenthesis_\xi \rrbracket_r^{min}=[\llparenthesis x\rrparenthesis_\xi]_{r}$. 
Hence, we assume that $x=a\beta_{m-\xi}^{m-1}+h$ where $0<h<\beta_{m-\xi}^{m-1}$. That is, $[x]_\xi=a$ and $[x]_{r^+}=a\beta_{m-\xi}^{m-r^+}+[h]_{r^+}$ ($\because \beta_{m-\xi}^{m-r^+}\in\mathbf{N}^+$), where  $0\!<\![h]_{r^+}\!<\!\beta_{m-\xi}^{m-r^+}$. Therefore, 
$[\llparenthesis x\rrparenthesis_{\xi}]_{r^+}\!-[x]_{r^+}=\frac{1}{2}\sum_{r^+< i\leq \xi}\beta_{m-i}^{m-r^+}-[h]_{r^+}$. 
Note that, $[x]_{r^+} \!\mbox{ mod } \beta_{m-\xi}^{m-r^+}=[h]_{r^+}$. Therefore, the fact holds.

\underline{\textit{Q.E.D. of Fact.}}\\[6pt]
Similarly, we can prove the following observation. 
\begin{fact}\label{app-equal-index-2-equal-len-fac3}
If $[\llparenthesis x\rrparenthesis_{\xi}]_{r^+}<[x]_{r^+}$, then
\begin{equation*}
[x^\prime]_{r^+} \mbox{ mod } \beta_{m-\xi}^{m-r^+}=
[x^\prime]_{r^+}-[\llparenthesis x^\prime\rrparenthesis_{\xi}]_{r^+}+\frac{1}{2}\sum_{r^+< i\leq \xi}\beta_{m-i}^{m-r^+}.
\end{equation*}  
\end{fact}

\begin{claim}\label{app-equal-index-2-equal-len-2}
$\left|[\llparenthesis x\rrparenthesis_{\xi}]_{r^+}-[\llparenthesis x\rrparenthesis_{b+1}]_{r^+}\right|$=$\left|[\llparenthesis x^{\prime}\rrparenthesis_{\xi}]_{r^+}-[\llparenthesis x^{\prime}\rrparenthesis_{b+1}]_{r^+}\right|$.
\end{claim}
\underline{\textit{Proof of Claim: }}

Assume that there are $d^\prime$ types of vertices, where $0\leq d^\prime\leq \xi-b-2$,  between $(\llparenthesis x\rrparenthesis_{b+1},y)$ and  $(\llparenthesis x\rrparenthesis_{\xi},y)$, whose indices are different and less                                                                                                                                                                                                                                                                                                                                                                                                                                                                                                                                                                                                                                                                                                                                                                                                                                                         than $\xi$ and greater than $\mathrm{idx}(\llparenthesis x\rrparenthesis_{b+1},y)$. And assume that the indices of these vertices (including $(\llparenthesis x\rrparenthesis_{b+1},y)$ and  $(\llparenthesis x\rrparenthesis_{\xi},y)$ are $\mathrm{IDX}(i)$ where 
\begin{itemize}
\item $1\leq i\leq d=d^\prime+2$;
\item $\mathrm{IDX}(i^\prime)<\mathrm{IDX}(i^{\prime}+1)$, where $1\leq i^\prime<d$;
\item $\mathrm{IDX}(1)=\mathrm{idx}(\llparenthesis x\rrparenthesis_{b+1},y)$;
\item $\mathrm{IDX}(d)=\mathrm{idx}(\llparenthesis x\rrparenthesis_{\xi},y)$.

\end{itemize}

Note that $\mathrm{idx}(\llparenthesis x\rrparenthesis_{\mathrm{IDX}(i)}\!,y)=\mathrm{IDX}(i)$, for $1<i<d$.

Assume for the purpose of a contradiction that

 $\hspace{10pt}\left|\left[\llparenthesis x\rrparenthesis_{\xi}\right]_{r^+}-\left[\llparenthesis x\rrparenthesis_{b+1}\right]_{r^+}\right|\neq \left|\left[\llparenthesis x^{\prime}\rrparenthesis_{\xi}\right]_{r^+}-\left[\llparenthesis x^{\prime}\rrparenthesis_{b+1}\right]_{r^+}\right|$. 

Then there must be a $c$, where $1<c\leq d$, such that 
\begin{enumerate}[(1$^\#$)]
\item
$\left|\left[\llparenthesis x\rrparenthesis_{\xi}\right]_{r^+}\!-\!\left[\llparenthesis x\rrparenthesis_{\mathrm{IDX}(c)}\right]_{r^+}\!\right|=\left|\left[\llparenthesis x^\prime\rrparenthesis_{\xi}\right]_{r^+}\!-\!\left[\llparenthesis x^\prime\rrparenthesis_{\mathrm{IDX}(c)}\right]_{r^+}\!\right|$;

\item 
%\begin{multline*}
$\left|\left[\llparenthesis x\rrparenthesis_{\mathrm{IDX}(c)}\right]_{r^+}\!-\!\left[\llparenthesis x\rrparenthesis_{\mathrm{IDX}(c-1)}\right]_{r^+}\!\right|\neq
\left|\left[\llparenthesis x^{\prime}\rrparenthesis_{\mathrm{IDX}(c)}\right]_{r^+}\!-\!\left[\llparenthesis x^{\prime}\rrparenthesis_{\mathrm{IDX}(c-1)}\right]_{r^+}\!\right|$.
%\end{multline*}
\end{enumerate}

Because                                                                                                                                                                                                                                                                                                                        $\mathrm{IDX}(1)\geq b+1>b$, and $\mathrm{IDX}(c-1)\geq \mathrm{IDX}(1)$, we have  $\mathrm{IDX}(c-1)>b$. Therefore, 
$\mathrm{idx}(\llparenthesis x\rrparenthesis_{\mathrm{IDX}(c-1)},y)=\mathrm{idx}(\llparenthesis x^{\prime}\rrparenthesis_{\mathrm{IDX}(c-1)},y)$. 

Assume that
\begin{equation}\label{app-app-equal-index-2-equal-len-2-eqn1}
\left|[\llparenthesis x\rrparenthesis_{\mathrm{IDX}(c)}]_{r^+}\!-\![\llparenthesis x\rrparenthesis_{\mathrm{IDX}(c-1)}]_{r^+}\!\right|>
\left|[\llparenthesis x^{\prime}\rrparenthesis_{\mathrm{IDX}(c)}]_{r^+}\!-\![\llparenthesis x^{\prime}\rrparenthesis_{\mathrm{IDX}(c-1)}]_{r^+}\!\right|. 
\end{equation}
The other case is symmetric.

Assume that $\mathrm{idx}(\llparenthesis x\rrparenthesis_{\mathrm{IDX}(c-1)}\!,y)\!=\!n$. Note that $n\!\geq\!\mathrm{IDX}(c\!-\!1)$. 

Furthermore, first assume that\\[-25pt] 

\begin{align}
\left[\llparenthesis x\rrparenthesis_{\mathrm{IDX}(c)}\right]_{r^+}\!&>\!\left[\llparenthesis x\rrparenthesis_{\mathrm{IDX}(c-1)}\right]_{r^+} \label{app-app-equal-index-2-equal-len-2-eqn2}\\
\left[\llparenthesis x\rrparenthesis_{\xi}\right]_{r^+}\!&\geq\!\left[\llparenthesis x\rrparenthesis_{\mathrm{IDX}(c)}\right]_{r^+} \label{app-app-equal-index-2-equal-len-2-eqn3}
\end{align}

Let 
\begin{multline*}
\zeta:=\left(\left[\llparenthesis x\rrparenthesis_{\mathrm{IDX}(c)}\right]_{r^+}\!-\!\left[\llparenthesis x\rrparenthesis_{\mathrm{IDX}(c-1)}\right]_{r^+}\!\right)-
\left|\left[\llparenthesis x^{\prime}\rrparenthesis_{\mathrm{IDX}(c)}\right]_{r^+}\!-\!\left[\llparenthesis x^{\prime}\rrparenthesis_{\mathrm{IDX}(c-1)}\right]_{r^+}\!\right|.
\end{multline*}

By (\ref{app-app-equal-index-2-equal-len-2-eqn1}) and (\ref{app-app-equal-index-2-equal-len-2-eqn2}), $\zeta>0$.

Hence, by (\ref{idx-distance-bound1}) and (\ref{app-app-equal-index-2-equal-len-2-eqn1}), 
\begin{equation*}
\zeta\geq \beta_{m-n}^{m-r^+}.
\end{equation*}

Hence, 
\begin{equation*}
\zeta> 2\times\frac{1}{2}\left(\beta_{m-n}^{m-r^+}-1\right)
\end{equation*}

That is, 
\begin{equation*}
\zeta>\displaystyle 2\times\frac{1}{2}\sum_{i=r+2}^{n} \left(\beta_{m-i}^{m-r^+}-\beta_{m-i+1}^{m-r^+}\right)
\end{equation*}
Therefore, by (\ref{idx-distance-bound2}) and Fact \ref{app-equal-index-2-equal-len-fac1}, we have 
\begin{eqnarray*}
\zeta&>&  \displaystyle\sum_{i=r+2}^{\mathrm{IDX}(c-1)}\left| [\llparenthesis x\rrparenthesis_{i}]_{r^+}-[\llparenthesis x\rrparenthesis_{i-1}]_{r^+}\right|+ \displaystyle\sum_{i=r+2}^{\mathrm{IDX}(c-1)}\left| [\llparenthesis x^\prime\rrparenthesis_{i}]_{r^+}-[\llparenthesis x^\prime\rrparenthesis_{i-1}]_{r^+}\right|.  \nonumber\\
\end{eqnarray*}
Therefore, 
\begin{eqnarray}\label{app-app-equal-index-2-equal-len-2-eqn4}
\zeta&>& \left|\displaystyle\sum_{i=r+2}^{\mathrm{IDX}(c-1)}\left([\llparenthesis x\rrparenthesis_{i}]_{r^+}-[\llparenthesis x\rrparenthesis_{i-1}]_{r^+}\right)\right|+\left|\displaystyle\sum_{i=r+2}^{\mathrm{IDX}(c-1)}\left([\llparenthesis x^\prime\rrparenthesis_{i}]_{r^+}-[\llparenthesis x^\prime\rrparenthesis_{i-1}]_{r^+}\right)\right|\nonumber\\
&=&\left|[\llparenthesis x\rrparenthesis_{\mathrm{IDX}(c-1)}]_{r^+}-[x]_{r^+} \right|+\left|[\llparenthesis x^\prime\rrparenthesis_{\mathrm{IDX}(c-1)}]_{r^+}-[x^\prime]_{r^+} \right|.  
\end{eqnarray} 

The above argument also tell us more about (1$^\#$). That is, 
\begin{equation*}
\left[\llparenthesis x\rrparenthesis_{\xi}\right]_{r^+}\!-\!\left[\llparenthesis x\rrparenthesis_{\mathrm{IDX}(c)}\right]_{r^+}\!=\left[\llparenthesis x^\prime\rrparenthesis_{\xi}\right]_{r^+}\!-\!\left[\llparenthesis x^\prime\rrparenthesis_{\mathrm{IDX}(c)}\right]_{r^+}.
\end{equation*}

Let $\psi:=\left|[\llparenthesis x\rrparenthesis_{\xi}]_{r^+}-[x]_{r^+}\right|$.\\[-20pt]

\begin{multline*}
\psi=\left|\left(\left[\llparenthesis x\rrparenthesis_{\mathrm{IDX}(c)}\right]_{r^+}\!-\!\left[\llparenthesis x\rrparenthesis_{\mathrm{IDX}(c-1)}\right]_{r^+}\!\right)+ \right.\\
\left([\llparenthesis x\rrparenthesis_{\mathrm{IDX}(c-1)}]_{r^+}-[x]_{r^+} \right)
+
\left.\left[\llparenthesis x\rrparenthesis_{\xi}\right]_{r^+}\!-\!\left[\llparenthesis x\rrparenthesis_{\mathrm{IDX}(c)}\right]_{r^+}\!\right|.
\end{multline*}
Recall that, by (\ref{app-app-equal-index-2-equal-len-2-eqn4})
\begin{equation*}
\left[\llparenthesis x\rrparenthesis_{\mathrm{IDX}(c)}\right]_{r^+}\!-\!\left[\llparenthesis x\rrparenthesis_{\mathrm{IDX}(c-1)}\right]_{r^+}\!> \\
\left| [\llparenthesis x\rrparenthesis_{\mathrm{IDX}(c-1)}]_{r^+}-[x]_{r^+} \right|.
\end{equation*}
Therefore, 
\begin{eqnarray*}
\psi&\geq& \left(\left[\llparenthesis x\rrparenthesis_{\mathrm{IDX}(c)}\right]_{r^+}\!-\!\left[\llparenthesis x\rrparenthesis_{\mathrm{IDX}(c-1)}\right]_{r^+}\!\right) - \\
&&\left|[\llparenthesis x\rrparenthesis_{\mathrm{IDX}(c-1)}]_{r^+}-[x]_{r^+} \right|
+\left|\left[\llparenthesis x\rrparenthesis_{\xi}\right]_{r^+}\!-\!\left[\llparenthesis x\rrparenthesis_{\mathrm{IDX}(c)}\right]_{r^+}\!\right|.\qquad [\mbox{by (\ref{app-app-equal-index-2-equal-len-2-eqn3})}]\\
\end{eqnarray*}
Therefore, 
\begin{eqnarray*}
\psi&>& \left|[\llparenthesis x^\prime\rrparenthesis_{\mathrm{IDX}(c-1)}]_{r^+}-[x^\prime]_{r^+} \right|+\left|\left[\llparenthesis x^{\prime}\rrparenthesis_{\mathrm{IDX}(c)}\right]_{r^+}\!-\!\left[\llparenthesis x^{\prime}\rrparenthesis_{\mathrm{IDX}(c-1)}\right]_{r^+}\!\right|+\\
&&\left|\left[\llparenthesis x^\prime\rrparenthesis_{\xi}\right]_{r^+}\!-\!\left[\llparenthesis x^\prime\rrparenthesis_{\mathrm{IDX}(c)}\right]_{r^+}\!\right| \hspace{10pt}[\mbox{by (\ref{app-app-equal-index-2-equal-len-2-eqn4}) and (1$^\#$)}] \\
\end{eqnarray*}
Therefore,
\begin{eqnarray*}
\psi&>& \left|[\llparenthesis x^\prime\rrparenthesis_{\mathrm{IDX}(c-1)}]_{r^+}-[x^\prime]_{r^+} \right.+\left.\left[\llparenthesis x^{\prime}\rrparenthesis_{\mathrm{IDX}(c)}\right]_{r^+}\!-\!\left[\llparenthesis x^{\prime}\rrparenthesis_{\mathrm{IDX}(c-1)}\right]_{r^+}\!\right.+\\
&&\left.\left[\llparenthesis x^\prime\rrparenthesis_{\xi}\right]_{r^+}\!-\!\left[\llparenthesis x^\prime\rrparenthesis_{\mathrm{IDX}(c)}\right]_{r^+}\!\right| \\ 
&=&\left|[\llparenthesis x^\prime\rrparenthesis_{\xi}]_{r^+}-[x^\prime]_{r^+}\right|.
\end{eqnarray*}

Because of the assumption that 
$\llbracket x\rrbracket_{r}^{min}\!-\![\llparenthesis x\rrparenthesis_\xi]_{r}\!\equiv\llbracket x^{\prime}\rrbracket_{r}^{min}\!-
\![\llparenthesis x^{\prime}\rrparenthesis_\xi]_{r}\not\equiv 0\hspace{3pt}(\mbox{mod }\beta_{m-\xi}^{m-r})$, it is easy to see that $\llparenthesis x\rrparenthesis_{\xi}\geq x\Leftrightarrow\llparenthesis x^\prime\rrparenthesis_{\xi}\geq x^\prime$.
Therefore, we have  $[x]_{r^+}\mbox{ mod }\beta_{m\!-\!\xi}^{m\!-\!r^+}\neq [x^{\prime}]_{r^+} \mbox{ mod } \beta_{m\!-\!\xi}^{m\!-\!r^+}$, by Fact \ref{app-equal-index-2-equal-len-fac2} and Fact \ref{app-equal-index-2-equal-len-fac3}.  But by Lemma \ref{approxi-copy-cat}, $[x]_{r^+}\equiv [x^{\prime}]_{r^+}$ mod $\beta_{m\!-\!\xi}^{m\!-\!r^+}$. A contradiction occurs. 

In the last arguments, if some of the assumptions (\ref{app-app-equal-index-2-equal-len-2-eqn1}), (\ref{app-app-equal-index-2-equal-len-2-eqn2}) and (\ref{app-app-equal-index-2-equal-len-2-eqn3}) 
  do not hold, then the arguments need to be revised a little bit, but are very similar. It turns out that, if even number of these three assumptions are violated, then 
$\left|[\llparenthesis x\rrparenthesis_{\xi}]_{r^+}\!-\![x]_{r^+}\right|\!>\!\left|[\llparenthesis x^\prime\rrparenthesis_{\xi}]_{r^+}\!-\![x^\prime]_{r^+}\right|$;
otherwise, $\left|[\llparenthesis x\rrparenthesis_{\xi}]_{r^+}\!-\![x]_{r^+}\right|<\left|[\llparenthesis x^\prime\rrparenthesis_{\xi}]_{r^+}\!-\![x^\prime]_{r^+}\right|$.

In summary, we arrive at a contradiction in all the cases. Therefore, the claim holds.

\underline{\textit{Q.E.D. of Claim}. }\\[-1pt]

Assume that 
\begin{equation}\label{app-app-equal-index-2-equal-len-2-eqn5}
\llparenthesis x\rrparenthesis_{\xi}\geq \llparenthesis x\rrparenthesis_{b+1}\geq \llparenthesis x\rrparenthesis_{b}\geq x.
\end{equation}

We have $\llparenthesis x^\prime\rrparenthesis_{\xi}\geq x^\prime$, since $\llbracket x\rrbracket_{r}^{min}\!-\![\llparenthesis x\rrparenthesis_\xi]_{r}\!\equiv\llbracket x^{\prime}\rrbracket_{r}^{min}\!-
\![\llparenthesis x^{\prime}\rrparenthesis_\xi]_{r}\not\equiv 0 \hspace{3pt}(\mbox{mod }\beta_{m-\xi}^{m-r})$. Note that, $\llbracket x\rrbracket_{r}^{min}\!-\![\llparenthesis x\rrparenthesis_\xi]_{r}\!\equiv\llbracket x^{\prime}\rrbracket_{r}^{min}\!-
\![\llparenthesis x^{\prime}\rrparenthesis_\xi]_{r}\hspace{3pt}(\mbox{mod }\beta_{m-\xi}^{m-r})$ implies that $\llbracket x\rrbracket_{r}^{min}\!-\![\llparenthesis x\rrparenthesis_\xi]_{r}\!=\llbracket x^{\prime}\rrbracket_{r}^{min}\!-
\![\llparenthesis x^{\prime}\rrparenthesis_\xi]_{r}$.

We show that  $x^\prime$ cannot be strictly between $\llparenthesis x^\prime\rrparenthesis_{b+1}$ and $\llparenthesis x^\prime\rrparenthesis_\xi$. Assume for a contradiction that $x^\prime$ is between $\llparenthesis x^\prime\rrparenthesis_{b+1}$ and $\llparenthesis x^\prime\rrparenthesis_\xi$. 
Note that $\mathrm{idx}(\llparenthesis x\rrparenthesis_{b+1})<\xi$ if $\llparenthesis x\rrparenthesis_{b+1}\neq \llparenthesis x\rrparenthesis_\xi$.
Then $\mathrm{idx}(\llparenthesis x\rrparenthesis_{b+1},y)=\mathrm{idx}(\llparenthesis x^\prime\rrparenthesis_{b+1},y)$. Assume that $\mathrm{idx}(\llparenthesis x^\prime\rrparenthesis_{b+1},y)=I_{b+1}$.
Because  $x^\prime$ is strictly between $\llparenthesis x^\prime\rrparenthesis_{b+1}$ and $\llparenthesis x^\prime\rrparenthesis_\xi$, $\frac{1}{2}\beta_{m-I_{b+1}}^{m-r^+}+1\leq [\llparenthesis x^\prime\rrparenthesis_\xi]_{r^+}-[x^\prime]_{r^+}\leq\beta_{m-I_{b+1}}^{m-r^+}-1$, 
and  $\mathpzc{U}_r^*+\frac{1}{2}\beta_{m-I_{b+1}}^{m-r}\leq [\llparenthesis x^\prime\rrparenthesis_\xi]_{r}-\llbracket x^\prime\rrbracket_r^{min}\leq\beta_{m-I_{b+1}}^{m-r}$.
On the other hand, $[\llparenthesis x\rrparenthesis_\xi]_{r^+}-[x]_{r^+}\geq \beta_{m-I_{b+1}}^{m-r^+}+1$ and 
$[\llparenthesis x\rrparenthesis_\xi]_{r}-\llbracket x\rrbracket_{r}^{min                        }\geq \beta_{m-I_{b+1}}^{m-r}+\mathpzc{U}_r^*$, 
because $\mathrm{idx}(\llparenthesis x\rrparenthesis_\xi,y)\in \mathbb{X}_{I_{b+1}}$ (cf. Lemma \ref{i=0theni-1=0}). Therefore, 
$\llbracket x\rrbracket_{r}^{min}\!-\![\llparenthesis x\rrparenthesis_\xi]_{r}\!\neq\llbracket x^{\prime}\rrbracket_{r}^{min}\!-
\![\llparenthesis x^{\prime}\rrparenthesis_\xi]_{r}$. We arrive at a contradiction.  
If $\llparenthesis x^\prime\rrparenthesis_{\xi}<\llparenthesis x^\prime\rrparenthesis_{b+1}$, then $\llparenthesis x^\prime\rrparenthesis_{b+1}$ would be $\llparenthesis x^\prime\rrparenthesis_\xi$ because $(\llparenthesis x^\prime\rrparenthesis_\xi,y)\in\mathbb{X}_{b+1}^*$, again a contradiction. 
Therefore, 
\begin{equation*}
\llparenthesis x^\prime\rrparenthesis_{\xi}\geq \llparenthesis x^\prime\rrparenthesis_{b+1}\geq x^\prime.
\end{equation*}
Moreover, if $\llparenthesis x^\prime\rrparenthesis_{b}>\llparenthesis x^\prime\rrparenthesis_{b+1}$, then $\llparenthesis x^\prime\rrparenthesis_b=\llparenthesis x^\prime\rrparenthesis_{b+1}$, since $(\llparenthesis x^\prime\rrparenthesis_{b+1},y)\in\mathbb{X}_b^*$, still a contradiction. Therefore,
\begin{equation}\label{app-app-equal-index-2-equal-len-2-eqn6}
\llparenthesis x^\prime\rrparenthesis_{\xi}\geq \llparenthesis x^\prime\rrparenthesis_{b+1}\geq \llparenthesis x^\prime\rrparenthesis_b\geq x^\prime.
\end{equation} 

By (\ref{app-app-equal-index-2-equal-len-2-eqn5}), 
\begin{eqnarray*}
\psi
&\geq& \left|[\llparenthesis x\rrparenthesis_{\xi}]_{r^+}-[\llparenthesis x\rrparenthesis_{b+1}]_{r^+}\right|+\left|[\llparenthesis x\rrparenthesis_{b+1}]_{r^+}-[\llparenthesis x\rrparenthesis_{b}]_{r^+}\right|
\end{eqnarray*}

Hence, 
\begin{eqnarray*}
\psi&\geq& \left|[\llparenthesis x\rrparenthesis_{\xi}]_{r^+}-[\llparenthesis x\rrparenthesis_{b+1}]_{r^+}\right|+\beta_{m-j}^{m-r^+}  \qquad [\mbox{ by (\ref{idx-distance-bound1}) }]\\
&\geq& \left|[\llparenthesis x\rrparenthesis_{\xi}]_{r^+}-[\llparenthesis x\rrparenthesis_{b+1}]_{r^+}\right|+\beta_{m-j^\prime-1}^{m-r^+}
\end{eqnarray*}

Therefore, 
\begin{eqnarray*}
\psi &>& \left|[\llparenthesis x\rrparenthesis_{\xi}]_{r^+}-[\llparenthesis x\rrparenthesis_{b+1}]_{r^+}\right|+\sum_{r^+\leq i\leq j^{\prime}}\left(\beta_{m-i-1}^{m-r^+}-\beta_{m-i}^{m-r^+}\right)
\end{eqnarray*}

By Claim \ref{app-equal-index-2-equal-len-2}, we have 
\begin{eqnarray*}
\psi &>& \left|[\llparenthesis x^{\prime}\rrparenthesis_{\xi}]_{r^+}-[\llparenthesis x^{\prime}\rrparenthesis_{b+1}]_{r^+}\right|+\sum_{r^+\leq i\leq j^{\prime}}\left(\beta_{m-i-1}^{m-r^+}-\beta_{m-i}^{m-r^+}\right) 
\end{eqnarray*}

By (\ref{idx-distance-bound2}) and Fact \ref{app-equal-index-2-equal-len-fac1}, we have
\begin{eqnarray*}
\psi &>&\left|[\llparenthesis x^{\prime}\rrparenthesis_{\xi}]_{r^+}-[\llparenthesis x^{\prime}\rrparenthesis_{b+1}]_{r^+}\right|+\displaystyle \sum_{r^+\leq i\leq b}\!\!\left([\llparenthesis x^\prime\rrparenthesis_{i+1}]_{r^+}\!\!-\![\llparenthesis x^\prime\rrparenthesis_{i}]_{r^+}\!\right) \hspace{3pt} \\ 
&=& [\llparenthesis x^{\prime}\rrparenthesis_{\xi}]_{r^+}-[x^\prime]_{r^+}  \qquad\qquad\quad [\mbox{ by (\ref{app-app-equal-index-2-equal-len-2-eqn6}) }]
\end{eqnarray*}
 
Finally, by Fact \ref{app-equal-index-2-equal-len-fac2} and Fact \ref{app-equal-index-2-equal-len-fac3}, $[x]_{r^+}\not\equiv [x^{\prime}]_{r^+}$  mod $\beta_{m-\xi}^{m-r^+}$. We arrive at a contradiction. 

If (\ref{app-app-equal-index-2-equal-len-2-eqn5}) does not hold, the last arguments need to be revised a little bit, but very similar.                                                                                                                                               
We need take care of the order of the vertices $\llparenthesis x\rrparenthesis_\xi, \llparenthesis x\rrparenthesis_{b+1}, \llparenthesis x\rrparenthesis_b$ and $x$. That is, if the order of (\ref{app-app-equal-index-2-equal-len-2-eqn5}) is changed, then the order of (\ref{app-app-equal-index-2-equal-len-2-eqn6}) will be changed accordingly, using similar arguments.  
\end{proof}

\end{remark}

%\end{appendix}

\newpage

\begin{thebibliography}{1}


\bibitem{Amano2010}
K. Amano,
``$k$-Subgraph isomorphism on $\mathrm{AC}^0$ circuits,''
 {Computational Complexity}, 19(2), pp.183--210, 2010.


\bibitem{Andreka1998Modal}
H. Andr\'{e}ka, I. N\'{e}meti, J. V. Benthem,
``Modal languages and bounded fragments of predicate logic,''
{Journal of Philosophical Logic}, 27(3), pp.217--274, 1998. 

\bibitem{Barringtion1990Uniformity}
D. M. Barrington, N. Immerman, and H. Straubing, 
``On uniformity within $NC^1$,'' {Journal of Computer and
System Sciences}, 41(3), pp.274--306, 1990.

\bibitem{Beame90Clique}
P. Beame,
``Lower bounds for recognizing small cliques on CRCW PRAM's,'' 
{Discrete Applied Mathematics}, 29(1), pp.3--20, 1990.

\bibitem{Berwanger2007Calculus}
D. Berwanger, E. Gr{\"a}del, and G. Lenzi,
``The variable hierarchy of the $\mu$-calculus is strict,''
{Theory of Computing Systems}, vol.40, pp.437--466, 2007.  

\bibitem{Dawar96Number}
A. Dawar, S. Lindell and S. Weinstein, 
``First order logic, fixed point logic and linear order,''
In {\em Computer Science Logic '95: Lecture Notes in Computer Science vol. 1092, Springer-Verlag}, pp.161--177, 1996. 

\bibitem{Dawar1998Ranks}
A. Dawar, K. Doets, S. Lindell, and S. Weinstein, 
 ``Elementary properties of finite ranks,''  
Mathematical Logic Quarterly, 44, pp.349--353, 1998.

\bibitem{DawarHowmany}
A. Dawar,
``How many first-order variables are needed on finite ordered structures?,''
In: {\em We Will Show Them: Essays in Honour of Dov Gabbay, Vol 1.}, 
College Publications, pp.489--520, 2005.

\bibitem{Dawar2011Ackermann}
A. Dawar, J. A. Makowsky, and D. Niwinski, 
The Ackermann award 2011, Report of the Jury, CSL 2011.


\bibitem{Denenberg1986circuits}
L. Denenberg, Y. Gurevich, and S. Shelah, 
``Definability
by constant-depth polynomial-size circuits,''
{Information and Control}, 70(2/3), pp.216--240, 1986.

\bibitem{EbbinghausFlum99FiniteM}
H-D. Ebbinghaus and J.~Flum,
{\em Finite Model Theory}.
 Springer, 1999.


\bibitem{Fagin1974Start}
R. Fagin,  
``Generalized First-Order Spectra and Polynomial-Time Recognizable Sets,'' {\em Complexity of Computation, (ed. R. Karp), SIAM-AMS Proc.}, 
7, pp.27--41, 1974. 


\bibitem{Grohe1998Variable}
M. Grohe,
``Finite variable logics in descriptive
complexity theory,'' {Bulletin of Symbolic Logic},
4(4), pp.345--398, 1998.

\bibitem{Gurevich1984Circuit2Logic}
Y. Gurevich and H. R. Lewis, 
``A logic for constant-depth circuits,''
 {Information and Control}, 61, pp.65--74, 1984.

\bibitem{Immerman1982Conj}
N. Immerman, 
``Upper and lower bounds for first order
expressibility,'' {Journal of Computer and System Sciences}, 
25(1), pp.76--98, 1982.

\bibitem{Immerman1983languageThat}
N. Immerman,
``Languages which capture complexity classes,'' 
 In: {\em STOC 1983: Proceedings
of the 15th Annual ACM Symposium on Theory of Computing,}
pp.347--354, 1983.


\bibitem{Immerman1989Parallel}
N. Immerman, 
``Expressibility and Parallel Complexity,''
 	SIAM Journal on Computing, 
 18, pp.625--638, 1989.

\bibitem{Immerman1991Variable}
N. Immerman,
``DSPACE[$n^k$]=VAR[$k+1$],'' {\em Sixth IEEE Structure in Complexity Theory Symposium}, 
pp.334--340, 1991.

\bibitem{Immerman1999Book}
N. Immerman, 
{\em Descriptive Complexity}.
Springer Graduate Texts in Computer Science, New York, 1999.


\bibitem{Lynch86Circuit}
J.F. Lynch,
``A Depth-Size Tradeoff for Boolean Circuits with Unbounded Fan-In,''
 In {\em Proceedings of the conference on Structure in Complexity Theory, Lecture Notes in Computer Science 223}, Springer-Verlag, pp.234--248, 1986.


\bibitem{Poizat1982ColorOrder}
B. Poizat,
``Deux ou trois choses que je sais de $L_n$,'' {Journal of 
Symbolic Logic}, 47(3), pp 641--658, 1982. 

\bibitem{RossmanStoc}
B. Rossman, 
``On the constant-depth complexity of $k$-clique,'' In 
{\em STOC'08: Proceedings
of the 40th Annual ACM Symposium on Theory of Computing,}
pp.721--730, 2008.

\bibitem{Rossman09Game}
B. Rossman, 
``Ehrenfeucht-Fra\"iss\' e games on random structures,'' 
In
{\em WoLLIC'09: Proceedings of the 15th Workshop on Logic, Language, Information and Computation,}
pp. 350--364, 2009.

\bibitem{RossmanThesis}
B. Rossman,
{\em Average-Case Complexity of Detecting Cliques.} 
PhD thesis. MIT, 2010.

\bibitem{RossmanUpperBound}
B. Rossman,
``A Tight Upper Bound on the Number of Variables for Average-Case $k$-Clique on Ordered Graphs,'' In
{\em WoLLIC'12: Proceedings of the 19th Workshop on Logic, Language, Information and Computation}, pp.282--290, 2012. 


\bibitem{Rowland2009BinaryRegularity}
E. Rowland,
``Regularity versus complexity in the binary representation of $3^n$,'' 
{Complex Systems},  
18, pp.367--377, 2009. 

\bibitem{Schweikardt05Arithmetic}
N. Schweikardt, 
``Arithmetic, first-order logic, and counting quantifiers,''
{ACM Transactions on Computational Logic,} 6(3), pp.634--671, 2005.

\bibitem{SchweikardtLinearOrder}
N. Schweikardt and T. Schwentick, 
``A note on the expressive power of linear orders,''
{Logical Methods in Computer Science}, 7(4), pp.1--13, 2011. 

\bibitem{Trakhtenbrot1950Decidability}
B. A. Trakhtenbrot, ``The Impossibility of an algorithm for the decidability problem on finite classes,''
{\em Proceeding of the USSR Academy of Sciences (in Russian)}, 70 (4), pp.569--572, 1950. 

\bibitem{Venema1990Infinite}
Y. Venema,
``Expressiveness and completeness of an interval tense logic,''
{Notre Dame Journal of Formal Logic}, 31, pp.529--547, 1990. 



\end{thebibliography}
\end{document}